%% file: sensitivity_analysis.tex
\pdfoutput=1
\documentclass[12pt]{article}

\usepackage{xr}
\input{preamble}
\hypersetup{%
  breaklinks = true,%
  colorlinks = true,%
  anchorcolor = webbrown,%
  citecolor = webbrown,%
  filecolor = webbrown,%
  linkcolor = webbrown,%
  menucolor = webbrown,%
  urlcolor= webbrown,%
  citebordercolor= 1 0 0,%
  menubordercolor=1 0 0,%
  urlbordercolor=1 0 0,%
  runbordercolor=1 0 0,%
  pdftitle=Sensitivity Analysis using Approximate Moment Condition Models,%
  pdfauthor=Tim Armstrong \& Michal Kolesár}

\title{Sensitivity Analysis using Approximate Moment Condition Models\thanks{We
    thank Isaiah Andrews, Gary Chamberlain, Tim Christensen, Mikkel
    Plagborg-Møller, Jesse Shapiro, Martin Weidner and participants at several
    conferences and seminars for helpful comments and suggestions, and Soonwoo
    Kwon for research assistance. All remaining errors are our own. Armstrong
    acknowledges support by National Science Foundation Grant SES-1628939.
    Kolesár acknowledges support by the Sloan Research Fellowship, and by the
    National Science Foundation Grant SES-1628878.}}
\author{Timothy B. Armstrong\thanks{email: \texttt{timothy.armstrong@yale.edu}}\\
  Yale University \and
  Michal Kolesár\thanks{email: \texttt{mkolesar@princeton.edu}}\\
  Princeton University}%
\date{\today}

\begin{document}

\maketitle

\begin{abstract}
  We consider inference in models defined by approximate moment conditions. We
  show that near-optimal confidence intervals (CIs) can be formed by taking a
  generalized method of moments (GMM) estimator, and adding and subtracting the
  standard error times a critical value that takes into account the potential
  bias from misspecification of the moment conditions. In order to optimize
  performance under potential misspecification, the weighting matrix for this
  GMM estimator takes into account this potential bias, and therefore differs
  from the one that is optimal under correct specification. To formally show the
  near-optimality of these CIs, we develop asymptotic efficiency bounds for
  inference in the locally misspecified GMM setting. These bounds may be of
  independent interest, due to their implications for the possibility of using
  moment selection procedures when conducting inference in moment condition
  models. We apply our methods in an empirical application to automobile demand,
  and show that adjusting the weighting matrix can shrink the CIs by a factor of
  3 or more.
\end{abstract}

\clearpage

\section{Introduction}\label{introduction_sec}

Economic models are typically viewed as approximations of reality. However,
conventional approaches to estimation and inference assume that a model holds
exactly. In this paper, we weaken this assumption, and consider inference in a
class of models characterized by moment conditions which are only required to
hold in an approximate sense. The failure of the moment conditions to hold
exactly may come from failure of exclusion restrictions (e.g.\ through omitted
variable bias or because instruments enter the structural equation directly in
an IV model), functional form misspecification, or other sources such as
measurement error, or data contamination.

We assume that we have a model characterized by a set of
population moment conditions $g(\theta)$. In the generalized method of moments
(GMM) framework, for instance, $g(\theta)=E[g(w_{i}, \theta)]$, which can be
estimated by the sample analog $\frac{1}{n}\sum_{i=1}^{n}g(w_{i}, \theta)$, based
on the sample $\{w_{i}\}_{i=1}^{n}$. When evaluated at the true parameter
value $\theta_{0}$, the population moment condition lies in a known set
specified by the researcher,
\begin{equation*}
  g(\theta_{0})=c/\sqrt{n}, \qquad c\in\mathcal{C}.
\end{equation*}
The set $\mathcal{C}$ formalizes the way in which the moment conditions may
fail, and it can then be varied as a form of sensitivity analysis, with
$\mathcal{C}=\{0\}$ reducing to the correctly specified case.

We focus on local misspecification: the scaling by $\sqrt{n}$ implies that the
specification error and the sampling error are of the same order of magnitude,
and it ensures that the asymptotic approximation captures the fact that it may
not be clear from the sample at hand whether the model is correctly specified.
It also leads to increased tractability, allowing us to deliver a simple method
for inference on a structural parameter of interest $h(\theta_{0})$, rather than
a pseudo-true parameter. This tractability has made local misspecification a
popular tool for sensitivity analysis in applied work, especially following the
recent influential paper by
\citet{andrews_measuring_2017}.\footnote{\label{fn:local-misspect-empirical}For
  recent empirical examples using local sensitivity analysis, see
  \citet{gayle2019}, or \citet{duflo2018}.} As with any asymptotic device, our
modeling of misspecification as local should not be taken to mean that we
literally believe that the model would be closer to correct if we had more data.
Rather, its usefulness should be judged by whether it yields accurate
approximations to the finite-sample behavior of estimators and confidence
intervals, which in our case requires that the set $\mathcal{C}/\sqrt{n}$ be
small relative to sampling uncertainty.

We propose a simple method for constructing asymptotically valid confidence
intervals (CIs) under this setup: one takes a standard estimator, such as the
GMM estimator, and adds and subtracts its standard error times a critical value
that takes into account the potential asymptotic bias of the estimator, in
addition to its variance. A key insight of this paper is that because the CIs
must be widened to take into account the potential bias, the optimal weighting
matrix for the correctly specified case (the inverse of the variance matrix of
the moments) is generally no longer optimal under local misspecification.
Rather, the optimal weighting matrix takes into account potential
misspecification in the moments in addition to the variance of their estimates:
it places less weight on moments that are allowed to be further from zero
according the researcher's specification of the set $\mathcal{C}$. We also show
that an analogous result holds for other performance criteria, such as
estimation under the mean-squared error: the optimal weighting matrix again
trades off the potential misspecification of the moments against their
precision, although the optimal tradeoff is different.

To illustrate the practical importance of this result, we apply our methods to
form misspecification-robust CIs in an empirical model of automobile demand
based on \citet{blp95}. We consider sets $\mathcal{C}$ motivated by the forms of
local misspecification considered in \citet{andrews_measuring_2017}, who
calculate the asymptotic bias of the usual GMM estimator in this model. We find
that adjusting the weighting matrix to account for potential misspecification
substantially reduces the potential bias of the estimator and, as a result,
leads to large efficiency improvements of the optimal CI relative to a CI based
on the GMM estimator that is optimal under correct specification: it shrinks the
CI by up to a factor of $3$ or more in our main specifications. As a result, we
obtain informative CIs in this model even under moderate amounts of
misspecification.

We show that the CIs we propose are near-optimal when the set $\mathcal{C}$ is
convex and centrosymmetric ($c\in\mathcal{C}$ implies $-c\in\mathcal{C}$). To
this end, we argue that the relevant ``limiting experiment'' for the locally
misspecified GMM model is isomorphic to an approximately linear model of
\citet{SaYl78}, which falls under a general framework studied by, among others,
\citet{donoho94}, \citet{CaLo04} and \citet{ArKo18optimal}. We derive asymptotic
efficiency bounds for CIs in the locally misspecified GMM model that formally
translate bounds from the approximately linear limiting experiment to the
locally misspecified GMM setting. In particular, these bounds imply that the
scope for improvement over our CIs by optimizing expected length at a particular
value of $\theta_{0}$ and $c=0$ (while still maintaining coverage over the whole
parameter space for $\theta$ and $\mathcal{C}$) is limited, even if one
optimizes expected length at the true values of $\theta_{0}$ and $c$.

These efficiency bounds address an important criticism of our CIs: they require
a priori specification of the set $\mathcal{C}$ that defines misspecification,
including both the \emph{magnitude} of misspecification and \emph{which} moments
are misspecified. In particular, one cannot substantively improve upon our CI
by, say, trying to use data-driven methods that gauge misspecification magnitude
or try to determine which moments are misspecified. These bounds have
implications for procedures proposed by \citet{andrews_hybrid_2009},
\citet{ditraglia_using_2016} and \citet{mccloskey_2017}, who consider the case
where some moments are known to be correctly specified and no a priori bound is
placed on the magnitude of misspecification of the remaining moments. As we
discuss in \Cref{sec:some-valid-some}, in this case our CI reduces to the usual
CI based on the $k_1$ correctly specified moments, and our efficiency bounds
show that CIs proposed in these papers cannot substantively improve upon it.

Because we cannot use the data to determine the magnitude $\Mbound$ of the set
$\mathcal{C}$, we recommend plotting our CIs as a function of the potential
misspecification magnitude $\Mbound$, or reporting the smallest value of
$\Mbound$ for which a particular finding breaks down. Such sensitivity analysis
is easy to conduct under out proposed implementation. In particular, we show
that, when the set $\mathcal{C}$ is characterized by $\ell_{p}$ constraints, the
class of weightings that trace out the optimal bias-variance tradeoff as a
function of how much relative weight we put on the bias can be easily computed
by recasting the problem as a penalized regression problem. By exploiting this
analogy, we develop a simple algorithm for computing this class under
$\ell_{\infty}$ constraints that is similar to the LASSO/LAR algorithm
\citep{ehjt04,rosset_piecewise_2007}; under $\ell_{2}$ constraints, the solution
admits a closed form.\footnote{An R package implementing our CIs under
  $\ell_{p}$ constraints is available at
  \url{https://github.com/kolesarm/GMMSensitivity}.} Furthermore, as we discuss
in \Cref{sec:implementation}, this class of weightings is entirely determined by
the shape of $\mathcal{C}$; its magnitude $M$ only determines the optimal
relative weight we should put on the bias. Thus, tracing out the optimal
weighting as a function of $\Mbound$ can be done at essentially no additional
computational cost. Furthermore, to avoid having to reoptimize the objective
function with respect to the new weighting matrix, one can also form the CIs by
adding and subtracting our critical value from a one-step estimator
\citep[see][Section 3.4]{newey_large_1994} based on any initial estimate that is
$\sqrt{n}$-consistent under correct specification. We illustrate this approach
in our empirical application in \Cref{empirical_sec}.

Our paper is related to several strands of literature. Our efficiency results
are related to those in \citet{chamberlain_asymptotic_1987} for point estimation
in the correctly specified setting (see also \citealp{hansen85}) and, more broadly,
semiparametric efficiency theory in correctly specified settings \citep[see,
e.g., Chapter 25 in][]{van_der_vaart_asymptotic_1998}. As we discuss in
\Cref{efficiency_examples_sec}, some of our efficiency results are novel
even in the correctly specified case, and may be of independent interest.
\citet{kitamura_robustness_2013} consider efficiency of point estimators
satisfying certain regularity conditions when the misspecification is bounded by
the Hellinger distance. As we discuss in more detail in
\Cref{sec:corr-spec-cress}, our results imply that under this form of
misspecification, the optimal weighting matrix remains the same as under correct
specification; both the usual GMM estimator and the estimator proposed by
\citet{kitamura_robustness_2013} can thus be used to form near-optimal CIs, and
both estimators have the same local asymptotic minimax properties.

Local misspecification has been used in a number of papers, which include, among
others, \citet{newey_generalized_1985}, \citet{berkowitz_validity_2012},
\citet{conley_plausibly_2012}, \citet{guggenberger_asymptotic_2012}, and
\citet{bugni_inference_2018}, and has antecedents in the literature on robust
statistics \citep[see][and references therein]{huber_robust_2011}.
\citet{andrews_measuring_2017} consider this setting and note that asymptotic
bias of a regular estimator can be calculated using influence function weights,
which they call the sensitivity, and show how such calculations can be used for
sensitivity analysis in applications (see also extensions of these ideas in
\citealt{andrews_informativeness_2018} and \citealt{mukhin_sensitivity_2018}).
Our results imply that, if one is interested in inference, conclusions of such
sensitivity analysis may be substantially sharpened by using the
misspecification-optimal weighting matrix, or, equivalently, the
misspecification-optimal sensitivity. In independent work,
\citet{bonhomme_minimizing_2018} provide a framework for estimation and
inference in misspecified likelihood models when the misspecification set
$\mathcal{C}$ is defined with respect to a larger class of models using
statistical notions of distance. Our focus is on overidentified moment condition
models, as in \citet{andrews_measuring_2017}, and we are agnostic about how
$\mathcal{C}$ is determined. The proposal to use estimators that optimize an
asymptotic bias-variance tradeoff using the influence function is common to both
papers. The efficiency bounds in \Cref{efficiency_sec} are unique to the present
paper.

The rest of this paper is organized as follows. \Cref{general_results_sec}
presents our misspecification robust CIs. \Cref{sec:implementation} gives
step-by-step instructions for computing our CIs, along with discussion of other
practical issues. \Cref{efficiency_sec} presents efficiency bounds for CIs in
locally misspecified models; it can be skipped by readers interested only in
implementing the methods. \Cref{applications_sec} discusses applications to
particular moment condition models. \Cref{empirical_sec} presents an empirical
application. Additional results and proofs are collected in appendices and an
online supplement.

\section{Misspecification-robust CIs}\label{general_results_sec}

We have a model that maps a vector of parameters
$\theta\in\Theta\subseteq \mathbb{R}^{d_{\theta}}$ to a $d_{g}$-dimensional
population moment condition $g(\theta)$ that restricts the distribution of the
observed data $\{w_{i}\}_{i=1}^{n}$. We allow the moment condition model to be
locally misspecified, so that at the true value $\theta_{0}$, the population
moment condition is not necessarily zero, but instead lies in a
$\sqrt{n}$-neighborhood of $0$:
\begin{equation}\label{eq:local-misspecification}
  g(\theta_{0})=c/\sqrt{n}, \qquad c\in\mathcal{C},
\end{equation}
where $\mathcal{C}\subseteq\mathbb{R}^{d_{g}}$ is a known set. The set
$\mathcal{C}$ may allow for misspecification in potentially all moment
conditions; we do not require that some elements of $c$ are zero. Our goal is to
construct a CI for a scalar $h(\theta_0)$, where
$h\colon \mathbb{R}^{d_\theta}\to\mathbb{R}$ is a known function. For example,
if we are interested in one of the elements $\theta_{j}$ of $\theta$, we would
take $h(\theta)=\theta_{j}$. More generally, the function $h$ will be nonlinear,
as is, for example, generally the case when $\theta$ is a vector of supply or
demand parameters, and $h(\theta)$ is an elasticity, or some counterfactual.

This setup allows (but does not require) both $\theta_{0}$ and $h(\theta_{0})$
to have the same interpretation as in the correctly specified case, so that our
CIs may still be interpreted as CIs for the structural parameter, elasticity, or
counterfactual of interest. For this interpretation, one typically needs to rule
out forms of misspecification that affect the mapping $\theta\mapsto h(\theta)$.
While we do not formally consider cases in which this mapping itself is
misspecified, such cases are covered under a mild generalization of our
framework, in which $h$ is a function of both $\theta$ and $c$.

Note that the interpretation of $h(\theta_{0})$, and the conceptual framework
defining $\theta_{0}$ is not affected by our modeling of the misspecification as
local: given a set $\Cglob_{n}=\mathcal{C}/\sqrt{n}$, the moment conditions
$g(\theta_{0})\in\Cglob_{n}$ describe the restrictions that the data generating
process and the researcher's modeling assumptions place on
$\theta_0$.\footnote{Formally, for a given sample size $n$, $\theta_{0}$ may be
  set identified, and the identified set under a distribution $P$ is defined as
  the set of parameters $\theta_0$ that satisfy the moment conditions
  $E_{P}g(w_i, \theta_0)\in \Cglob_{n}$ where $E_P$ denotes expectation under
  $P$. We construct CIs that cover $h(\theta_0)$ for points $\theta_0$ in the
  identified set \citep[see][for a discussion of this notion of coverage]{im04}.
  See \Cref{efficiency_sec} and \Cref{efficiency_sec_append} for formal
  definitions of coverage and optimality of our CIs.} The plausibility of these
restrictions is evaluated for a given sample size at hand; it doesn't depend on
assumptions about how $\Cglob_{n}$ changes with $n$. While we focus on sequences
$\Cglob_n=\mathcal{C}/\sqrt{n}$, we discuss in \Cref{remark:global
  misspecification} how our insights can be used to construct CIs that are valid
global misspecification, when $\Cglob_{n}$ is fixed with $n$.

To formalize the notion of asymptotic validity and efficiency of CIs, we will
need to allow the true parameter value $\theta_0$ as well as the vector $c$ and
the data generating process (and hence the map $\theta\mapsto g(\theta)$) to
vary with the sample size. For clarity of exposition, we focus here on the case
in which these parameters are fixed. See
\Cref{main_text_efficiency_bound_thm} and
\Cref{efficiency_sec_append} for the general case. Under some forms of
misspecification, such as functional form misspecification, there may be
additional higher-order terms on the right-hand side
of~\eqref{eq:local-misspecification}; our results remain unchanged if this is
the case. Again, for clarity of exposition, we focus on the case in
which~\eqref{eq:local-misspecification} holds exactly.

We assume that the sample moment condition $\hat{g}(\theta)$, constructed using
the data $\{w_{i}\}_{i=1}^{n}$, satisfies
\begin{equation}\label{g_clt_assump}
  \sqrt{n}(\hat g(\theta_{0})-g(\theta_{0}))\indist \mathcal{N}(0,\Sigma),
\end{equation}
where $\indist$ denotes convergence in distribution as $n\to\infty$. In the GMM
model, the population and sample moment conditions are given by
$g(\theta)=E[g(w_{i}, \theta)]$ and
$\hat{g}(\theta)=\frac{1}{n}\sum_{i=1}^{n}g(w_{i}, \theta)$, respectively, where
$g(\cdot, \cdot)$ is a known function. However, to cover other minimum distance
problems, we do not require that the moment conditions necessarily take this
form. We further assume that the moment condition is smooth enough so that
\begin{equation}\label{linear_approx_assump}
  \text{for any $\theta_n=\theta_0+\mathcal{O}_P(1/\sqrt{n})$,}\qquad
  \hat g(\theta_n)-\hat g(\theta_0)=\Gamma (\theta_n-\theta_0)+o_P(1/\sqrt{n}),
\end{equation}
where $\Gamma$ is the $d_{g}\times d_{\theta}$ derivative matrix of $g$ at
$\theta_{0}$. Conditions~\eqref{g_clt_assump} and~\eqref{linear_approx_assump}
are standard regularity conditions in the literature on linear and nonlinear
estimating equations; see \citet{newey_large_1994} for primitive conditions.
Finally, we also assume that $h$ is continuously differentiable with the
$1\times d_{\theta}$ derivative matrix at $\theta_{0}$ given by $H$.

\subsection{CIs based on asymptotically linear
  estimators}\label{sec:asympt-line-estim}

Under correct specification, when $\mathcal{C}=\{0\}$, standard estimators
$\hat{h}$ of $h(\theta)$ are asymptotically linear in $\hat{g}(\theta_{0})$.
This will typically extend to our locally misspecified case, so that for some
vector $k\in\mathbb{R}^{d_g}$,
\begin{equation}\label{h_asymptotically_linear_eq}
  \sqrt{n}(\hat h-h(\theta_0))
  =k'\sqrt{n}\hat g(\theta_0)+o_P(1)
  \indist \mathcal{N}(k' c, k'\Sigma k),
\end{equation}
where the convergence in distribution follows
by~\eqref{eq:local-misspecification} and~\eqref{g_clt_assump}. If in addition,
the estimator is regular (so that equality in \Cref{h_asymptotically_linear_eq}
holds uniformly for $\theta$ in a $\sqrt{n}$-neighborhood of $\theta_{0}$), then
$k$ will satisfy (see e.g. Section 2 in \citealp{newey_semiparametric_1990})
\begin{equation}\label{l_kGamma_eq}
  H=-k'\Gamma.
\end{equation}
For example, in a GMM model, if we take $\hat h=h(\hat{\theta}_{W})$ where
\begin{equation}\label{eq:gmm-w}
  \hat{\theta}_{W}=\argmin_{\theta} \hat g(\theta)'W\hat g(\theta),
\end{equation}
is the GMM estimator with weighting matrix $W$,
\Cref{h_asymptotically_linear_eq,l_kGamma_eq} will hold with
$k'=-H(\Gamma' W\Gamma)^{-1}\Gamma' W$ \citep[see][]{newey_generalized_1985}.
Because the vector $k$ determines the local asymptotic bias of the estimator, we
follow \citet{andrews_measuring_2017}, and refer to $k$ as the
\emph{sensitivity} of $\hat{h}$.

We now show how to construct misspecification-robust CIs based on an
asymptotically linear estimator $\hat{h}$ with a given sensitivity $k$. In
\Cref{sec:optimal-cis}, we show how to choose this sensitivity optimally, to
achieve the shortest CI among those based on regular asymptotically linear
estimators. In \Cref{efficiency_sec}, we will show that, under this choice of
$k$, the resulting CI is (near) optimal not only within the class of CIs based
on regular asymptotically linear estimators, but among \emph{all} CIs that
satisfy the asymptotic coverage requirement.

Let $\hat{k}$ and $\hat{\Sigma}$ be consistent estimates of $k$ and $\Sigma$.
Then by Slutsky's theorem,
\begin{equation*}
  \frac{\sqrt{n}(\hat h-h(\theta_0))}{\sqrt{\hat k'\hat \Sigma\hat k}}
  \indist \mathcal{N}\left(\frac{k'c}{\sqrt{k'\Sigma k}},1\right).
\end{equation*}
Under correct specification, the right-hand side corresponds to a standard
normal distribution, and we can form a CI with asymptotic coverage
$100\cdot (1-\alpha)\%$ as
$\hat{h}\pm z_{1-\alpha/2}\sqrt{\hat{k}'\hat{\Sigma}\hat{k}/n}$, where
$z_{1-\alpha/2}$ is the $1-\alpha/2$ quantile of a $\mathcal{N}(0,1)$
distribution; this is the usual Wald CI\@.

When we allow for misspecification, the Wald CI will no longer be valid.
However, note that the asymptotic bias ${k'c}/{\sqrt{k'\Sigma k}}$ is bounded in
absolute value by ${\maxbias_{\mathcal{C}}(k)}/{\sqrt{k'\Sigma k}}$ where
$\maxbias_{\mathcal{C}}(k)\equiv \sup_{c\in\mathcal{C}}\abs{k'c}$. Therefore,
given $c$, the $z$-statistic in the preceding display is asymptotically
$\mathcal{N}(t,1)$ where
$\abs{t}\le \maxbias_{\mathcal{C}}(k)/\sqrt{k'\Sigma k}$. This leads to the CI
\begin{equation}\label{h_ci_eq}
  \hat h \pm
  \cv_\alpha
  \left(\frac{\maxbias_{\mathcal{C}}(\hat k)}{\sqrt{\hat k'\hat \Sigma \hat k}}\right)\cdot
  \sqrt{\hat k'\hat \Sigma \hat k}/\sqrt{n},
\end{equation}
where $\cv_\alpha(\overline{t})$ is the $1-\alpha$ quantile of $\abs{Z}$, with
$Z\sim \mathcal{N}(\overline{t},1)$. In particular,
$\cv_{\alpha}(0)=z_{1-\alpha/2}$, so that in the correctly specified
case,~\eqref{h_ci_eq} reduces to the usual Wald CI\@. As we discuss in
\Cref{efficiency_sec}, in the limiting experiment, this CI becomes equivalent to
the fixed-length CI proposed by \citet{donoho94}.

To form a one-sided CI based on an estimator $\hat{h}$ with sensitivity $k$, one
can simply subtract its maximum bias, in addition to the standard error:
\begin{equation}\label{h_oci_eq}
  \hor{\hat{h}-\maxbias_{\mathcal{C}}(\hat{k})/\sqrt{n}
    -z_{1-\alpha}\sqrt{\hat{k}'\hat{\Sigma}\hat{k}/n}, \infty}.
\end{equation}
One could also form a valid two-sided CI by adding and subtracting the
worst-case bias $\maxbias_{\mathcal{C}}(\hat k)/\sqrt{n}$ from $\hat{h}$, in
addition to adding and subtracting
$z_{1-\alpha/2}\sqrt{\hat{k}'\hat\Sigma\hat k/n}$; however, since $\hat{h}$
cannot simultaneously have a large positive and a large negative bias, such CI
will be conservative, and longer than the CI in~\eqref{h_ci_eq}.

\subsection{Optimal CIs}\label{sec:optimal-cis}

The asymptotic length of the CI in \Cref{h_ci_eq} is given by
\begin{equation}\label{eq:ci-length}
  2\cdot \cv_\alpha\left(\maxbias_{\mathcal{C}}(k)/\sqrt{k'
      \Sigma k}\right)\cdot \sqrt{k' \Sigma k}/\sqrt{n}.
\end{equation}
To attain the shortest possible CI, we therefore need to use an estimator with
sensitivity that minimizes this expression. We restrict attention to
asymptotically linear estimators that are regular, so that we need to
minimize~\eqref{eq:ci-length} subject to~\eqref{l_kGamma_eq}. The CI length in
\Cref{eq:ci-length} depends on $\theta$ only through $\Sigma$. Furthermore, it
depends on the sensitivity only through the maximum bias,
$\maxbias_{\mathcal{C}}(k)$, and the variance $k'\Sigma k$. Therefore, rather
than minimizing~\eqref{eq:ci-length} directly over all sensitivities $k$, one
can first minimize the variance subject to a bound $\overline{B}$ on the
worst-case bias,
\begin{equation}\label{eq:bias-variance-optimization}
  \min_{k} k'\Sigma k\qquad \text{s.t.~\eqref{l_kGamma_eq} and}\quad
  \sup_{c\in\mathcal{C}}\abs{k'c}\leq \overline{B},
\end{equation}
and then vary the bound $\overline{B}$ to find the bias-variance trade-off that
leads to the shortest CI\@. In our implementation in \Cref{sec:implementation},
we focus on the case where $\mathcal{C}$ is characterized by $\ell_p$
constraints, in which case a closed-form expression for the worst-case bias
$\sup_{c\in\mathcal{C}}\abs{k'c}$ is available, and it is computationally
trivial to solve~\eqref{eq:bias-variance-optimization} directly or in Lagrange
multiplier form. In general, when the set $\mathcal{C}$ is convex, one can
reformulate~\eqref{eq:bias-variance-optimization} as a convex optimization
problem, leading to a computationally tractable solution (see
\Cref{efficiency_sec}). One can also use~\eqref{eq:bias-variance-optimization}
to determine the optimal sensitivity for constructing one-sided CIs, if we use
quantiles of excess length as the criterion for choosing a CI\@. We provide
details in \Cref{efficiency_sec_append}.

Once the optimal sensitivity has been determined, we can implement an estimator
with this sensitivity as a one-step estimator. In particular, let
$\hat\theta_{\text{initial}}$ be an initial $\sqrt{n}$-consistent estimator of
$\theta_0$, let $\hat{k}= k+o_{P}(1)$ be a consistent estimator of the desired
sensitivity $k$. Then the one-step estimator
\begin{equation*}
  \hat h=h(\hat\theta_{\text{initial}})+\hat k'\hat g(\hat\theta_{\text{initial}})
\end{equation*}
will have the desired sensitivity. This follows from the Taylor expansion
\begin{equation*}
  \begin{split}
    \sqrt{n}(\hat h-h(\theta_0))
    &=H\sqrt{n}(\hat\theta_{\text{initial}}-\theta_0)+\hat k'\sqrt{n}\hat g(\hat\theta_{\text{initial}})+o_P(1) \\
    &=(H+\hat{k}'\Gamma)\sqrt{n}(\hat\theta_{\text{initial}}-\theta_0) +
    \hat{k}'\sqrt{n}\hat{g}(\theta_0)+o_P(1),
  \end{split}
\end{equation*}
where the second line follows from~\eqref{linear_approx_assump}. It then follows
from~\eqref{l_kGamma_eq} that the first term converges in probability to zero,
and $\hat{h}$ satisfies~\eqref{h_asymptotically_linear_eq}.

\section{Practical implementation}\label{sec:implementation}

We now give step-by-step instructions for computing our CI\@. To make it easy to
determine the sensitivity of the CI to the magnitude of misspecification, we
consider sets of the form
$\mathcal{C}=\mathcal{C}(M)=\{Mc\colon c\in\mathcal{C}(1)\}$, where the scalar
$M$ measures the magnitude of misspecification. We discuss the exact
specification of the set $\mathcal{C}(\Mbound)$ in \Cref{remark:choice-of-C} below.

The fact that $M$ simply scales the potential magnitude of misspecification
leads to a simplification when tracing out the optimal CI as a function of
$\Mbound$. In particular, let $\{k_{\lambda}\}_{\lambda\geq 0}$ be the
bias-variance optimizing class of sensitivities that traces out the solutions to
\Cref{eq:bias-variance-optimization} as we vary the bound $\overline{B}$ when
$\mathcal{C}=\mathcal{C}(1)$. The index $\lambda$ determines the relative weight
on the bias; it can correspond to the Lagrange multiplier in a Lagrangian
formulation of~\eqref{eq:bias-variance-optimization}, or we can simply take
$\lambda=\overline B$ if we are solving~\eqref{eq:bias-variance-optimization}
directly. Let $\overline{B}_{\lambda}=\maxbias_{\mathcal{C}(1)}(k_\lambda)$. It
then follows by a change-of-variables argument that
$\maxbias_{\mathcal{C}(\Mbound)}(k_\lambda)= \Mbound\overline B_\lambda$, and
that $k_\lambda$ minimizes the asymptotic variance subject to this bound on
worst-case bias over $\mathcal{C}(\Mbound)$. Thus,
$\{k_\lambda\}_{\lambda\geq 0}$ is also a bias-variance optimizing class of
sensitivities for $\mathcal{C}(\Mbound)$. We therefore only need to compute the
class $\{k_{\lambda}\}_{\lambda\geq 0}$ only once, even when a range of values
$\Mbound$ is considered.

With this simplification, we can construct CIs for a range of values of $M$ as
follows:
\begin{enumerate}
\item\label{item:step1} Obtain an initial estimate $\hat\theta_{\text{initial}}$
  and estimates $\hat H$, $\hat \Gamma$ and $\hat \Sigma$ of $H$, $\Gamma$ and
  $\Sigma$.

  In particular, for the GMM model, when
  $\hat{g}(\theta) =\frac{1}{n}\sum_{i=1}^{n} g(w_i, \theta)$, we can take
  $\hat\theta_{\text{initial}}$ to be the GMM estimator
  $\hat{\theta}_{W}=\argmin_{\theta} \hat g(\theta)'W\hat g(\theta)$ for some
  weight matrix $W$. The remaining objects are the usual quantities used to
  estimate the asymptotic variance of this estimator:
  $\hat{\Sigma}= \frac{1}{n}\sum_{i=1}^{n} g(w_i, \hat\theta_{\text{initial}})
  g(w_i, \hat\theta_{\text{initial}})'$ (or, in the case of dependent
  observations, an autocorrelation robust version of this estimate),
  $\hat \Gamma=\frac{d}{d\theta'} \hat{g}(\theta)
  \big|_{\theta=\hat\theta_{\text{initial}}}$ (or, if $g(\theta, w)$ is
  nonsmooth, a numerical derivative as in \citealt{hong_extremum_2015}, or
  Section 7.3 of \citealt{newey_large_1994}) and
  $\hat{H}=\frac{d}{d\theta'}h(\theta)\big|_{\theta=\hat\theta_{\text{initial}}}$.
\item\label{item:step2} Compute the bias-variance optimizing class
  $\{\hat{k}_{\lambda}\}_{\lambda\geq 0}$ that
  solves~\eqref{eq:bias-variance-optimization} with $\mathcal{C}=\mathcal{C}(1)$
  and with $\hat\Sigma$ in place of $\Sigma$. Algorithms and closed-form
  solutions for computing $\{\hat{k}_{\lambda}\}_{\lambda \geq 0}$ for
  particular choices of $\mathcal{C}(M)$ are discussed in
  \Cref{remark:choice-of-C}. Let
  $\overline B_\lambda=\sup_{c\in\mathcal{C}(1)} |\hat k_\lambda'c|$. For each
  $\Mbound$, let $\lambda^*_{\Mbound}$ minimize the CI length\footnote{The
    critical value $\cv_{\alpha}(b)$ can easily be computed in statistical
    software as the square root of the $1-\alpha$ quantile of a non-central
    $\chi^{2}$ distribution with $1$ degree of freedom and non-centrality
    parameter $b^{2}$.}
  $2\cv_\alpha(\Mbound \overline B_\lambda/ \sqrt{\hat{k}_{\lambda}' \hat\Sigma
    \hat k_{\lambda}})\cdot \sqrt{\hat{k}_{\lambda}' \hat\Sigma
    \hat{k}_{\lambda}}$ over $\lambda$.
\item\label{item:step3} For each $M$, construct the one-step estimator
  $\hat{h}_{\lambda^{*}_{M}} =
  h(\hat\theta_{\text{initial}})+\hat{k}_{\lambda^{*}_{M}}'
  \hat{g}(\hat\theta_{\text{initial}})$, and report the misspecification-robust
  CI under $\mathcal{C}(\Mbound)$
  \begin{equation}\label{eq:optimal-CI-formula}
    \hat{h}_{\lambda_{M}^{*}}\pm
    \cv_{\alpha}
    \left(\Mbound \overline B_{\lambda^{*}_{M}}/
      \sqrt{\hat{k}_{\lambda_{M}^*}' \hat\Sigma \hat k_{\lambda_{M}^*}}\right)\cdot
    \sqrt{\hat{k}_{\lambda_{M}^*}' \hat\Sigma \hat k_{\lambda^*} / n}.
  \end{equation}
\end{enumerate}

\begin{remark}[Choice of $\mathcal{C}(\Mbound)$]\label{remark:choice-of-C}
  A simple and flexible way of forming the set $\mathcal{C}$ is to take
  \begin{equation}\label{eq:C_Bgamma}
    \mathcal{C}=\mathcal{C}(M)=\{B\gamma\colon \norm{\gamma}\le \Mbound\},
  \end{equation}
  where $B$ is a $d_g\times d_\gamma$ matrix and $\norm{\cdot}$ is some norm.
  The matrix $B$ can be used to standardize moments, account for their
  correlations, or to pick out which moments are believed to be misspecified.
  For instance, setting $B$ to the last $d_{\gamma}$ columns of the
  $d_{g}\times d_{g}$ identity matrix allows for misspecification in the last
  $d_{\gamma}$ moments, while maintaining that the first $d_{g}-d_{\gamma}$
  moments are valid.

  In light of our result in \Cref{efficiency_sec} that it is not possible to
  determine the set $\mathcal{C}$ in a data-driven way, the normalizing matrix
  $B$ and the baseline misspecification magnitude $\Mbound$ used should be
  chosen to reflect application-specific arguments about which forms of
  misspecification are plausible; we can then vary $\Mbound$ over other
  plausible choices as a form of sensitivity analysis. We illustrate this in the
  context of our application in \Cref{empirical_sec}, and we refer the reader to
  \citet{conley_plausibly_2012} and \citet{andrews_measuring_2017} for
  additional examples and discussion. Alternatively, one can also use measures
  of statistical distance such as the probability of detecting that the model is
  misspecified to aid with interpretation of $\Mbound$, as suggested in
  \citet{HaSa08robustness} or \citet{bonhomme_minimizing_2018}.

  While it is not possible to determine $\Mbound$ automatically, it is possible
  to obtain a lower CI $[\Mbound_{\min}, \infty]$ for $\Mbound$, which can be
  used as a diagnostic check verifying that the values of $M$ considered are not
  too small. We develop such tests by generalizing the $J$-test of
  overidentifying restrictions in \Cref{sec:specification-test}. We recommend
  reporting the lower bound $\Mbound_{\min}$ along with the plot of the optimal
  CI as a function of $\Mbound$.

  The norm $\norm{\cdot}$ determines how the researcher's bounds on each element
  of $\gamma$ interact. With the $\ell_{\infty}$ norm, one places separate
  bounds on each element of $\gamma$, which leads to a simple interpretation: no
  single element of $\gamma$ can be greater than $\Mbound$. Under an $\ell_{p}$
  norm with $1\leq p< \infty$, the bounds on each element of $\gamma$ interact
  with each other, so that larger amounts of misspecification in one element is
  allowed if other elements are correctly specified. Depending on whether such
  interactions are desirable, we recommend setting $p=2$, or $p=\infty$.

  For these choices of the norm, computing the class of optimal sensitivities
  $\{\hat{k}_{\lambda}\}_{\lambda\geq 0}$ is particularly simple. In particular,
  when $\norm{\cdot}$ corresponds to an $\ell_{p}$ norm, the worst-case bias has
  a closed form, since by Hölder's inequality,
  $\maxbias_{\mathcal{C}(M)}(k)=\sup_{\norm{\gamma}_{p}\leq
    1}M\abs{k'B\gamma}=\Mbound\norm{B'k}_{p'}$, where $p'$ is the Hölder
  complement of $p$ ($p'=1$ if $p=\infty$, while $p'=2$ if $p=2$), and the
  optimal sensitivities $\{\hat{k}_{\lambda}\}_{\lambda\geq 0}$ can be computed
  by casting the problem as a penalized regression problem. We explain the
  connection to penalized regression, and provide details in
  \Cref{sec:lp-bounds-appendix}.

  When $p=2$, so that $\norm{\cdot}$ corresponds to the Euclidean norm, the
  problem is analogous to ridge regression, and the optimal sensitivities in
  Step~\ref{item:step2} of the implementation take the form
  $\hat{k}_{\lambda}'=-H(\Gamma'W_{\lambda}\Gamma)^{-1}\Gamma'W_{\lambda}$,
  where $W_{\lambda}=(\lambda BB'+\hat{\Sigma})^{-1}$, with
  $\overline{B}_\lambda=\norm{B'\hat{k}_{\lambda}}_{2}$. As an alternative to
  using the one-step estimator in Step~\ref{item:step3} of the implementation,
  one can implement this sensitivity directly as a GMM estimator with weighting
  matrix $W_{\lambda}$ (see also \Cref{weighting_matrix_remark} below). Relative
  to the optimal weighting matrix $\Sigma^{-1}$ under correct specification, the
  matrix $W_{\lambda}$ trades off precision of the moments against their
  potential misspecification.

  When $p=\infty$, the penalized regression analogy leads a simple algorithm for
  computing the optimal sensitivities $\{\hat{k}_{\lambda}\}_{\lambda\geq 0}$
  that is similar to the LASSO/LAR algorithm \citep{ehjt04}. We give details on
  the algorithm in \Cref{sec:lp-bounds-appendix}. It follows from this algorithm
  if $B$ corresponds to columns of the identity matrix, as $\Mbound$ grows, the
  optimal sensitivity successively drops the ``least informative'' moments, so
  that in the limit, if $d_{g}\leq d_{\gamma}+d_{\theta}$, the optimal
  sensitivity corresponds to that of an exactly identified GMM estimator based
  on the $d_{\theta}$ ``most informative'' moments only, where
  ``informativeness'' is given by both the variability of a given moment, and
  its potential misspecification. If $d_{g}>d_{\gamma}+d_{\theta}$, one simply
  drops all invalid moments in the limit.
\end{remark}

\begin{remark}\label{weighting_matrix_remark}
  In Step~\ref{item:step3} of our implementation, we use a one-step estimator
  $\hat{h}_{\lambda}$ to compute a CI that is asymptotically valid and optimal.
  Due to concerns about finite-sample behavior (analogous to concerns about
  finite sample behavior of one-step estimators in the correctly specified
  case), one may prefer using a different estimator that is asymptotically
  equivalent to $\hat{h}_{\lambda}$. In general, one can implement an estimator
  with sensitivity $k$ as a GMM or minimum distance estimator by using an
  appropriate weighting matrix, so that one can in particular replace
  $\hat{h}_{\lambda}$ by $h(\hat{\theta}_{W})$, with the weighting matrix $W$
  appropriately chosen. To give the formula for the weighting matrix, let
  $\Gamma_{\perp}$ denote a $d_{g}\times (d_{g}-d_{\theta})$ matrix that's
  orthogonal to $\Gamma$, so that $\Gamma_{\perp}'\Gamma=0$, and let
  $\hat{\Gamma}_{\perp}$ denote a consistent estimate. Let $S$ denote a
  $d_{g}\times d_{\theta}$ matrix that satisfies $S'\hat{\Gamma}=-I$ and
  $\hat{k}_{\lambda}=S\hat{H}'$. Then we can set
  $W=S W_{1} S'+\hat{\Gamma}_{\perp}W_{2}\hat{\Gamma}_{\perp}'$ for some
  non-singular matrix $W_{1}$, and an arbitrary conformable matrix $W_{2}$. It
  can be verified by simple algebra that $\hat{\theta}_{W}$ will have
  sensitivity $k_{\lambda}$.
\end{remark}

\begin{remark}[Global misspecification]\label{remark:global misspecification}
  While we focus on the local misspecification setting, in which the set
  $\mathcal{C}/\sqrt{n}$ shrinks with $n$ at a $1/\sqrt{n}$, one can use our
  insights about optimal weighting to construct a CI that retains the
  near-optimality properties of the above CI under local misspecification, while
  having correct coverage under asymptotics in which this set shrinks more
  slowly or stays fixed with the sample size (the latter is termed ``global
  misspecification'' in the literature). Let $W$ be a weighting matrix that
  leads to the optimal sensitivity, as described in
  \Cref{weighting_matrix_remark} above, and let $\mathcal{I}_\cglob$ be a CI
  constructed from the GMM estimator with moment conditions
  $\theta\mapsto g(w_i, \theta)-\cglob$. Let
  $\mathcal{I}=\cup_{\cglob\in \mathcal{C}/\sqrt{n}}\mathcal{I}_\cglob$ be the
  union of these CIs over possible values of $\cglob$ in the set
  $\mathcal{C}/\sqrt{n}$. Such an approach was suggested in the context of
  misspecified linear IV by \citet{conley_plausibly_2012}, although they did not
  consider adjusting the weighting matrix. The resulting CI has correct
  asymptotic coverage under both local and global misspecification, and, for
  one-sided CI construction, is asymptotically equivalent under local
  misspecification to the CI discussed above. We provide further details in
  \Cref{global_misspec_sec_append}. In the \namecref{global_misspec_sec_append},
  we also discuss a second approach to constructing CIs valid under global
  misspecification based on misspecification-robust standard errors
  \citep{HaIn03}, which is applicable if the estimate of the worst-case bias
  under global misspecification is asymptotically normal. The resulting one- and
  two-sided CIs are asymptotically equivalent under local misspecification to
  the optimal CIs discussed above.
\end{remark}

\begin{remark}[Other performance criteria]\label{remark:other-performance}
  In addition to constructing a CI, one may be interested in a point estimate of
  $h(\theta_0)$, using mean squared error (MSE) as the criterion. The steps to
  forming the MSE optimal point estimate are exactly the same as above, except
  that, rather than minimizing CI length in Step~\ref{item:step2}, we choose
  $\lambda$ to minimize
  $\maxbias_{\mathcal{C}}(\hat{k}_\lambda)^2+ \hat{k}_\lambda'\hat{\Sigma}
  \hat{k}_\lambda=\Mbound \overline B_{\lambda} + \hat{k}_\lambda'\hat{\Sigma}
  \hat{k}_\lambda$. Similar ideas apply to other criteria, such as mean absolute
  deviation or quantiles of excess length of one-sided CIs (discussed in
  \Cref{efficiency_sec_append}). If $\lambda$ is chosen differently in
  Step~\ref{item:step2} the CI computed in Step~\ref{item:step3} will be longer
  than the one computed at $\lambda_{M}^*$, but it will still have correct
  coverage.
\end{remark}

\section{Efficiency bounds and near optimality}\label{efficiency_sec}

The CI given in \Cref{eq:optimal-CI-formula} has the apparent defect that the
local misspecification vector $c$ is reflected in the length of the CI only
through the a priori restriction $\mathcal{C}$ imposed by the researcher. Thus,
if the researcher is conservative about misspecification, the CI will be wide,
even if it turns out that $c$ is in fact much smaller than the a priori bounds
defined by $\mathcal{C}$. Moreover, this approach requires the researcher to
explicitly specify the set $\mathcal{C}$, including any tuning parameters such
as the parameter $\Mbound$ if the set takes the form
$\mathcal{C}=\mathcal{C}(M)$ that we considered in \Cref{sec:implementation}.
One may therefore seek to improve upon this CI by estimating the magnitude of
$c$, or by estimating the tuning parameters, and constructing a CI that is
shorter if these estimates indicate that misspecification is mild. Similarly, it
may be restrictive to require that the CI be centered at an asymptotically
linear estimator: this rules out, for example, using a $J$-test to decide which
moments to use.

The main result of this section shows that, when $\mathcal{C}$ is convex and
centrosymmetric ($c\in\mathcal{C}$ implies $-c\in\mathcal{C}$), the scope for
improving on the CI in~\eqref{eq:optimal-CI-formula} is nonetheless limited: no
sequence of CIs that maintain coverage under all local misspecification vectors
$c\in\mathcal{C}$ can be substantially tighter, even under correct
specification. This result can be interpreted as translating results from a
``limiting experiment'' that is an extension of the linear regression model. We
first give a heuristic derivation of this limiting experiment and explain our
result in the context of this limiting experiment. We then present the formal
asymptotic result, and discuss its implications in some familiar settings.
Readers who are interested only in implementing the methods, rather than
efficiency results, can skip this section.

We restrict attention in this section to the GMM model, in which
$\hat g(\theta)=\frac{1}{n}\sum_{i=1}^{n}g(w_i, \theta)$, and we further
restrict the data $\{w_{i}\}_{i=1}^{n}$ to be independent and identically
distributed (i.i.d.). Similar to semiparametric efficiency theory in the
standard, correctly specified case, this facilitates parts of the formal
statements and proofs, such as the definition of the set of distributions under
which coverage is required and the construction of least favorable submodels. We
expect that analogous results could be obtained in other settings.

\subsection{Limiting experiment}\label{limiting_experiment_sec}

As discussed in~\Cref{sec:asympt-line-estim}, we can form CIs based on linear
estimators with asymptotic distribution $\mathcal{N}(k'c, k'\Sigma k)$. This
suggests that the problem of constructing an asymptotically valid CI for
$h(\theta)$ in the model~\eqref{eq:local-misspecification} is asymptotically
equivalent to the problem of constructing a CI for the parameter $H\theta$ in
the approximately linear model
\begin{equation}\label{limit_experiment_eq}
  Y=-\Gamma \theta+c+\Sigma^{1/2}\varepsilon, \quad c\in\mathcal{C}, \qquad
  \varepsilon\sim \mathcal{N}(0,I),
\end{equation}
where $\Gamma$, $H$ and $\Sigma^{1/2}$ are known, and we observe $Y$. One can
think of this model as an ``approximately'' linear regression model, with
$-\Gamma$ playing the role of the design matrix of the (fixed) regressors, and
$c$ giving the approximation error. This model dates back at least to
\citet{SaYl78}, who considered estimation in this model when $\mathcal{C}$ is a
rectangular set and $\Sigma$ is diagonal. The analog of the asymptotically
linear estimator $\hat{h}$ in~\eqref{h_asymptotically_linear_eq} is the linear
estimator $k'Y$. To see the analogy, note that $k'Y-H\theta$ is distributed
$\mathcal{N}((-k'\Gamma-H)\theta+k'c, k'\Sigma k)$, and restricting ourselves to
estimators that do not have infinite worst-case bias when $\theta$ is
unrestricted gives the condition~\eqref{l_kGamma_eq}. In the limiting
experiment, the analog of the CI~\eqref{h_ci_eq} is given by the CI
$k'Y\pm \cv_\alpha(\maxbias_{\mathcal{C}}(k)/\sqrt{k'\Sigma k})\cdot
\sqrt{k'\Sigma k}$. Finding weights that minimize the length of this CI is
isomorphic to the problem of finding the sensitivity that minimizes the
asymptotic CI length in~\eqref{eq:ci-length}.

For a general convex set $\mathcal{C}$, the bias-variance optimization problem
in \Cref{eq:bias-variance-optimization} can be reformulated as a convex
programming problem, as shown in \citet{low95}. In particular, when the set
$\mathcal{C}$ is centrosymmetric (see \Cref{efficiency_sec_append} for the
general case), the bias-variance optimizing class of weights
$\{k_{\lambda}\}_{\lambda> 0}$ is given by the class
$\{k_{\delta}\}_{\delta> 0}$, where
\begin{equation}\label{optimal_weights_eq}
  k_\delta'
  =k_{\delta, \Sigma, \Gamma, H, \mathcal{C}}'
  =\frac{-(c_\delta-\Gamma \theta_\delta)'\Sigma^{-1}}{(c_\delta-\Gamma \theta_\delta)
    '\Sigma^{-1}\Gamma H'/HH'},
\end{equation}
and, for each $\delta$, $c_\delta, \theta_\delta$ are the solutions to the
convex program
\begin{equation}\label{half_modulus_eq}
  \sup_{\theta, c} H\theta \quad \text{s.t.}\quad c\in\mathcal{C},
  \quad (c-\Gamma \theta)'\Sigma^{-1}(c-\Gamma \theta)\le \delta^2/4.
\end{equation}

It then follows from \citet{donoho94} that among fixed-length CIs based on
linear estimators (CIs that take the form $k'Y \pm \chi$ for some constant
$\chi$), the shortest CI in the limiting experiment takes the form
\begin{equation}\label{eq:optimal-FLCI-limit_experiment}
  k_{\delta^*}'Y\pm \cv_\alpha(\maxbias_{\mathcal{C}}(k_{\delta^*})/\sqrt{k_{\delta^*}'\hat \Sigma k_{\delta^*}})\cdot
  \sqrt{k_{\delta^*}'\Sigma k_{\delta^*}},
\end{equation}
where $\maxbias_{\mathcal{C}}(k_\delta) =-k_\delta' c_\delta$, and
$\delta^{*}=\argmin_{\delta>0} 2\cv_\alpha(\maxbias_{\mathcal{C}}(k_{\delta}) /
\sqrt{k_{\delta}' \Sigma k_{\delta}})\cdot \sqrt{k_{\delta}'\Sigma k_{\delta}}$
is chosen to minimize the CI length. The CI in~\eqref{eq:optimal-CI-formula} is
an analog of this CI, with $\delta$ playing the role of the index $\lambda$.

The CI in~\eqref{eq:optimal-FLCI-limit_experiment} takes a familiar form in the
special case in which $\mathcal{C}$ is a linear subspace of
$\mathbb{R}^{d_{g}}$, so that for some $d_{g}\times d_{\gamma}$ full-rank matrix
$B$ with $d_{\gamma}\leq d_{g}-d_{\theta}$,
$\mathcal{C}=\{B\gamma\colon \gamma\in\mathbb{R}^{d_{\gamma}}\}$. Let
$B_{\perp}$ denote a $d_{g}\times (d_{g}-d_{\gamma})$ matrix that's orthogonal
to $B$. Then for any $\delta>0$, $k_{\delta}'={k}'_{LS, B}$, where
\begin{equation}\label{eq:linear-subspace-sensitivity}
  {k}_{LS, B}'=-H(\Gamma'B_{\perp}(B_{\perp}'\Sigma
  B_{\perp})^{-1}B_{\perp}'\Gamma)^{-1}
  \Gamma'B_{\perp}(B_{\perp}'\Sigma B_{\perp})^{-1}B_{\perp}'
\end{equation}
is the sensitivity of the GLS estimator after
pre-multiplying~\eqref{limit_experiment_eq} by $B_{\perp}'$, (which effectively
picks out the observations with zero misspecification). Since this estimator is
unbiased, the CI in~\eqref{eq:optimal-FLCI-limit_experiment} becomes
${k}_{LS, B}'Y\pm z_{1-\alpha/2}\sqrt{k_{LS, B}'\Sigma k_{LS, B}}$.

Like the asymptotic CI~\eqref{eq:optimal-CI-formula}, the CI
in~\eqref{eq:optimal-FLCI-limit_experiment} has the potential drawback that its
length is determined by the worst possible misspecification in $\mathcal{C}$,
leaving open the possibility of efficiency improvements when $c$ turns out to be
close to zero. As a best-case scenario for such improvements, consider the
problem: among confidence sets with coverage at least $1-\alpha$ for all
$\theta\in\mathbb{R}^{d_\theta}$ and $c\in\mathcal{C}$, minimize expected length
when $\theta=\theta^*$ and $c=0$. Note that this setup is even more favorable
for potential improvements on our CI, since it allows the researcher to guess
correctly that $\theta$ is equal to some $\theta^*$, and it allows for
confidence sets that are not intervals (in this case, length is defined as
Lebesgue measure). Let $\kappa_{*}(H, \Gamma, \Sigma, \mathcal{C})$ denote the
ratio of this optimized expected length relative to the length of the CI
in~\eqref{eq:optimal-FLCI-limit_experiment} (it can be shown that this ratio
does not depend on $\theta^*$).

If $\mathcal{C}$ is convex, a formula for
$\kappa_{*}(H, \Gamma, \Sigma, \mathcal{C})$ follows from applying the general
results in Corollary 3.3 in \citet{ArKo18optimal} to the limiting model. If
$\mathcal{C}$ is also centrosymmetric, then
\begin{equation}\label{kappa_conv_cs_eq}
  \kappa_{*}(H, \Gamma, \Sigma, \mathcal{C})
  =\frac{(1-\alpha)E\left[\omega(2(z_{1-\alpha}-Z))| Z\le z_{1-\alpha}\right]}{2\min_\delta
    \cv_\alpha\left(\frac{\omega(\delta)}{2\omega'(\delta)}-\frac{\delta}{2}\right)\omega'(\delta)},
\end{equation}
where $Z\sim \mathcal{N}(0,1)$ and $\omega(\delta)$ is two times the optimized
value of~\eqref{half_modulus_eq}. Furthermore, we show in
\Cref{theorem:sharp-unversal-bound} that the right-hand side is lower-bounded by
$(z_{1-\alpha}(1-\alpha)-\tilde{z}_{\alpha}\Phi(\tilde{z}_{\alpha})+
\phi(z_{1-\alpha})-\phi(\tilde{z}_{\alpha}))/ z_{1-\alpha/2}$, where
$\tilde{z}_{\alpha}=z_{1-\alpha}-z_{1-\alpha/2}$ for any $H$, $\Gamma$, $\Sigma$
and $\mathcal{C}$, where $\phi(\cdot)$ denotes the standard normal density. For
$\alpha=0.05$, this universal lower bound evaluates to 71.7\%. Evaluating
$\kappa_{*}$ for particular choices of $H$, $\Gamma$, $\Sigma$, and
$\mathcal{C}$ often yields even higher efficiency.

If $\mathcal{C}$ is a linear subspace, then $\omega(\delta)$ is linear, and
\begin{equation}\label{eq:kappa-linear-subspace}
  \kappa_{*}(H, \Gamma, \Sigma, \mathcal{C})=\frac{(1-\alpha)z_{1-\alpha}
    +\phi(z_{1-\alpha})}{z_{1-\alpha/2}}\geq \frac{z_{1-\alpha}}{z_{1-\alpha/2}},
\end{equation}
where the lower bound follows since $\phi(z_{1-\alpha})\geq \alpha z_{1-\alpha}$
by the Gaussian tail bound $1-\Phi(x)\leq \phi(x)/x$ for $x>0$. This bound
corresponds to that in \citet{pratt61} for the case of a univariate normal mean.
The potential efficiency improvement essentially comes from using prior
knowledge of $\theta^*$ to turn a two-sided critical value into a one-sided
critical value. Furthermore, it follows from \citet{joshi69} that the CI
${k}_{LS, B}'Y\pm z_{1-\alpha/2}\sqrt{k_{LS, B}'\Sigma k_{LS, B}}$ is the unique
CI that achieves minimax expected length. Thus, not only is the scope for
improvement at a particular $\theta^{*}$ bounded
by~\eqref{eq:kappa-linear-subspace}, any CI with shorter expected length at some
$\theta^{*}$ must necessarily perform worse elsewhere in the parameter space.

For the one-sided CI~\eqref{h_oci_eq}, the analogous CI in the limiting
experiment is
$\hor{k'Y-\maxbias_{\mathcal{C}}(k)-z_{1-\alpha}\sqrt{k'\Sigma k}, \infty}$,
and, as we discuss in \Cref{efficiency_sec_append}, to choose the
optimal sensitivity $k$, one can consider optimizing a given quantile of its
worst-case excess length.
The results in \citet{ArKo18optimal} again give an efficiency bound for improvement
at $c=0$ and a particular $\theta^*$, analogous to~\eqref{kappa_conv_cs_eq} for
the two-sided case.  See \Cref{efficiency_sec_append} for details.
If $\mathcal{C}$ is a linear subspace, then
optimizing quantiles of worst-case excess length yields the CI
$\hor{k_{LS, B}'Y-z_{1-\alpha}\sqrt{k_{LS, B}'\Sigma k_{LS, B}}, \infty}$,
independently of the quantile one is optimizing. Furthermore, the efficiency
bound implies that this one-sided CI is in fact fully optimal over all quantiles
of excess length and all values of $\theta, c$ in the local parameter space.

These efficiency results for the CI~\eqref{eq:optimal-FLCI-limit_experiment}
in the limiting experiment suggest that the scope for improvement over the CI
in~\eqref{eq:optimal-CI-formula} should be limited in large samples.
\Cref{main_text_efficiency_bound_thm}, stated in the next section, uses the
analogy with the approximately linear model~\eqref{limit_experiment_eq} along
with Le Cam-style arguments involving least favorable submodels to show that
this bound indeed translates to the locally misspecified GMM model. For
one-sided CIs, we state an analogous result in \Cref{efficiency_sec_append}. We
discuss the implications of these results in \Cref{efficiency_examples_sec}.

\subsection{Asymptotic efficiency bound}\label{efficiency_bound_sec_main}

To make precise our statements about coverage and efficiency, we need the notion
of uniform (in the underlying distribution) coverage of a confidence interval.
This requires additional notation, which we now introduce. Let $\mathcal{P}$
denote a set of distributions $P$ of the data $\{w_{i}\}_{i=1}^{n}$, and let
$\Theta_n\subseteq\mathbb{R}^{d_\theta}$ denote the parameter space for
$\theta$. We require coverage for all pairs
$(\theta, P)\in \Theta_n\times \mathcal{P}$ such that
$\sqrt{n}g_P(\theta)\in\mathcal{C}$, where the subscript $P$ on the population
moment condition makes it explicit that it depends on the distribution of the
data.\footnote{To be precise, we should also subscript all other quantities such
  as $\Gamma$ and $\Sigma$ by $P$. To prevent notational clutter, we drop this
  index in the main text unless it causes confusion.} Letting
$\mathcal{S}_n=\{(\theta, P)\in \Theta_n\times\mathcal{P}:
\sqrt{n}g_P(\theta)\in\mathcal{C}\}$ denote this set, the condition for coverage
at confidence level $1-\alpha$ can be written
\begin{equation}\label{coverage_eq}
  \liminf_{n\to\infty} \inf_{(\theta, P)\in\mathcal{S}_n} P(h(\theta)\in\mathcal{I}_{n})
  \ge 1-\alpha.
\end{equation}
We say that a confidence set $\mathcal{I}_{n}$ is asymptotically valid
(uniformly over $\mathcal{S}_{n}$) at confidence level $1-\alpha$ if this
condition holds.\footnote{In general, $\theta_0$ and $h(\theta_0)$ may be set
  identified for a given sample size $n$ (although our assumptions imply that
  the identified set will shrink at a root-$n$ rate). The coverage
  requirement~\eqref{coverage_eq} states that the CI must cover points in the
  identified set for $h(\theta)$, as in \citet{im04}; see
  \Cref{efficiency_sec_append}.}

Among two-sided CIs of the form $\hat{h}\pm \hat{\chi}$ that are asymptotically
valid, we prefer CIs with shorter expected length. To avoid issues with
convergence of moments, we use truncated expected length, and define the
asymptotic expected length of a two-sided CI at $P_{n}\in\mathcal{P}$ as
$\liminf_{T\to\infty}\liminf_{n\to\infty}E_{P_{n}}\min\{\sqrt{n}\cdot
2\hat{\chi}, T\}$, where $E_{P}$ denotes expectation under $P$.

We are now ready to state the main efficiency result.

\begin{theorem}\label{main_text_efficiency_bound_thm}
  Suppose that $\mathcal{C}$ is convex and centrosymmetric. Let
  $\hat h_{\lambda^*}$ and $\hat\chi^*_{\lambda^*}$ be formed as in
  \Cref{sec:implementation}. Suppose that
  \Cref{theta_init_assump,g_clt_assump_append,,GammaHSigma_assump,scriptB_assump,,deltastar_continuous_assump}
  in \Cref{efficiency_sec_append} hold. Suppose that the data
  $\{w_{i}\}_{i=1}^{n}$ are i.i.d.\ under all $P\in\mathcal{P}$. Let
  $(\theta^*, P_0)$ be correctly specified (i.e. $g_{P_0}(\theta^*)=0$) such that
  $\mathcal{P}$ contains a submodel through $P_0$ satisfying
  \Cref{submodel_assump}. Then:
  \begin{enumerate}[label=(\roman{enumi})]
  \item The CI $\hat{h}_{\lambda^{*}}\pm \hat \chi^*_{\lambda^*}$ is
    asymptotically valid, and its half-length $\hat\chi^*_{\lambda^*}$ satisfies
    $\sqrt{n}\hat \chi^*_{\lambda^*}=\chi(\theta, P)+o_P(1)$ uniformly over
    $(\theta, P)\in\mathcal{S}_n$ where
    \begin{equation*}
      \chi(\theta, P)=\min_k \cv_\alpha(\maxbias_{\mathcal{C}}(k)/
      \sqrt{k'\Sigma_{\theta, P}k})\sqrt{k'\Sigma_{\theta, P}k}
    \end{equation*}
    with $\maxbias_{\mathcal{C}}(k)$ calculated with $\Gamma=\Gamma_{\theta, P}$
    and $H=H_{\theta}$.
  \item For any other asymptotically valid CI $\hat{h}\pm \hat{\chi}$,
    \begin{equation*}
      \frac{\liminf_{T\to\infty}\liminf_{n\to\infty}E_{P_0}\min\{\sqrt{n}\cdot
        2\hat{\chi}, T\}} {2\chi(\theta^*, P_0)}
      \ge \kappa_{*}(H_{\theta^*}, \Gamma_{\theta^*, P_0}, \Sigma_{\theta^*, P_0}, \mathcal{C}),
    \end{equation*}
    where $\kappa_{*}(H, \Gamma, \Sigma, \mathcal{C})$ is defined
    in~\eqref{kappa_conv_cs_eq}. Furthermore, for any $H, \Sigma, \Gamma$, and
    $\mathcal{C}$, $\kappa_{*}$ admits the universal lower bound
    $(z_{1-\alpha}(1-\alpha)-\tilde{z}_{\alpha}\Phi(\tilde{z}_{\alpha})+
    \phi(z_{1-\alpha})-\phi(\tilde{z}_{\alpha}))/ z_{1-\alpha/2}$, where
    $\tilde{z}_{\alpha}=z_{1-\alpha}-z_{1-\alpha/2}$ and $\phi(\cdot)$ denotes
    the standard normal density.
  \end{enumerate}

\end{theorem}
The proof for this theorem is given in \Cref{efficiency_sec_append}, which also
gives an analogous result for one-sided confidence intervals.
\Cref{theta_init_assump,g_clt_assump_append,,GammaHSigma_assump,scriptB_assump,,deltastar_continuous_assump},
stated in \Cref{efficiency_sec_append}, require that the conditions in
\Cref{general_results_sec} hold in a uniform sense over the class $\mathcal{P}$.
In the supplemental materials, we give primitive conditions for these
assumptions in the misspecified linear IV model. \Cref{submodel_assump}, also
stated in the appendix, requires that the class $\mathcal{P}$ be rich enough to
contain a submodel that is least favorable for the GMM problem, so that the
class doesn't implicitly impose any other conditions that could be used to make
inference easier. In the supplemental materials, we provide a general way of
constructing a submodel satisfying these conditions.

The universal lower bound on $\kappa_{*}$ is new and may be of independent
interest. For $\alpha=0.05$, it evaluates to 71.7\%. The universal lower bound
is sharp in the sense that there exist $\Gamma, \Sigma, H$ and $\mathcal{C}$ for
which $\kappa_{*}$ equals this lower bound. In particular applications, the
efficiency bound $\kappa_{*}$ can be computed at estimates of $\Gamma$, $\Sigma$
and $H$, and often, this gives much higher efficiencies. We illustrate these
bounds in the empirical application in \Cref{empirical_sec}.

\subsection{Discussion}\label{efficiency_examples_sec}

To help build intuition for the efficiency bound in
\Cref{main_text_efficiency_bound_thm}, and to relate this result to the
literature, we now consider some special cases. We first discuss the (standard)
correctly specified case. Second, we consider the case in which some moments are
known to be valid, and the misspecification in the remaining moments is
unrestricted. This case may be of interest in its own right. We then discuss
the general case.  Finally, we discuss the connection to certain statistical
measures of distance considered in the literature.

\subsubsection{Correctly specified case}\label{sec:correctly-specified-case}
Suppose that $\mathcal{C}=\{0\}$. This is in particular a linear subspace of
$\mathbb{R}^{d_{g}}$, with $B=0$, and $B_{\perp}=I$, the $d_{g}\times d_{g}$
identity matrix. Thus, in the limiting experiment, the optimal CI uses the GLS
estimator $k_{LS,0}'Y$, with $k_{LS,0}$ given
in~\eqref{eq:linear-subspace-sensitivity} (with $B=0$). For testing the null
hypothesis $H\theta=h_{0}$ against the one-sided alternative
$H\theta\geq h_{0}$, the one-sided $z$-statistic based on ${k}_{LS,0}'Y$ is
uniformly most powerful \citep[Proposition 15.2]{van_der_vaart_asymptotic_1998}.
Inverting these tests yields the CI
$\hor{k_{LS,0}'Y-z_{1-\alpha}\sqrt{k_{LS,0}'\Sigma k_{LS,0}}, \infty}$. Since
the underlying tests are uniformly most powerful, this CI achieves the shortest
excess length, simultaneously for all quantiles and all possible values of the
parameter $\theta$. For two-sided CIs, the results described in
\Cref{limiting_experiment_sec} imply that the CI
$h_{LS,0}'Y\pm z_{1-\alpha/2}\sqrt{k_{LS,0}'\Sigma k_{LS,0}}$ is the unique CI
that achieves minimax expected length, and the efficiency of this CI relative to
a CI that optimizes its expected length at a single value $\theta^*$ of $\theta$
when indeed $\theta=\theta^*$ is given in \Cref{eq:kappa-linear-subspace}. It
evaluates to 84.99\% at $\alpha=0.05$.

Applying \Cref{main_text_efficiency_bound_thm} to the case
$\mathcal{C}=\{0\}$ gives an asymptotic version of the two-sided efficiency
bound. Furthermore, the CI in \Cref{main_text_efficiency_bound_thm}
reduces to the usual two-sided CI based on $\hat{\theta}_{\Sigma^{-1}}$. Thus,
in this case, \Cref{main_text_efficiency_bound_thm} shows that very
little can be gained over the usual two-sided CI by optimizing the CI relative
to a particular distribution $P_0$. Results in the appendix give an analogous
result for one-sided CIs. In the one-sided case, this asymptotic result is
essentially a version of a classic result from the semiparametric efficiency
literature for one-sided tests, applied to CIs (see Chapter 25.6 in
\citealp{van_der_vaart_asymptotic_1998}). In the two-sided case, the result is, to
our knowledge, new.

\subsubsection{Some valid and some invalid moments}\label{sec:some-valid-some}

Consider now the case in which the first $d_{g}-d_{\gamma}$ moments are known to
be valid, with the potential misspecification for the remaining $d_{\gamma}$
moments unrestricted. Then
$\mathcal{C}=\{(0', \gamma')'\colon \gamma\in\mathbb{R}^{d_{\gamma}}\}$
corresponds to a linear subspace with $B$ given by the last $d_{\gamma}$ columns
of the identity matrix, and $B_{\perp}$ given by the first $d_{g}-d_{\gamma}$
columns.
Optimal CIs in the limiting experiment therefore use the estimator $k_{LS,
  B}'Y$, which is the GLS estimator based
only on the observations with no misspecification.

The one-sided CI based on $k_{LS, B}'Y$ achieves the shortest excess length,
simultaneously for all quantiles and all possible values of the parameter
$\theta$. The two-sided CI
$k_{LS, B}'Y\pm z_{1-\alpha/2}\sqrt{k_{LS, B}'\Sigma k_{LS, B}}$ is optimal in
the same sense as the usual CI in \Cref{sec:correctly-specified-case}: it
achieves minimax expected length, and its efficiency, relative to a CI that
optimizes its length at a single $\theta^{*}$ and $\gamma=0$, is lower-bounded
by $z_{1-\alpha}/z_{1-\alpha/2}$. \Cref{main_text_efficiency_bound_thm} formally
translates the efficiency bound from the limiting model to the GMM model, so
that the usual two-sided CI based on $h(\hat{\theta}_{W(B)})$ is asymptotically
efficient in the same sense as the usual CI based on
$h(\hat{\theta}_{\Sigma^{-1}})$ discussed in \Cref{sec:correctly-specified-case}
under correct specification. Just as with the results in
\Cref{sec:correctly-specified-case}, this asymptotic result is, to our
knowledge, new. The one-sided analog follows from the results in
\Cref{efficiency_sec_append}. These results stand in sharp contrast to the
results for estimation, where the MSE improvement at small values of $\gamma$
may be substantial.

An important consequence of these results is that asymptotically valid one-sided
CIs based on shrinkage or model-selection procedures, such as one-sided versions
of the CIs proposed in \citet{andrews_hybrid_2009}, \citet{ditraglia_using_2016} or
\citet{mccloskey_2017} must have \emph{worse} excess length performance than the
usual one-sided CI based on the GMM estimator $h(\hat{\theta}_{W(B)})$ that uses
valid moments only. While it is possible to construct two-sided CIs that improve
upon the usual CI based on $h(\hat{\theta}_{W(B)})$ at particular values of
$\theta$ and $\gamma$, the scope for such improvement is smaller than the ratio
of one- to two-sided critical values. Furthermore, any such improvement must
come at the expense of worse performance at other points in the parameter
space.\footnote{Consistently with these results, in a simulation study
  considered in \citet{ditraglia_using_2016}, the post-model selection CI that
  he proposes is shown to be wider on average than the usual CI around a GMM
  estimator that uses valid moments only.} Therefore, in order to tighten CIs
based on valid moments only, it is \emph{necessary} to make a priori
restrictions on the potential misspecification of the remaining moments.

\subsubsection{General case}\label{sec:general-case}

According to the results in \Cref{sec:some-valid-some}, one must place a
priori bounds on the amount of misspecification in order to use misspecified
moments. This leads us to the general case, where we place the local
misspecification vector $c$ in some set $\mathcal{C}$ that is not necessarily a
linear subspace. One can then form a CI centered at an estimate formed from
these misspecified moments using the methods in
\Cref{sec:implementation}. In the case where $\mathcal{C}$ is convex and
centrosymmetric, \Cref{main_text_efficiency_bound_thm} shows that this CI
is near optimal, in the sense that no other CI can improve upon it by more than
a factor of $\kappa_{*}$, even in the favorable case of correct specification.
Since the width of the CI is asymptotically constant under local parameter
sequences $\theta_n\to\theta^*$ and sufficiently regular probability
distributions $P_n\to P_0$ (for example, $P_n\to P_0$ along submodels satisfying
\Cref{submodel_assump}), this also shows that the CI is near optimal
in a local minimax sense. In the general case,
\Cref{main_text_efficiency_bound_thm}, as well as the analogous results for one-sided CIs in \Cref{efficiency_sec_append} are,
to our knowledge, new.

As we discuss in \Cref{sec:implementation}, we recommend reporting results for a
range of sets $\mathcal{C}(\Mbound)$ indexed by a scalar $\Mbound$ that bounds
the magnitude of misspecification. One may instead wish to report a single CI
based on a data-driven estimate of $\Mbound$, for example, by using a
first-stage $J$ test to assess plausible magnitudes of misspecification.
Formally, one would seek a CI that is valid over
$\mathcal{C}(\overline{\Mbound})$ while improving length when in fact
$\norm{\gamma}\ll \overline{\Mbound}$, where $\overline{\Mbound}$ is some
initial conservative bound. When $\mathcal{C}$ is convex and centrosymmetric,
\Cref{main_text_efficiency_bound_thm} shows that the scope for such improvements
is limited: the average length of any such CI cannot be much smaller than the CI
that uses the most conservative choice $\overline{\Mbound}$, even when $c=0$.
The impossibility of choosing $\Mbound$ based on the data is related to the
impossibility of using specification tests to form an upper bound for $\Mbound$.
On the other hand, it is possible to obtain a lower bound for $\Mbound$ using
such tests. We develop lower CIs for $\Mbound$ in \Cref{sec:specification-test}.

\subsubsection{Cressie-Read divergences}\label{sec:corr-spec-cress}

\citet{andrews_informativeness_2018} have shown that defining misspecification
in terms of the magnitude of any divergence in the
\citet{cressie_multinomial_1984} family leads to a set $\mathcal{C}$ that
asymptotically takes the form
$\mathcal{C}=\{\Sigma^{1/2}\gamma:\|\gamma\|_2\le \Mbound\}=\{c\colon
c'\Sigma^{-1}c\le \Mbound^2\}$. The Cressie-Read family includes the Hellinger
distance used by \citet{kitamura_robustness_2013}, who consider minimax point
estimation among estimators satisfying certain regularity conditions. Since this
set $\mathcal{C}$ takes the form discussed in \Cref{remark:choice-of-C} with
$p=2$ and $B=\Sigma^{1/2}$, it follows from the discussion in
\Cref{remark:choice-of-C} that the optimal sensitivity corresponds to the GMM
estimator with weighting matrix
$(\lambda BB'+\Sigma)^{-1}=(\lambda+1)^{-1}\Sigma^{-1}$. Since this is
proportional to the weighting matrix $\Sigma^{-1}$ that is optimal under correct
specification, we obtain the same optimal sensitivity
$k_{LS,0}'=-H(\Gamma' \Sigma^{-1}\Gamma)^{-1}\Gamma'\Sigma^{-1}$ as in the
correctly specified case discussed in \Cref{sec:correctly-specified-case}. As we
show in \Cref{sec:cress-read-diverg}, this form of $\mathcal{C}$ leads to a
closed form solution for the efficiency bound $\kappa_*$.

The results above imply that
any estimator with sensitivity $k_{LS,0}$ is near-optimal for CI
construction. In line with these results, the estimator in
\citet{kitamura_robustness_2013} has sensitivity $k_{LS,0}$. Thus, the usual GMM
estimator $h(\hat{\theta}_{\Sigma^{-1}})$ and the estimator in
\citet{kitamura_robustness_2013} are both near-optimal for CI construction, even
if one allows for arbitrary CIs that are not necessarily centered at estimators
that satisfy the regularity conditions in \citet{kitamura_robustness_2013}.
Also, because they have the same sensitivity, under this form of
misspecification, the usual GMM estimator $h(\hat{\theta}_{\Sigma^{-1}})$ and
the estimator in \citet{kitamura_robustness_2013} have the same local
asymptotic minimax properties.

\subsection{Extensions: asymmetric constraints and constraints on
  \texorpdfstring{$\theta$}{theta}}

If the set $\mathcal{C}$ is convex but asymmetric (such as when $\mathcal{C}$
includes bounds on a norm as well as sign restrictions, or when $\mathcal{C}$
includes equality and sign restrictions, as in \citet{moon_estimation_2009}),
one can still apply bounds from \citet{ArKo18optimal} to the limiting model
described in \Cref{limiting_experiment_sec}. Our general asymptotic efficiency
bounds in \Cref{efficiency_sec_append} translate these results to the locally
misspecified GMM model so long as $\mathcal{C}$ is convex. Since the negative
implications for efficiency improvements under correct specification use
centrosymmetry of $\mathcal{C}$, introducing asymmetric restrictions, such as
sign restrictions, is one possible way of getting efficiency improvements at
some smaller set $\mathcal{D}\subseteq\mathcal{C}$ while maintaining coverage
over $\mathcal{C}$. We derive efficiency bounds and optimal CIs for this problem
in \Cref{efficiency_sec_append}. Interestingly, the scope for efficiency
improvements can be different for one- and two-sided CIs, and can depend on the
direction of the CI in this case. To get some intuition for this, note that, in
the instrumental variables model with a single instrument and single endogenous
regressor, sign restrictions on the covariance of an instrument with the error
term can be used to sign the direction of the bias of the instrumental variables
estimator, which is useful for forming a one-sided CI only in one direction.

Finally, while we focus on restrictions on $c$, one can also incorporate local
restrictions on $\theta$. Our general results in
\Cref{efficiency_sec_append} give efficiency bounds that cover this
case. Similar to the discussion above, these results have implications for using
prior information about $\theta$ to determine the amount of misspecification, or
to shrink the width of a CI directly. In particular, while it is possible to use
prior information on $\theta$ (say, an upper bound on $\|\theta\|$ for some norm
$\|\cdot\|$) to shrink the width of the CI, the width of the CI and the
estimator around which it is centered must depend on the a priori upper bounds
on the magnitude of $\theta$ and $c$ when this prior information takes the form
of a convex, centrosymmetric set for $(\theta', c')'$. This rules out, for
example, choosing the moments based on whether the resulting estimate for
$\theta$ is in a plausible range.

\section{Applications}\label{applications_sec}

This section describes particular applications of our approach, along with a
discussion of implementation details appropriate to each application.

\subsection{Instrumental variables}\label{iv_sec}

The single equation linear instrumental variables (IV) model is given by
\begin{equation}\label{eq:iv-structural}
  y_{i}=x_i'\theta_0+\varepsilon_i
\end{equation}
where, in the correctly specified case,
$E\varepsilon_{i}z_{i}=E(y_{i}-x_i'\theta_0)z_{i}=0$, with $z_{i}$ a $d_{g}$-vector
of instruments. This is an instance of a GMM model with
$g(\theta)=E(y_{i}-x_i'\theta)z_{i}$ and
$\hat{g}(\theta)=\frac{1}{n}\sum_{i=1}^{n}z_{i}(y_{i}-x_{i}'\theta)$.

One common reason for misspecification in this model is that the instruments do
not satisfy the exclusion restriction, because they appear directly in the
structural equation~\eqref{eq:iv-structural}, so that
$\varepsilon_{i}=z_{Ii}'\gamma/\sqrt{n}+\eta_{i}$, where $E[z_{i}\eta_{i}]=0$,
and $z_{Ii}$ corresponds to a subset $I$ of the instruments, the validity of
which one is worried about. This form of misspecification has previously been
considered in a number of papers, including \citet{hahn_iv_2005},
\citet{conley_plausibly_2012}, and \citet{andrews_measuring_2017}, among others.
Bounding the norm of $\gamma$ using some norm $\norm{\cdot}$ then leads to the
set given in \Cref{eq:C_Bgamma}, with $B=E[z_{i}z_{Ii}']$.

Although the matrix $B$ is unknown, for the purposes of estimating the optimal
sensitivity and constructing asymptotically valid CIs, it can be replaced by the
sample analog $\hat{B}=n^{-1}\sum_{i=1}^{n}z_{i}z_{Ii}'$. This does not affect
the asymptotic validity or coverage properties of the resulting CI\@. Under this
setup, the parameter $M$ bounds that magnitude of $\gamma$, the direct effect of
the instruments on the outcome. Therefore, the appropriate choice of $M$ will
depend on the plausible magnitude of these direct effects---see, for example,
\citet{conley_plausibly_2012} for examples and a discussion.

The linearity of the moment condition leads to simplifications in our
implementation in~\Cref{sec:implementation}. In Step~\ref{item:step1}, as the
initial estimator, one can use the two-stage least squares (2SLS) estimator
\begin{equation*}
  \textstyle\hat\theta_{\text{initial}}
  =\left[\left(\sum_{i=1}^{n}z_{i}x_i'\right)'\left(\sum_{i=1}^{n}z_{i}z_{i}'\right)^{-1}
    \left(\sum_{i=1}^{n}z_{i}x_i'\right)\right]^{-1}\left(\sum_{i=1}^{n}z_{i}x_i'\right)'
  \left(\sum_{i=1}^{n}z_{i}z_{i}'\right)^{-1}
  \sum_{i=1}^{n}z_{i}y_{i}.
\end{equation*}
This leads to the estimates $\hat{\Gamma}=-\frac{1}{n}\sum_{i=1}^{n}z_{i}x_i'$
and
$\hat{\Sigma}=\frac{1}{n}\sum_{i=1}^{n}
(y_{i}-x_i'\hat\theta_{\text{initial}})^2z_{i}z_{i}'$. Alternatively, if we
assume homoskedasticity, we can use the estimator
$\hat\Sigma_{H}=\frac{1}{n}\sum_{i=1}^{n}(y_{i}-x_i'\hat\theta_{\text{initial}})^2\cdot
\frac{1}{n}\sum_{i=1}^{n}z_{i}z_{i}'$. In the correctly specified case, the 2SLS
estimator is only optimal under homoskedasticity, while the GMM estimator with
weighting matrix $\hat\Sigma^{-1}$ is optimal in general. Due to concerns with
finite sample performance, however, it is common to use the 2SLS estimator along
with standard errors based on a robust variance estimate, even when
heteroskedasticity is suspected. Mirroring this practice, one can use
$\hat{\Sigma}_{H}$ when forming the optimal sensitivity
$\hat{k}_{\lambda^{*}_{M}}$ and worst-case bias in Step~\ref{item:step2}, but
use the robust variance estimate $\hat{\Sigma}$ in Step~\ref{item:step3} when
forming the final CI in \Cref{eq:optimal-CI-formula}. The CI will then be
optimal under homoskedasticity, but it will remain valid under
heteroskedasticity, just like the usual CI based on 2SLS with robust standard
errors in the correctly specified case.

If the parameter of interest linear in $\theta$, $h(\theta)=H\theta$, then the
one-step estimator $\hat{h}_{\lambda^{*}_{M}}$ in Step~\ref{item:step3} does not
depend on the choice of the initial estimator (except possibly through the
estimate of $\Sigma$ when forming the desired sensitivity):
\begin{equation*}
  \begin{split}
    \hat{h}_{\lambda^{*}_{M}} & =H\hat\theta_{\text{initial}}+
    \hat{k}_{\lambda^{*}_{M}}'\frac{1}{n}\sum_{i=1}^{n}(y_{i}-x_i'\hat\theta_{\text{initial}})z_{i}
    = \hat{k}_{\lambda^{*}_{M}}'\frac{1}{n}\sum_{i=1}^{n}y_{i}z_{i}
    +\left(H-\hat{k}_{\lambda^{*}_{M}}'\frac{1}{n}\sum_{i=1}^{n}z_{i}x_i'\right)\hat\theta_{\text{initial}}\\
    & =\hat{k}_{\lambda^{*}_{M}}'\frac{1}{n}\sum_{i=1}^{n}y_{i}z_{i},
  \end{split}
\end{equation*}
where the second line follows since the sensitivities
$\hat{k}_{\lambda^{*}_{M}}$ satisfy
$H=-\hat{k}_{\lambda^{*}_{M}}'\hat\Gamma=
\hat{k}_{\lambda^{*}_{M}}'\frac{1}{n}\sum_{i=1}^{n}z_{i}x_i'$. Since the
estimator $\hat{h}_{\lambda^{*}_{M}}$ is linear, the worst-case bias
calculations are the same under global misspecification, when the magnitude $M$
of $\gamma$ in \Cref{eq:C_Bgamma} grows at the rate $\sqrt{n}$. By using
variance estimates that are valid under global misspecification in place of the
variance estimate
$\hat{k}_{\lambda^{*}_{M}}'\hat\Sigma \hat{k}_{\lambda^{*}_{M}}$ in the CI
construction, one can ensure that the resulting CI also remains valid under
global misspecification. See \Cref{sec:cis-based-missp} for details.

\begin{remark}\label{iv_controls_remark}
  This framework can also be used to incorporate a priori restrictions on the
  magnitude of coefficients on control variables in an instrumental variables
  regression. Suppose that we have a set of controls $w_{i}$, that appear in the
  structural equation~\eqref{eq:iv-structural}, so that
  $y_{i}=x_i'\theta+w_i'\gamma/\sqrt{n}+\epsilon_{i}$, and $\epsilon_{i}$ is
  uncorrelated with $w_{i}$ as well as vector of instruments $\tilde{z}_{i}$. If
  one is willing to restrict the magnitude of the coefficient vector $\gamma$,
  so that $\norm{\gamma}\leq \Mbound$, then one can add $w_{i}$ to the original
  vector of instruments $\tilde{z}_{i}$, $z_{i}=(\tilde{z}_{i}', w_{i}')'$. For
  example, if one is concerned with functional form misspecification, one can
  define the control variables to be higher order series terms. We then obtain
  the misspecified IV model with the set $\mathcal{C}$ given
  by~\eqref{eq:C_Bgamma}, with $B=E[z_{i}w_{i}']$. Thus, we can interpret this
  model as a locally misspecified version of a model with $w_{i}$ used as an
  excluded instrument.
\end{remark}

\begin{remark}
  Instead of bounding the coefficient vector $\gamma$, one can alternatively bound
  the magnitude of the direct effect $z_{Ii}'\gamma$. If all instruments are
  potentially invalid, $z_{Ii}=z_{i}$, and one sets
  $\mathcal{C}=\{\gamma\colon E[(z_{i}'\gamma)^{2}]\leq \Mbound\}$, then under
  homoscedasticity, this corresponds to the case discussed in
  \Cref{sec:corr-spec-cress}, where the uncertainty from potential
  misspecification is exactly proportional to the asymptotic sampling
  uncertainty in $\hat{g}(\theta)$. Consequently, in this case the optimal
  sensitivity is the same as that given by the 2SLS estimator.
\end{remark}

\subsection{Omitted variables bias in linear regression}

Specializing to the case where $z_{i}=x_i$, the misspecified IV model of
\Cref{iv_sec} gives a misspecified linear regression model as a special case.
This can be used to assess sensitivity of regression results to issues such as
omitted variables bias. In particular, consider the linear regression model
\begin{equation*}
  y_{i}=x_i'\theta+w_i^*+\tilde \varepsilon_i, \quad Ex_i\tilde\varepsilon_i=0
\end{equation*}
where $x_i$ and $y_{i}$ are observed and $w_i^*$ is a (possibly unobserved)
omitted variable. Correlation between $w_i^*$ and $x_i$ will lead to omitted
variables bias in the OLS regression of $y_{i}$ on $x_i$. If $w_{i}^{*}$ is
unobserved, then we obtain our framework by making the assumption
$\sqrt{n}Ew_{i}^{*}x_{i}\in \mathcal{C}$, for some set $\mathcal{C}$, and
letting $\hat{g}(\theta)=\frac{1}{n}\sum_{i=1}^{n}x_{i}(y_{i}-x_{i}'\theta)$.
This setup can also cover choosing between different sets of control variables.
Suppose that $w_{i}^{*}=w_{i}'\gamma$, where $w_{i}$ is a vector of observed
control variables that the researcher is considering not including in the
regression. If $\gamma$ is unrestricted, then by the results in
\Cref{sec:some-valid-some}, the long regression of $y_{i}$ on both $x_{i}$ and
$w_{i}$ yields nearly optimal CIs. If one is willing to restrict the magnitude
of $\gamma$, it is possible to tighten these CIs, with the setting reducing to
that in \Cref{iv_controls_remark}, with $z_{i}=(x_{i}, w_{i}')'$. The same
framework can be used to incorporate selection bias by defining $w_i^*$ to be
the inverse Mills ratio term in the formula for
$E[y_{i}\mid x_i, i\;\text{observed}]$ in \citet{heckman_sample_1979}.

\subsection{Functional form misspecification}\label{sec:functional-form-misspecification}

Our setup allows for misspecification in moment conditions arising from
functional form misspecification.  To apply our setup, one must relate this
misspecification to the bounds $\mathcal{C}$ on the moment conditions at the
true parameter value.  One approach to bounding functional form misspecification
is to use smoothness conditions from the nonparametric statistics literature,
such as bounds on derivatives \citep[see, for example,][for an introduction to
this literature]{tsybakov09}.  Since these sets are typically convex (taking a convex
combination of two functions that satisfy a given bound on a given derivative
gives a function that also satisfies this bound), they typically lead
to convex sets $\mathcal{C}$, so that our framework can be applied.

As a simple example, consider a nonparametric IV model with discrete covariates:
\begin{equation*}
  E[y_i-m(x_i)\mid z_i]=0.
\end{equation*}
Suppose $x$ takes values in the finite set
$\mathcal{X}=\{\tilde x_1,\ldots, \tilde x_{N_x}\}$ and $z_i$ takes values in
the finite set $\mathcal{Z}=\{\tilde z_1,\ldots, \tilde z_{N_z}\}$. This setting
was considered by \citet{freyberger_identification_2015}, who place only
nonparametric smoothness or shape restrictions on the unknown function $m$. To
see the connection with our setting, we note that such restrictions can be
interpreted as bounds on specification error from a parametric model. If one
models these restrictions as local to a parametric family, one obtains our
setting. In particular, let $m(x_i)=f(x_i, \theta_0)+n^{-1/2}r(x_i)$,
$r\in\mathcal{R}$, where $\mathcal{R}$ is a nonparametric smoothness class. For
example, if $x_i$ is univariate, we can let
$f(x_i, \theta)=\theta_1+\theta_2x_i$ and define $\mathcal{R}$ to be the class
of functions with $r(0)=r'(0)=r''(0)$ and second derivative bounded by some
constant $\Mbound$. This is equivalent to placing the bound $n^{-1/2}\Mbound$ on
the second derivative of $m(\cdot)$, which corresponds to a Hölder smoothness
class. We can then map this to a misspecified GMM model, with the $j$th element
of the moment function given by
$g_j(x_i, y_i, \theta)=(y_i-f(x_i, \theta_0))I(z_i=\tilde z_j)$ and $j$th
element of the misspecification vector $c$ given by
$Er(x_i)I(z_i=\tilde z_j)=\sum_{\tilde x\in\mathcal{X}} r(\tilde x)
P(x_i=\tilde{x}, z_i=\tilde z_j)$. Stacking these equations, we see that
$c=B\gamma$ where $B$ is a matrix composed of the elements
$P(x_i=\tilde x, z_i=\tilde z_j)$ and
$\gamma=(r(\tilde x_1), \ldots, r(\tilde x_{N_x}))'$. As with the IV setting in
\Cref{iv_sec}, $B$ is unknown, but can be replaced by a consistent estimate
based on the sample analogue. So long as the set $\mathcal{R}$ is convex, we
obtain convex restrictions on $\gamma$ and therefore $c$, so that our framework
applies.

This example brings up an important point about the interpretation of
$h(\theta)$. If the object of interest is a functional of
$m(x)=f(x, \theta_0)+n^{-1/2}r(x)$, then we will need to allow the object of
interest $h(\cdot)$ to depend on the misspecification vector directly, as well
as on $\theta$. As discussed at the beginning of \Cref{general_results_sec},
this falls into a mild extension of our framework. Alternatively, under a
suitable parametrization of $f$ and $r$, it is often possible to define the
object of interest to be function of $\theta$ alone. For example, if we are
interested in the derivative $m'(x_0)$ at a particular point $x_0$ under a bound
on the second derivative of $m(\cdot)$, we can let
$f(x, \theta)=\theta_1+\theta_2 x$ and define $\mathcal{R}$ to be the class of
functions with $r(x_0)=r'(x_0)=r''(x_0)=0$ and second derivative bounded by
$\Mbound$. Then $m'(x_0)=\theta_2$.

\subsection{Treatment effect extrapolation}\label{sec:treatm-effect-extr}

Often, the average effect of a counterfactual policy on a particular subset of a
population is of interest, and we would like to weaken the assumptions under
which this effect is point-identified. We have available estimates
$\hat{\tau}=(\hat\tau_{1}, \dotsc, \hat{\tau}_{m})'$ of the parameter $\tau$, with
$\hat{g}(\theta)=\hat{\tau}-A\theta$, and $g(\theta)=\tau-A\theta$ for a known
matrix $A$. We would like to extrapolate from these estimates to learn about the
parameter of interest $h(\theta_{0})=H\theta_{0}$. The potential extrapolation
bias is captured by the assumption that $g(\theta_{0})\in\Cglob_{n}$, some
convex set.

Note that, because the moment condition is linear in the parameter of interest,
asymptotic validity of our CIs does not require that the set $\Cglob_{n}$ takes
the form $\Cglob_{n}=\mathcal{C}/\sqrt{n}$. CIs given in \Cref{h_ci_eq} based on
linear estimators of the form $\hat{h}=k'\hat{\tau}$ (such as minimum distance
estimators), with $\mathcal{C}=\sqrt{n}\Cglob_{n}$ are valid under both local
and global misspecification (i.e.\ under the assumption that the set
$\Cglob_{n}$ is fixed as $n\to\infty$).

One example that falls into this setup are differences-in-differences designs
when the parallel trends assumption is violated. Here there are $m$ time
periods, with treatment taking place in period $T_{0}$. The $(m-T_{0})$-vector
$\theta$ corresponds to a vector of dynamic treatment effects on the treated,
$A\theta=(0', \theta')'$, and $g(\theta_{0})$ is a vector of trend differences
between the treated and untreated, with $g(\theta_{0})=0$ if the parallel trends
assumption holds. \citet{RaRo19} build on the framework in this paper to develop
CIs in this setting.

Another example that has been of recent interest involves nonseparable models
with endogeneity. Under conditions in \citet{imbens_identification_1994} and
\citet{heckman_structural_2005}, instrumental variables estimates
$\hat{\tau}_{m}$ with different instruments are consistent for average treatment
effects for different subpopulations. A recent literature
\citep{kowalski_doing_2016,brinch_beyond_2017,mogstad_using_2017} has focused on
using assumptions on treatment effect heterogeneity to extrapolate these
estimates to other populations. Our framework applies if these assumptions
amount to placing the differences between the estimated treatment effects and
the effect of interest in a known convex set.

\section{Empirical application}\label{empirical_sec}

This section illustrates the confidence intervals developed in
\Cref{general_results_sec} in an empirical application to automobile
demand based on the data and model in \citet[BLP hereafter]{blp95}. We use the
version of the model as implemented by \citet{andrews_measuring_2017}, who
calculate the asymptotic bias of the GMM estimator with weighting matrix
$\Sigma^{-1}$ under local misspecification in this setting.\footnote{The
  dataset for this empirical application has been downloaded from the
  \citet{andrews_measuring_2017} replication files, available at
  \url{https://doi.org/10.7910/DVN/LLARSN}.}

\subsection{Model description and implementation}\label{sec:model-description}

In this model, the utility of consumer $i$ from purchasing a vehicle $j$,
relative to the outside option, is given by a random-coefficient logit model
$U_{ij}=\sum_{k=1}^{K}x_{jk}({\beta}_{k}+\sigma_{k}v_{ik})-\alpha
p_{j}/y_{i}+\xi_{j}+\epsilon_{ij}$, where $p_{j}$ is the price of the vehicle,
$x_{jk}$ the $k$th observed product characteristic, $\xi_{j}$ is an unobserved
product characteristic, and $\epsilon_{ij}$ is has an i.i.d.\ extreme value
distribution. The income of consumer $i$ is assumed to be log-normally
distributed, $y_{i}=e^{m+\varsigma v_{i0}}$, where the mean $m$ and the variance
$\varsigma$ of log-income are assumed to be known and set to equal to estimates
from the Current Population Survey. The unobservables
$v_{i}=(v_{i0}, \dotsc, v_{iK})$ are i.i.d.\ standard normal, while the
distribution of the unobserved product characteristic $\xi_{j}$ is unrestricted.

The marginal cost $mc_{j}$ for producing vehicle $j$ is given by
$\log(mc_{j})=w_{j}'\nu+\omega_{j}$, where $w_{j}$ are observable
characteristics, and $\omega_{j}$ is an unobservable characteristic. The full
vector of model parameters is given by
$\theta=(\sigma', \alpha, \beta', \nu')'$. Given this vector, and given a vector
of unobservable characteristics, one can compute the market shares implied by
utility maximization, which can be inverted to yield the unobservable
characteristic as a function of $\theta$, $\xi_{j}(\theta)$. One can similarly
invert the unobserved cost component, writing it as a function of $\theta$,
$\omega_{j}(\theta)$, under the assumption that firms set prices to maximize
profits in a Bertrand-Nash equilibrium. Given a vector $z_{dj}$ of demand-side
instruments, and a vector $z_{sj}$ of supply-side instruments, this
yields the moment condition $g(\theta)=E[\hat{\gamma}(\theta)]$, where
\begin{equation*}
  \hat{g}(\theta)=
  \frac{1}{n}\sum_{j=1}^{n} \begin{pmatrix}
    z_{dj}\xi_{j}(\theta)\\
    z_{sj}\omega_{j}(\theta)
  \end{pmatrix}.
\end{equation*}

The BLP data spans the period 1971 to 1990, and includes information on
essentially all $n=999$ models sold during that period (for simplicity, we have
suppressed the time dimension in the description above). There are 5 observable
characteristics $x_{j}$: a constant, horsepower per 10 pounds of weight (HPWt),
a dummy for whether air-conditioning is standard (Air), mileage per 10 dollars
(MP\$) defined as MPG over average gas price in a given year, and car size
(Size), defined as length times width. The vector $z_{dj}$ consists of $x_{j}$,
plus the sum of $x_{j}$ across models other than $j$ produced by the same firm,
and for rival firms. There are 6 cost variables $w_{j}$: a constant, log of
HPWt, Air, log of MPG, log of Size, and a time trend. The vector
$z_{sj}$ consists of these variables, MP\$, and the sums of $w_{j}$ for own-firm
products other than $j$, and for rival firms. After excluding collinear
instruments, this gives a total of $d_{g}=31$ instruments, 25 of which are
excluded to identify $d_{\theta}=17$ model parameters. The parameter of interest
is average markup,
$h(\theta)=\frac{1}{n}\sum_{j}(p_{j}-mc_{j}(\theta))/p_{j}$.

One may worry that some of these instruments are invalid, because elements of
$z_{dj}$ or $z_{sj}$ may appear directly in the utility or cost function with
the coefficient on the $\ell$th element given by
$\delta_{d\ell}\gamma_{d\ell}/\sqrt{n}$ or
$\delta_{s\ell}\gamma_{s\ell}/\sqrt{n}$, respectively. Here $\delta_{d\ell}$ and
$\delta_{s\ell}$ are scaling constants so that, given the sample size $n=999$ at
hand, $\gamma_{d\ell}$ has the interpretation that the consumer willingness to
pay for one standard deviation change in the $\ell$th demand-side instrument
$z_{dj\ell}$ is $\gamma_{d\ell}\%$ of the average 1980 car price, and changing
the $\ell$th supply-side instrument $z_{sj\ell}$ by one standard deviation
changes the marginal cost by $\gamma_{s\ell}$\% of the average car price.
\citet{andrews_measuring_2017} use this scaling in their sensitivity analysis,
and they discuss economic motivation for concerns about this form of
misspecification. By way of comparison, the estimates of the parameters $\beta$
and $\nu$ in the utility and cost function imply that consumers are on average
willing to pay between 2.2 and 10.0\% of the average car price for a standard
deviation change in one of the included car characteristics, and that a standard
deviation change in the included cost characteristics changes the marginal cost
by between 3.8 and 11.1\% of the average car price. We therefore interpret
specifications of the set $\mathcal{C}$ that allow for
$\abs{\gamma_{s\ell}}\approx1\text{--}2$ (or
$\abs{\gamma_{d\ell}}\approx1\text{--}2$) as allowing for moderate amounts of
misspecification in the $\ell$th supply-side (or demand-side) instrument.

We follow the implementation in \Cref{sec:implementation}. Given a set $I$ of
potentially invalid instruments, we follow~\Cref{remark:choice-of-C} and consider
sets $\mathcal{C}$ of the form~\eqref{eq:C_Bgamma}, with $\norm{\cdot}$
corresponding to an $\ell_{p}$ norm with $p\in\{2,\infty\}$, and
$B=\tilde{B}_{I}\cdot \#I^{1/p}$, where $\tilde{B}_{I}$ is given by the columns of
\begin{equation*}
  \tilde{B}=\bigl(\begin{smallmatrix}
    E[z_{dj}z_{dj}']\operatorname{diag}(\delta_{d\ell})& 0\\
    0 & E[z_{sj}z_{sj}']\operatorname{diag}(\delta_{d\ell})
  \end{smallmatrix}\bigr),
\end{equation*}
and $\#{I}$ is the number of potentially invalid instruments. The scaling by
$(\#{I})^{1/p}$ ensures that the vector $\gamma=\Mbound(1,\dotsc,1)'$ is always
included in the set. \citet{andrews_measuring_2017} report the sensitivity of
the usual GMM estimator under this form of misspecification, considering
misspecification in each instrument individually (so that $I$ contains a single
element), and setting $M=1$. However, if one is concerned about the validity of
several instruments, it is natural to allow $I$ to contain all instruments the
validity of which is questionable. In our analysis, we vary the set of
potentially misspecified instruments. We also vary $M$ in order to assess the
sensitivity of conclusions to different amounts of misspecification. As we will
see below, different choices of $\mathcal{C}$ lead to different sensitivities
for the optimal estimator, and using the optimal sensitivity can reduce the
width of the CI substantially relative to CIs based on the usual GMM estimator.

We use the estimate $\hat{\theta}_{\text{initial}}$ that corresponds to the GMM
estimator based on the weight matrix that's optimal under correct specification,
as reported in \citet{andrews_measuring_2017}, and the estimates $\hat{\Gamma}$,
$\hat{H}$, and $\hat{\Sigma}$ are computed following Step~\ref{item:step1} of
the implementation.

\subsection{Results}

To illustrate that using the sensitivity that is optimal under local
misspecification can yield substantially tighter CIs, \Cref{fig:l2-blp}
plots the confidence intervals based on the optimal sensitivity, as well as
those based on $\hat{\theta}_{\text{initial}}$ under different sets $I$ of
potentially invalid instruments and $\ell_{2}$ constraints on $\gamma$. It is
clear from the figure that using the optimal sensitivity yields substantially
tighter confidence intervals, relative to simply adjusting the usual CI by using
the critical value $\cv_{\alpha}(\cdot)$ to take into account the potential bias
of $h(\hat{\theta}_{\text{initial}})$, by as much as a factor of 3.4.
The intuitive reason for this is that by adjusting the sensitivity of the
estimator, it is possible to substantially reduce its bias at little cost in
terms of an increase in variance. Thus, for example, while the CI for the
average markup based on the estimate $\hat{\theta}_{\text{initial}}$ is
essentially too wide to be informative when the set of potentially invalid
instruments corresponds to all excluded instruments, the CI based on the optimal
sensitivity, $[46.0, 66.0]\%$,
is still quite tight.

If a researcher is ex ante unsure what form of misspecification one should worry
about, as a sensitivity check, it is useful to consider the effects of different
forms of misspecification. In \Cref{fig:l2_lI-blp}, we plot the optimal
confidence intervals for different subsets of invalid instruments, under both
$\ell_{2}$ and $\ell_{\infty}$ norms for $\gamma$. Although the choice of norm
matters when the number of potentially misspecified instruments is greater than
one, the results are qualitatively similar. Comparing the results for different
choices of the set of potentially invalid instruments suggests that allowing
supply-side instruments to be invalid generally increases the average markup
estimate, while allowing demand-side instruments to be invalid has the opposite
effect.

As it may be ex ante unclear what magnitude of misspecification is reasonable to
allow for, as discussed in \Cref{sec:implementation}, it is useful to plot the
optimal CI for multiple choices of $\Mbound$. We do this in \Cref{fig:l2_M-blp}
for $p=2$, and we allow all excluded instruments to be potentially invalid. One
can see that while the CI is unstable for values of $\Mbound$ smaller than about
$0.4$, for larger values of $\Mbound$, the estimate is quite stable and equal to
about 50\%. Even at $\Mbound=2$, one rejects the hypothesis that the optimal
markup is equal to the initial estimate
$h(\hat{\theta}_{\text{initial}})=32.7\%$. This suggests that ignoring
misspecification in the BLP model likely leads to a downward bias in the
estimate of the average markup. At the same time, it is possible to obtain
reasonably tight CIs for the average markup even under a moderate amount of
misspecification.

The $J$-statistic for testing the hypothesis that all moments are correctly
specified equals 426.7. Consequently, the hypothesis is rejected at the usual
significance levels. Furthermore, it can be seen from \Cref{fig:l2_lI-blp} that
the CIs for ``all excluded'' (that allow all excluded instruments to be invalid
at $\Mbound=1$), and ``all excluded demand'' (that assume validity of
supply-side instruments) do not overlap. This implies that either the
misspecification in the demand-side instruments must be greater than 1\% of the
average care price ($\Mbound=1$), or else the supply-side instruments must also
be invalid. \Cref{tab:Jtest} implements the specification test from
\Cref{sec:specification-test} that gives lower CI $[M_{\min}, \infty]$ for $M$.
The results suggest that if one assumes only a subset of the instruments is
invalid, the misspecification in the potentially invalid instruments must be
quite large. For example, if we assume that all instruments are valid except
potentially the demand-side instruments based on rival firms' product
characteristics, then the misspecification in these instruments must be greater
than $\Mbound=5.36$. If we allow all instruments to be invalid, then
$\Mbound\geq 1.13$.

Finally, to illustrate the implication of
\Cref{main_text_efficiency_bound_thm} that one cannot substantively
improve upon the CIs that we construct, we calculate the efficiency bound
$\kappa_{*}$ for these CIs in \Cref{tab:efficiency}. The table shows that
the bound is at least as high as the efficiency bound for the usual CI under
correct specification (given in~\eqref{eq:kappa-linear-subspace} and equal to
84.99\% at $\alpha=0.05$). Thus, the asymptotic scope for improvement over the
CIs reported in \Cref{fig:l2_lI-blp} at particular values of $\theta$ and
$c=0$ is even smaller than the scope for improvement over the usual CI at
particular values of $\theta$ under correct specification.

\appendix

\section{Details of calculations}\label{sec:details-calculations}

This appendix contains the details of calculations for particular sets
$\mathcal{C}$ discussed in \Cref{sec:corr-spec-cress} and
\Cref{remark:choice-of-C}.

\subsection{Cressie-Read divergences}\label{sec:cress-read-diverg}

Consider the problem~\eqref{half_modulus_eq} under constraints of the
form $\{c\colon c\Sigma^{-1}c\leq \Mbound^{2}\}$. The Lagrangian for this problem can
be written as
\begin{equation*}
2H\theta + \lambda_{1}(\delta^{2}/4-
 (c-\Gamma \theta)'\Sigma^{-1}(c-\Gamma \theta))
+\lambda_{2}(\Mbound^{2}-c'\Sigma^{-1}c).
\end{equation*}
(we multiply the objective function by $2$ so that its optimized value equals
$\omega(\delta)$). The first-order conditions imply that at optimum,
$c=\frac{\lambda_{1}}{\lambda_{1}+\lambda_{2}}\Gamma\theta$, and
$\theta=\frac{\lambda_{1}+\lambda_{2}}{\lambda_{1}\lambda_{2}}
(\Gamma'\Sigma^{-1}\Gamma)^{-1}H'$. Plugging these expressions into the
constraints yields
$\Mbound^{2}=H(\Gamma'\Sigma^{-1}\Gamma)^{-1}H'/\lambda_{2}^{2}$ and
$\delta^{2}/4=H(\Gamma'\Sigma^{-1}\Gamma)^{-1}H'/\lambda_{1}^{2}$. Since
$H(\Gamma'\Sigma^{-1}\Gamma)^{-1}H'=k_{LS,0}'\Sigma k_{LS,0}$, solving for
$\lambda_{1}$ and $\lambda_{2}$, and plugging into the expression for $\theta$
yields
\begin{equation*}
  \theta=\frac{\delta/2+\Mbound}{\sqrt{k_{LS,0}'\Sigma k_{LS,0}}}\cdot
(\Gamma'\Sigma^{-1}\Gamma)^{-1}H'.
\end{equation*}
Thus,
$\omega(\delta)=2H\theta= (\delta+2\Mbound)\sqrt{k_{LS,0}'\Sigma k_{LS,0}}$.
With this form of $\omega$, the bound in~\eqref{kappa_conv_cs_eq} becomes
\begin{equation*}
  \kappa_{*}(H, \Gamma, \Sigma, \mathcal{C})=\frac{(1-\alpha)(z_{1-\alpha}+\Mbound)
    +\phi(z_{1-\alpha})}{\cv_{\alpha}(\Mbound)}.
\end{equation*}
This efficiency equals at least $\min\{\kappa^{L}_{*, \alpha}, 1-\alpha\}$, where
$\kappa^{L}_{*, \alpha}=((1-\alpha)z_{1-\alpha}+\phi(z_{1-\alpha}))/z_{1-\alpha/2}$
denotes the efficiency given in~\eqref{eq:kappa-linear-subspace} when
$\mathcal{C}$ is a linear subspace. To show this, observe that
$\cv_{\alpha}'(\Mbound)\leq 1$ for all $\Mbound\geq 0$. Therefore, the derivative of
\begin{equation*}
  (1-\alpha)(z_{1-\alpha}+\Mbound)+\phi(z_{1-\alpha})-\min\{1-\alpha, \kappa^{L}_{*, \alpha}\}\cv_{\alpha}(\Mbound)
\end{equation*}
with respect to $\Mbound$, given by
$1-\alpha-\min\{1-\alpha, \kappa^{L}_{*, \alpha}\}\cv_{\alpha}'(\Mbound)$, is
always non-negative. Since the expression in the above display equals
$(\kappa^{L}_{*, \alpha}-\min\{\kappa^{L}_{*,
  \alpha},1-\alpha\})z_{1-\alpha/2}\geq 0$ at $\Mbound=0$, it follows that it is
always non-negative. Rearranging it then yields
$\kappa_{*}(H, \Gamma, \Sigma, \mathcal{C})\geq \min\{\kappa^{L}_{*,
  \alpha},1-\alpha\}$. Furthermore, it follows from
\Cref{eq:one-sided-efficiency} below that the efficiency of one-sided CIs at
$c=0$ is given by $\kappa_{*}^{\operatorname{OCI}, \beta}=1$.

\subsection{\texorpdfstring{$\ell_{p}$}{lp} Bounds}\label{sec:lp-bounds-appendix}

We now consider the form of the optimal sensitivity when the set
$\mathcal{C}=\mathcal{C}(\Mbound)$ takes the form in \Cref{eq:C_Bgamma}, and
$\norm{\cdot}$ corresponds to an $\ell_{p}$ norm, as discussed in
\Cref{remark:choice-of-C}. First, we explain the connection with penalized
estimation. Since $c=B\gamma$, one can write the approximately linear
model~\eqref{limit_experiment_eq} as
\begin{equation*}
  Y=-\Gamma \theta+B\gamma+\Sigma^{1/2}\varepsilon,
\end{equation*}
which one can think of as a regression model with correlated errors, design
matrix $(-\Gamma, B)$, and coefficient vector $(\theta', \gamma')'$. With this
interpretation, it is clear that if the number of regressors
$d_\theta+d_{\gamma}$ is greater than the number of observations $d_g$, the
constraint on the norm of $\gamma$ is necessary to make the model informative.
Using the observation from \Cref{remark:choice-of-C} that
$\maxbias_{\mathcal{C}(1)}(k)=\norm{B'k}_{p'}$, it follows that the optimization
problem~\eqref{eq:bias-variance-optimization} under $\mathcal{C}=\mathcal{C}(1)$
is equivalent to
\begin{equation}\label{eq:bias-variance-lp}
  \min_{k} k'\Sigma k\qquad \text{s.t.}\quad H=-k'\Gamma \qquad\text{and}\quad
  \norm{B'k}_{p'} \leq \overline{B}.
\end{equation}
We now specialize the results to the cases $p=2$ and $p=\infty$. We discuss the
case $p=1$ in a working paper version of this paper \citep{ArKo20apr_sensitivity}.

\subsubsection{\texorpdfstring{$p=2$}{p=2}}

In this case, the Lagrangian form of~\eqref{eq:bias-variance-lp} becomes
\begin{equation*}
  \min_{k} k'(\Sigma +\lambda BB')k\qquad \text{s.t.}\quad H=-k'\Gamma,
\end{equation*}
with the Lagrange multiplier $\lambda$ giving the relative weight on bias.
Optimizing this objective is isomorphic to deriving the minimum variance
unbiased estimator of $H\theta$ in a regression model with design matrix
$-\Gamma$ and variance $\Sigma+\lambda BB'$, so the Gauss-Markov theorem implies
that the optimal sensitivities are
$k_{\lambda}'=-H(\Gamma'W_{\lambda}\Gamma)^{-1}\Gamma'W_{\lambda}$ where
$W_{\lambda}=[\Sigma+\lambda BB']^{-1}$.

\subsubsection{\texorpdfstring{$p=\infty$}{p=infinity}}
Write the Lagrangian form of~\eqref{eq:bias-variance-lp} as
\begin{equation}\label{eq:pinf-lagrangian}
  \min_{k} k'\Sigma k/2+\lambda \norm{B'k}_{1}\qquad \text{s.t.}\quad H=-k'\Gamma.
\end{equation}
It will be convenient to transform the problem so that the $\ell_{1}$ constraint
only involves $d_{\gamma}$ elements of $k$. Let
\begin{equation}\label{eq:T-definition}
  T=\begin{pmatrix}
    B_{\perp}'\\
    (B'B)^{-1}B'\end{pmatrix}, \qquad T^{-1}=
  \begin{pmatrix}
    B_{\perp}&B
  \end{pmatrix},
\end{equation}
where $B_{\perp}$ is an orthonormal matrix that's orthogonal to $B$. Then, since
$TB=(0,I_{d_{\gamma}})'$, the above minimization problem is equivalent to the
problem
\begin{equation*}
  \min_{\tk}\tk'\tSigma\tk/2+\tlambda\sum_{i\in
    I}\abs{\tk_{i}}
  \qquad\text{s.t.}\quad H'=-\tGamma'\tk,
\end{equation*}
where $\tk={T'}^{-1}k$, $\tSigma=T\Sigma T'$, $\tGamma=T\Gamma$, and
$I=\{d_{g}-d_{\gamma}, \dotsc, d_{g}\}$ indexes the last $d_{\gamma}$ elements
of $\tk$.

To minimize the above display and give the solution path as $\tlambda$ varies,
we use arguments similar to those in Theorem 2 of \citet{rosset_piecewise_2007}.
For $i\in I$, write $\tk_{i}=\tk_{+, i}-\tk_{-, i}$, where
$\tk_{+, i}=\max\{\tk_{i},0\}$ and $\tk_{-, i}=-\min\{\tk_{i},0\}$. We minimize
the objective function in the preceding display over
$\{\tk_{+, i}, \tk_{-, i}, \tk_{j}\colon i\in I, j\not\in I\}$ subject to the
constraints $\tk_{+, i}\geq 0$ and $\tk_{-, i}\geq 0$. Let $\mu$ denote a vector
of Lagrange multipliers on the restriction $-H'=\tGamma' \tk$. Then the
Lagrangian can be written as
\begin{equation*}
  \tk'\tSigma
  \tk/2+\tlambda\sum_{i\in I}(\tk_{+, i}+\tk_{-, i})+\mu'(H'+\tGamma' \tk)-\sum_{i\in
    I}(\tlambda_{+, i}\tk_{+, i} +\tlambda_{-, i}\tk_{-, i}).
\end{equation*}
The first-order conditions are given by
\begin{align}
  e_{i}'\tSigma \tk+e_{i}'\tGamma \mu&=0 & i&\in I^{C},\label{eq:foc-iv}\\
  e_{i}'\tSigma \tk+e_{i}'\tGamma \mu+\tlambda&=\tlambda_{+, i}&i&\in I,\label{eq:foc-i+}\\
  -(e_{i}'\tSigma \tk+e_{i}'\tGamma\mu)+\tlambda&=\tlambda_{-, i}&i&\in I.\label{eq:foc-i-}
\end{align}
The complementary slackness conditions are given by $\tlambda_{+, i}\tk_{+, i}=0$
and $\tlambda_{-, i}\tk_{-, i}=0$ for $i\in I$, and the feasibility constraints
are $\tlambda_{+, i}\geq 0$, $\tlambda_{-, i}\geq 0$ for $i\in I$ and
$-H'=\tGamma' \tk$.

Let $\mathcal{A}^{C}=\{i\colon i\in I, \tk_{i}=0\}$, and let
$\mathcal{A}=\{i\colon i\not\in\mathcal{A}^{C}\}$ denote the set of active
constraints. Let $s$ denote a vector of length $\abs{\mathcal{A}}$ with elements
$s_{i}=\sign(\tk_{i})$ if $i\in I$ and $s_{i}=0$ otherwise.

The slackness and feasibility conditions imply that if for $i\in I$,
$\tk_{i}> 0$, then $\tlambda_{+, i}=0$, and if $\tk_{i}<0$ or $\tlambda_{-, i}=0$.
It therefore follows from~\eqref{eq:foc-i+} and~\eqref{eq:foc-i-} that
$e_{i}'\tSigma \tk+e_{i}'\tGamma \mu=-\sign(\tk_{i})\tlambda=-s_{i}\tlambda$. We
can combine this condition with~\eqref{eq:foc-iv} and write
\begin{equation}\label{eq:foc-A}
  e_{i}'\tSigma \tk+e_{i}'\tGamma\mu=-s_{i}\tlambda, \qquad i\in\mathcal{A}.
\end{equation}
On the other hand, if $i\in\mathcal{A}^{C}$, then since $\tlambda_{+, i}$ and
$\tlambda_{-, i}$ are non-negative, it follows from~\eqref{eq:foc-i+}
and~\eqref{eq:foc-i-} that
\begin{equation}\label{eq:AC-constraint}
  \abs{e_{i}'\tSigma \tk+e_{i}'\tGamma \mu}\leq \tlambda=\abs{e_{j}'\tSigma \tk+e_{j}'\tGamma\mu},\qquad
  i\in\mathcal{A}^{C}, j\in\mathcal{A}.
\end{equation}
Let $\tk_{\mathcal{A}}$ denote the subset of $\tk$ corresponding to the active
moments, $\tGamma_{\mathcal{A}}$ denote the corresponding rows of $\tGamma$, and
$\tSigma_{\mathcal{A}\mathcal{A}}$ the corresponding submatrix of $\tSigma$.
Then we can write the condition~\eqref{eq:foc-A} together with the feasibility
constraint $\tGamma' \tk=-H'$ compactly as
\begin{equation*}
  \begin{pmatrix}
    0 & \tGamma_{\mathcal{A}}'\\
    \tGamma_{\mathcal{A}}&\tSigma_{\mathcal{A}\mathcal{A}}
  \end{pmatrix}
  \begin{pmatrix}
    \mu\\\tk_{\mathcal{A}}
  \end{pmatrix}=
  \begin{pmatrix}
    -H' \\
    -s\tlambda
  \end{pmatrix}.
\end{equation*}
Using the block matrix inverse formula, this implies
\begin{align*}
  \mu&= (\tGamma_{\mathcal{A}}'\tSigma_{\mathcal{A}\mathcal{A}}^{-1}\tGamma_{\mathcal{A}})^{-1}\left(
       H'-\tGamma_{\mathcal{A}}'\tSigma_{\mathcal{A}\mathcal{A}}^{-1}s\tlambda\right)\\
  \tk_{\mathcal{A}}&=-\tSigma_{\mathcal{A}\mathcal{A}}^{-1}\tGamma_{\mathcal{A}}\mu
                     -\tSigma_{\mathcal{A}\mathcal{A}}^{-1}s\tlambda\\
     &=\tSigma_{\mathcal{A}\mathcal{A}}^{-1}\tGamma_{\mathcal{A}}
       (\tGamma_{\mathcal{A}}'\tSigma_{\mathcal{A}\mathcal{A}}^{-1}\tGamma_{\mathcal{A}})^{-1}\left(
       \tGamma_{\mathcal{A}}'\tSigma_{\mathcal{A}\mathcal{A}}^{-1}s\tlambda -H'\right)
       -\tSigma_{\mathcal{A}\mathcal{A}}^{-1}s\tlambda
\end{align*}
Consequently, if we're in a region in where the solution path is differentiable
with respect to $\tlambda$, we have
\begin{equation}\label{eq:k-direction}
  \frac{\partial \tk_{\mathcal{A}}}{\partial \tlambda}=
  \tSigma_{\mathcal{A}\mathcal{A}}^{-1}\tGamma_{\mathcal{A}}
  (\tGamma_{\mathcal{A}}'\tSigma_{\mathcal{A}\mathcal{A}}^{-1}\tGamma_{\mathcal{A}})^{-1}
  \tGamma_{\mathcal{A}}'\tSigma_{\mathcal{A}\mathcal{A}}^{-1}s
  -\tSigma_{\mathcal{A}\mathcal{A}}^{-1}s.
\end{equation}
The differentiability of path is violated if either (a) the
constraint~\eqref{eq:AC-constraint} is violated for some $i\in\mathcal{A}^{C}$
if $\tk(\tlambda)$ keeps moving in the same direction, and we add $i$ to
$\mathcal{A}$ at a point at which~\eqref{eq:AC-constraint} holds with equality;
or else (b) the sensitivity $\tk_{i}(\tlambda)$ for some $i\in \mathcal{A}$
reaches zero. In this case, drop $i$ from $\mathcal{A}$. In either case, we need
to re-calculate the direction~\eqref{eq:k-direction} using the new definition of
$\mathcal{A}$.

Based on the arguments above and the fact that
$\tk(0)= -\tSigma^{-1}\tGamma (\tGamma'\tSigma^{-1}\tGamma)^{-1}H'$, we can
derive the following algorithm, similar to the LAR-LASSO algorithm, to generate
the path of optimal sensitivities $\tk(\tlambda)$:
\begin{enumerate}
\item Initialize $\tlambda=0$, $\mathcal{A}=\{1,\dotsc, d_{g}\}$,
  $\mu=(\tGamma'\tSigma^{-1}\tGamma)^{-1}H'$, $\tk=-\tSigma^{-1}\tGamma \mu$.
  Let $s$ be a vector of length $d_{g}$ with elements
  $s_{i}=\1{i\in I}\sign(\tk_{i})$, and calculate initial directions as
  $\mu_{\Delta}=-(\tGamma'\tSigma^{-1}\tGamma)^{-1}\tGamma'\tSigma^{-1}s$,
  $\tk_{\Delta}=-\tSigma^{-1}(\tGamma\mu_{\Delta} +s)$
\item While $(\abs{\mathcal{A}}>\max\{d_{g}-d_{\gamma}, d_{\theta}\})$:
  \begin{enumerate}
  \item Set step size to $d=\min\{d_{1}, d_{2}\}$, where
    \begin{align*}
      d_{1}&=\min\{d> 0\colon \tk_{i}+d \tk_{\Delta, i}=0,i\in\mathcal{A}\cap \mathcal{I}\}\\
      d_{2}&=\min\{d> 0\colon
             \abs{e_{i}'(\tSigma \tk+\tGamma\mu)+d e_{i}'(\tSigma \tk_{\Delta}+\tGamma \mu_{\Delta})}=\tlambda+d, i\in\mathcal{A}^{C}\}
    \end{align*}
    Take step of size $d$: $\tk\mapsto \tk+d\tk_{\Delta}$,
    $\mu\mapsto \mu+d\mu_{\Delta}$, and $\tlambda\mapsto \tlambda+d$.
  \item If $d=d_{1}$, drop $\argmin(d_{1})$ from $\mathcal{A}$, and if
    $d=d_{2}$, then add $\argmin(d_{2})$ to $\mathcal{A}$. Let $s$ be a
    vector of length $d_{g}$ with elements
    $s_{i}=-\1{i\in I}\sign(e_{i}'\tSigma \tk+e_{i}'\tGamma \mu)$, and calculate
    new directions as
    \begin{align*}
      \mu_{\Delta}&= -(\tGamma_{\mathcal{A}}'\tSigma_{\mathcal{A}\mathcal{A}}^{-1}\tGamma_{\mathcal{A}})^{-1}
                    \tGamma_{\mathcal{A}}'\tSigma_{\mathcal{A}\mathcal{A}}^{-1}s_{\mathcal{A}}\\
      (\tk_{\Delta})_{\mathcal{A}}&=-\tSigma_{\mathcal{A}\mathcal{A}}^{-1}(\tGamma_{\mathcal{A}}\mu_{\Delta}
                                    +s_{\mathcal{A}})\\
      (\tk_{\Delta})_{\mathcal{A}^{C}}&=0
    \end{align*}
  \end{enumerate}
\end{enumerate}
The solution path $k(\lambda)$ is then obtained as $k(\lambda)=T'\tk(\lambda)$.

Finally, we show that in the limit $\Mbound\to\infty$, the optimal sensitivity
corresponds to a method of moments estimator based on the most informative set
of $d_{\theta}$ moments, with the remaining $d_{g}-d_{\theta}$ moments dropped.
The optimal sensitivity as $M\to\infty$ obtains by
solving~\eqref{eq:pinf-lagrangian} as $\lambda\to\infty$. If $B$ corresponds to
columns of the identity matrix, then this is equivalent to minimizing
$\norm{k_{I}}_{1}$ subject to $H=-k'\Gamma$. This can be written as a linear
program $\min k_{I, +}+k_{I, -i}$ st $-H'=\Gamma'(k_{+}-k_{-})$,
$k_{+}, k_{-}\geq 0$. The minimization problem is done on a
$d_{\theta}$-dimensional hyperplane, and solution must occur at a boundary point
of the feasible set, where only $d_{\theta}$ variables are non-zero. So the
optimal $k$ has $d_{\theta}$ non-zero elements.

\section{Specification test}\label{sec:specification-test}

One can test the null hypothesis of correct specification (i.e.\ the null
hypothesis that $c=0$) using the $J$ statistic
\begin{equation*}
  J=n\min_{\theta}\hat{g}({\theta})'\hat\Sigma^{-1}\hat{g}({\theta})=n\hat{g}(\hat{\theta})'\hat\Sigma^{-1}\hat{g}(\hat{\theta}),
\end{equation*}
where
$\hat{\theta}=\operatorname{argmin}_{\theta}\hat{g}({\theta})'\hat\Sigma^{-1}\hat{g}({\theta})$.
Alternatively, letting $\hat\Sigma^{-1/2}$ denote the symmetric square root of
$\hat\Sigma^{-1}$, one can project $\hat\Sigma^{-1/2}\hat{g}(\tilde{\theta})$, where
$\tilde{\theta}$ is some consistent estimate, onto the complement of the space
spanned by $\hat\Sigma^{-1/2}\hat\Gamma$,
\begin{equation*}
  S=n\hat{g}(\tilde{\theta})'\hat\Sigma^{-1/2}\hat R\hat\Sigma^{-1/2}\hat{g}(\tilde{\theta}),
\end{equation*}
where
$\hat R=I-\hat\Sigma^{-1/2}\hat\Gamma(\hat\Gamma'\hat\Sigma^{-1}\hat\Gamma)^{-1}\hat\Gamma'\hat\Sigma^{-1/2}$.
If the model is correctly specified, so that $c=0$, $S$ and $J$ are asymptotically
equivalent \citep[p.~2231]{newey_large_1994}, and distributed
$\chi^{2}_{d_{g}-d_{\theta}}$.

Under local misspecification, the $J$ statistic has a noncentral $\chi^2$
distribution, with noncentrality parameter depending on $c$
\citep{newey_generalized_1985}, and the asymptotic equivalence of $J$ and $S$
still holds. In this section, we use this observation to form a test of the null
hypothesis $H_0\colon c\in\mathcal{C}$. When $\mathcal{C}$ takes the form in
\Cref{eq:C_Bgamma} for some norm $\norm{\cdot}$, inverting these tests gives a
lower CI for $\Mbound$. We begin with a lemma deriving the asymptotic
distribution of $S$ and $J$ under local misspecification.

\begin{lemma}\label{j_test_lemma}
  Suppose that \Cref{eq:local-misspecification,g_clt_assump,,linear_approx_assump} hold, and that $\hat{\theta}$ and
  $\tilde{\theta}$ satisfy, for some $K$ and
  $K_{opt}'=-(\Gamma'\Sigma^{-1}\Gamma)^{-1}\Gamma'\Sigma^{-1}$,
  \begin{align*}
  \sqrt{n}(\hat{\theta}-\theta_0)&=K_{opt}'\sqrt{n}\hat{g}(\theta_0),&&\text{and}&
  \sqrt{n}(\tilde{\theta}-\theta_0)&=K'\sqrt{n}\hat{g}(\theta_0).
  \end{align*}
  Suppose that $\hat\Sigma$ and $\hat\Gamma$ are consistent estimates of
  $\Sigma$ and $\Gamma$, and that $\Sigma$ and $\Gamma$ are full rank. Then
  $S=J+o_P(1)$ and $S$ and $J$ converge in distribution to a noncentral
  chi-square distribution with $d_g-d_\theta$ degrees of freedom and
  noncentrality parameter $c'\Sigma^{-1/2}R\Sigma^{-1/2}c$ where
  $R=I-\Sigma^{-1/2}\Gamma(\Gamma'\Sigma^{-1}\Gamma)^{-1}\Gamma\Sigma^{-1/2}$.
\end{lemma}
\begin{proof}
  By \Cref{eq:local-misspecification,,g_clt_assump,linear_approx_assump},
  $\sqrt{n}\hat{g}(\tilde{\theta})=(I+\Gamma
  K')\Sigma^{1/2}(\Sigma^{-1/2}c+Z_n)+o_{P}(1)$ where
  $Z_n=\Sigma^{-1/2}[\sqrt{n}
  \hat{g}(\theta_0)-c]\stackrel{d}{\to}\mathcal{N}(0,I_{d_g})$, so that
  \begin{multline*}
    S=
    (\Sigma^{-1/2}c+Z_n)'\Sigma^{1/2}(\Sigma^{-1/2}+\Sigma^{-1/2}\Gamma K')'
    R(\Sigma^{-1/2}+\Sigma^{-1/2}\Gamma K')\Sigma^{1/2}(\Sigma^{-1/2}c+Z_n)+o_{P}(1)\\
    =(\Sigma^{-1/2}c+Z_n)'R(\Sigma^{-1/2}c+Z_n)+o_P(1)
    \stackrel{d}{\to} (\Sigma^{-1/2}c+Z)'R(\Sigma^{-1/2}c+Z)
  \end{multline*}
  where $Z\sim\mathcal{N}(0,I_{d_g})$ and we use the fact that
$R(I+\Sigma^{-1/2}\Gamma K'\Sigma^{1/2})=R$.
  Similarly,
  \begin{equation*}
      \sqrt{n}\hat{g}(\hat{\theta})=(I-\Gamma(\Gamma'\Sigma^{-1}\Gamma)\Gamma'\Sigma^{-1})(c+\Sigma^{1/2}Z_n)+o_{P}(1)
      = \Sigma^{1/2}R(\Sigma^{-1/2}c+Z_n)+o_{P}(1),
  \end{equation*}
  so that $J= (\Sigma^{-1/2}c+Z_n)'R(\Sigma^{-1/2}c+Z_n)+o_{P}(1)=S+o_{P}(1)$. To
  prove the second claim, decompose $R=P_{1}P_{1}'$, where
  $P_{1}\in \mathbb{R}^{d_{\theta}\times (d_{g}-d_{\theta})}$ corresponds to the
  eigenvectors associated with non-zero eigenvalues of $R$. Then
  \begin{equation*}
    (\Sigma^{-1/2}c+Z)'R(\Sigma^{-1/2}c+Z)=(P_{1}'\Sigma^{-1/2}c+P_{1}'Z)'(P_{1}'\Sigma^{-1/2}c+P_{1}'Z).
  \end{equation*}
  Since $P_{1}'Z\sim\mathcal{N}(0,I_{d_{g}-d_{\theta}})$, it follows that the
  random variable in the preceding display has a non-central $\chi^{2}$
  distribution with $d_{g}-d_{\theta}$ degrees of freedom and non-centrality
  parameter $c'\Sigma^{-1/2}R\Sigma^{-1/2}c$.
\end{proof}

\Cref{j_test_lemma} can be interpreted in using the limiting experiment
described in \Cref{limiting_experiment_sec}. In particular, the asymptotic
distribution of the $S$ and $J$ statistics is given by the distribution of the
statistic $Y'\Sigma^{-1/2}R\Sigma^{-1/2}Y$ in the limiting experiment.

The quantiles of a non-central chi-square distribution are increasing in the
noncentrality parameter \citep{sun_monotonicity_2010}. Thus, to test the null
hypothesis $H_{0}\colon c\in\mathcal{C}$, the appropriate critical value for
tests based on the $J$ or $S$ statistic is based on a non-central chi-squared
distribution, with non-centrality parameter
\begin{equation*}
  \bar{\lambda}=\sup_{c\in\mathcal{C}}c'\Sigma^{-1/2}R\Sigma^{-1/2}c.
\end{equation*}
If $\mathcal{C}=\{B\gamma\colon \norm{\gamma}_{p}\leq \Mbound\}$, then this becomes
\begin{equation*}
  \bar{\lambda}=\sup_{\norm{t}_{p}\leq \Mbound}t'B'\Sigma^{-1/2}R\Sigma^{-1/2}Bt
  =\sup_{\norm{t}_{p}\leq 1}\Mbound^{2}\norm{R\Sigma^{-1/2}Bt}_{2}^{2}
  =\Mbound^{2}\norm{A}^{2}_{p,2}.
\end{equation*}
where the second equality uses the fact that $R$ is idempotent,
$A=R\Sigma^{-1/2}B$, and
$\norm{A}_{p, q}=\max_{\norm{x}_{p}\leq 1}\norm{Ax}_{q}$ is the $(p, q)$
operator norm. For $p=2$, the operator norm has a closed form, which gives
$\bar{\lambda}=\Mbound\max\operatorname{eig}(B'\Sigma^{-1/2}R\Sigma^{-1/2}B)$.

\section{Asymptotic coverage and efficiency}\label{efficiency_sec_append}

This appendix contains the asymptotic coverage and efficiency results discussed
in \Cref{efficiency_sec}. In particular, we prove
\Cref{main_text_efficiency_bound_thm}. In order to allow for stronger
statements, we state upper and lower bounds separately.
\Cref{main_text_efficiency_bound_thm} then follows by combining these
results. \Cref{main_text_efficiency_bound_thm} focuses on two-sided CIs
in the case where $\mathcal{C}$ is centrosymmetric, in addition to being convex.
In this appendix, we also prove analogous results for one-sided CIs, and we
generalize these results to the case where $\mathcal{C}$ is a convex but
asymmetric set. When $\mathcal{C}$ is convex but asymmetric, the negative
results about the scope for improvement when $c$ is close to zero no longer
hold. Therefore, we consider the general problem of optimizing quantiles of
excess length over a set $\mathcal{D}\subseteq\mathcal{C}$, which may be a
strict subset of $\mathcal{C}$.

The remainder of this appendix is organized as follows. \Cref{setup_sec}
presents notation and definitions, as well as an overview of the results.
\Cref{submodel_sec} contains results on least favorable submodels as well
as a two-point testing lemma used in later proofs. We then use this to obtain
efficiency bounds for one-sided CIs in
\Cref{oneside_efficiency_bound_sec}, and for two-sided CIs in
\Cref{twoside_efficiency_bound_sec}. \Cref{achieving_bound_sec}
shows that our CIs achieve (or, for two-sided CIs, nearly achieve) these bounds.
\Cref{adaptation_sec_append} shows how
\Cref{main_text_efficiency_bound_thm} follows from these results, and
also gives a one-sided version of this theorem. Primitive conditions for the
misspecified linear IV model, as well as a general construction of a least
favorable submodel satisfying the assumptions used in this section, are given in
the supplemental appendix.

\subsection{Setup}\label{setup_sec}

While our focus is on parameter spaces that place restrictions on $c$, we will
also allow for local restrictions on $\theta$ in some results. This
allows us to bound the scope for ``directing power'' at particular values of
$\theta$. Formally, for some parameter $\theta^*$, we consider the local
parameter space that restricts $(\sqrt{n}(\theta-\theta^*)', c')'$ to some set
$\mathcal{F}\subseteq \mathbb{R}^{d_\theta+d_g}$. The unrestricted case
considered throughout most of the main text corresponds to
$\mathcal{F}=\mathbb{R}^{d_\theta}\times \mathcal{C}$ (in which case $\theta^*$
does not affect the definition of the parameter space). We also allow for
additional restrictions on $\theta$ by placing it in some set $\Theta_n$.
Finally, we use $\mathcal{P}$ to denote the set of distributions $P$ over which
we require coverage.

With this notation, the set of values of $\theta$ that are consistent with the
model under $P$ (i.e.\ the identified set under $P$) is
\begin{equation*}
  \Theta_I(P)=\Theta_I(P;\mathcal{F}, \Theta_n)
  =\left\{\theta\in\Theta_n: \sqrt{n}((\theta-\theta^*)', g_P(\theta)')'\in\mathcal{F}\right\},
\end{equation*}
and the set of pairs $(\theta, P)$ over which coverage is required is given by
\begin{equation*}
\mathcal{S}_n=\{(\theta, P)\in\Theta_n\times \mathcal{P}: \theta\in\Theta_I(P)\}
=\{(\theta, P)\in\Theta_n\times \mathcal{P}:
\sqrt{n}((\theta-\theta^*)', g_P(\theta)')'\in\mathcal{F}\},
\end{equation*}
which reduces to the definition in \Cref{main_text_efficiency_bound_thm}
when $\mathcal{F}=\mathbb{R}^{d_\theta}\times \mathcal{C}$. The coverage
requirement for a CI $\mathcal{I}_n$ is then given by~\eqref{coverage_eq} with
this definition of $\mathcal{S}_n$. To compare one-sided CIs
$\hor{\hat{c}, \infty}$, we will consider the $\beta$ quantile of excess length.
Rather than restricting ourselves to the minimax criterion, we consider
worst-case excess length over a potentially smaller parameter space
$\mathcal{G}$, which may place additional restrictions on $\theta$ and $c$. Let
\begin{equation*}
  q_{\beta, n}(\hat c;\mathcal{P}, \mathcal{G}, \Theta_n)
  =\sup_{P\in\mathcal{P}}\sup_{\theta\in\Theta_I(P;\mathcal{G}, \Theta_n)} q_{P, \beta}(h(\theta)-\hat c)
\end{equation*}
where $q_{P, \beta}$ denotes the $\beta$ quantile under $P$. We will also
consider bounds on $q_{P, \beta}(h(\theta)-\hat c)$ at a single $P$, which
corresponds to the optimistic case of optimizing length at a single
distribution. For two-sided CIs, we will consider expected length.

Our efficiency bounds can be thought of as applying the bounds in
\citet{ArKo18optimal} to a local asymptotic setting, which corresponds to the
limiting model~\eqref{limit_experiment_eq} with $\Gamma=\Gamma_{\theta^*, P_0}$,
$\Sigma=\Sigma_{\theta^*, P_0}$ and $H=H_{\theta^*}$. The between class modulus
of continuity for this model is
\begin{multline}\label{between_class_modulus_eq}
  \omega(\delta;\mathcal{F}, \mathcal{G}, H, \Gamma, \Sigma)
  =\sup\, \, H(s_1-s_0)\quad \text{s.t.}\quad (s_0', c_0')'\in \mathcal{F}, (s_1', c_1')'\in \mathcal{G},\\
             [(c_1-c_0)-\Gamma (s_1-s_0)]'\Sigma^{-1}[(c_1-c_0)-\Gamma (s_1-s_0)]\le \delta^2.
\end{multline}
We use the notation $\omega(\delta)$ and
$\omega(\delta;\mathcal{F}, \mathcal{G})$ when the context is clear. In the case
where $\mathcal{G}=\mathcal{F}=\mathbb{R}^{d_\theta}\times \mathcal{C}$ and
$\mathcal{C}$ is centrosymmetric, the solution satisfies $s_1=-s_0$ and
$c_1=-c_0$, which gives the same optimization problem
as~\eqref{half_modulus_eq}, with the objective multiplied by two (this matches
the definition of $\omega(\cdot)$ used to define $\kappa_{*}$ in the main text).

For one-sided CIs, we show that, for any CI satisfying the coverage
condition~\eqref{coverage_eq} for a rich enough class $\mathcal{P}$, we will
have
\begin{equation}\label{onesided_bound_eq}
  \liminf_{n\to\infty} \sqrt{n}q_{\beta, n}(\hat c;\mathcal{P}, \mathcal{G}, \Theta_n)\ge \omega(\delta_\beta;\mathcal{F}, \mathcal{G}, H, \Gamma, \Sigma)
\end{equation}
where $\delta_\beta=z_{1-\alpha}+z_\beta$, where $z_{\tau}$ denotes the $\tau$
quantile of the $\mathcal{N}(0,1)$ distribution. For bounds on excess length at
a single $P_0$ with $E_{P_0}g(w_i, \theta^*)=0$, we obtain this bound with
$\mathcal{G}=\{0\}$:
\begin{equation}\label{onesided_directed_power_bound_eq}
  \liminf_{n\to\infty} \sqrt{n}q_{P_0,\beta}(h(\theta^*)-\hat c)\ge \omega(\delta_\beta;\mathcal{F}, \{0\}, H, \Gamma, \Sigma).
\end{equation}
These results can be thought of as a local asymptotic version of Theorem 3.1 in
\citet{ArKo18optimal} applied to our setting.

For two-sided CIs, we show that, if a CI $\mathcal{I}_n=\{\hat h\pm \hat \chi\}$
satisfies the coverage condition~\eqref{coverage_eq} for a rich enough class
$\mathcal{P}$, then, for any $P_0$ with $E_{P_0}g(w_i, \theta^*)=0$, expected
length satisfies
\begin{multline}\label{two_sided_bound_eq}
  \liminf_{T\to\infty}\liminf_{n\to\infty} E_{P_0}\min\{\sqrt{n}2\hat\chi, T\} \\
  \ge (1-\alpha)E[\omega(z_{1-\alpha}-Z;\{0\}, \mathcal{F}, H, \Gamma, \Sigma)
    +\omega(z_{1-\alpha}-Z;\mathcal{F}, \{0\}, H, \Gamma, \Sigma)|Z\le z_{1-\alpha}],
\end{multline}
where $Z\sim \mathcal{N}(0,1)$. The above bound uses truncated expected length
to avoid technical issues with convergence of moments when achieving the bound
(note however that this bound immediately implies the same bound on excess
length without truncation). Our results constrain the CI to take the form of an
interval. We conjecture that the bound applies to arbitrary confidence sets
(with length defined as Lebesgue measure) under additional regularity
conditions.

Here, ``rich enough'' means that $\mathcal{P}$ contains a least favorable
submodel. \Cref{submodel_sec} begins the derivation of our efficiency results by
giving conditions on this submodel. In \Cref*{gmm_submodel_sec}, we construct a
submodel satisfying these conditions under mild conditions.

\subsection{Least favorable submodel}\label{submodel_sec}

Let $P_0$ be a distribution with $E_{P_0}g(w_i, \theta^*)=0$ (i.e.\ the model
holds for this data-generating process with $\theta=\theta^*$ and $c=0$), and
consider a parametric submodel $P_t$ indexed by $t\in\mathbb{R}^{d_g}$ (i.e.\
the dimension of $t$ is the same as the dimension of the values of
$g(w_i, \theta)$) with $P_t$ equal to $P_0$ at $t=0$. We assume that
$\{w_i\}_{i=1}^{n}$ are i.i.d.\ under $P_t$. Let $\pi_t(w_i)$ denote the density
of a single observation with respect to its distribution under $P_0$, so that
$E_{P_t}f(w_i)=E_{P_0}f(w_i)\pi_t(w_i)$ for any function $f$. We expect that the
least favorable submodel for this problem will be the one that makes estimating
$E_{P}g(W_i, \theta^*)$ most difficult. This corresponds to any subfamily with
score function $g(w_i, \theta^*)$. We also place additional conditions on this
submodel, given in the following assumption.

\begin{assumption}\label{submodel_assump}
  The data are i.i.d.\ under $P_t$ for all $t$ in a neighborhood of zero, and
  the density $\pi_t(w_i)$ for a single observation is quadratic mean
  differentiable at $t=0$ with score function $g(w_i, \theta^*)$, where
  $E_{P_0}g(w_i, \theta^*)=0$. In addition, the function
  $(t', \theta')'\mapsto E_{p_t} g(w_i, \theta)$ is continuously differentiable
  at $(0', {\theta^*}')'$ with
  \begin{equation}\label{t_theta_derivative_eq}
    \left[\frac{d}{d(t', \theta')} E_{p_t} g(w_i, \theta)\right]_{t=0,\theta=\theta^*}
    =(\Sigma, \Gamma)
  \end{equation}
  where $\Sigma$ and $\Gamma$ are full rank.
\end{assumption}

To understand \Cref{submodel_assump}, note that Problem 12.17 in \citet{LeRo05}
gives the Jacobian with respect to $t$ as $\Sigma$ in the case where
$g(w_i, \theta^*)$ is bounded, and the Jacobian with respect to $\theta$ is
equal to $\Gamma$ by definition. \Cref{submodel_assump} requires the slightly
stronger condition that $E_{p_t} g(w_i, \theta)$ is continuously differentiable
with respect to $(t', \theta')'$ for $t$ close to $0$ and $\theta$ close to
$\theta^*$. This is needed to apply the Implicit Function Theorem in the
derivations that follow. In the supplemental materials, we give a construction
of a quadratic mean differentiable family satisfying this condition, without
requiring boundedness of $g(w_i, \theta^*)$ ({\Cref*{gmm_submodel_lemma}} in
\Cref*{gmm_submodel_sec}).

The bounds in \citet{ArKo18optimal} are obtained by bounding the power of a
two-point test (simple null and simple alternative) where the null and
alternative are given by the points that achieve the modulus. To obtain
analogous results in our setting, we use a bound on the power of a two-point
test in a least favorable submodel.

Consider sequences of local parameter values $(\theta_{0,n}', c_{0,n}')'$ and
$(\theta_{1,n}', c_{1,n}')'$ where, for some $s_0$, $c_0$ $s_1$ and $c_1$,
\begin{equation}\label{null_sequence_def_eq}
  \theta_{d, n}=\theta^*+(s_{d} +o(1))/\sqrt{n}, \,
                 c_{d, n}=c_{d}+o(1)\quad d\in\{0,1\}
\end{equation}
Consider a sequence of tests of $(\theta_{0,n}', c_{0,n}')'$ vs
$(\theta_{1,n}', c_{1,n}')'$. Formally, for any $(\theta', c')'$, let
\begin{equation}\label{Pnc_eq}
  \mathcal{P}_n(\theta, c)=\left\{P\in\mathcal{P}: E_{P} g(w_i, \theta)=c/\sqrt{n}\right\}
\end{equation}
be the set of probability distributions in $\mathcal{P}$ that are consistent
with the parameter values $(\theta', c')'$. We derive a bound on the asymptotic
minimax power of a level $\alpha$ test of
\begin{equation}\label{H0_H1_def_eq}
  H_{0,n}: P\in \mathcal{P}_n(\theta_{0,n}, c_{0,n})
  \quad\text{vs}\quad
  H_{1,n}: P\in \mathcal{P}_n(\theta_{1,n}, c_{1,n}),
\end{equation}
as well as a bound on the power of a test of $H_{0,n}$ at $P_0$. Let $\Phi$ be
the standard normal cdf and let
\begin{equation*}
  \overline\beta(s_{0}, c_{0}, s_{1}, c_{1})
  =\Phi\left(\sqrt{[c_1-c_0-\Gamma (s_1-s_0)]'
      \Sigma ^{-1}[c_1-c_0-\Gamma (s_1-s_0)]}-z_{1-\alpha}\right).
\end{equation*}

\begin{lemma}\label{two_point_lemma}
  Let $\mathcal{P}$ be a class of distributions that contains a family $P_t$
  that satisfies \Cref{submodel_assump}. Then, for any sequence of
  tests $\phi_n$ satisfying
  $\limsup_n \sup_{P\in \mathcal{P}_n(\theta_{0,n}, c_{0,n})}E_{P}\phi_n\le
  \alpha$, we have
  \begin{equation*}
    \limsup_n E_{P_0}\phi_n
    \le \overline\beta(s_{0}, c_{0}, 0, 0)
    \quad\text{and}\quad
    \limsup_n \inf_{P\in \mathcal{P}_n(\theta_{1, n}, c_{1, n})} E_{P}\phi_{n}
    \le \overline\beta(s_0, c_0, s_1, c_1).
  \end{equation*}
\end{lemma}

\Cref{two_point_lemma} says that the asymptotic minimax power of any test of
$H_{0,n}$ vs $H_{1,n}$ is bounded by $\overline\beta(s_0, c_0, s_1, c_1)$.
Furthermore, if we take $s_1=0$ and $c_1=0$, then this bound is achieved at
$P_0$. Note that, in keeping with the analogy with the linear
model~\eqref{limit_experiment_eq}, $\overline\beta(s_0, c_0, s_1, c_1)$ is the
power of the optimal (Neyman-Pearson) test of the simple null $(s_0', c_0')$ vs
the simple alternative $(s_1', c_1')$ in the model~\eqref{limit_experiment_eq}.

\begin{proof}[Proof of \Cref{two_point_lemma}]

  The proof involves two steps. First, we use the Implicit Function Theorem to
  find sequences $t_{0,n}$ and $t_{1,n}$ such that $P_{t_0,n}$ satisfies
  $H_{0,n}$ and $P_{t_1,n}$ satisfies $H_{1,n}$. Next, we apply a standard
  result on testing in quadratic mean differentiable families to obtain the
  limiting power of the optimal test of $P_{t_0,n}$ vs $P_{t_1,n}$, which gives
  an upper bound on the limiting minimax power of any test of $H_{0,n}$ vs
  $H_{1,n}$.

  Let $f(t, \theta, a)=E_{P_t} g(w_i, \theta)-a$ so that $(\theta', c')'$ is
  consistent with $P_t$ iff. $f(t, \theta, c/\sqrt{n})=0$.
  Under \Cref{submodel_assump}, it follows from the Implicit Function
  Theorem that there exists a function $r(\theta, a)$ such that, for $\theta$ in
  a neighborhood of $\theta^*$ and $a$ in a neighborhood of zero,
  \begin{equation*}
    E_{P_{r(\theta, a)}} g(w_i, \theta)-a=f(r(\theta, a), \theta, a)=0.
  \end{equation*}
  Thus, letting $t_{0,n}=r(\theta_{0,n}, c_{0,n}/\sqrt{n})$ and
  $t_{1,n}=r(\theta_{1,n}, c_{1,n}/\sqrt{n})$, $P_{t_{0,n}}$ satisfies $H_{0,n}$
  and $P_{t_{1,n}}$ satisfies $H_{1,n}$. Furthermore,
  \begin{equation*}
    \left[\frac{d}{d(\theta', a')} r(\theta, a)\right]_{(\theta', a')=(\theta^*,0)}
    =-\Sigma^{-1} (\Gamma, -I_{d_g})
  \end{equation*}
  so that
  \begin{equation*}
    r(\theta, a)=\Sigma^{-1}a-\Sigma^{-1} \Gamma
    (\theta-\theta^*)+o(\|\theta-\theta^*\|+\|a\|).
  \end{equation*}
  Thus, letting $t_{0,\infty}=\Sigma^{-1}c_0-\Sigma^{-1} \Gamma s_0$, we have
  \begin{equation*}
    \begin{split}
      t_{0,n}&=r(\theta_{0,n}, c_{0,n}/\sqrt{n})
      =\Sigma^{-1}c_{0,n}/\sqrt{n}-\Sigma^{-1} \Gamma (\theta_{0,n}-\theta^*)+o(\|\theta_{0,n}-\theta^*\|+\|c_{0,n}\|/\sqrt{n}) \\
      &=\Sigma^{-1}c_0/\sqrt{n}-\Sigma^{-1} \Gamma s_0/\sqrt{n}+o(1/\sqrt{n})
      =t_{0,\infty}/\sqrt{n}+o(1/\sqrt{n}).
    \end{split}
  \end{equation*}
  Similarly, $t_{1,n}=t_{1,\infty}/\sqrt{n}+o(1/\sqrt{n})$ where
  $t_{1,\infty}=\Sigma^{-1}c_1-\Sigma^{-1} \Gamma s_1$.

  Since the information matrix for this submodel evaluated at $t=0$ is $\Sigma$,
  it follows from the arguments in Example 12.3.12 in \citet{LeRo05}, extended
  to the case where the null and alternative are both drifting sequences (rather
  than just the alternative), that the limit of the power of the Neyman-Pearson
  test of $P_{t_{0,n}}$ vs $P_{t_{1,n}}$ is
  \begin{equation*}
    \Phi\left(
    \sqrt{[t_{1,\infty}-t_{0,\infty}]'
    \Sigma [t_{1,\infty}-t_{0,\infty}]}-z_{1-\alpha}
    \right)
    =\overline\beta(s_0,c_0,s_1,c_1).
  \end{equation*}
  This gives the required bound on minimax power over $H_{1,n}$. To obtain the
  bound on power at $P_0$, note that, for $\theta_{1,n}=\theta^*$ and
  $c_{1,n}=0$, $t_{0,n}=0$, the bound also corresponds to the power of a test
  that is optimal for $P_{t_{0,n}}$ vs $P_0$.
\end{proof}

\subsection{One-sided CIs}\label{oneside_efficiency_bound_sec}

We prove the following efficiency bound for one-sided CIs.

\begin{theorem}\label{oneside_efficiency_bound_thm}
  Let $\mathcal{P}$ be a class of distributions that contains a submodel $P_t$
  satisfying \Cref{submodel_assump}. Let
  $\Theta_n(C)=\{\theta|\|\theta-\theta^*\|\le C/\sqrt{n}\}$ for some constant
  $C$, and let $\mathcal{F}$ be given. Let $\hor{\hat c, \infty}$ be a sequence
  of CIs such that, for all $C$, the coverage condition~\eqref{coverage_eq}
  holds with $\Theta_n=\Theta_n(C)$. Let $\mathcal{G}\subseteq \mathcal{F}$ be a
  set such that the limiting modulus $\omega$ is well-defined and continuous for
  all $\delta$. Then the asymptotic lower bounds~\eqref{onesided_bound_eq}
  and~\eqref{onesided_directed_power_bound_eq} hold.
\end{theorem}
\begin{proof}
  Consider a sequence of simple null and alternative values of $\theta$ and $c$
  that satisfy~\eqref{null_sequence_def_eq} for some $s_0, c_0, s_1, c_1$, with
  $(\sqrt{n}(\theta_{0,n}-\theta^*)', c_{0,n}')'\in\mathcal{F}$ and
  $(\sqrt{n}(\theta_{1,n}-\theta^*)', c_{1,n}')'\in\mathcal{G}$, for each $n$.
  Note that
  \begin{equation*}
    \lim_{n\to\infty}\sqrt{n}[h(\theta_{1,n})-h(\theta_{0,n})]=H (s_1-s_0).
  \end{equation*}
  Consider the testing problem
  $H_{0,n}:P\in\mathcal{P}_n(\theta_{0,n}, c_{0,n})$ vs
  $H_{1,n}:P\in\mathcal{P}_n(\theta_{1,n}, c_{1,n})$ defined in~\eqref{Pnc_eq}
  and~\eqref{H0_H1_def_eq}. Suppose that
  \begin{equation}\label{qbetan_eq}
    q_{\beta, n}(\hat c;\mathcal{P}, \mathcal{G}, \Theta_n)<h(\theta_{1,n})-h(\theta_{0,n}).
  \end{equation}
  Let $\phi_n$ denote the test that rejects when
  $h(\theta_{0,n})\notin\hor{\hat{c}, \infty}$. Since, for any
  $P\in\mathcal{P}_{n}(\theta_{1,n}, c_{1,n})$, we have
  $q_{P, \beta}(h(\theta_{1,n})-\hat c)\le q_{\beta, n}(\hat{c}; \mathcal{P},
  \mathcal{G}, \Theta_n)$ by construction, it follows that, for all
  $P\in\mathcal{P}_{n}(\theta_{1,n}, c_{1,n})$,
\begin{equation*}
  E_{P}\phi_n
    =P(h(\theta_{1,n})-\hat c<h(\theta_{1,n})-h(\theta_{0,n}))
    \ge P(h(\theta_{1,n})-\hat c\le q_{P, \beta}(h(\theta_{1,n})-\hat c))
    \ge \beta,
\end{equation*}
where the last step follows from properties of quantiles \citep[Lemma 21.1
in][]{van_der_vaart_asymptotic_1998}. The coverage
requirement~\eqref{coverage_eq} implies that the test $\phi_n$ that rejects when
$h(\theta_{0,n})\notin\hor{\hat c, \infty}$ has asymptotic level $\alpha$ for
$H_{0,n}$. Thus, by \Cref{two_point_lemma}, we must have
$\beta\le \overline\beta(s_0,c_0,s_1,c_1)$ if~\eqref{qbetan_eq} holds infinitely
often.

It follows that, if $\overline\beta(s_0,c_0,s_1,c_1)<\beta$, we must have
\begin{equation*}
  \liminf_{n\to\infty} \sqrt{n}q_{\beta, n}(\hat c;\mathcal{P}, \mathcal{G}, \Theta_n)
  \ge H(s_1-s_0)
\end{equation*}
since otherwise, \Cref{qbetan_eq} would hold infinitely often. Since the
sequences and limiting $(s_0', c_0')\in\mathcal{F}$ and
$(s_1', c_1')\in\mathcal{G}$ were arbitrary, the above bound holds for any
$(s_0', c_0')\in\mathcal{F}$ and $(s_1', c_1')\in\mathcal{G}$ with
$\overline\beta(s_0,c_0,s_1,c_1)\le \beta-\eta$, where $\eta>0$ is arbitrary.
The maximum of the right-hand side over $s_0,c_0,s_1,c_1$ in this set is equal
to $\omega(\delta_{\beta-\eta};\mathcal{F}, \mathcal{G}, H, \Gamma, \Sigma)$ by
definition, so taking $\eta\to 0$ gives the result.
\end{proof}

\subsection{Two-sided CIs}\label{twoside_efficiency_bound_sec}

We prove the following efficiency bound for two-sided CIs.

\begin{theorem}\label{twoside_efficiency_bound_thm}
  Suppose that, for all $C$, $\{\hat h\pm \hat \chi\}$ satisfies the local
  coverage condition~\eqref{coverage_eq} with
  $\Theta_n=\Theta_n(C)=\{\theta|\|\theta-\theta^*\|\le C/\sqrt{n}\}$, where
  $\mathcal{P}$ contains a submodel $P_t$ satisfying
  \Cref{submodel_assump}. Suppose also that
  $0_{d_\theta+d_g}\in\mathcal{F}$ and a minimizer
  $(s_{\vartheta}', c_{\vartheta}')'$ of $(c-\Gamma s)'\Sigma^{-1}(c-\Gamma s)$
  subject to $Hs=\vartheta$ and $(s', c')'\in\mathcal{F}$ exists for all
  $\vartheta\in\mathbb{R}$. Then the asymptotic lower
  bound~\eqref{two_sided_bound_eq} holds.
\end{theorem}

In the case where $\mathcal{F}=\mathbb{R}^{d_\theta}\times \mathcal{C}$, which is the focus of the main text, a sufficient condition for the existence of the minimizer $(s_{\vartheta}', c_{\vartheta}')'$ is that $\mathcal{C}$ is compact, $H$ is not equal to the zero vector and $\Gamma$ is full rank.

\begin{proof}
  For each $\vartheta\in\mathbb{R}$, let
  $\tilde\theta_{\vartheta, n}=\theta^*+s_{\vartheta}/\sqrt{n}$, and let
  $\phi_{\vartheta, n}=I(h(\tilde\theta_{\vartheta, n})\notin\{\hat h\pm
  \hat\chi\})$ be the test that rejects when $h(\tilde\theta_{\vartheta, n})$ is
  not in the CI\@. When the constant $C$ defining $\Theta_n=\Theta_n(C)$ is
  large enough, the asymptotic coverage condition~\eqref{coverage_eq} implies
  that $\phi_{\vartheta, n}$ is an asymptotic level $\alpha$ test for
  $H_{0,n}:P\in \mathcal{P}_n(\tilde\theta_{\vartheta, n}, c_{\vartheta})$
  defined in~\eqref{Pnc_eq} and~\eqref{H0_H1_def_eq}. Thus, by
  \Cref{two_point_lemma},
\begin{equation}\label{delta_vartheta_bound_eq}
  \limsup_{n\to\infty} E_{P_0}\phi_{\vartheta, n}\le
  \Phi(\delta_{\vartheta}-z_{1-\alpha})
  \quad\text{where}\quad
  \delta_{\vartheta}=\sqrt{(c_\vartheta-\Gamma s_\vartheta)'\Sigma^{-1}(c_\vartheta-\Gamma s_\vartheta)}.
\end{equation}

We apply this bound to a grid of values of $\vartheta$.
Let $\mathcal{E}_n(m)$ denote the grid centered at zero with length $2m$ and meshwidth $1/m$:
\begin{equation*}
\mathcal{E}_n(m)=\{j/m: j\in \mathbb{Z}, |j|\le m^2\}.
\end{equation*}
Let
\begin{equation*}
\widetilde{\mathcal{E}}_n(m)=\{\sqrt{n}[h(\tilde\theta_{\vartheta, n})-h(\theta^*)]: \vartheta\in\mathcal{E}_n(m)\}.
\end{equation*}
Note that
$h(\tilde\theta_{\vartheta,
  n})=h(\theta^*)+(1+o(1))Hs_\vartheta/\sqrt{n}=h(\theta^*)+(1+o(1))\vartheta/\sqrt{n}$.
Thus, letting $a_1,\ldots, a_{2m^2 +1}$ denote the ordered elements in
$\mathcal{E}_n(m)$ and $\tilde a_1,\ldots, \tilde a_{m^2+1}$ the ordered
elements in $\widetilde{\mathcal{E}}_n$, we have $\tilde a_j\to a_j$ for each
$j$ as $n\to\infty$.

Let $\mathcal{N}(n, m)$ be the number of elements $\tilde a_j$ in $\widetilde{\mathcal{E}}_n$ such that $h(\theta^*)+\tilde a_j/\sqrt{n}=h(\tilde\theta_{a_j, n})\in \{\hat h\pm \hat \chi\}$.  Then
\begin{equation*}
E_{P_0}\mathcal{N}(n, m)
=\sum_{j=1}^{2m^2+1} E_{P_0}I(h(\tilde\theta_{a_j, n})\in\{\hat h\pm\hat \chi\})
=\sum_{j=1}^{2m^2+1} [1-E_{P_0}\phi_{a_j, n}].
\end{equation*}
It follows from~\eqref{delta_vartheta_bound_eq} that (assuming the constant $C$
that defines $\Theta_n(C)$ is large enough),
\begin{equation*}
  \liminf_{n\to\infty}E_{P_0}\mathcal{N}(n, m)\ge \sum_{j=1}^{2m^2+1} [1-\Phi(\delta_{a_j}-z_{1-\alpha})]
  =\sum_{j=1}^{2m^2+1} \Phi(z_{1-\alpha}-\delta_{a_j}).
\end{equation*}
Note that $2\hat \chi\ge n^{-1/2}[\mathcal{N}(n, m)-1]\cdot \min_{1\le j\le 2m^2}(\tilde a_{j+1}-\tilde a_j)=n^{-1/2}[\mathcal{N}(n, m)-1]\cdot m^{-1}\cdot (1+\varepsilon_n)$ where $\varepsilon_n=\min_{1\le j\le 2m^2}(\tilde a_{j+1}-\tilde a_j)/m^{-1}-1$ is a nonrandom sequence converging to zero.  This, combined with the above display, gives
\begin{equation*}
\liminf_{n\to\infty}E_{P_0}\min\{2n^{1/2}\hat \chi, T\}
  \ge \left[m^{-1}\sum_{j=1}^{2m^2+1} \Phi(z_{1-\alpha}-\delta_{a_j})-m^{-1}\right]
\end{equation*}
for any $T>2m$.
We have
\begin{equation}\label{aj_sum_eq}
m^{-1}\sum_{j=1}^{2m^2+1} \Phi(z_{1-\alpha}-\delta_{a_j})
=m^{-1}\sum_{j=1}^{2m^2+1} \int I(\delta_{a_j}\le z_{1-\alpha}-z)d\Phi(z).
\end{equation}
Following the proof of Theorem 3.2 in \citet{ArKo18optimal}, note that, for $\vartheta\ge 0$, $t\ge 0$, we have $\delta_\vartheta\le t$ iff. $\vartheta\le \omega(t;\{0\}, \mathcal{F})$.  Indeed, note that $\omega(\delta_\vartheta;\{0\}, \mathcal{F})\ge H s_{\vartheta}=\vartheta$ by feasibility of $0$ and $s_\vartheta, c_\vartheta$ for this modulus problem.  Since the modulus is increasing, this means that, if $\delta_\vartheta\le t$, we must have $\vartheta\le \omega(t;\{0\}, \mathcal{F})$.  Now suppose $\vartheta\le \omega(t;\{0\}, \mathcal{F})$.
Then $H s_{\omega(t;\{0\}, \mathcal{F})}\ge \vartheta$, so, for some $\lambda\in [0,1]$, $(s_\lambda', c_\lambda')=\lambda (s_{\omega(t;\{0\}, \mathcal{F})}', c_{\omega(t;\{0\}, \mathcal{F})}')$ satisfies $H s_\lambda=\vartheta$, which means that $\delta_\vartheta\le \sqrt{(c_\lambda-\Gamma s_\lambda)'\Sigma^{-1}(c_\lambda-\Gamma s_\lambda)}\le t$ as claimed.

Thus, the part of the expression in~\eqref{aj_sum_eq} corresponding to terms in
the sum with $a_j\ge 0$ is given by
\begin{multline*}
  m^{-1}\sum_{j=1}^{2m^2+1} \int
  I(0\le a_j\le \omega(z_{1-\alpha}-z;\{0\}, \mathcal{F}))\, d\Phi(z) \\
  \ge \int_{z\le z_{1-\alpha}} \min\{\omega(z_{1-\alpha}-z;\{0\}, \mathcal{F})-1/m, m\}d\Phi(z).
\end{multline*}
By the Dominated Convergence Theorem, this converges to
$\int_{z\le z_{1-\alpha}}\omega(z_{1-\alpha}-z;\{0\}, \mathcal{F})d\Phi(z)$ as
$m\to\infty$. Similarly, for $\vartheta< 0$, $t\ge 0$, we have
$\delta_\vartheta\le t$ iff. $-\vartheta\le \omega(t;\mathcal{F}, \{0\})$, so
that an analogous argument shows that, for arbitrary $\varepsilon>0$, there
exists $m$ such that
$\int_{z\le z_{1-\alpha}}\omega(z_{1-\alpha}-z;\mathcal{F},
\{0\})d\Phi(z)-\varepsilon$ is an asymptotic lower bound for the part of the
expression~\eqref{aj_sum_eq} that corresponds to terms in the sum with $a_j<0$.
Thus, for any $\varepsilon>0$, there exist constants $C$ and $T$ such that, if
the coverage condition~\eqref{coverage_eq} holds with $\Theta_{n}=\Theta_{n}(C)$,
\begin{equation*}
  \liminf_{n\to\infty} E_{P_0}\min\{n^{1/2}2\hat \chi, T\}
  \ge \int_{z\le z_{1-\alpha}}[\omega(z_{1-\alpha}-z;\{0\}, \mathcal{F})
  +\omega(z_{1-\alpha}-z;\mathcal{F}, \{0\})]d\Phi(z)-2\varepsilon.
\end{equation*}
This gives the result.
\end{proof}

\subsection{Achieving the bound}\label{achieving_bound_sec}

This section gives formal results showing that the CIs proposed in the main text
are asymptotically valid, and that, if the sensitivities are chosen optimally,
they achieve the efficiency bound in \Cref{oneside_efficiency_bound_thm} in the
one-sided case, and nearly achieve the bound in
\Cref{twoside_efficiency_bound_thm} in the two-sided case (where ``nearly''
means up to the sharp efficiency bound $\kappa_{*}$ in the limiting model, given
in~\eqref{kappa_conv_cs_eq}, in the case where $\mathcal{C}$ is
centrosymmetric).

We specialize to the case considered in the main text where we require coverage
without local restrictions on $\theta$. In the notation of
\Cref{oneside_efficiency_bound_sec,twoside_efficiency_bound_sec}, this
corresponds to $\mathcal{F}=\mathbb{R}^{d_\theta}\times \mathcal{C}$ for a
convex (but possibly asymmetric) set $\mathcal{C}$.

In the main text, we focused on the case where $\mathcal{C}$ is centrosymmetric.
To allow for general convex $\mathcal{C}$, we use estimators that are
asymptotically affine, rather than linear. We focus on one-step estimators,
which take the form
\begin{equation*}
\hat h=h(\hat\theta_{\text{initial}})+\hat k' g(\hat\theta_{\text{initial}})+\hat a/\sqrt{n}.
\end{equation*}
for some vector $\hat{k}$ and some scalar $\hat a$. We continue to require the
condition
\begin{equation}\label{finite_bias_est_eq}
\hat H=-\hat k'\hat \Gamma,
\end{equation}
where $\hat\Gamma$ is an estimator of $\Gamma$ satisfying conditions to be given below.

To deal with asymmetric $\mathcal{C}$, and to state results involving worst-case quantiles of excess length over different sets, it will be helpful to separately define worst-case upper and lower bias.  For a set $\mathcal{C}\in\mathbb{R}^{d_g}$, let
\begin{equation*}
  \maxbias_{\mathcal{C}}(k, a)=\sup_{c\in\mathcal{C}}k'c+a,
  \quad \minbias_{\mathcal{C}}(k, a)=\inf_{c\in\mathcal{C}}k'c+a
\end{equation*}
A one-sided asymptotic $1-\alpha$ CI is given by $\hor{\hat{c}, \infty}$ where
\begin{equation*}
  \begin{split}
    \hat c &= \hat h-\maxbias_{\mathcal{C}}(\hat k, \hat a)/\sqrt{n}-z_{1-\alpha}\sqrt{\hat k'\hat \Sigma \hat k}/\sqrt{n} \\
    &=h(\hat\theta_{\text{initial}})+\hat k' g(\hat\theta_{\text{initial}})+
    \hat{a}/\sqrt{n}
    -\maxbias_{\mathcal{C}}(\hat k, \hat a)/\sqrt{n}-z_{1-\alpha}\sqrt{\hat k'\hat \Sigma \hat k}/\sqrt{n} \\
    &=h(\hat\theta_{\text{initial}})+\hat k' g(\hat\theta_{\text{initial}})
    -\maxbias_{\mathcal{C}}(\hat{k},
    0)/\sqrt{n}-z_{1-\alpha}\sqrt{\hat{k}'\hat{\Sigma} \hat k}/\sqrt{n},
  \end{split}
\end{equation*}
and $\hat \Sigma$ is an estimate of $\Sigma$. Thus, the intercept term $\hat a$
does not matter for the one-sided CI and can be taken to be zero in this case.
For two-sided CIs, however, the choice of $\hat a$ matters, and we assume that
$\hat a$ is chosen so that the estimator is centered:
\begin{equation}\label{centered_a_eq}
  \maxbias_{\mathcal{C}}(\hat k, \hat a)=\sup_{c\in\mathcal C}\hat{k}'c+\hat a
  =-\left(\inf_{c\in\mathcal C}\hat k'c+\hat a\right)=-\minbias_{\mathcal{C}}(\hat k, \hat a).
\end{equation}
A two-sided asymptotic $1-\alpha$ CI is then given by $\hat h\pm \hat \chi$ where
\begin{equation*}
  \hat \chi =
  \cv_\alpha\left(\maxbias_{\mathcal{C}}(\hat k, \hat a)/\sqrt{\hat k'\hat \Sigma \hat k}\right)
  \sqrt{\hat k'\hat \Sigma \hat k}/\sqrt{n},
  \quad
  \text{where $\cv_\alpha(t)$ is the $1-\alpha$ quantile of $|\mathcal{N}(t,1)|$}.
\end{equation*}

For both forms of CIs, we first state a result for general $\hat k$, $\hat a$,
and then specialize to optimal weights. For the one-sided case, we consider CIs
that optimize worst-case length over $(\sqrt{n}(\theta-\theta^*)', c')'$ in some
set $\mathcal{G}$, subject to coverage over
$\mathcal{F}=\mathbb{R}^{d_\theta}\times \mathcal{C}$. In principle, this allows
for confidence sets that ``direct power'' not only at particular values of $c$
but also at particular values of $\theta$. However, \Cref*{invariance_lemma} in
\Cref*{sec:repl-mathbbrd_th-wit} shows that the optimal weights for this problem
are the same as the optimal weights when $\mathcal{G}$ is replaced by
$\mathbb{R}^{d_\theta}\times \mathcal{D}(\mathcal{G})$, where
$\mathcal{D}(\mathcal{G})=\{c\colon\text{there exists $s$ s.t.}\;(s', c')'\in
\mathcal{G}\}$. Thus, it is without loss of generality to consider weights that
optimize worst-case excess length over $c\in\mathcal{D}$ subject to coverage
over $c\in\mathcal{C}$ where $\mathcal{D}\subseteq \mathcal{C}$ is a compact
convex set.

The optimal weights take the form
$\hat k=k(\delta_\beta, \hat H, \hat\Gamma, \hat \Sigma)$ where
\begin{equation}\label{khat_def_eq}
k(\delta, H, \Gamma, \Sigma)'=\frac{((c_{1,\delta}^*-c_{0,\delta}^*)
- \Gamma (s_{1,\delta}^*-s_{0,\delta}^*))' \Sigma^{-1}}{((c_{1,\delta}^*-c_{0,\delta}^*)
- \Gamma (s_{1,\delta}^*-s_{0,\delta}^*))' \Sigma^{-1} \Gamma H'/ H H'}
\end{equation}
and $c_{0,\delta}$, $s_{0,\delta}$, $c_{1,\delta}$, $s_{1,\delta}$ solve the
between class modulus problem~\eqref{between_class_modulus_eq} with
$\mathcal{F}=\mathbb{R}^{d_\theta}\times \mathcal{C}$ and
$\mathcal{G}=\mathbb{R}^{d_\theta}\times \mathcal{D}$. For a two-sided CI of the
form given above, the optimal weights take this form with
$\mathcal{D}=\mathcal{C}$, $\delta$ minimizing $\hat \chi$, and with $\hat a$
chosen to center the CI so that~\eqref{centered_a_eq} holds. We note that, in
the case where $\mathcal{D}=\mathcal{C}$ and $\mathcal{C}$ is centrosymmetric,
$s_{1,\delta}^*=s_{0,\delta}^*$ and $c_{1,\delta}^*=c_{0,\delta}^*$,
and~\eqref{between_class_modulus_eq} reduces to two times the optimization
problem~\eqref{half_modulus_eq}. The weights $\hat k$ then take the form given
in~\eqref{optimal_weights_eq} in the main text, and, since $\mathcal{C}$ is
centrosymmetric, $\hat a=0$, which gives the two-sided CI proposed in the main
text.

For our general result showing coverage for possibly suboptimal weights
$\hat k$, $\hat a$, we make the following assumptions. In the following, for a
set $\mathcal{A}_n$, random variables $A_{n, \theta, P}$ and $B_{n, \theta, P}$ and
a sequence $a_n$, we say $A_{n, \theta, P}=B_{n, \theta, P}+o_P(a_n)$ uniformly over
$(\theta, P)$ in $\mathcal{A}_n$ if, for all $\varepsilon>0$,
$\sup_{(\theta, P)\in\mathcal{A}_n}
P(a_n^{-1}\|A_{n, \theta, P}-B_{n, \theta, P}\|>\varepsilon)\to 0$. We say
$A_{n, \theta, P}=B_{n, \theta, P}+\mathcal{O}_P(a_n)$ uniformly over $(\theta, P)$
in a set $\mathcal{A}_n$ if
$\lim_{C\to\infty}\limsup_{n\to\infty}\sup_{(\theta, P)\in\mathcal{A}_n}
P(a_n^{-1}\|A_{n, \theta, P}-B_{n, \theta, P}\|>C)= 0$. In the following, the set
$\mathcal{S}_n$ defined in \Cref{setup_sec} over which coverage is
required is defined with $\mathcal{F}=\mathbb{R}^{d_\theta}\times \mathcal{C}$.

\begin{assumption}\label{theta_init_assump}
The set $\mathcal{C}$ is compact
or takes the form $\widetilde{\mathcal{C}}\times \mathbb{R}^{d_{g_2}}$ where $d_{g_1}+d_{g_2}=d_g$ and $\widetilde{\mathcal{C}}$ is a compact subset of $\mathbb{R}^{d_{g_1}}$.
In addition,
$\hat\theta_{\text{initial}}-\theta=\mathcal{O}_P(1/\sqrt{n})$,
$\hat g(\hat\theta_{\text{initial}})-\hat g(\theta)=\Gamma_{\theta, P}(\hat\theta_{\text{initial}}-\theta)+o_P(1/\sqrt{n})$
and $h(\hat\theta_{\text{initial}})-h(\theta)=H_\theta (\hat\theta_{\text{initial}}-\theta)+o_P(1/\sqrt{n})$
uniformly over $(\theta, P)\in\mathcal{S}_n$.
\end{assumption}

\begin{assumption}\label{g_clt_assump_append}
  $\hat g(\theta)-g_P(\theta)=\mathcal{O}(1/\sqrt{n})$ uniformly over
  $(\theta, P)\in\mathcal{S}_n$. Furthermore, for a collection of matrices
  $\Sigma_{\theta, P}$ such that $k_{\theta, P}'\Sigma_{\theta, P}k_{\theta, P}$ is
  bounded away from zero and infinity,
\begin{equation*}
  \sup_{t\in\mathbb{R}}\sup_{(\theta, P)\in\mathcal{S}_n}
  \left|P\left(\frac{\sqrt{n}k_{\theta, P}'(\hat g(\theta)-g_P(\theta))}{\sqrt{k_{\theta, P}'\Sigma_{\theta, P}k_{\theta, P}}}\le t\right)
- \Phi\left(t\right)\right|
\to 0.
\end{equation*}
\end{assumption}

\begin{assumption}\label{khat_assump}
  $\hat k-k_{\theta, P}=o_P(1)$ uniformly over $(\theta, P)\in\mathcal{S}_n$, and
  similarly for $\hat a$, $\hat \Gamma$, $\hat H$ and $\hat \Sigma$.
  Furthermore, $k_{\theta, P}$, $a_{\theta, P}$, $\Gamma_{\theta, P}$, $H_{\theta}$
  and $\Sigma_{\theta, P}$ are bounded uniformly over
  $(\theta, P)\in\mathcal{S}_n$. In the case where
  $\mathcal{C}=\widetilde{\mathcal{C}}\times \mathbb{R}^{d_{g_2}}$, assume that
  the last $d_{g_2}$ elements of $\hat k$ are zero with probability one for all
  $P\in\mathcal{P}$.
\end{assumption}

\begin{theorem}\label{khat_CI_thm}
  Suppose that \Cref{theta_init_assump,g_clt_assump_append,,khat_assump} hold
  and let $\hat c$ be defined above with $\hat k$, $\hat \Gamma$ and $\hat H$
  satisfying~\eqref{finite_bias_est_eq}. Then
\begin{equation*}
  \liminf_{n\to\infty}\inf_{(\theta,P)\in \mathcal{S}_n}
  P(h(\theta)\in \hor{\hat{c}, \infty})\ge 1-\alpha,
\end{equation*}
and
\begin{multline*}
  \limsup_{n\to\infty}\sup_{P\in\mathcal{P}}\sup_{\theta\in\Theta_I(P;\mathbb{R}^{d_\theta}\times \mathcal{D}, \Theta_n)}
\left\{\sqrt{n}q_{\beta, P}(h(\theta)-\hat c) \Big. \right.  \\
\left.
-\left[\maxbias_{\mathcal{C}}(k_{\theta, P},0)-\minbias_{\mathcal{D}}(k_{\theta, P},0)
+(z_{1-\alpha}+z_\beta)\sqrt{k_{\theta, P}'\Sigma_{\theta, P}k_{\theta, P}}
\right]\right\}\le 0.
\end{multline*}
\end{theorem}
\begin{proof}
  If $\mathcal{C}=\widetilde{\mathcal{C}}\times \mathbb{R}^{d_{g_2}}$ with
  $\widetilde{\mathcal{C}}$ compact, the theorem can equivalently be stated as
  holding with $\hat{k}$ redefined to be the vector in $\mathbb{R}^{d_{g_1}}$
  that contains the first $d_{g_1}$ elements of the original sensitivity
  $\hat k$, and with other objects redefined similarly. Therefore, it suffices
  to consider the case where $\mathcal{C}$ is compact.

  Note that
  \begin{multline*}
    \sqrt{n}(\hat h-h(\theta))
    =H_\theta\sqrt{n}(\hat\theta_{\text{initial}}-\theta)+\hat k
    \sqrt{n}\hat{g}(\theta) +
    \hat{k} \sqrt{n}(\hat{g}(\hat\theta_{\text{initial}})-\hat g(\theta))
    +\hat a+o_P(1) \\
    =H_\theta\sqrt{n}(\hat\theta_{\text{initial}}-\theta)+
    \hat k \sqrt{n}(\hat{g}(\theta)-g_P(\theta)) +\hat k'c
    +\hat{k}\sqrt{n}\Gamma_{\theta, P}(\hat\theta_{\text{initial}}-\theta)
    +\hat a+o_P(1) \\
    =(H_\theta+k_{\theta, P}'\Gamma_{\theta, P})\sqrt{n}(\hat\theta_{\text{initial}}-\theta)
    +k_{\theta, P}'c +a_{\theta, P}
    +k_{\theta, P}'\sqrt{n}(\hat{g}(\theta)-g_P(\theta))+o_P(1),
  \end{multline*}
  where $c=\sqrt{n}g_P(\theta)$ and the $o_P(1)$ terms are uniform over
  $(\theta, P)\in\mathcal{S}_n$ (the last equality uses the fact that
  $\mathcal{C}$ is compact). By \Cref{khat_assump}
  and~\eqref{finite_bias_est_eq}, $H_\theta+k_{\theta, P}'\Gamma_{\theta, P}=0$
  so this implies
\begin{equation}\label{hhat_asymptotic_approximation_eq}
\sqrt{n}(\hat h-h(\theta))=k_{\theta, P}'c+a_{\theta, P}+k_{\theta, P}'\sqrt{n}(\hat g(\theta)-g_P(\theta))+o_P(1)
\end{equation}
uniformly over $(\theta, P)\in\mathcal{S}_n$.
By compactness of $\mathcal{C}$ and \Cref{khat_assump}, we also have
\begin{equation*}
  \maxbias_{\mathcal{C}}(\hat k, \hat a)=\maxbias_{\mathcal{C}}(k_{\theta, P}, a_{\theta, P})+o_P(1),
  \quad \hat k'\hat\Sigma\hat k=k_{\theta, P}'\Sigma_{\theta, P} k_{\theta, P}+o_P(1)
\end{equation*}
uniformly over $(\theta, P)\in\mathcal{S}_n$.  Thus,
\begin{multline}\label{oneside_asymptotic_eq}
  \sqrt{n}(\hat c-h(\theta))
  =\sqrt{n}(\hat{h}-h(\theta))-\maxbias_{\mathcal{C}}(\hat k,
  \hat{a})-z_{1-\alpha}\sqrt{\hat k'\hat\Sigma\hat k}
  =k_{\theta, P}'c+\\
  a_{\theta, P}+k_{\theta, P}'\sqrt{n}(\hat g(\theta)-g_P(\theta))
  -\maxbias_{\mathcal{C}}(k_{\theta, P}, a_{\theta,
    P})-z_{1-\alpha}\sqrt{k_{\theta, P}'\Sigma_{\theta, P} k_{\theta, P}}+o_P(1)
\end{multline}
uniformly over $(\theta, P)\in\mathcal{S}_n$.
Since $k_{\theta, P}'c+a_{\theta, P}-\maxbias_{\mathcal{C}}(k_{\theta, P}, a_{\theta, P})\le 0$ by definition, the first part of the theorem (coverage) now follows from \Cref{g_clt_assump_append}.
For the last part of the theorem, note that, using the above display and the fact that $k_{\theta, P}'c+a_{\theta, P}\ge \minbias_{\mathcal{D}}(k_{\theta, P}, a_{\theta, P})$ for any $(\theta, P)$ with $c=\sqrt{n}E_{P}g(w_i, \theta)\in\mathcal{D}$, it follows that
$\sqrt{n}(h(\theta)-c)$ is less than or equal to
\begin{equation*}
\maxbias_{\mathcal{C}}(k_{\theta, P}, a_{\theta, P})-\minbias_{\mathcal{D}}(k_{\theta, P}, a_{\theta, P})
+z_{1-\alpha}\sqrt{k_{\theta, P}'\Sigma_{\theta, P} k_{\theta, P}}
+k_{\theta, P}'\sqrt{n}(\hat g(\theta)-g_P(\theta))+o_P(1)
\end{equation*}
uniformly over $(\theta, P)$ with $\sqrt{n}E_{P}g(w_i, \theta)\in\mathcal{D}$.
This, along with \Cref{g_clt_assump_append}, gives the last part of
the theorem.
\end{proof}

\begin{theorem}\label{khat_twoside_thm}
  Suppose that \Cref{theta_init_assump,g_clt_assump_append,,khat_assump} hold
  and let $\hat h$ and $\hat \chi$ be defined above with $\hat k$, $\hat a$,
  $\hat \Gamma$ and $\hat H$ satisfying~\eqref{finite_bias_est_eq}
  and~\eqref{centered_a_eq}. Then
\begin{equation*}
\liminf_{n\to\infty}\inf_{(\theta, P)\in\mathcal{S}_n}
P\left(h(\theta)\in \{\hat h \pm \hat \chi\}\right)\ge 1-\alpha.
\end{equation*}
In addition, we have
\begin{equation*}
  \sqrt{n}\hat \chi -\cv_\alpha\left(\frac{\maxbias_{\mathcal{C}}(k_{\theta, P}, a_{\theta, P})}{\sqrt{k_{\theta, P}'\Sigma_{\theta, P} k_{\theta, P}}}\right)\sqrt{k_{\theta, P}'\Sigma_{\theta, P} k_{\theta, P}}
 \stackrel{p}{\to} 0
\end{equation*}
uniformly over $(\theta, P)\in \mathcal{S}_n$.
\end{theorem}
\begin{proof}
  As with \Cref{khat_CI_thm}, it suffices to consider the case where
  $\mathcal{C}$ is compact. Let $(\theta_n,P_n)$ be a sequence in
  $\mathcal{S}_n$ and let $c_n=\sqrt{n}g_{P_n}(\theta_n)$. Let
  $b_n=k_{\theta_n,P_n}'c_n+a_{\theta_n,P_n}$,
  $\text{sd}_n=\sqrt{k_{\theta_n,P_n}'\Sigma_{\theta_n,P_n} k_{\theta_n,P_n}}$
  and $\overline b_n=\maxbias_{\mathcal{C}}(k_{\theta_n,P_n}, a_{\theta_n,P_n})$.
  Note that, by~\eqref{centered_a_eq},
  $\maxbias_{\mathcal{C}}(k_{\theta_n,P_n}, a_{\theta_n, P_n})=-\minbias_{\mathcal{C}}(k_{\theta_n,P_n}, a_{\theta_n,P_n})$
  when \Cref{khat_assump} holds. It therefore follows that
  $-\overline b_n\le b_n\le \overline b_n$.

  Let
  $Z_n = \sqrt{n}k_{\theta_n,P_n}'(\hat{g}(\theta_n)-g_{P_n}(\theta_n)) /
  \text{sd}_{n}$. Note that $Z_n$ converges in distribution (under $P_n$) to a
  $\mathcal{N}(0,1)$ random variable by \Cref{g_clt_assump_append}.
  By~\eqref{hhat_asymptotic_approximation_eq},
\begin{equation*}
\sqrt{n}(\hat h-h(\theta_n))
=b_n+\text{sd}_n Z_n+o_{P_n}(1).
\end{equation*}
Using the fact that $\text{sd}_n$ is bounded away from zero and $\sqrt{\hat k'\hat \Sigma\hat k}/\text{sd}_n$ converges in probability to one under $P_n$, it also follows that
\begin{equation*}
\sqrt{n}(\hat h-h(\theta_n))/\sqrt{\hat k'\hat \Sigma\hat k}=b_n/\text{sd}_n + Z_n+o_{P_n}(1).
\end{equation*}
Also, by \Cref{khat_assump}, we have, for a large enough constant $K$,
\begin{equation*}
  \left|\cv_\alpha\left(\frac{\maxbias_{\mathcal{C}}(\hat k, \hat a)}{
        \sqrt{\hat k'\hat \Sigma \hat k}}\right)
    -\cv_\alpha\left(\frac{\overline b_n}{\text{sd}_n}\right)\right|
  \le K\left\{\left[\maxbias_{\mathcal{C}}(\hat k, \hat a)
      - \overline b_n\right]
    +\left[\sqrt{\hat k'\hat \Sigma \hat k}-\text{sd}_n\right]\right\}
  \stackrel{p}{\to} 0.
\end{equation*}
This, along with the fact that $\sqrt{\hat k'\hat \Sigma\hat k}/\text{sd}_n$ converges in probability to one under $P_n$, gives the second part of the theorem.
Furthermore, it follows from the above display that
\begin{multline*}
   P_n\left(h(\theta_n)>\hat h+\hat \chi
  \right)
   =P_n\left(\frac{\sqrt{n}\left(\hat h-h(\theta_n)\right)}{\sqrt{\hat k'\hat \Sigma \hat k}}<-\cv_\alpha\left(\maxbias_{\mathcal{C}}(\hat k, \hat a)/\sqrt{\hat k'\hat \Sigma \hat k}\right)\right) \\
  =P_n\left(b_n/\text{sd}_n+Z_n<-\cv_\alpha\left(\overline{b}_n/\text{sd}_n\right)+o_{P_n}(1)\right)
=\Phi(-b_n/\text{sd}_n-\cv_\alpha\left(\overline{b}_n/\text{sd}_n\right))+o(1).
\end{multline*}
Similarly,
\begin{multline*}
  P_n\left(h(\theta_n)<\hat h-\hat \chi
  \right)
=P_n\left(\frac{\sqrt{n}\left(\hat h-h(\theta_n)\right)}{\sqrt{\hat k'\hat \Sigma \hat k}}>\cv_\alpha\left(\maxbias_{\mathcal{C}}(\hat k, \hat a)/\sqrt{\hat k'\hat \Sigma \hat k}\right)\right) \\
  =P_n\left(b_n/\text{sd}_n+Z_n>\cv_\alpha\left(\overline{b}_n/\text{sd}_n\right)+o_{P_n}(1)\right)
=1-\Phi(-b_n/\text{sd}_n+\cv_\alpha\left(\overline{b}_n/\text{sd}_n\right))+o(1).
\end{multline*}
Thus, the probability of the CI not covering is given, up to $o(1)$, by
\begin{equation*}
  1-\Phi(-b_n/\text{sd}_n+\cv_\alpha\left(\overline{b}_n/\text{sd}_n\right))
  +\Phi(-b_n/\text{sd}_n-\cv_\alpha\left(\overline{b}_n/\text{sd}_n\right)).
\end{equation*}
This is the probability that the absolute value of a $\mathcal{N}(b_n/\text{sd}_n,1)$ variable is greater than $\cv_\alpha\left(\overline{b}_n/\text{sd}_n\right)$, which is less than $1-\alpha$ since $|b_n|\le \overline b_n$.
\end{proof}

We now specialize to the case where the optimal weights are used. We make a
uniform consistency assumption on $\hat \Gamma$, $\hat H$ and $\hat \Sigma$, as
well as assumptions on the rank of $H$, $\Gamma$ and $\Sigma$. The latter are
standard regularity conditions for the correctly specified ($\mathcal{C}=\{0\}$)
case.

\begin{assumption}\label{GammaHSigma_assump}
  The estimators $\hat \Gamma$, $\hat H$ and $\hat \Sigma$ are full rank with
  probability one and satisfy $\hat \Gamma-\Gamma_{\theta, P}=o_P(1)$,
  $\hat H-H_\theta=o_P(1)$ and $\hat \Sigma-\Sigma_{\theta, P}=o_P(1)$ uniformly
  over $(\theta, P)\in\mathcal{S}_n$.
\end{assumption}

\begin{assumption}\label{scriptB_assump}
  There exists a compact set $\mathcal{B}$ that contains the set
  $\{(H_\theta, \Gamma_{\theta, P}, \Sigma_{\theta, P}): \theta\in\Theta_n,
  P\in\mathcal{P}\}$ for all $n$, such that (i) in the case where
  $\mathcal{C}$ is compact, $H\ne 0$ and $\Gamma$ and $\Sigma$ are full rank for
  any $(H, \Gamma, \Sigma)\in\mathcal{B}$ or (ii) in the case where
  $\mathcal{C}=\widetilde{\mathcal{C}}\times \mathbb{R}^{d_{g_2}}$ with
  $\widetilde{\mathcal{C}}$ compact, the same holds for the sub-matrices
  corresponding to the first $d_{g_1}$ moments.
\end{assumption}

Using these assumptions, we can verify that \Cref{khat_assump} holds
with weights $k_{\theta, P}$ that achieve the efficiency bound in
\Cref{oneside_efficiency_bound_thm} and nearly achieve the efficiency
bound in \Cref{twoside_efficiency_bound_thm}. This gives the following
results.

\begin{theorem}\label{optimal_khat_oneside_thm}
  Suppose that
  \Cref{theta_init_assump,g_clt_assump_append,,GammaHSigma_assump,scriptB_assump}
  hold and let $\hat c$ be defined above with
  $\hat k=k(\delta_\beta, \hat H, \hat \Gamma, \hat \Sigma)$. Then
  \begin{equation*}
    \liminf_{n\to\infty}\inf_{(\theta, P)\in \mathcal{S}_n}
    P(h(\theta)\in\hor{\hat{c}, \infty})\ge 1-\alpha
  \end{equation*}
  and
\begin{equation*}
  \limsup_{n\to\infty}
  \sup_{P\in\mathcal{P}}\sup_{\theta\in\Theta_I(P;\mathbb{R}^{d_\theta}\times \mathcal{D}, \Theta_n)}
  \left[\sqrt{n}q_{\beta, P}(h(\theta)-\hat c)
    -\omega(\delta_\beta;\mathbb{R}^{d_\theta}\times \mathcal{C}, \mathbb{R}^{d_\theta}\times
    \mathcal{D}, H_\theta, \Gamma_{\theta, P}, \Sigma_{\theta, P})
  \right]\le 0.
\end{equation*}
\end{theorem}
\begin{proof}
  In the case where $\mathcal{C}$ is compact, it follows from
  \Cref*{argmax_continuity_lemma} in \Cref*{optimal_weights_continuity_sec},
  $k(\delta, H,\Gamma, \Sigma)$ is continuous on $\{\delta\}\times\mathcal{B}$.
  Since $\mathcal{B}$ is compact, this means that $k(\delta, H,\Gamma, \Sigma)$
  is uniformly continuous. Thus, \Cref{GammaHSigma_assump} implies that $\hat k$
  satisfies \Cref{khat_assump} with
  $k_{\theta, P}=k(\delta, H_\theta, \Gamma_{\theta, P}, \Sigma_{\theta, P})$.
  Furthermore, $\hat k$ satisfies~\eqref{finite_bias_est_eq} by assumption. By
  properties of the modulus \citep[Equation (24) in][]{ArKo18optimal},
\begin{multline*}
\maxbias_{\mathcal{C}}(k_{\theta, P},0)-\minbias_{\mathcal{D}}(k_{\theta, P},0)
+(z_{1-\alpha}+z_\beta)\sqrt{k_{\theta, P}'\Sigma_{\theta, P}k_{\theta, P}} \\
=\omega(\delta_\beta;\mathbb{R}^{d_\theta}\times \mathcal{C}, \mathbb{R}^{d_\theta}\times \mathcal{D}, H_\theta, \Gamma_{\theta, P}, \Sigma_{\theta, P})
\end{multline*}
for this $k_{\theta, P}$.  Applying \Cref{khat_CI_thm} gives the result.

In the case where $\mathcal{C}=\widetilde{\mathcal{C}}\times \mathbb{R}^{d_{g_2}}$ with $\widetilde{\mathcal{C}}$ compact, the last $d_{g_2}$ elements of $\hat k$ are equal to zero as required by \Cref{khat_assump}, and the first $d_{g_1}$ elements are the same as the weights computed from the modulus problem with the last $d_{g_2}$ components thrown away and $H$, $\Gamma$ and $\Sigma$ redefined to be the sub-matrices corresponding to the first $d_{g_1}$ elements of the moments.  Thus, the same arguments apply in this case.
\end{proof}

For two-sided CIs, we consider weights
$\hat k=k(\delta^*(\hat{H}, \hat{\Gamma}, \hat{\Sigma}),
\hat{H}, \hat\Gamma, \hat\Sigma)$ given by~\eqref{khat_def_eq} with
$\mathcal{G}=\mathcal{F}=\mathbb{R}^{d_\theta}\times \mathcal{C}$, where
$\delta^*$ may depend on the data through $\hat H$, $\hat \Gamma$ and
$\hat \Sigma$. If $\delta^*$ is chosen to optimize CI length, it will be given
by $\delta_\chi(\hat H, \hat \Gamma, \hat\Sigma)$ where
\begin{equation}\label{flci_opt_delta_eq}
  \delta_\chi(H, \Gamma, \Sigma)=\argmin_\delta
  \cv_\alpha\left(\frac{\omega(\delta)}{2\omega'(\delta)}
  -\frac{\delta}{2}\right)\omega'(\delta)
\end{equation}
where $\omega(\delta)=\omega(\delta;\mathbb{R}^{d_\theta}\times \mathcal{C}, \mathbb{R}^{d_\theta}\times \mathcal{C}, H,\Gamma, \Sigma)$ is the single class modulus
\citep[see Section 3.4 in][]{ArKo18optimal}.

We make a continuity assumption on $\delta^*$.

\begin{assumption}\label{deltastar_continuous_assump}
$\delta^*$ is a continuous function of its arguments on the set $\mathcal{B}$ given in \Cref{scriptB_assump}.
\end{assumption}

\begin{theorem}\label{optimal_khat_twoside_thm}
  Suppose that
  \Cref{theta_init_assump,g_clt_assump_append,,GammaHSigma_assump,,scriptB_assump,deltastar_continuous_assump}
  hold and let $\hat h$ be defined above with
  $\hat k=k(\delta^*(\hat H, \hat \Gamma, \hat \Sigma), \hat{H}, \hat{\Gamma},
  \hat{\Sigma})$. Then the conclusion of \Cref{khat_twoside_thm} holds. If, in
  addition, $\delta^*=\delta_\chi(\hat H, \hat\Gamma, \hat\Sigma)$ for
  $\delta_\chi$ the CI length optimizing choice of $\delta$ given
  in~\eqref{flci_opt_delta_eq}, then the half-length $\hat\chi$ satisfies
  $\sqrt{n}\hat\chi=\chi(\theta, P)+o_P(1)$ uniformly over
  $(\theta, P)\in\mathcal{S}_n$, where
\begin{equation*}
  \chi(\theta, P)=\min_\delta \cv_\alpha\left(\frac{\omega(\delta)}{2\omega'(\delta)}-\frac{\delta}{2}\right)\omega'(\delta), \quad \omega(\delta)=\omega(\delta;\mathbb{R}^{d_\theta}\times \mathcal{C}, \mathbb{R}^{d_\theta}\times \mathcal{C}, H_{\theta}, \Gamma_{\theta, P}, \Sigma_{\theta, P}).
\end{equation*}
\end{theorem}
\begin{proof}
  The result follows from using the same arguments as in the proof of
  \Cref{optimal_khat_oneside_thm}, along with continuity of $\delta^*$,
  to verify \Cref{khat_assump}. The form of the limiting half-length
  for the optimal weights follows from properties of the modulus \citep[see
  Section 3.4 in][]{ArKo18optimal}.
\end{proof}

\subsection{Centrosymmetric case}\label{adaptation_sec_append}

\Cref{main_text_efficiency_bound_thm} in \Cref{efficiency_sec}
gives a bound for two-sided CIs in the case where $\mathcal{C}$ is
centrosymmetric. This follows from applying
\Cref{optimal_khat_twoside_thm,twoside_efficiency_bound_thm} in the
centrosymmetric case. In particular, comparing the asymptotic length in
\Cref{optimal_khat_twoside_thm} to the bound in
\Cref{twoside_efficiency_bound_thm} and using the fact that
$\omega(\delta;\mathbb{R}^{d_\theta}\times
\mathcal{C}, \{0\}, H_\theta, \Gamma_{\theta, P}, \Sigma_{\theta, P})=\omega(\delta;\{0\}, \mathbb{R}^{d_\theta}\times
\mathcal{C}, H_\theta, \Gamma_{\theta, P}, \Sigma_{\theta, P})=\frac{1}{2}\omega(2\delta;\mathbb{R}^{d_\theta}\times
\mathcal{C}, \mathbb{R}^{d_\theta}\times
\mathcal{C}, H_\theta, \Gamma_{\theta, P}, \Sigma_{\theta, P})$ when $\mathcal{C}$ is
centrosymmetric gives the bound
$\kappa_{*}(H_{\theta}, \Gamma_{\theta, P_0}, \Sigma_{\theta, P_0}, \mathcal{C})$ from
the statement of \Cref{main_text_efficiency_bound_thm}. This corresponds
to the bound in Corollary 3.3 of \citet{ArKo18optimal}. The universal lower
bound for $\kappa_{*}$ follows from the following result:
  \begin{theorem}\label{theorem:sharp-unversal-bound}
    For any $H, \Gamma, \Sigma$ and $\mathcal{C}$, the efficiency $\kappa_{*}$
    given in~\eqref{kappa_conv_cs_eq} is lower bounded by
    \begin{equation*}
      (z_{1-\alpha}(1-\alpha)-\tilde{z}_{\alpha}\Phi(\tilde{z}_{\alpha})+
      \phi(z_{1-\alpha})-\phi(\tilde{z}_{\alpha}))/ z_{1-\alpha/2}
    \end{equation*}
    where $\tilde{z}_{\alpha}=z_{1-\alpha}-z_{1-\alpha/2}$ and $\Phi$ and $\phi$
    denote the standard normal cdf, and pdf respectively. The lower bound is sharp in the sense that it holds with
    equality if $\omega(\delta)=K_{0}\min\{\delta,2z_{1-\alpha/2}\}$, for some
    constant $K_{0}$.
  \end{theorem}
  \begin{proof}
    Since $\cv_{\alpha}(b)\leq b+z_{1-\alpha/2}$, the
    denominator in~\eqref{kappa_conv_cs_eq} is upper-bounded by
    \begin{multline}\label{eq:denom-bound}
     \min_{\delta} 2\cv_{\alpha}\left(\frac{\omega(\delta)}{2
          \omega'(\delta)} -\frac{\delta}{2} \right)
      \omega'(\delta) \leq \\2\cv_{\alpha}\left(\frac{          \omega(2z_{1-\alpha/2})}{2 \omega'(2z_{1-\alpha/2})} -z_{1-\alpha/2}
      \right) \omega'(2z_{1-\alpha/2})\leq \omega(2z_{1-\alpha/2}).
    \end{multline}

    On the other hand, the numerator in~\eqref{kappa_conv_cs_eq}
    can be decomposed as
    \begin{multline*}
      (1-\alpha)E\left[\omega(2(z_{1-\alpha}-Z))\mid Z\leq z_{1-\alpha} \right]
      = E\left[\omega(2(z_{1-\alpha}-Z))\1{Z\leq z_{1-\alpha}-z_{1-\alpha/2}} \right]\\
      +E\left[\omega(2(z_{1-\alpha}-Z))\1{z_{1-\alpha}-z_{1-\alpha/2}\leq Z\leq
          z_{1-\alpha}} \right].
    \end{multline*}
    Since the modulus $\omega(\delta)$ is non-decreasing, the first summand is
    lower-bounded by
    \begin{equation*}
      E\left[\omega(2z_{1-\alpha/2})\1{Z\leq z_{1-\alpha}-z_{1-\alpha/2}}
      \right]= \omega(2z_{1-\alpha/2})\Phi(z_{1-\alpha}-z_{1-\alpha/2}).
    \end{equation*}
    Since the modulus $\omega(\delta)$ is concave,
    $\omega(2(z_{1-\alpha}-Z))\geq
    (z_{1-\alpha}-Z)/z_{1-\alpha/2}\cdot\omega(2z_{1-\alpha/2})$, so that the
    second summand is lower-bounded by
    \begin{multline*}
      \frac{\omega(2z_{1-\alpha/2})}{z_{1-\alpha/2}}E\left[(z_{1-\alpha}-Z)\1{z_{1-\alpha}-z_{1-\alpha/2}\leq
          Z\leq
          z_{1-\alpha}} \right]\\
      = \frac{\omega(2z_{1-\alpha/2})}{z_{1-\alpha/2}}
      \left(z_{1-\alpha}(1-\alpha-\Phi(z_{1-\alpha}-z_{1-\alpha/2}))
        +\phi(z_{1-\alpha})-\phi(z_{1-\alpha}-z_{1-\alpha/2}) \right),
    \end{multline*}
    where the equality follows by the formula for the expectation of a truncated
    normal random variable. Combining the two preceding displays then yields
    \begin{multline}\label{eq:num-bound}
      (1-\alpha)E\left[\omega(2(z_{1-\alpha}-Z))\mid Z\leq z_{1-\alpha} \right]\\
      \geq
      \omega(2z_{1-\alpha/2})\frac{z_{1-\alpha}(1-\alpha)-\tilde{z}_{\alpha}
        \Phi(\tilde{z}_{\alpha})
        +\phi(z_{1-\alpha})-\phi(\tilde{z}_{\alpha})}{z_{1-\alpha/2}},
    \end{multline}
    where $\tilde{z}_{\alpha}=z_{1-\alpha}-z_{1-\alpha/2}$. Combining this with
    the bound in~\eqref{eq:denom-bound} then yields the result. The sharpness of
    the bound for the case $\omega(\delta)=K_{0}\min\{\delta,2z_{1-\alpha/2}\}$
    follows from by noting that in this case, both~\eqref{eq:denom-bound}
    and~\eqref{eq:num-bound} hold as equalities.
  \end{proof}

For the one-sided case, we obtain the following bound.

\begin{theorem}\label{theorem:onesided-efficiency}
  Consider the setting of \Cref{optimal_khat_oneside_thm}, with
  $\mathcal{C}$ centrosymmetric. Then the weights
  $\hat k=\hat k(\delta_\beta, \hat H, \hat \Gamma, \hat \Sigma)$ with
  $\mathcal{D}=\mathcal{C}$ are identical to the weights
  $\hat k(\delta_{\tilde\beta}, \hat H, \hat \Gamma, \hat \Sigma)$ computed with
  $\mathcal{D}=\{0\}$, but with $\tilde \beta=\Phi((z_\beta-z_{1-\alpha})/2)$.
  Furthermore, letting $\hat c_{\operatorname{minimax}}$ denote the lower
  endpoint of the CI computed with these weights
  ($\hat k(\delta_\beta, \hat H, \hat \Gamma, \hat \Sigma)$ with
  $\mathcal{D}=\mathcal{C}$), we have
\begin{equation*}
  \limsup_{n\to\infty}
  \sup_{P\in\mathcal{P}}\sup_{\theta\in\Theta_I(P;\mathbb{R}^{d_\theta}\times \{0\}, \Theta_n)}
  \left\{\sqrt{n}q_{\beta, P}(h(\theta)-\hat c_{\operatorname{minimax}})
  -\frac{1}{2}\left[
  \omega_{\theta, P}(\delta_\beta)+\delta_\beta \omega'_{\theta, P}(\delta_\beta)
  \right]
  \right\}\le 0
\end{equation*}
where
$\omega_{\theta, P}(\delta)=\omega(\delta;\mathbb{R}^{d_\theta}\times
\mathcal{C}, \mathbb{R}^{d_\theta}\times
\mathcal{C}, H_\theta, \Gamma_{\theta, P}, \Sigma_{\theta, P})$. For $\hat c$
computed instead with $\mathcal{D}=\{0\}$, we obtain
\begin{equation*}
  \limsup_{n\to\infty}
  \sup_{P\in\mathcal{P}}\sup_{\theta\in\Theta_I(P;\mathbb{R}^{d_\theta}\times \{0\}, \Theta_n)}
\left\{\sqrt{n}q_{\beta, P}(h(\theta)-\hat c)
  -\frac{1}{2}
      \omega_{\theta, P}(2\delta_\beta)
\right\}\le 0.
\end{equation*}
\end{theorem}
\begin{proof}
The first statement follows from Corollary~3.2 in \citet{ArKo18optimal}.
The second statement follows from applying \Cref{khat_CI_thm} as in the proof of \Cref{optimal_khat_oneside_thm}, noting that $\minbias_{\{0\}}(k_{\theta, P},0)=0$, and using arguments from the proof of Corollary~3.2 in \citet{ArKo18optimal}.
The last statement follows from \Cref{optimal_khat_oneside_thm} and the fact that
$\omega(\delta;\mathbb{R}^{d_\theta}\times \mathcal{C}, \mathbb{R}^{d_\theta}\times \{0\}, H_\theta, \Gamma_{\theta, P}, \Sigma_{\theta, P})=\frac{1}{2}\omega(2\delta;\mathbb{R}^{d_\theta}\times \mathcal{C}, \mathbb{R}^{d_\theta}\times \mathcal{C}, H_\theta, \Gamma_{\theta, P}, \Sigma_{\theta, P})$.
\end{proof}

Thus, directing power toward the correctly specified case yields the same
one-sided CI once one changes the quantile over which one optimizes excess
length. If one does attempt to direct power, the scope for doing so is bounded
by a factor of
\begin{equation}\label{eq:one-sided-efficiency}
\kappa_{*}^{\operatorname{OCI}, \beta}(H_{\theta}, \Gamma_{\theta, P_0}, \Sigma_{\theta, P_0}, \mathcal{C})
=\frac{\omega_{\theta, P}(2\delta_\beta)}{\omega_{\theta, P}(\delta_\beta)+\delta_\beta \omega'_{\theta, P}(\delta_\beta)}.
\end{equation}
This gives a bound for the one-sided case analogous to the bound $\kappa_{*}$
in~\eqref{kappa_conv_cs_eq} for two-sided CIs.

A consistent estimate of these bounds can be obtained by plugging in
$\omega(\delta;\mathbb{R}^{d_\theta}\times
\mathcal{C}, \mathbb{R}^{d_\theta}\times \mathcal{C}, \hat H, \hat \Gamma,
\hat{\Sigma})$ for
$\omega_{\theta, P}(\delta)=\omega(\delta;\mathbb{R}^{d_\theta}\times
\mathcal{C}, \mathbb{R}^{d_\theta}\times
\mathcal{C}, H_\theta, \Gamma_{\theta, P}, \Sigma_{\theta, P})$.
\Cref{tab:efficiency} reports estimates of this bound under different forms
of misspecification in the empirical application in \Cref{empirical_sec}.

\section{Global Misspecification}\label{global_misspec_sec_append}

We now describe two approaches to the construction of CIs that are robust to
global misspecification. The first approach is generally applicable, and, for a
one-sides CIs, yields CIs that are asymptotically equivalent under
local misspecification to the CIs proposed in the main text. The second approach
also exhibits this equivalence property for two-sided CIs, but the regularity
conditions it imposes may not be satisfied in all applications.

Before describing the procedures
in~\Cref{sec:cis-based-recent,sec:cis-based-missp} below, let us briefly
describe the setup. Under global misspecification, the true parameter $\theta_0$
satisfies
\begin{equation}\label{global_misspec_eq}
  g_P(\theta_0)=E_{P}g(w_i, \theta_0)=\cglob, \quad \cglob\in \Cglob,
\end{equation}
where $\Cglob$ is fixed with the sample size $n$. To accommodate both local and
global misspecification with the same notation, we consider a sequence
$\Cglob=\Cglob_n$ of sets. Under global misspecification $\Cglob_n$ is fixed
with $n$, whereas, under local misspecification,
$\Cglob_n=\mathcal{C}/\sqrt{n}=\{c/\sqrt{n}: c\in\mathcal{C}\}$ where
$\mathcal{C}$ is fixed with $n$. The rest of the setup is the same as the formal
setup in \Cref{efficiency_bound_sec_main,efficiency_sec_append}: we are
interested in a CI for $h(\theta)$ that satisfies the coverage
condition~\eqref{coverage_eq}, where
$\mathcal{S}_n=\{(\theta, P)\in \Theta_n\times\mathcal{P}: g_P(\theta)\in
\Cglob\}$ denotes the set of pairs $(\theta, P)$ such that $\theta$ is in the
identified set under $P$.

For concreteness, we focus on GMM estimators. We treat the weighting matrix $W$
as given, and construct CIs that are asymptotically equivalent to the CIs given
in \Cref{sec:asympt-line-estim} with sensitivity
$k'=-H(\Gamma'W\Gamma)^{-1}\Gamma' W$. To make these CIs optimal under local
misspecification, we need to choose the weighting matrix $W$ so that this
sensitivity is optimal under local misspecification. This can be done by
computing the optimal sensitivity $\hat{k}$ under local misspecification using
first stage estimates, following our implementation in
\Cref{sec:implementation}, and then computing an equivalent GMM weighting matrix
as described in \Cref{weighting_matrix_remark}.

\subsection{CIs based on recentering the moments}\label{sec:cis-based-recent}

Let $\mathcal{I}_\cglob$ be a family of CIs, indexed by $\cglob\in\Cglob$, such
that, for each $\cglob$, $\mathcal{I}_\cglob$ is asymptotically valid for the
GMM model defined by the moment function $\theta\mapsto g(w_i, \theta)-\cglob$.
We consider the CI $\mathcal{I}=\cup_{\cglob\in\Cglob}\mathcal{I}_\cglob$. Since
this CI contains a CI based on the moment conditions
$\theta\mapsto \hat g(\theta)-\cglob_0$, where $\cglob_0=E_{P}g(w_i, \theta_0)$
is the true value of $\cglob$, it will have correct asymptotic coverage under
standard conditions. (While we omit a formal statement, we note that this
follows by showing that $\mathcal{I}_{c_n}$ has correct coverage under a
drifting sequence of moment functions
$\tilde g_n(w_i, \theta)=g(w_i, \theta)-\cglob_n$ for sequences
$\cglob_n\in\Cglob$. This follows from the usual arguments for asymptotic
coverage of GMM estimators under correct specification.) As we discuss in
\Cref{global_computation_sec} below, this CI can be computed using constrained
optimization, where the objective function will typically be smooth so long as
$\theta\mapsto g(w_i, \theta)$ is smooth. However, since the problem involves
minimization of a GMM objective function, it will typically not be convex.

We now show that this approach can be used to construct a one-sided CI that is
asymptotically equivalent under local misspecification to the CI proposed in the
main text. We consider CIs based on the GMM estimator
$\hat\theta_{W, \cglob}= \arg\min_\theta [\hat g(\theta)-\cglob]'W
[\hat{g}(\theta)-\cglob]$ based on the moment function $g(w_i, \theta)-\cglob$.
Let $\hat\Gamma_\theta$ denote an estimate of the Jacobian matrix
$\Gamma_{\theta, P}=\frac{d}{d t'}E_{P}g(w_i, t)|_{t=\theta}$. Under local
misspecification, the identified set for $\theta$ shrinks toward a point, so
that the same estimate $\hat\Sigma$ is consistent for $\Sigma_{\theta, P}$
uniformly over the identified set. This is no longer the case under global
misspecification, and we instead need to use a class of estimates
$\hat\Sigma_{\theta}$ for $\Sigma_{\theta, P}$ indexed by $\theta$. In the
i.i.d.\ case, we can take
$\hat\Sigma_{\theta}=\frac{1}{n}\sum_{i=1}^{n} g(w_i, \theta)g(w_i, \theta)'$,
in which case the estimate of $\hat\Sigma$ in \Cref{sec:implementation}
corresponds to $\hat{\Sigma}=\hat\Sigma_{\hat{\theta}_{\text{initial}}}$. We
allow for the possibility that $W$ is data dependent, in which case we will
assume that it converges to some $W_{P}$, which may depend on the underlying
distribution $P$. Let
$\hat k_{\theta}'=-H_\theta\left( \hat\Gamma_{\theta}'W\hat\Gamma_{\theta}
\right)^{-1}\hat\Gamma_{\theta}'W$ denote an estimate of the sensitivity
$k_{\theta, P}'=-H_\theta \left(\Gamma_{\theta, P}'W_P\Gamma_{\theta, P}
\right)^{-1}\Gamma_{\theta, P}'W$. This gives the standard error
$\se_{\hat\theta_{W, \cglob}}$ where
$\se_{\theta}=\sqrt{\hat k_{\theta}'\hat\Sigma_{\theta}\hat k_{\theta}/n}$. The
nominal $1-\alpha$ one-sided CI based on the given $\cglob$ is then
$\hor{h(\hat\theta_{W, \cglob})-z_{1-\alpha}\se_{W, \cglob}, \infty}$. The lower
endpoint of a one-sided CI that is asymptotically valid for the set $\Cglob$ is
then given by $\hor{\hat c_{\text{glob}}, \infty}$ where
\begin{equation*}
  \hat c_{\text{glob}} =
  \inf_{\cglob\in\Cglob} \, \left[h(\hat\theta_{W, \cglob})-z_{1-\alpha}\se_{\hat\theta_{W, \cglob}} \right].
\end{equation*}
The CI proposed in the main text takes the form $\hor{\hat c_{\text{loc}}, \infty}$ where
\begin{equation*}
   \hat c_{\text{loc}} = h(\hat\theta_{\text{initial}}) + \hat k' \hat g(\hat\theta_{\text{initial}})
  - \sup_{c\in \mathcal{C}} \hat k' c/\sqrt{n} - z_{1-\alpha} \sqrt{\hat k\hat\Sigma\hat k}/\sqrt{n}.
\end{equation*}
We assume that the sensitivity $\hat k$ corresponds to the same GMM weighting
matrix, so that \Cref{theta_init_assump,,g_clt_assump_append,khat_assump} hold
with $k_{\theta, P}$, $\Gamma_{\theta, P}$, $H_\theta$ and $\Sigma_{\theta, P}$
given above. In addition, we use some further regularity conditions for the
estimator $\hat\theta_{W, \cglob}$.

\begin{assumption}\label{glob_misspec_ci_local_assump}
  The estimator $\hat\theta_{W, \cglob}$ satisfies
  \begin{equation*}
   \sup_{c\in\mathcal{C}}\left| h(\hat\theta_{W, c/\sqrt{n}})-h(\theta)
  - k_{\theta,P}'[\hat g(\theta)-c/\sqrt{n}] \right|=o_P(1/\sqrt{n})
  \end{equation*}
  and
  \begin{equation*}
    \sup_{c\in\mathcal{C}}\left| \hat k_{\hat\theta_{W, c/\sqrt{n}}}'\hat\Sigma_{\hat\theta_{W, c/\sqrt{n}}}\hat k_{\hat\theta_{W, c/\sqrt{n}}}- k_{\theta, P}'\Sigma_{\theta, P}k_{\theta, P} \right|=o_P(1)
  \end{equation*}
  uniformly over $(\theta, P)\in\mathcal{S}_n$, where $\Cglob=\Cglob_n=\mathcal{C}/\sqrt{n}$.
\end{assumption}
\Cref{glob_misspec_ci_local_assump} imposes an influence function representation
for the GMM estimator $\hat\theta_{W, c}$, and a uniform consistency condition
for the asymptotic variance estimator. We verify this assumption under primitive
conditions in the linear IV setting in
\Cref*{global_misspec_linear_iv_sec_append}. Note that, while the CI is robust
to global misspecification, \Cref{glob_misspec_ci_local_assump} imposes
conditions that are required only under local misspecification.

\begin{theorem}\label{glob_misspec_ci_local_thm}
  Let $\hat c_{\text{loc}}$ and $\hat c_{\text{glob}}$ be given above with
  $\Cglob=\mathcal{C}/\sqrt{n}$, where $\mathcal{C}$ is a compact set. Suppose
  that \Cref{theta_init_assump,g_clt_assump_append,,khat_assump} hold for
  $\hat c_{\text{loc}}$, and \Cref{glob_misspec_ci_local_assump} holds for
  $\hat c_{\text{glob}}$, with the same $k_{\theta, P}$ and $\Sigma_{\theta, P}$,
  and assume that, for each $P\in\mathcal{P}$, there exists $\theta\in\Theta_n$
  such that $(\theta, P)\in \mathcal{S}_n$ (i.e.\ the identified set for $\theta$
  under each $P\in\mathcal{P}$ is nonempty). Then
  $\sqrt{n}(\hat c_{\text{loc}}-\hat c_{\text{glob}})$ converges in probability
  to zero uniformly over $P\in\mathcal{P}$.
\end{theorem}
\begin{proof}
Under these assumptions,
\begin{multline*}
  \hat c_{\text{glob}}=\inf_{c\in\mathcal{C}} h(\theta) + k_{\theta, P}'\hat g(\theta) - k'_{\theta,P}c/\sqrt{n}
  - z_{1-\alpha}\sqrt{k_{\theta, P}'\Sigma_{\theta, P}k_{\theta, P}}/\sqrt{n}+o_P(1/\sqrt{n}) \\
  = h(\theta) + k_{\theta, P}'\hat g(\theta) - \maxbias_{\mathcal{C}}(k_{\theta,P},0)/\sqrt{n}
  - z_{1-\alpha}\sqrt{k_{\theta, P}'\Sigma_{\theta, P}k_{\theta, P}}/\sqrt{n}+o_P(1/\sqrt{n})
\end{multline*}
uniformly over $(\theta, P)\in\mathcal{S}_n$. The result follows by noting that
this matches the asymptotic expression for $\hat c_{\text{loc}}$ given in
\Cref{oneside_asymptotic_eq} in the proof of \Cref{khat_CI_thm}.
\end{proof}

The CI $\hor{\hat c_{\text{glob}}, \infty}$ is valid under global
misspecification under standard conditions, and \Cref{glob_misspec_ci_local_thm}
shows that this CI is asymptotically equivalent to the CI
$\hor{\hat c_{\text{loc}}, \infty}$ under local misspecification.

\subsubsection{Computation}\label{global_computation_sec}

The problem of computing this CI can be written as a nested optimization
problem, in which one minimizes the GMM objective function for a given $\cglob$,
and then minimizes the lower endpoint of the CI over $\cglob$.  Alternatively,
in the spirit of recent papers on MPEC \citep[e.g.][]{dube_improving_2012},
one can write this as an optimization problem over $\theta, \cglob$ subject to
the constraint that $\theta$ minimizes the GMM objective under $\cglob$:
\begin{equation*}
  \min_{\theta, \cglob} h(\theta) - z_{1-\alpha} \se_{\theta}
  \quad\text{s.t.}\quad \theta=\arg\min_{t} \left[ \hat g(t)-\cglob \right]'W \left[ \hat g(t)-\cglob \right], \quad \cglob\in\Cglob.
\end{equation*}
In the case where $\hat g(\theta)$ is smooth, one may also relax the constraint
that $\theta$ minimizes the GMM objective by instead imposing only the first
order conditions:
\begin{equation*}
  \min_{\theta, \cglob} h(\theta) - z_{1-\alpha} \se_{\theta}
  \quad\text{s.t.}\quad
  \hat\Gamma_\theta' W\left[ \hat g(\theta)-\cglob \right]=0,
  \quad \cglob\in\Cglob.
\end{equation*}
Since a constraint is relaxed, this can only make the resulting CI more conservative.

\subsection{CIs based on misspecification-robust standard errors}\label{sec:cis-based-missp}

In some cases, it may be possible to construct estimates $\hat{B}$ of the
worst-case asymptotic bias of the estimator $\hat{h}$ that are asymptotically
normal. In such cases, one can construct CIs that are valid under global
misspecification by using misspecification-robust standard errors.

Let $\theta^{*}_{c}=\argmin_{\theta}(g_{P}(\theta)-c)'W(g_{P}(\theta)-c)$, so
that under the model~\eqref{global_misspec_eq}, the true parameter is given by
$\theta_{0}=\theta^{*}_{\cglob}$, and the pseudo-true parameter, the estimand of
the GMM estimator $\hat{\theta}$ with weighting matrix $W$, is given by
$\theta^{*}_{0}$. Note that $\theta^{*}_{c}$ depends on $P$, but we leave this
dependence implicit for clarity of notation. \citet{HaIn03} show that under
regularity conditions, if $\Cglob$ is fixed with $n$,
\begin{equation*}
  \sqrt{n}(\hat{h}-h(\theta_{0}^{*}))\indist \mathcal{N}(0,\Omega_{h}),
\end{equation*}
where, unlike in the correctly specified case, the asymptotic variance
$\Omega_{h}$ generally depends on the weighting matrix $W$. \citet{HaIn03} also
show how to construct consistent estimates of $\Omega_{h}$. Let
$B_{\Cglob}=\max_{\cglob\in\Cglob}\abs{h(\theta^{*}_{0})-h(\theta^{*}_{\cglob})}$
denote the worst-case bias. Suppose that we have available an estimator
$\hat{B}$ of $B_{\Cglob}$, such that $(\hat{h},\hat{B})$ are jointly
asymptotically normal. Then, by the delta method, for some asymptotic variances
$\Omega_{+}$ and $\Omega_{-}$, it holds that
\begin{align}\label{eq:global_misspec_normal}
  \frac{\sqrt{n}(\hat{h}-\hat{B}-(h(\theta^{*}_{0})-B_{\Cglob}))}{\Omega^{1/2}_{-}}
  & \indist\mathcal{N}(0, 1),
  & \frac{\sqrt{n}(\hat{h}
    +\hat{B}-(h(\theta^{*}_{0})+B_{\Cglob}))}{\Omega^{1/2}_{+}}
  & \indist\mathcal{N}(0,1).
\end{align}
Note that $B_{\Cglob}$, $\Omega_{+}$ and $\Omega_{-}$ may depend on $P$,
although we leave this implicit in the notation.

If one has available estimators $\hat{\Omega}_{+}$ and $\hat{\Omega}_{-}$ that satisfy
\begin{align}\label{eq:global_misspec_variance}
  \hat{\Omega}_{+}/\Omega_{+}&\inprob 1,&
  \hat{\Omega}_{-}/\Omega_{-}&\inprob 1,
\end{align}
one can construct one-sided CIs as
$\hor{\hat{h}-\hat{B}-z_{1-\alpha} \hat{\Omega}_{-}^{1/2}/\sqrt{n}, \infty}$,
and $\hol{-\infty,\hat{h}+\hat{B}+z_{1-\alpha}
  \hat{\Omega}_{+}^{1/2}/\sqrt{n}}$. Two-sided CIs can be constructed as
\begin{equation*}
  \widetilde{CI}= \left[\hat{h}-
    \cv_{\alpha}(\sqrt{n}\hat{B}/\hat{\Omega}_{-}^{1/2})\cdot \hat{\Omega}_{-}^{1/2}/\sqrt{n},
    \hat{h}+\cv_{\alpha}(\sqrt{n}\hat{B}/\hat{\Omega}_{+}^{1/2})\cdot \hat{\Omega}_{+}^{1/2}/\sqrt{n}
  \right].
\end{equation*}
The next result shows that these CIs are valid under global misspecification.

\begin{theorem}\label{theorem:global_misspec_validity}
  Suppose that the convergence in~\eqref{eq:global_misspec_normal}
  and~\eqref{eq:global_misspec_variance} holds uniformly over
  $(\theta,P)\in \mathcal{S}_{n}=\{(\theta,P)\in \Theta_n\times\mathcal{P}:
  g_P(\theta)\in \Cglob\}$, with $\Cglob$ is fixed. Then
  \begin{equation*}
    \liminf_{n\to\infty}\inf_{(\theta,P)\in\mathcal{S}_{n}}P(h(\theta)\in
\hor{\hat{h}-\hat{B}-z_{1-\alpha} \hat{\Omega}_{-}^{1/2}/\sqrt{n}, \infty})\geq 1-\alpha.
  \end{equation*}
  Suppose, in addition, that $\hat{B}/B_{\Cglob}\inprob 1$ and
  $\abs{\Omega_{+}-\Omega_{-}}/\sqrt{n}B_{\Cglob}\to 0$ uniformly over
  $(\theta,P)\in \mathcal{S}_{n}$. Then
  \begin{equation*}
    \liminf_{n\to\infty}\inf_{(\theta,P)\in\mathcal{S}_{n}}P(h(\theta)\in\widetilde{CI})\geq 1-\alpha.
  \end{equation*}
\end{theorem}
The proof of this \namecref{theorem:global_misspec_validity} is deferred to
\Cref{sec:auxiliary-results}. Under global misspecification, when
$\sqrt{n}B_{\Cglob}\to\infty$, the condition
$\abs{\Omega_{+}-\Omega_{-}}/\sqrt{n}B_{\Cglob}\to 0$ holds if $\Omega_{+}$ and
$\Omega_{-}$ are of the same order, which is typically the case. In this case,
$\widetilde{CI}$ is asymptotically equivalent to the CI
$[\hat{h}-\hat{B}-z_{1-\alpha}\hat{\Omega}_{-}/\sqrt{n},\hat{h}+\hat{B}+
z_{1-\alpha}\hat{\Omega}_{+}/\sqrt{n}]$. Since in large samples, the uncertainty
about the endpoints of the identified set
$[h(\theta^{*}_{0})-B_{\Cglob},h(\theta^{*}_{0})+B_{\Cglob}]$ dominates by the
uncertainty about the location of the endpoints, it suffices to use a one-sided
critical value $z_{1-\alpha}$ (see \citealp{im04}, for a discussion).

Under local misspecification, if the estimator $\hat{\theta}$ is asymptotically
linear with sensitivity $k$,
$\sqrt{n}B_{\Cglob}=\maxbias_{\sqrt{n}\mathcal{\Cglob}}(k)$ is bounded, so that
the condition $\abs{\Omega_{+}-\Omega_{-}}/\sqrt{n}B_{\Cglob}\to 0$ holds if
$\Omega_{+}-\Omega_{-}\to 0$. This is indeed the case if
$\sqrt{n}\hat{B}=\maxbias_{\sqrt{n}\mathcal{\Cglob}}(k)+o_{p}(1)$, since then
$\Omega_{+}$ and $\Omega_{-}$ both equal $k'\Sigma k$. In this case
$\widetilde{CI}$ is asymptotically equivalent to the CI in \Cref{h_ci_eq}. The
CI thus automatically adapts to the misspecification magnitude.

\begin{example}[linear IV model]\label{example:linear_iv_glob_misspec}
  To give an example of a setting in which the
  condition~\eqref{eq:global_misspec_normal} holds, consider the linear
  instrumental variables (IV) model from \Cref{iv_sec}. In particular, suppose
  $h=H\theta$, and suppose that
  $g(\theta)=E[z_{i}(y_{i}-x_{i}'\theta)]\in \Cglob$, with
  $\Cglob=\{\sqrt{n} B\gamma\colon\norm{\gamma}\leq M_{n}\}$,
  $B=E[z_{i}z_{Ii}']$, and for concreteness, suppose $\norm{\cdot}$ corresponds
  to an $\ell_{2}$ norm. If $M_{n}=M$ is fixed, this reduces to the local
  misspecification setup in the main text, but if $M_{n}=\sqrt{n}M$, the
  misspecification is global. Consider the 2SLS estimator
  $\hat{h}=\hat{k}'\sum_{i=1}^{n}z_{i}y_{i}$, with
  $\hat{k}=-H(\hat{\Gamma}'\hat{W}\hat{\Gamma})^{-1}\hat{\Gamma}'\hat{W}$,
  $\hat{\Gamma}=-n^{-1}\sum_{i=1}^{n}z_{i}x_{i}'$ and
  $\hat{W}^{-1}=\sum_{i=1}^{n}z_{i}z_{i}'$; let $\Gamma=-E[z_{i}x_{i}']$,
  $W=E[z_{i}z_{i}']^{-1}$, and $k=-H(\Gamma'W\Gamma)^{-1}\Gamma'W$.

  Then $h(\theta^{*}_{0})=H\theta+k'E[z_{i}z_{Ii}']\gamma/\sqrt{n}$, and
  $B_{\Cglob}=\norm{k'E[z_{i}z_{Ii}']}M_{n}/\sqrt{n}$. Consider the estimator
  $\hat{B}=\norm{\hat{k}'\frac{1}{n}\sum_{i=1}^{n}z_{i}z_{Ii}'}M_{n}/\sqrt{n}$
  of the worst-case bias, which is the same as the estimator
  $\maxbias_{\sqrt{n}\Cglob}(\hat{k})/\sqrt{n}$ under local misspecification
  used in the main text. Since $\hat{B}$ and $\hat{h}$ depend on the data only
  through the sample means
  $S_{n}=n^{-1}(\sum_{i}z_{i}z_{i}',\sum_{i}x_{i}z_{i}',\sum_{i}x_{i}y_{i}')$,
  \Cref{eq:global_misspec_normal} holds by the delta method, and consistent
  estimates $\hat{\Omega}_{+}$ and $\hat{\Omega}_{-}$ of $\Omega_{+}$ and
  $\Omega_{-}$ can be constructed using a consistent estimator of the asymptotic
  variance of $S_{n}$, which yields the CI
  \begin{equation*}
    \widetilde{CI}=
    \left[\hat{h}-
      \cv_{\alpha}(\maxbias_{\sqrt{n}\Cglob}(\hat{k})/
      \hat{\Omega}_{-}^{1/2})\cdot \hat{\Omega}_{-}^{1/2}/\sqrt{n},
      \hat{h}+\cv_{\alpha}(\maxbias_{\sqrt{n}\Cglob}(\hat{k})
      /\hat{\Omega}_{+}^{1/2})\cdot \hat{\Omega}_{+}^{1/2}/\sqrt{n}
    \right].
  \end{equation*}
  Thus, relative to the CI described in the main text, $\widetilde{CI}$ differs
  only in that it uses variance estimates $\hat{\Omega}_{+}$ and
  $\hat{\Omega}_{-}$ that are valid under global misspecification. If $M_{n}=M$,
  so that misspecification is local, $\hat{\Omega}_{+}=k'\Sigma k+o_{p}(1)$ and
  $\hat{\Omega}_{+}=k'\Sigma k+o_{p}(1)$, and the CI is asymptotically
  equivalent to the CI described in the main text.
\end{example}

\subsubsection{Proof of
Theorem~\ref{theorem:global_misspec_validity}}\label{sec:auxiliary-results}

\begin{proof}
  Let
  $Z_{-}=\sqrt{n}(\hat{h}-\hat{B}-(h(\theta^{*}_{0})-B_{\Cglob}))/\Omega^{1/2}_{-}$,
  and let
  $Z_{+}=\sqrt{n}(\hat{h}+\hat{B}-(h(\theta^{*}_{0})+B_{\Cglob}))/\Omega^{1/2}_{+}$. Then
  \begin{multline*}
      P(h(\theta)\geq
      \hat{h}-\hat{B}-z_{1-\alpha}\hat{\Omega}_{-}^{1/2}/\sqrt{n}) \\=
      P(Z_{-}\leq
      \sqrt{n}(B_{\Cglob}-(h(\theta^{*}_{0})-h(\theta)))/\Omega_{-}^{1/2}+
      z_{1-\alpha}\hat{\Omega}_{-}^{1/2}/\Omega_{-}^{1/2})\\
      \geq      P(Z_{-}\leq
      z_{1-\alpha}\hat{\Omega}_{-}^{1/2}/\Omega_{-}^{1/2})\geq 1-\alpha+o(1),
  \end{multline*}
  where the equality follows by definition of $Z_{-}$, the first inequality
  follows since $B_{\Cglob}\geq h(\theta^{*}_{0})-h(\theta)$ by definition of
  $B_{\Cglob}$, and the second inequality follows since
  $\hat{\Omega}_{-}^{1/2}/\Omega_{-}^{1/2}\inprob 1$,
  $Z_{-}\indist \mathcal{N}(0,1)$, and since convergence in distribution to a
  continuous distribution implies uniform convergence of the cdfs \citep[Lemma
  2.11]{van_der_vaart_asymptotic_1998}. To show the result for the two-sided CI,
  let $b=h(\theta^{*})-h(\theta)$ denote the asymptotic bias. Then
\begin{equation*}
  \begin{split}
    P(h(\theta)\in\widetilde{CI}) &= P(
    \cv_{\alpha}(\sqrt{n}\hat{B}/\hat{\Omega}_{-}^{1/2})\cdot
    \hat{\Omega}_{-}^{1/2}/\sqrt{n} \geq \hat{h}-h(\theta)\geq
    -\cv_{\alpha}(\sqrt{n}\hat{B}/\hat{\Omega}_{+}^{1/2})\cdot
    \hat{\Omega}_{+}^{1/2}/\sqrt{n} )\\
    &= P( Z_{-}+\sqrt{n}b/\Omega^{1/2}_{-}\leq A_{-} )+ P(
    -Z_{+}-\sqrt{n}b/\Omega^{1/2}_{+}\leq A_{+}) -1,
  \end{split}
\end{equation*}
where
$A_{-}=\sqrt{n}(B_{\Cglob}-\hat{B})/\Omega^{1/2}_{-}+
\cv_{\alpha}(\sqrt{n}\hat{B}/\hat{\Omega}_{-}^{1/2})\cdot
\hat{\Omega}_{-}^{1/2}/\Omega^{1/2}_{-} $ and
$A_{+}= \cv_{\alpha}(\sqrt{n}\hat{B}/\hat{\Omega}_{+}^{1/2})\cdot
\hat{\Omega}_{+}^{1/2}
/\Omega_{+}^{1/2}+\sqrt{n}(B_{\Cglob}-\hat{B})/\Omega^{1/2}_{+}$. Now, since
$\hat{\Omega}_{-}/\Omega_{-}\inprob 1$, applying first \Cref{eq:cv1} and
next~\Cref{eq:cv2} in~\Cref{lemma:auxiliary-results-cv} below yields
\begin{equation*}
  A_{-}=\sqrt{n}(B_{\Cglob}-\hat{B})/\Omega^{1/2}_{-}+
  \cv_{\alpha}(\sqrt{n}\hat{B}/\Omega_{-}^{1/2})+o_{p}(1)
  =\cv_{\alpha}(\sqrt{n}B_{\Cglob} / \Omega_{-}^{1/2})+o_{p}(1),
\end{equation*}
where the $o_{p}(1)$ term is asymptotically negligible uniformly over
$\mathcal{S}_{n}$. By similar argument,
$A_{+}=\cv_{\alpha}(\sqrt{n}B_{\Cglob} / \Omega_{+}^{1/2})+o_{p}(1)$. By
\Cref{eq:global_misspec_normal}, it therefore follows that, up to a term that's
asymptotically negligible uniformly over $\mathcal{S}_{n}$,
$P(h(\theta)\in \widetilde{CI})$ equals
\begin{equation}\label{eq:coverage-bound}
  P( Z+\sqrt{n}b/\Omega^{1/2}_{-}\leq \cv_{\alpha}(\sqrt{n}B_{\Cglob} / \Omega_{-}^{1/2}) )
  + P(Z-\sqrt{n}b/\Omega^{1/2}_{+}\leq \cv_{\alpha}(\sqrt{n}B_{\Cglob} / \Omega_{+}^{1/2})) -1,
\end{equation}
where $Z$ denotes as standard normal random variable. Fix $\epsilon>0$. To
conclude the proof, we will show that for $n$ large enough, this expression is
bounded below by $1-\alpha-\epsilon$.

Since~\Cref{eq:coverage-bound} is symmetric in $\Omega_{+}$ and $\Omega_{-}$,
suppose without loss of generality that $\Omega_{+}>\Omega_{-}$. By the
assumption of the~\namecref{theorem:global_misspec_validity}, for $n$ large
enough and $\eta>0$ specified below,
$\abs{\Omega_{-}^{1/2}/\Omega^{1/2}_{+}-1}\leq \eta
\sqrt{n}B_{\Cglob}/\Omega^{1/2}_{+}$.

We'll consider two cases,$\sqrt{n}B_{\Cglob}/\Omega^{1/2}_{+}>z_{1-\epsilon}$,
and $\sqrt{n}B_{\Cglob}/\Omega^{1/2}_{+}\leq z_{1-\epsilon}$. Suppose first
$\sqrt{n}B_{\Cglob}/\Omega^{1/2}_{+}>z_{1-\epsilon}$. Then, if $b<0$,
\Cref{eq:coverage-bound} is bounded below by
$\Phi(\cv_{\alpha}(\sqrt{n}B_{\Cglob} / \Omega_{-}^{1/2}))+1-\alpha -1
\geq\Phi(\cv_{\alpha}(\sqrt{n}B_{\Cglob} / \Omega_{+}^{1/2}))-\alpha \geq
1-\epsilon-\alpha$, where the last inequality follows since
$\cv_{\alpha}(t)\geq t$. If $b$ is positive, it is bounded below by
$1-\alpha+\Phi(\cv_{\alpha}(\sqrt{n}B_{\Cglob} / \Omega_{+}^{1/2})) -1\geq
1-\alpha-\epsilon$.

Next, suppose, $\sqrt{n}B_{\Cglob}/\Omega^{1/2}_{+}\leq z_{1-\epsilon}$. Then
$\abs{\Omega_{-}^{1/2}/\Omega^{1/2}_{+}-1}\leq \eta z_{1-\epsilon}$, so that by
\Cref{eq:cv2},
\begin{equation*}
  \abs{ \cv_{\alpha}(\sqrt{n}B_{\Cglob} / \Omega_{-}^{1/2})-\sqrt{n}B_{\Cglob} / \Omega_{-}^{1/2}
    -  \cv_{\alpha}(\sqrt{n}B_{\Cglob} / \Omega_{+}^{1/2})
    +\sqrt{n}B_{\Cglob} / \Omega_{+}^{1/2}}\leq
  \nu,
\end{equation*}
where $\nu=\eta z_{1-\epsilon}(z_{1-\alpha/2}-z_{1-\alpha})$. Therefore,
\begin{multline*}
  P( Z+\sqrt{n}b/\Omega^{1/2}_{-}\leq \cv_{\alpha}(\sqrt{n}B_{\Cglob} / \Omega_{-}^{1/2}) )\\
  \geq P( Z+\sqrt{n}b/\Omega^{1/2}_{-}\leq \cv_{\alpha}(\sqrt{n}B_{\Cglob} /
  \Omega_{+}^{1/2})- \sqrt{n}B_{\Cglob} / \Omega_{+}^{1/2}+\sqrt{n}B_{\Cglob} /
  \Omega_{-}^{1/2} -\nu
  )\\
  = P\left( Z +\sqrt{n}b/\Omega^{1/2}_{+} \leq \cv_{\alpha}(\sqrt{n}B_{\Cglob} /
    \Omega_{+}^{1/2})+\sqrt{n}(B_{\Cglob}-b)(\Omega_{-}^{-1/2}-
    \Omega_{+}^{-1/2}) -\nu \right)\\
  \geq P\left( Z +\sqrt{n}b/\Omega^{1/2}_{+} \leq
    \cv_{\alpha}(\sqrt{n}B_{\Cglob} / \Omega_{+}^{1/2})
    -\nu \right)\\ \geq P\left( Z
    +\sqrt{n}b/\Omega^{1/2}_{+} \leq \cv_{\alpha}(\sqrt{n}B_{\Cglob} /
    \Omega_{+}^{1/2})\right)+1-2 \Phi(\nu/2),
\end{multline*}
where the last equality follows since
$\inf_{x}\{\Phi(x-\nu)-\Phi(x)\}=1-2\Phi(\nu/2)$. It therefore follows
that the coverage probability in~\Cref{eq:coverage-bound} is bounded below by
\begin{multline*}
  P( Z +\sqrt{n}b/\Omega^{1/2}_{+} \leq \cv_{\alpha}(\sqrt{n}B_{\Cglob} /
    \Omega_{+}^{1/2}))+ P(Z-\sqrt{n}b/\Omega^{1/2}_{+}\leq
  \cv_{\alpha}(\sqrt{n}B_{\Cglob} / \Omega_{+}^{1/2})) \\
  -1+(1-2 \Phi(\nu/2 )) \geq 1-\alpha+(1-2 \Phi(\nu/2 )).
\end{multline*}
where the inequality follows by definition of $\cv_{\alpha}$. Setting
$\eta= 2z_{1/2+\epsilon/2}/(z_{1-\epsilon}(z_{1-\alpha/2}-z_{1-\alpha}))$ then
implies that the right-hand side evaluates to $1-\alpha-\epsilon$. Thus,
\Cref{eq:coverage-bound} is bounded below by $1-\alpha-\epsilon$, concluding the
proof.
\end{proof}

\begin{lemma}\label{lemma:auxiliary-results-cv}
  The critical value $\cv_{\alpha}(t)$ satisfies, for any $a> 0$,
  \begin{equation}\label{eq:cv1}
    \sup_{b\geq 0}\abs{\cv_{\alpha}(ab)/a-\cv_{\alpha}(b)}\leq
    z_{1-\alpha/2}\frac{\abs{1-a}}{\max\{a, 1\}},
  \end{equation}
  and
  \begin{equation}\label{eq:cv2}
    \sup_{b\geq 0}\abs{\cv_{\alpha}(ab)-ab-\cv_{\alpha}(b)+b}\leq (z_{1-\alpha/2}-z_{1-\alpha})
    \frac{\abs{1-a}}{\max\{a, 1\}}.
  \end{equation}
\end{lemma}
\begin{proof}
  Since the function $\cv_{\alpha}$ is increasing and convex, with slope bounded
  by $1$, for $b_{1},b_{2}\geq 0$
  \begin{equation}\label{eq:property1}
    \cv_{\alpha}(b_{1}+b_{2})\leq
    b_{1}+\cv_{\alpha}(b_{2}),
  \end{equation}
 and for $a\geq 1$ and $b\geq 0$,
 \begin{equation}\label{eq:property2}
   \cv_{\alpha}(ba)/a+\cv_{\alpha}(0)(1-1/a)\geq \cv_{\alpha}(b).
 \end{equation}
 Suppose $a\geq 1$. Then by~\Cref{eq:property1}
 \begin{equation*}
   \cv_{\alpha}(ab)/a-\cv_{\alpha}(b)\leq
   (\cv_{\alpha}(b)-b)(1/a-1)\leq 0,
 \end{equation*}
 since $\cv_{\alpha}(b)-b$ is bounded below by $z_{1-\alpha}$. On the other
 hand, by \Cref{eq:property2}, the left-hand side is greater than
 $-\cv_{\alpha}(0)(1-1/a)$. If $a\leq 1$, the same argument with $ab$ and $b$
 reversed then yields~\Cref{eq:cv1}.

 To show \Cref{eq:cv2}, suppose first that $a\geq 1$. By~\Cref{eq:property1},
 $\cv_{\alpha}(ab)-ab-\cv_{\alpha}(b)+b\leq 0$. On the other hand, by \Cref{eq:property2},
 \begin{equation*}
   \cv_{\alpha}(ab)-ab-\cv(b)+b
   \geq
   (\cv_{\alpha}(ab)-ab-cv_{\alpha}(0))(1-1/a)
   \geq (1-1/a)(z_{1-\alpha}-z_{1-\alpha/2})
 \end{equation*}
 where the second inequality follows since $\cv_{\alpha}(ab)-ab\geq z_{1-\alpha}$
 and $\cv_{\alpha}(0)=z_{1-\alpha/2}$. If $a<1$, the same argument with $ab$ and
 $b$ reversed then yields~\Cref{eq:cv2}.
\end{proof}

\bibliography{../../np-testing-library}

\clearpage

\begin{table}[p]
  \caption{$J$-test of overidentifying restrictions in the application to
    \citet{blp95} under different forms of $\ell_{p}$ misspecification.}\label{tab:Jtest}
  \centering
  \begin{tabular}[htp]{@{}lrr@{}}
Instrument set & $p=2$ & $p=\infty$\\
\midrule
D/F\@: \# cars                      & 10.21& 10.21 \\
S/F\@: \# cars                      & 15.00& 15.00 \\
Supply: Miles/dollar                & 16.31& 16.31 \\
All D/F                             &  2.71&  2.71 \\
All D/R                             &  5.36&  5.55 \\
All S/F                             &  2.54&  2.56 \\
All S/R                             &  4.06&  6.84 \\
All excluded demand                 &  1.80&  1.97 \\
All excluded supply                 &  1.60&  1.72 \\
All excluded                        &  1.13&  2.56 \\[1ex]
  \end{tabular}
  \medskip
  \begin{minipage}{\linewidth}\small
    \emph{Notes:} The table gives the minimum value of $\Mbound$ such that the
    test of overidentifying restrictions has $p$-value equal to $0.05$. ``D/F'':
    Demand-side instrument based on characteristics of other cars produced by
    the same firm. ``S/F'': Supply-side instrument based on characteristics of
    other cars produced by the same firm. ``D/R'': Demand-side instrument based
    on characteristics of cars produced by rivals. ``S/R'': Supply-side
    instrument based on characteristics of cars produced by rivals. ``All
    excluded'': All excluded instruments are potentially invalid.
  \end{minipage}
\end{table}

\begin{table}[p]
  \centering
  \caption{Efficiency bounds (in \%) for one and two-sided 95\% confidence
    intervals at $c=0$ under $\ell_{p}$ misspecification in the application to
    \citet{blp95}.}\label{tab:efficiency}
  \vspace{1ex}
  \begin{tabular}[ht]{@{}lrrrr@{}}
 & \multicolumn{2}{c}{Two-sided} & \multicolumn{2}{c}{One-sided}\\
    \cmidrule(rl){2-3}\cmidrule(rl){4-5}
Instrument set &  $p=2$ & $p=\infty$ & $p=2$ & $p=\infty$\\
\midrule
D/F\@: \# cars       &     85.9 &       85.9 &     100.0 &      100.0 \\
S/F\@: \# cars       &     90.1 &       90.1 &      99.8 &       99.8 \\
Supply: Miles/dollar &     85.0 &       85.0 &     100.0 &      100.0 \\
All D/F              &     85.5 &       85.7 &     100.0 &      100.0 \\
All D/R              &     94.8 &       95.3 &      93.9 &       95.3 \\
All S/F              &     88.6 &       89.1 &      99.7 &       99.7 \\
All S/R              &     89.4 &       89.2 &      98.5 &       99.5 \\
All excluded demand  &     95.4 &       96.4 &      95.0 &       97.3 \\
All excluded supply  &     90.3 &       90.1 &      98.2 &       99.6 \\
All excluded         &     97.0 &       97.5 &      99.5 &       98.2 \\[1ex]
  \end{tabular}
  \medskip
  \begin{minipage}{\linewidth}\small
    \emph{Notes:} For two-sided confidence intervals, the table calculates the
    ratio of the expected length of a 95\% confidence interval that minimizes
    its length at $c=0$ relative to the length of the CI
    in~\eqref{eq:optimal-FLCI-limit_experiment}, given
    in~\eqref{kappa_conv_cs_eq}. For one-sided confidence intervals, the table
    calculates an analogous bound, given in \Cref{adaptation_sec_append}, when
    the confidence interval optimizes the 80\% quantile of excess length.
    Instrument set labels are describe in notes to \Cref{tab:Jtest}.
  \end{minipage}
\end{table}

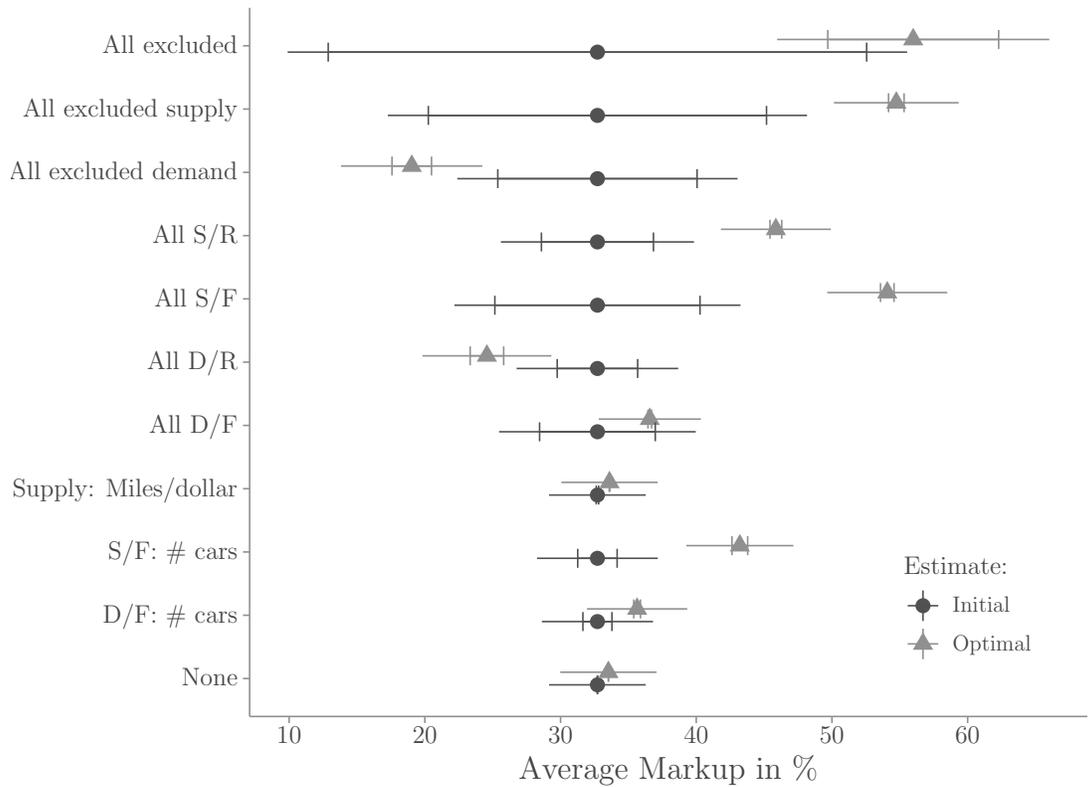
\begin{figure}[p]
  \centering%
  \input{./blp_l2.tex}
  \caption{Confidence intervals under $\ell_{2}$ misspecification and $\Mbound=1$ in
    the application to \citet{blp95}.}\label{fig:l2-blp}
  \vspace{1ex}
  \begin{minipage}[h]{1.0\linewidth}
    \footnotesize Vertical lines correspond to the estimate $\pm$ the worst case
    bias, and horizontal lines correspond to 95\% confidence intervals.
    Different rows correspond to assuming that different subsets of instruments
    are potentially invalid. ``None'': correct specification. ``D/F'':
    Demand-side instrument based on characteristics of other cars produced by
    the same firm. ``S/F'': Supply-side instrument based on characteristics of
    other cars produced by the same firm. ``D/R'': Demand-side instrument based
    on characteristics of cars produced by rivals. ``S/R'': Supply-side
    instrument based on characteristics of cars produced by rivals. ``All
    excluded'': All excluded instruments are potentially invalid.
  \end{minipage}
\end{figure}

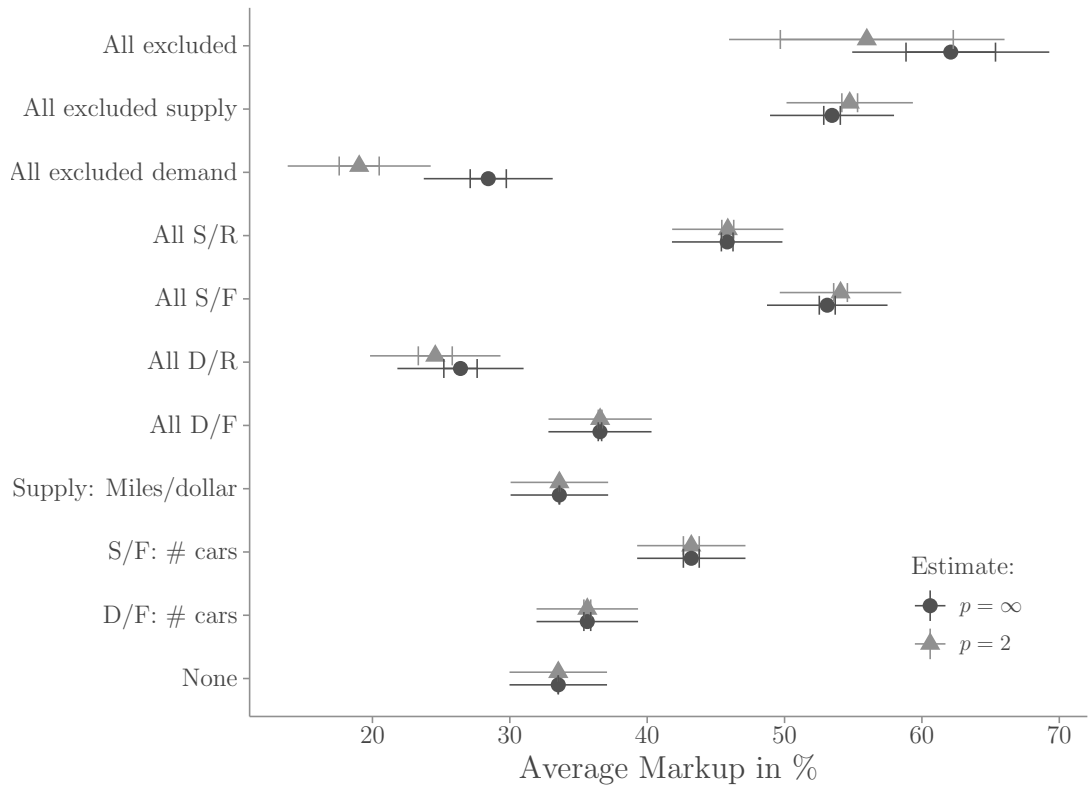
\begin{figure}[p]
  \centering%
  \input{./blp_l2_lI.tex}
  \caption{Optimal Confidence intervals under $\ell_{2}$, and
    $\ell_{\infty}$ misspecification and $\Mbound=1$ in the application to
    \citet{blp95}.}\label{fig:l2_lI-blp}
  \vspace{1ex}
  \begin{minipage}[h]{1.0\linewidth}
    \footnotesize Vertical lines correspond to the estimate $\pm$ the worst case
    bias, and horizontal lines correspond to 95\% confidence intervals.
    Different rows correspond to assuming that different subsets of instruments
    are potentially invalid. ``None'': correct specification. ``D/F'':
    Demand-side instrument based on characteristics of other cars produced by
    the same firm. ``S/F'': Supply-side instrument based on characteristics of
    other cars produced by the same firm. ``D/R'': Demand-side instrument based
    on characteristics of cars produced by rivals. ``S/R'': Supply-side
    instrument based on characteristics of cars produced by rivals. ``All
    excluded'': All excluded instruments are potentially invalid.
  \end{minipage}
\end{figure}

\begin{figure}[p]
  \centering%
  \input{./blp_l2_M.tex}
  \caption{Optimal confidence intervals under $\ell_{2}$ misspecification the
    application to \citet{blp95} as a function of misspecification parameter
    $\Mbound$, when all excluded instruments are allowed to be potentially
    invalid.}\label{fig:l2_M-blp}
  \vspace{1ex}
  \begin{minipage}[h]{1.0\linewidth}
    \footnotesize Dotted line corresponds to point estimate, shaded region
    denotes the estimate $\pm$ its worst-case bias, and a 95\% confidence band
    is denoted by solid lines.
  \end{minipage}
\end{figure}
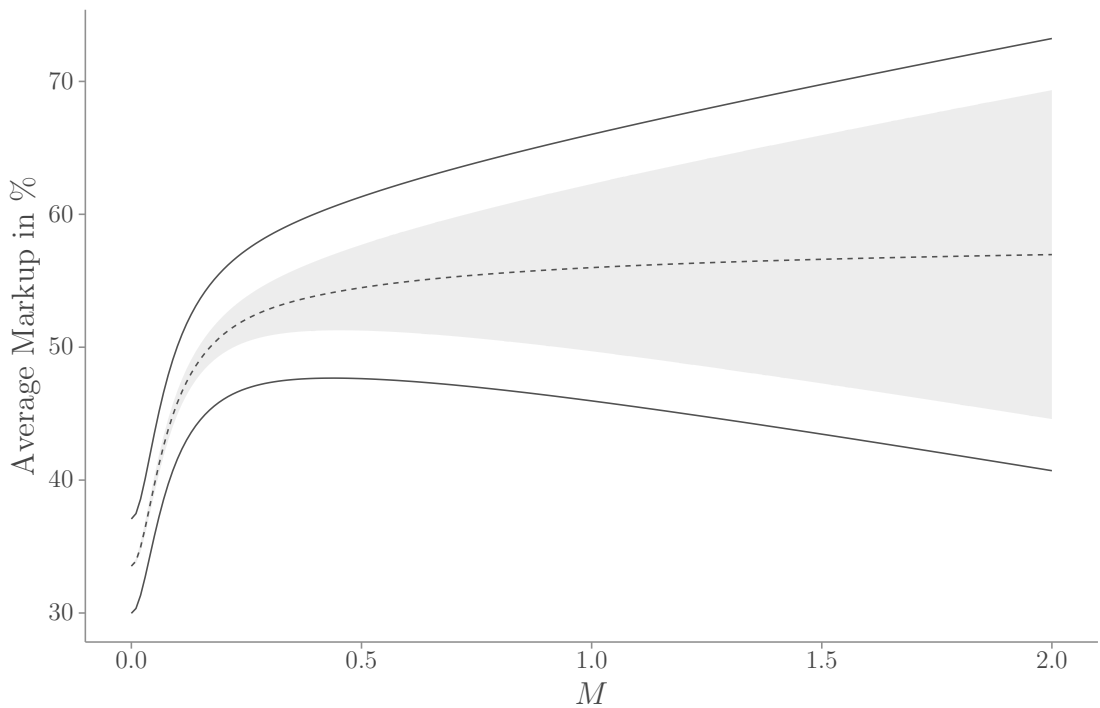

\end{document}


\maketitle

\renewcommand{\theequation}{S\arabic{equation}}
\renewcommand{\thetable}{S\arabic{table}}
\renewcommand{\thefigure}{S\arabic{figure}}
\renewcommand{\thepage}{S\arabic{page}}

\begin{appendices}
  \setcounter{section}{4}
  \crefalias{section}{sappsec}
  \crefalias{subsection}{sappsubsec}
  \crefalias{subsubsection}{sappsubsubsec}

  \Cref{sec:addit-asympt-results} gives a construction of a submodel satisfying
  \Cref*{submodel_assump}, verifies the conditions in
  \Cref*{efficiency_sec_append} in the misspecified IV model, and collects
  auxiliary results used in \Cref*{efficiency_sec_append}.

\section{Additional asymptotic results}\label{sec:addit-asympt-results}

\subsection{Construction of a submodel satisfying
  \texorpdfstring{\Cref*{submodel_assump}}{
  Assumption~\ref{submodel_assump}}}\label{gmm_submodel_sec}

We give here a construction of a submodel satisfying \Cref*{submodel_assump}
under mild conditions on the class $\mathcal{P}$. The construction follows
Example 25.16 (p. 364) of \citet{van_der_vaart_asymptotic_1998}.

\begin{lemma}\label{gmm_submodel_lemma}
  Suppose that $g(w_i, \theta)$ is continuously differentiable almost surely in a
  neighborhood of $\theta^*$ where $E_{P_0}g(w_i, \theta^*)=0$, and that, for
  some $\varepsilon>0$,
  \begin{equation*}
    E_{P_0}\sup_{\|\theta-\theta^*\|\le \varepsilon}|g(w_i, \theta)g(w_i, \theta)'|<\infty
    \quad\text{and}\quad
    E_{P_0}\sup_{\|\theta-\theta^*\|\le \varepsilon}\left\|\frac{d}{d\theta'}g(w_i, \theta)\right\|<\infty.
  \end{equation*}
  Let
  \begin{equation*}
    \pi_t(w_i) =C(t)h(t'g(w_i, \theta^*))
    \quad\text{where}\quad
    h(x)=2\left[1+\exp\left(-2 x\right)\right]^{-1}
  \end{equation*}
  with $C(t)^{-1}=E_{P_0}h(t'g(w_i, \theta^*))$. This submodel satisfies
  \Cref*{submodel_assump}, and the bounds on the moments in the above display
  hold with $P_0$ replaced by $P_t$.
\end{lemma}
\begin{proof}
  Quadratic mean differentiability follows from Problem 12.6 in \citet{LeRo05},
  so we just need to show that \Cref*{t_theta_derivative_eq} holds, and that the
  derivative is continuous in a neighborhood of
  $(t', \theta')'=(0', {\theta^*}')'$. For this, it suffices to show that each
  partial derivative exists and is continuous as a function of $(t', \theta')'$
  in a neighborhood of $(0', {\theta^*}')'$, and that the Jacobian matrix of
  partial derivatives takes the form in \Cref*{t_theta_derivative_eq} at
  $(t', \theta')'=(0', {\theta^*}')'$ \citep[see Theorem 4.5.3
  in][]{shurman_calculus_2016}.

  To this end, we first show that $C(t)$ is continuously differentiable, and
  derive its derivative at $0$. It can be checked that $h(x)$ is continuously
  differentiable, with $h(0)=h'(0)=1$, and that $h(x)$ and $h'(x)$ are bounded.
  We have, for some constant $K$,
  \begin{equation*}
    \left|\frac{d}{d t_j} h(t'g(w_i, \theta^*))\right|
    =|h'(t'g(w_i, \theta^*))g_j(w_i, \theta^*)|
    \le K |g_j(w_i, \theta^*)|
  \end{equation*}
  so, since $E_{P_0}|g_j(w_i, \theta^*)|<\infty$, we have, by a corollary of the Dominated Convergence Theorem \citep[Corollary 5.9 in][]{bartle_elements_1966},
  \begin{equation*}
    \frac{d}{d t_j}E_{P_0}h(t'g(w_i, \theta^*))
    =E_{P_0}\frac{d}{d t_j}h(t'g(w_i, \theta^*))
    =E_{P_0}h'(t'g(w_i, \theta^*))g_j(w_i, \theta^*).
  \end{equation*}
  By boundedness of $h'$ and the Dominated Convergence Theorem, this is continuous in $t$.  Thus, $C(t)$ is continuously differentiable in each argument, with
  \begin{equation*}
    \frac{d}{d t_j} C(t)
    =-\left[E_{P_0}h(t'g(w_i, \theta^*))\right]^{-2}
    E_{P_0}h'(t'g(w_i, \theta^*))g_j(w_i, \theta^*)
  \end{equation*}
  which gives $\left[\frac{d}{dt_j}C(t)\right]_{t=0}=E_{P_0}g_j(w_i, \theta^*)=0$.

Now consider the derivative of
\begin{equation*}
E_{P_t}g(w_i, \theta)
=E_{P_0}g(w_i, \theta)\pi_t(w_i)
=C(t)E_{P_0}g(w_i, \theta)h(t'g(w_i, \theta^*))
\end{equation*}
with respect to elements of $\theta$ and $t$.
We have, for each $j, k$
\begin{equation*}
\frac{d}{dt_j}g_k(w_i, \theta)h(t'g(w_i, \theta^*))
=g_k(w_i, \theta)h'(t'g(w_i, \theta^*))g_j(w_i, \theta^*).
\end{equation*}
This is bounded by a constant times $|g_k(w_i, \theta)g_j(w_i, \theta^*)|$ by boundedness of $h'$.
Also,
\begin{equation*}
\frac{d}{d\theta_j}g_k(w_i, \theta)h(t'g(w_i, \theta^*))
\end{equation*}
is bounded by a constant times $\frac{d}{d\theta_j}g_k(w_i, \theta)$ by boundedness of $h$.
By the conditions of the lemma, the quantities in the above two displays are bounded uniformly over $(t', \theta')'$ in a neighborhood of $({\theta^*}',0')'$ by a function with finite expectation under $P_0$.  It follows that we can again apply Corollary 5.9 in \citet{bartle_elements_1966} to obtain the derivative of $E_{P_0}g(w_i, \theta)h(t'g(w_i, \theta^*))$ with respect to each element of $\theta$ and $t$ by differentiating under the expectation.  Furthermore, the bounds above and continuous differentiability of $g(w_i, \theta)$ along with the Dominated Convergence Theorem imply that the derivatives are continuous in $(t', \theta')'$.

Thus, $E_{P_t}g(w_i, \theta)$ is differentiable with respect to each argument of
$t$ and $\theta$, with the partial derivatives continuous with respect to
$(\theta', t')'$. It follows that $(t', \theta')'\mapsto E_{P_t}g(w_i, \theta)$
is differentiable at $t=0, \theta=\theta^*$. To calculate the Jacobian, note that
\begin{equation*}
  \frac{d}{dt'}E_{P_t}g(w_i, \theta)
  =C(t)E_{P_0}g(w_i, \theta)g(w_i, \theta^*)'h'(t'g(w_i, \theta^*))
  +E_{P_0}g(w_i, \theta)h(t'g(w_i, \theta^*))\frac{d}{dt'}C(t).
\end{equation*}
Evaluating this at $t=0$, $\theta=\theta^*$, the second term is equal to zero by
calculations above, and the first term is given by
$E_{P_0}g(w_i, \theta^*)g(w_i, \theta^*)$. For the derivative with respect to
$\theta$ at $\theta=\theta^*$, $t=0$, this is equal to $\Gamma_{\theta^*, P_0}$
by definition. Thus, \Cref*{submodel_assump} holds. Furthermore, the bounds on
the moments of $g(w_i, \theta)$ hold with $P_t$ replacing $P_0$ by boundedness
of $\pi_t(w_i)$.
\end{proof}

\subsection{Example: misspecified linear IV}\label{iv_sec_append}

We verify our conditions in the misspecified linear IV model, defined by the equation
\begin{equation*}
g_P(\theta)=E_{P}(y_{i}-x_i'\theta)z_{i}=c/\sqrt{n}, \,
c\in\mathcal{C}
\end{equation*}
where $\mathcal{C}$ is a compact convex set, $y_{i}$ is a scalar valued random variable, $x_i$ is a $\mathbb{R}^{d_\theta}$ valued random variable and $z_{i}$ is a $\mathbb{R}^{d_g}$ valued random variable, with $d_g\ge d_\theta$.
The derivative matrix and variance matrix are
\begin{equation*}
\Gamma_{\theta, P}=\frac{d}{d\theta'}g_P(\theta)=-E_{P}z_{i}x_i',
\quad\quad
\Sigma_{\theta, P}=\text{var}_{P}((y_{i}-x_i'\theta)z_{i}).
\end{equation*}
Let $\Theta\subset\mathbb{R}^{d_\theta}$ be a compact set and let $h:\Theta\to\mathbb{R}$ be continuously differentiable with nonzero derivative at all $\theta\in\Theta$.  Let $\varepsilon$ be given and let $\mathcal{P}$ be a set of probability distributions $P$ for $(x_i', z_{i}', y_{i})'$.  We make the following assumptions on $\mathcal{P}$.

\begin{assumption}\label{iv_example_assump}
For all $P\in\mathcal{P}$, the following conditions hold.
\begin{enumerate}
\item\label{iv_moment_bounds} For all $j$,
  $E_{P}|x_{i, j}|^{4+\varepsilon}<1/\varepsilon$,
  $E_{P}|z_{i, j}|^{4+\varepsilon}<1/\varepsilon$ and
  $E_{P}|y_{i}|^{4+\varepsilon}<1/\varepsilon$.
\item\label{iv_deriv} The matrix $E_{P}z_{i}x_i'$ is full rank and
  $\|E_{P}z_{i}x_i' u\|/\|u\|> 1/\varepsilon$ for all
  $u\in\mathbb{R}^{d_g}\backslash\{0\}$ (i.e.\ the singular values of
  $E_{P}z_{i}x_i'$ are bounded away from zero).
\item The matrix $\Sigma_{\theta, P}=\text{var}_{P}((y_{i}-x_i'\theta)z_{i})$
  satisfies $u'\Sigma_{\theta, P}u/\|u\|^2> \varepsilon$ for all
  $u\in\mathbb{R}^{d_g}\backslash\{0\}$ and all $\theta$ such that there exists
  $c\in\mathcal{C}$ and $n\ge 1$ such that
  $E_{P}(y_{i}-x_i'\theta)z_{i}=c/\sqrt{n}$.
\end{enumerate}
\end{assumption}

Note that, applying Cauchy-Schwartz, the first condition implies
$E_{P}|v_1v_2v_3v_4|^{1+\varepsilon/4} < 1/\varepsilon$ for any
$v_1,v_2,v_3,v_4$ where each $v_k$ is an element of $x_i$, $z_{i}$ or $y_{i}$.
In particular, $z_{i}(y_{i}-x_i'\theta)$ has a bounded $2+\varepsilon/2$ moment
uniformly over $\theta\in\Theta$ and $P\in\mathcal{P}$.

\subsubsection{Conditions for \texorpdfstring{\Cref*{optimal_khat_oneside_thm,optimal_khat_twoside_thm}}{Theorems~\ref{optimal_khat_oneside_thm}
   and~\ref{optimal_khat_twoside_thm}}}\label{iv_example_achieving_bound_sec}

We first verify the conditions of \Cref*{achieving_bound_sec}. To verify the
conditions of \Cref*{optimal_khat_oneside_thm,optimal_khat_twoside_thm} (which
show that the plug-in optimal weights
$\hat k=k(\delta, \hat H, \hat\Gamma, \hat \Sigma)$ lead to CIs that achieve or
nearly achieve the efficiency bounds in \Cref*{oneside_efficiency_bound_thm} and
\Cref*{twoside_efficiency_bound_thm}), we must verify
\Cref*{theta_init_assump,g_clt_assump_append,,GammaHSigma_assump,,scriptB_assump}.

Let
\begin{equation*}
  \hat\theta_{\text{initial}}=
  \left(\sum_{i=1}^{n}z_{i}x_i' W_n \sum_{i=1}^{n} x_{i}z_{i}'\right)^{-1}\sum_{i=1}^{n} z_{i}x_i' W_n \sum_{i=1}^{n}z_{i}y_{i}
\end{equation*}
where $W_n=W_P+o_P(1)$ uniformly over $P\in\mathcal{P}$ and $W_P$ is a positive definite matrix with $u'W_{P}u/\|u\|^2$ bounded away from zero uniformly over $P\in\mathcal{P}$.  Let $\hat H=H_{\hat \theta}$ where $H_\theta$ is the derivative of $h$ at $\theta$.
Let
\begin{equation*}
  \hat \Gamma=-\frac{1}{n}\sum_{i=1}^{n}z_{i}x_i',
  \quad\quad
  \hat \Sigma=\frac{1}{n}\sum_{i=1}^{n}z_{i}z_{i}'(y_{i}-x_i'\hat\theta_{\text{initial}})^2.
\end{equation*}

First, let us verify \Cref*{g_clt_assump_append}. Indeed, it follows from a CLT
for triangular arrays (\Cref{uniform_clt_lemma} with
$v_i=u_n'\left[z_{i}(y_{i}-x_i'\theta)-Ez_{i}(y_{i}-x_i'\theta)\right]$ with
$u_n$ an arbitrary sequence with $\|u_n\|=1$ all $n$) that
\begin{equation*}
\sup_{u\in\mathbb{R}^{d_g}}\sup_{t\in\mathbb{R}}\sup_{(\theta', c')\in\Theta\times\mathcal{C}}
\sup_{P\in\mathcal{P}_n(\theta, c)} \left|P\left(\frac{\sqrt{n}u'(\hat g(\theta)-g_P(\theta))}{\sqrt{u'\Sigma_{\theta, P}u}}\le t\right)
- \Phi\left(t\right)\right|
\to 0
\end{equation*}
(note that $u$ can be taken to satisfy $\|u\|=1$ without loss of generality, since the formula inside the probability statement is invariant to scaling).  Note that this, along with compactness of $\mathcal{C}$, also implies that $\frac{1}{\sqrt{n}}\sum_{i=1}^{n}z_{i}(y_{i}-x_i'\theta)=\sqrt{n}\hat g(\theta)=\mathcal{O}_P(1)$ uniformly over $\theta$ and $P$ with $P\in\mathcal{P}(\theta, c)$ for some $c$.

For \Cref*{theta_init_assump}, we have
\begin{equation*}
\sqrt{n}\left(\hat\theta_{\text{initial}}-\theta\right)
=\left(\frac{1}{n}\sum_{i=1}^{n}z_{i}x_i' W_n \frac{1}{n}\sum_{i=1}^{n} x_{i}z_{i}'\right)^{-1}\frac{1}{n}\sum_{i=1}^{n} z_{i}x_i' W_n \frac{1}{\sqrt{n}}\sum_{i=1}^{n}z_{i}(y_{i}-x_i'\theta).
\end{equation*}
Since $\frac{1}{n}\sum_{i=1}^{n}z_{i}x_i'$ converges in probability to
$-\Gamma_{\theta, P}$ uniformly over $P$ by \Cref{uniform_lln_lemma} and
$\frac{1}{\sqrt{n}}\sum_{i=1}^{n}z_{i}(y_{i}-x_i'\theta)=\mathcal{O}_P(1)$
uniformly over $P$ by the verification of \Cref*{g_clt_assump_append} above, it
follows that this display is $\mathcal{O}_P(1)$ uniformly over $P$ and $\theta$,
as required. For the second part of the assumption, we have
\begin{equation*}
  \hat g(\hat\theta_{\text{initial}})-g(\theta)
  =-\frac{1}{n}\sum_{i=1}^{n}z_{i}x_i'(\hat\theta_{\text{initial}}-\theta)
  =\Gamma_{\theta, P}(\hat\theta_{\text{initial}}-\theta)
  +(\hat\Gamma-\Gamma_{\theta, P})(\hat\theta_{\text{initial}}-\theta).
\end{equation*}
The last term is uniformly $o_P(1/\sqrt{n})$ as required since $(\hat\theta_{\text{initial}}-\theta)=\mathcal{O}_P(1/\sqrt{n})$ as shown above and $\hat\Gamma-\Gamma_{\theta, P}$ converges in probability to zero uniformly by an LLN for triangular arrays (\Cref{uniform_lln_lemma}).
For the last part of the assumption, we have, by the mean value theorem,
\begin{equation*}
h(\hat\theta_{\text{initial}})-h(\theta)=
H_{\theta^*(\hat\theta_{\text{initial}})}(\hat\theta_{\text{initial}}-\theta)
=H_{\theta}(\hat\theta_{\text{initial}}-\theta)
+\left(H_{\theta^*(\hat\theta_{\text{initial}})}-H_{\theta}\right)(\hat\theta_{\text{initial}}-\theta)
\end{equation*}
where $\theta^*(\hat\theta_{\text{initial}})-\theta$ converges uniformly in probability to zero.  Since $\theta\mapsto H_\theta$ is uniformly continuous on $\theta$ (since it is continuous by assumption and $\Theta$ is compact), it follows that $H_{\theta^*(\hat\theta_{\text{initial}})}-H_{\theta}$ converges uniformly in probability to zero, which, along with the verification of the first part of the assumption above, gives the required result.

For \Cref*{GammaHSigma_assump}, the first two parts of the assumption
(concerning uniform consistency of $\hat\Gamma$ and $\hat H$) follow from
arguments above. For the last part (uniform consistency of $\hat\Sigma$), note
that
\begin{equation*}
  \hat\Sigma=\frac{1}{n}\sum_{i=1}^{n}z_{i}z_{i}'(y_{i}-x_i'\hat\theta_{\text{initial}})^2
  =\frac{1}{n}\sum_{i=1}^{n}z_{i}z_{i}'(y_{i}-x_i'\theta)^2
  +\frac{1}{n}\sum_{i=1}^{n}z_{i}z_{i}'\big[(y_{i}-x_i'\hat\theta_{\text{initial}})^2-(y_{i}-x_i'\theta)^2\big].
\end{equation*}
The first term converges uniformly in probability to $\Sigma_{\theta, P}$ by an LLN for triangular arrays (\Cref{uniform_lln_lemma}).  The last term is equal to
\begin{equation*}
  \frac{1}{n}\sum_{i=1}^{n}z_{i}z_{i}'(x_i'\hat\theta_{\text{initial}}+x_i'\theta-2y_{i})x_i'(\hat\theta_{\text{initial}}-\theta).
\end{equation*}
This converges in probability to zero by an LLN for triangular arrays (\Cref{uniform_lln_lemma}) and the moment bound in \Cref{iv_example_assump}(\ref{iv_moment_bounds})

Finally, \Cref*{scriptB_assump} follows by
\Cref{iv_example_assump}(\ref{iv_deriv}), and the condition that the derivative
is nonzero for all $\theta$.

\subsubsection{Conditions for \texorpdfstring{\Cref*{oneside_efficiency_bound_thm,twoside_efficiency_bound_thm}}{Theorems~\ref{oneside_efficiency_bound_thm}
  and~\ref{twoside_efficiency_bound_thm}}}

We now verify the conditions of the lower bounds,
\Cref*{oneside_efficiency_bound_thm,twoside_efficiency_bound_thm}. Given
$P_0\in\mathcal{P}$ with $E_{P_0}g(w_i, \theta^*)=0$, we need to show that a
submodel $P_t$ satisfying \Cref*{submodel_assump} exists with
$P_t\in\mathcal{P}$ for $\|t\|$ small enough. To verify this condition, we take
$\mathcal{P}$ to be the set of all distributions satisfying
\Cref{iv_example_assump}, and we assume that $\theta^*$ is in the interior of
$\Theta$.

Let $P_t$ be the subfamily given in \Cref{gmm_submodel_lemma}. This
satisfies \Cref*{submodel_assump} by
\Cref{gmm_submodel_lemma} (the moment conditions needed for this lemma hold by
\Cref{iv_example_assump}(\ref{iv_moment_bounds})), so we just need to check that
$P_t\in\mathcal{P}$ for $t$ small enough. For this, it suffices to show that
$E_{P_t}|x_{i, j}|^{4+\varepsilon}$, $E_{P_t}|z_{i, j}|^{4+\varepsilon}$,
$E_{P_t}|y_{i}|^{4+\varepsilon}$, $E_{P_t}z_{i}x_i'$ and
$\text{var}_{P_t}(z_{i}(y_{i}-x_i'\theta))$ are continuous in $t$ at $t=0$,
which holds by the Dominated Convergence Theorem since the likelihood ratio
$\pi_t(w_i)$ for this family is bounded and continuous with respect to $t$.

\subsubsection{Conditions for \texorpdfstring{\Cref*{global_misspec_sec_append}}{Appendix~\ref{global_misspec_sec_append}}}\label{global_misspec_linear_iv_sec_append}

In \Cref*{global_misspec_sec_append}, we proposed a CI that is asymptotically
valid under global misspecification and asymptotically equivalent to the CIs
considered in the rest of the paper under local misspecification. Specializing
to the present setting with misspecified IV, the CI is the union over $\cglob$
of CIs that use the GMM estimator $\hat\theta_{W, \cglob}$ based on the moment
function $\theta\mapsto z_i(y_i-x_i'\theta)-\cglob$. This estimator is given by
$\theta_{W, \cglob}=-\left(\hat \Gamma'W\hat\Gamma \right)^{-1}\hat\Gamma'\hat W
\left(\frac{1}{n}\sum_{i=1}^{n}z_{i}y_{i} - \cglob \right)$ where
$\hat\Gamma=-\frac{1}{n}\sum_{i=1}^{n}z_{i}x_{i}'$ as defined above. We estimate
$k_{\theta, P}'=-H_\theta(\Gamma_{\theta, P}'W_{P} \Gamma_{\theta,
  P})^{-1}\Gamma_{\theta, P}'W_P$ using
$\hat k_{\theta}'=-H_\theta(\hat\Gamma'W\hat\Gamma)^{-1}\Gamma'W_P$. We estimate
$\Sigma_{\theta, P}=var_P(z_i(y_i-x_i'\theta))$ using
$\hat\Sigma_\theta=\frac{1}{n}\sum_{i=1}^{n} z_{i}z_i'(y_i-x_i'\hat\theta_{W,
  \cglob})^2 = \frac{1}{n}\sum_{i=1}^n z_{i}z_{i}'[y_i-x_i'(\hat{\Gamma}' W
\hat{\Gamma})^{-1}\hat\Gamma'\hat W (\frac{1}{n}\sum_{i=1}^{n}z_{i}y_{i} -
\cglob) ]^2$. In addition to \Cref{iv_example_assump}, we assume that the
weighting matrix is given by a (possibly data dependent) sequence $W_n$ such
that $W_{n}-W_{P}=o_P(1)$ uniformly over $(\theta, P)\in\mathcal{S}_n$, where
$W_{P}$ is some family of limiting weighting matrices with $u'W_{P} u /\|u\|^2$
bounded away from zero and infinity uniformly over $P\in\mathcal{P}$. The
population influence function weights are then given by
$k'_{\theta, P}=H_\theta \left(\Gamma_{\theta, P}'W_{\theta, P}\Gamma_{\theta,
    P} \right)^{-1}\Gamma_{\theta, P}'$.

To verify the asymptotic equivalence result
(\Cref*{glob_misspec_ci_local_assump}), we need to verify
\Cref*{glob_misspec_ci_local_thm}. To this end, first note that
$\hat\Gamma-\Gamma_{\theta, P}=o_P(1)$ uniformly over
$(\theta, P)\in\mathcal{S}_n$ by a law of large numbers
(\Cref{uniform_lln_lemma}). Thus, by the bounds on $E_{P} z_{i}x_{i}'$ in
\Cref{iv_example_assump},
$\sup_{c\in\mathcal{C}}\left| \hat\theta_{W, c/\sqrt{n}}-\hat\theta_{W,0}
\right|= \sup_{c\in\mathcal{C}}\left| \left(\hat \Gamma'W\hat\Gamma
  \right)^{-1}\hat\Gamma' W c/\sqrt{n} \right|=\mathcal{O}_P(1/\sqrt{n})$
uniformly over $(\theta, P)\in\mathcal{S}_n$. Note that
$\hat\theta_{W,0}=\hat\theta_{\text{initial}}$ where
$\hat\theta_{\text{initial}}$ is defined in
\Cref{iv_example_achieving_bound_sec} above, so it follows from arguments in
that section that $\hat\theta_{W,0}-\theta=\mathcal{O}_P(1/\sqrt{n})$ uniformly
over $(\theta, P)\in\mathcal{S}_n$. Thus,
$\sup_{c\in\mathcal{C}}\left| \hat\theta_{W, c/\sqrt{n}}-\theta \right|=
\mathcal{O}_P(1/\sqrt{n})$ uniformly over $(\theta, P)\in\mathcal{S}_n$.
Similarly, we have
$\sup_{c\in\mathcal{C}}|\hat\Sigma_{\theta_{W,
    c/\sqrt{n}}}-\hat\Sigma_{\theta_{W, c/\sqrt{n}}}|=o_P(1)$ and
$\hat\Sigma_{\theta_{W,0}}$ corresponds to the estimate used in
\Cref{iv_example_achieving_bound_sec} above, so that
$\sup_{c\in\mathcal{C}}|\hat\Sigma_{\theta_{W, c/\sqrt{n}}}-\Sigma_{\theta,
  P}|=o_P(1)$ uniformly over $(\theta, P)\in\mathcal{S}_n$ by arguments in
\Cref{iv_example_achieving_bound_sec}.

The last part of \Cref*{glob_misspec_ci_local_assump} will now follow if we can
show that
$\sup_{c\in\mathcal{C}}\abs{\hat k_{\hat\theta_{W, c/\sqrt{n}}}-k_{\theta, P}
}=o_P(1)$. Since we have already shown uniform consistency of $\hat\Gamma$, this
will follow so long as
$\sup_{c\in\mathcal{C}}\left| H_{\hat\theta_{W, c/\sqrt{n}}}-H_\theta
\right|=o_P(1)$ uniformly over $(\theta, P)\in\mathcal{S}_n$. This follows by
the fact that
$\sup_{c\in\mathcal{C}}\left| \hat\theta_{W, c/\sqrt{n}}-\theta \right|=o_P(1)$
uniformly over $(\theta, P)\in\mathcal{S}_n$ along with uniform continuity of
$\theta\mapsto H_\theta$ on $\Theta$ (since $\Theta$ is compact, continuity
implies uniform continuity).

Finally, for the first display of \Cref*{glob_misspec_ci_local_assump}, note
that, for some $\theta^*(c)$ on the line segment between $\theta$ and
$\hat\theta_{W, c}$,
\begin{align*}
  &h(\hat\theta_{W, c/\sqrt{n}})-h(\theta)-k_{\theta, P}'[\hat g(\theta)-c/\sqrt{n}] \\
  &= H_{\theta^*(c)}(\hat\theta_{W,0}-\theta) - k_{\theta, P}'\hat g(\theta)
    +H_{\theta^*(c)}(\hat\Gamma'W\hat\Gamma)^{-1}\hat\Gamma'W c/\sqrt{n}
    + k'_{\theta, P}c/\sqrt{n} \\
  &= H_{\theta}(\hat\theta_{W,0}-\theta) - k_{\theta, P}'\hat g(\theta)
    +H_{\theta}(\hat\Gamma'W\hat\Gamma)^{-1}\hat\Gamma'W c/\sqrt{n}
    + k'_{\theta, P}c/\sqrt{n}
    +R_{n, \theta, P}(c) \\
  &= \left[-H_{\theta}(\hat\Gamma'W\hat\Gamma)^{-1}\hat\Gamma'W- k_{\theta, P}'\right]
    \left[\frac{1}{n}\sum_{i=1}^{n}z_i(y_i-x_i'\theta)-c/\sqrt{n}\right]
    +R_{n, \theta, P}(c)
\end{align*}
where $\sup_{c\in\mathcal{C}}\sqrt{n}|R_{n, \theta, P}(c)|=o_P(1)$ uniformly
over $(\theta, P)\in\mathcal{S}_n$. The first display of
\Cref*{glob_misspec_ci_local_assump} now follows from the fact that
$\frac{1}{n}\sum_{i=1}^{n}
z_i(y_i-x_i'\theta)-c/\sqrt{n}=\mathcal{O}_{P}(1/\sqrt{n})$ (by
\Cref{uniform_clt_lemma}) and
$-H_{\theta}(\hat\Gamma'W\hat\Gamma)^{-1}\hat\Gamma'W- k_{\theta, P}'=o_P(1)$
uniformly over $(\theta, P)\in\mathcal{S}_n$.

\subsection{Auxiliary results}\label{auxiliary_sec}

This section contains auxiliary results used in \Cref*{efficiency_sec_append}.
\Cref{sec:repl-mathbbrd_th-wit} shows that optimizing length over a set of the
form $\mathcal{G}=\mathbb{R}^{d_\theta}\times\mathcal{D}$ is without loss of
generality, as claimed in \Cref*{achieving_bound_sec}.
\Cref{optimal_weights_continuity_sec} contains a result on the continuity of the
optimal weights with respect to $\delta$, $\Gamma$, $\Sigma$ and $H$.
\Cref{clt_lln_sec} states a law of large numbers and central limit theorem for
triangular arrays.

It will be convenient to state some of these results in the general setup of \citet{donoho94}, \citet{low95}, and
\citet{ArKo18optimal}. Using the notation in \citet{ArKo18optimal},
the between class modulus problem is given by
\begin{align}
  \label{eq:general_modulus}
  \omega(\delta)=\omega(\delta;\mathcal{F}, \mathcal{G}, L, K)=\sup L(g-f);\text{s.t}\;
  \|K(g-f)\|\le \delta, f\in\mathcal{F}, g\in\mathcal{G},
\end{align}
where $\mathcal{F}$ and $\mathcal{G}$ are convex sets with
$\mathcal{G}\subseteq\mathcal{F}$, $L$ is a linear functional and $K$ is a
linear operator from $\mathcal{F}$ to a Hilbert space with norm $\|\cdot\|$. In
our case, this is given by \Cref*{between_class_modulus_eq} in the main text,
which fits into this setting with $(\theta', c')'$ playing the role of $f$,
$\mathbb{R}^{d_\theta}\times \mathcal{C}$ playing the role of $\mathcal{F}$, $K$
given by the transformation $(\theta', c')'\mapsto -\Gamma \theta+c$, and with
the norm defined using the inner product $\langle x, y\rangle=x'\Sigma^{-1}y$.
The linear functional $L$ is given by $(\theta', c')'\mapsto H\theta$.

\subsubsection{Replacing \texorpdfstring{$\mathbb{R}^{d_\theta}\times\mathcal{D}$ with a general
  set $\mathcal{G}$}{Rd times D with a general set G}}\label{sec:repl-mathbbrd_th-wit}

In \Cref*{achieving_bound_sec}, we mentioned that directing power at sets that
do not restrict $\theta$ is without loss of generality when we require coverage
over a set that does not make local restrictions on $\theta$. This holds by the
following lemma (applied with
$\mathcal{U}=\mathbb{R}^{d_\theta}\times \{0\}^{d_g}$).

\begin{lemma}\label{invariance_lemma}
  Let $\mathcal{U}$ be a set with $0\in\mathcal{U}$ such that
  $\mathcal{F}=\mathcal{F}-\mathcal{U}$ (i.e. $\mathcal{F}$ is invariant to
  adding elements in $\mathcal{U}$). Then, for any solution
  $\tilde f^*, \tilde g^*$ to the modulus problem
\begin{equation*}
\sup L(g-f)\;\text{s.t.}\;
  \|K(g-f)\|\le \delta, f\in\mathcal{F}, g\in\mathcal{G}+\mathcal{U},
\end{equation*}
where $K$ is a linear operator, there is a solution $f^*, g^*$ to the modulus
problem~\eqref{eq:general_modulus} for $\mathcal{F}$ and $\mathcal{G}$ with
$g^*-f^*=\tilde g^*-\tilde f^*$. Furthermore, any solution to the modulus
problem~\eqref{eq:general_modulus} for $\mathcal{F}$ and $\mathcal{G}$ is also a
solution to the modulus problem for $\mathcal{F}$ and $\mathcal{G}+\mathcal{U}$.
\end{lemma}
\begin{proof}
  Let $\tilde f$, $\tilde g+\tilde u$ be a solution to the modulus problem for
  $\mathcal{F}$ and $\mathcal{G}+ \mathcal{U}$ with $\tilde g\in\mathcal{G}$ and
  $\tilde u\in\mathcal{U}$. Then $f=\tilde f-\tilde u$, and $g=\tilde g$ is
  feasible for $\mathcal{F}$ and $\mathcal{G}$ and achieves the same value of
  the objective function. Since it achieves the maximum for the objective
  function over the larger set $\mathcal{F}\times (\mathcal{G}+\mathcal{U})$ and
  is in $\mathcal{F}\times\mathcal{G}$, it must maximize the objective function
  over $\mathcal{F}\times\mathcal{G}$. Thus, $f, g$ achieves the modulus for
  $\mathcal{F}$ and $\mathcal{G}$ and also for $\mathcal{F}$ and
  $\mathcal{G}+\mathcal{U}$. Since the modulus for $\mathcal{F}$ and
  $\mathcal{G}$ is the same as the modulus over $\mathcal{F}$ and the larger set
  $\mathcal{G}+\mathcal{U}$, it also follows that any solution to the former
  modulus problem is a solution to the latter modulus problem.
\end{proof}

\subsubsection{Continuity of optimal weights}\label{optimal_weights_continuity_sec}

We first give some lemmas under the general setup~\eqref{eq:general_modulus}.

\begin{lemma}\label{general_delta_continuity_lemma}
  For each $\delta$, let $(f^*_\delta, g^*_\delta)$ be a solution to the modulus
  problem~\eqref{eq:general_modulus}, and let
  $h^*_\delta=g^*_\delta-f^*_\delta$. Let $\delta_0,\delta_1$ be given, and
  suppose that $\omega$ is strictly increasing on an open interval containing
  $\delta_0$ and $\delta_1$, and that a solution to the modulus problem exists
  for $\delta_0$ and $\delta_1$. Then $Kh^*_{\delta_0}$ and $Kh^*_{\delta_1}$
  are defined uniquely (i.e.\ they do not depend on the particular solution
  $(f^*_\delta, g^*_\delta)$) and
\begin{equation*}
\|Kh^*_{\delta_0}-Kh^*_{\delta_1}\|^2\le 2|\delta_1^2-\delta_0^2|
\end{equation*}
\end{lemma}
\begin{proof}
Let $f_0=f^*_{\delta_0}$, $f_1=f^*_{\delta_1}$ and similarly for $g_0$, $g_1$, $h_0$ and $h_1$.  Let $\tilde h=(h_0+h_1)/2$.  Note that $\tilde h=\tilde g-\tilde f$ where $\tilde g=(g_0+g_1)/2\in\mathcal{G}$ and $\tilde f=(f_0+f_1)/2\in\mathcal{F}$ by convexity.  Thus, $\omega(\|K\tilde h\|)\ge L\tilde h=[\omega(\delta_0)+\omega(\delta_1)]/2\ge \min\{\omega(\delta_0), \omega(\delta_1)\}$.  From this and the fact that $\omega$ is strictly increasing on an open interval containing $\delta_0$ and $\delta_1$, it follows that $\|K\tilde h\|\ge \min\{\delta_0,\delta_1\}$.

Note that $h_1=\tilde h+(h_1-h_0)/2$ and $\langle K\tilde h, K(h_1-h_0)/2\rangle=\|Kh_1\|^2/4-\|Kh_0\|^2/4=(\delta_1^2-\delta_0^2)/4$ (the last equality uses the fact that the constraint on $\|K(f-g)\|$ binds at any $\delta$ at which the modulus is strictly increasing).  Thus,
\begin{align*}
\delta_1^2&=\|Kh_1\|^2=\|K\tilde h\|^2+\|K(h_1-h_0)/2\|^2+(\delta_1^2-\delta_0^2)/2 \\
&\ge \min\{\delta_0^2,\delta_1^2\}+\|K(h_1-h_0)/2\|^2+(\delta_1^2-\delta_0^2)/2.
\end{align*}
Thus,
$\|K(h_1-h_0)\|^2/4\le \delta_1^2-\min\{\delta_0^2,\delta_1^2\}-(\delta_1^2-\delta_0^2)/2 =|\delta_1^2-\delta_0^2|/2$
as claimed.  The fact that $Kh^*_{\delta_0}$ is defined uniquely follows from applying the result with $\delta_1$ and $\delta_0$ both given by $\delta_0$.
\end{proof}

\begin{lemma}\label{argmax_lemma}
  For each $\delta$, let $(f^*_\delta, g^*_\delta)$ be a solution to the modulus
  problem~\eqref{eq:general_modulus}, and let
  $h^*_\delta=g^*_\delta-f^*_\delta$. Let $\delta_0$ and $\varepsilon>0$ be
  given, and suppose that $\omega$ is strictly increasing in a neighborhood of
  $\delta_0$, and that the modulus is achieved at $\delta_0$. Let
  $g\in\mathcal{G}$ and $f\in\mathcal{F}$ satisfy
  $L(g-f)> \omega(\delta_0)-\varepsilon$ with $\|K(g-f)\|\le \delta_0$, and let
  $h=g-f$. Then
\begin{equation*}
\|K(h-h^*_{\delta_0})\|^2< 4[\delta_0^2-\omega^{-1}(\omega(\delta_0)-\varepsilon)^2].
\end{equation*}
\end{lemma}
\begin{proof}
Let $h^*=h^*_{\delta_0}$, $g^*=g^*_{\delta_0}$ and $f^*=f^*_{\delta_0}$.
Using the fact that $\langle K(h+h^*)/2,K(h-h^*)/2\rangle=\|Kh\|^2/4-\|Kh^*\|^2/4$, we have
\begin{equation*}
\|Kh\|^2=\|K(h+h^*)/2\|^2+\|K(h-h^*)/2\|^2+\|Kh\|^2/2-\|Kh^*\|^2/2.
\end{equation*}
Rearranging this gives
\begin{equation}\label{h_hstar_difference_eq}
  \|K(h-h^*)/2\|^2
  =[\|Kh\|^2+\|Kh^*\|^2]/2
  -\|K(h+h^*)/2\|^2.
\end{equation}
Let $\delta'=\omega^{-1}(\omega(\delta_0)-\varepsilon)$. Since
$Lh> \omega(\delta')$ and $Lh^*=\omega(\delta_0)$, it follows that
$L(h+h^*)/2> [\omega(\delta')+\omega(\delta)]/2\ge \omega(\delta')$. Since
$(h+h^*)/2=(g+g^*)/2-(f+f^*)/2$ with $(g+g^*)/2\in\mathcal{G}$ and
$(f+f^*)/2\in\mathcal{F}$, this means that $\|K(h+h^*)/2\|>\delta'$. Using this
and the fact that $[\|Kh\|^2+\|Kh^*\|^2]/2\le \delta_0^2$, it follows that
$\|K(h-h^*)/2\|^2\le \delta_0^2-{\delta'}^2$ as claimed.
\end{proof}

\begin{lemma}\label{general_h_continuity_lemma}
  Let
  $h^*_{\delta, \mathcal{F}, \mathcal{G}, L, K} = g^*_{\delta, \mathcal{F},
    \mathcal{G}, L, K}-f^*_{\delta, \mathcal{F}, \mathcal{G}, L, K}$ where
  $g^*_{\delta, \mathcal{F}, \mathcal{G}, L, K}$,
  $f^*_{\delta, \mathcal{F}, \mathcal{G}, L, K}$ is a solution to the modulus
  problem~\eqref{eq:general_modulus}. Let
  $\delta_0,L_0, K_0, \mathcal{F}_0,\mathcal{G}_0$ and
  $\{\delta_n,L_n,K_n, \mathcal{F}_n, \mathcal{G}_n\}_{n=1}^{\infty}$ be given.

  Let
  $\mathcal{H}(\delta, K, \mathcal{F}, \mathcal{G})=\{g-f: f\in\mathcal{F},
  g\in\mathcal{G}, \|K(g-f)\|\le \delta\}$ denote the feasible set of values of
  $g-f$ for the modulus problem for $\delta, K, \mathcal{F}, \mathcal{G}$.
  Suppose that, for any $\varepsilon>0$, we have, for large enough $n$,
  $\mathcal{H}(\delta_0-\varepsilon, K_0,\mathcal{F}_0,\mathcal{G}_0)\subseteq
  \mathcal{H}(\delta_n,K_n, \mathcal{F}_n, \mathcal{G}_n)\subseteq
  \mathcal{H}(\delta_0+\varepsilon, K_0,\mathcal{F}_0,\mathcal{G}_0)$. Suppose
  also that $L_n h - L_0 h\to 0$ and $\|(K_n-K_0)h\|\to 0$ uniformly over $h$ in
  $\mathcal{H}(\delta_0+\varepsilon, K_0,\mathcal{F}_0,\mathcal{G}_0)$ for
  $\varepsilon$ small enough. Suppose also that
  $\omega(\delta;\mathcal{F}_0,\mathcal{G}_0,L_0,K_0)$ is strictly increasing
  for $\delta$ in a neighborhood of $\delta_0$. Then
  $\|K_{n}h^*_{\delta_n, \mathcal{F}_n,
    \mathcal{G}_n,L_n,K_n}-K_0h^*_{\delta_0,\mathcal{F}_0,\mathcal{G}_0,L_0,K_0}\|\to
  0$ and
  $L_{n} h^*_{\delta_n, \mathcal{F}_n,
    \mathcal{G}_n,L_n,K_n}-L_0h^*_{\delta_0,\mathcal{F}_0,\mathcal{G}_0,L_0,K_0}\to
  0$.
\end{lemma}
\begin{proof}
For any $\varepsilon>0$,
$g^*_{\delta_0-\varepsilon, \mathcal{F}_0,\mathcal{G}_0, L_0, K_0}$,
$f^*_{\delta_0-\varepsilon, \mathcal{F}_0,\mathcal{G}_0, L_0, K_0}$
is feasible for the modulus problem under
$\delta_n, \mathcal{F}_n, \mathcal{G}_n,L_n, K_n$ for large enough $n$.
Thus, for large enough $n$,
\begin{equation*}
\omega(\delta_0-\varepsilon, \mathcal{F}_0,\mathcal{G}_0,L_0,K_0)
=L h^*_{\delta_0-\varepsilon, \mathcal{F}_0,\mathcal{G}_0,L_0,K_0}
\le L_n h^*_{\delta_n, \mathcal{F}_n, \mathcal{G}_n,L_n,K_n}.
\end{equation*}
Taking limits and using the fact that $(L_n-L) h^*_{\delta_n, \mathcal{F}_n, \mathcal{G}_n,L_n,K_n}\to 0$, it follows that,
\begin{equation*}
\omega(\delta_0-\varepsilon;\mathcal{F}_0,\mathcal{G}_0,L_0,K_0)
-\varepsilon
\le L h^*_{\delta_n, \mathcal{F}_n, \mathcal{G}_n,L_n,K_n}
\end{equation*}
for large enough $n$. By continuity of the modulus in $\delta$, for any $\eta>0$
the left-hand side is strictly greater than
$\omega(\delta_0+\varepsilon;\mathcal{F}_0,\mathcal{G}_0,L_0,K_0)-\eta$ for
$\varepsilon$ small enough. Since
$g^*_{\delta_n, \mathcal{F}_n, \mathcal{G}_n,L_n,K_n}$,
$f^*_{\delta_n, \mathcal{F}_n, \mathcal{G}_n,L_n,K_n}$ is feasible for
$\delta_0+\varepsilon, \mathcal{F}_0,\mathcal{G}_0,L_0,K_0$ for $n$ large
enough, it follows from \Cref{argmax_lemma} that
\begin{multline*}
\|K_{0}(h^*_{\delta_n, \mathcal{F}_n, \mathcal{G}_n,L_n,K_n}-
h^*_{\delta_0+\varepsilon, \mathcal{F}_0,\mathcal{G}_0,L_0,K_0})\| \\
< 4 [(\delta_0+\varepsilon)^2
-\omega^{-1}(\omega(\delta_0+\varepsilon;\mathcal{F}_0,\mathcal{G}_0,L_0,K_0)-\eta;\mathcal{F}_0,\mathcal{G}_0,L_0,K_0)^2].
\end{multline*}
By continuity of the modulus and inverse modulus, the right-hand side can be made arbitrarily close to zero by taking $\varepsilon$ and $\eta$ small.  Thus,
\begin{equation*}
  \lim_{\varepsilon\downarrow 0}\limsup_n
  \|K_{0}(h^*_{\delta_n, \mathcal{F}_n, \mathcal{G}_n,L_n,K_n}-
  h^*_{\delta_0+\varepsilon, \mathcal{F}_0,\mathcal{G}_0,L_0,K_0})\|=0.
\end{equation*}
It then follows from \Cref{general_delta_continuity_lemma} that
$\lim_{n\to\infty} \|K_{0}(h^*_{\delta_n, \mathcal{F}_n, \mathcal{G}_n,L_n,K_n}-
h^*_{\delta_0,\mathcal{F}_0,\mathcal{G}_0,L_0,K_0})\|=0$.
The result then follows from the assumption that $\|(K_0-K_n)h\|\to 0$ uniformly over $\mathcal{H}(\delta_0+\varepsilon, K_0,\mathcal{F}_0,\mathcal{G}_0)$.
\end{proof}

We now specialize to our setting. Let
$f^*_{\delta, H, \Gamma, \Sigma}=({s_0^*}', {c_0^*}')$ and
$g^*_{\delta, H, \Gamma, \Sigma}=({s_1^*}', {c_1^*}')$ denote solutions to the
modulus problem in \Cref*{between_class_modulus_eq} with
$\mathcal{F}=\mathbb{R}^{d_\theta}\times \mathcal{C}$ and
$\mathcal{G}=\mathbb{R}^{d_\theta}\times \mathcal{D}$. Let
$\omega(\delta;H, \Gamma, \Sigma)=\omega(\delta;\mathbb{R}^{d_\theta}\times
\mathcal{C}, \mathbb{R}^{d_\theta}\times \mathcal{D}, H, \Gamma, \Sigma)$ denote
the modulus. Let
$h^*_{\delta, H, \Gamma, \Sigma}=f^*_{\delta, H, \Gamma, \Sigma}-g^*_{\delta, H,
  \Gamma, \Sigma}$ and let
$K_{\Gamma, \Sigma}=\Sigma^{-1/2}(-\Gamma, I_{d_g\times d_g})$. Note that
$h^*_{\delta, H, \Gamma, \Sigma, \mathcal{C}}=({s^*}', {c^*}')'$ where
$({s^*}', {c^*}')'$ solves
\begin{equation}\label{difference_modulus_eq}
  \sup H s\;\text{s.t.}\;
  (c-\Gamma s)'\Sigma^{-1}(c-\Gamma s)\le \delta^2,
  c\in\mathcal{D}-\mathcal{C}, s\in \mathbb{R}^{d_\theta}.
\end{equation}
Furthermore, a solution does indeed exist so long as $\mathcal{C}$ and $\mathcal{D}$ are compact and $\Gamma$ and $\Sigma$ are full rank, since this implies that the constraint set is compact.

Let $\delta_0$, $H_0$, $\Gamma_0$ and $\Sigma_0$ be such that $\delta_0>0$,
$H_0\ne 0$ and such that $\Gamma_0$ and $\Sigma_0$ are full rank. We wish to
show that $K_{\Gamma, \Sigma}h^*_{\delta, H, \Gamma, \Sigma}$ is continuous as a
function of $\delta$, $H$, $\Gamma$ and $\Sigma$ at
$(\delta_0,H_0,\Gamma_0,\Sigma_0)$. To this end, let $\delta_n$, $H_n$,
$\Gamma_n$ and $\Sigma_n$ be arbitrary sequences converging to $\delta_0$,
$H_0$, $\Gamma_0$ and $\Sigma_0$ (with $\Sigma_n$ symmetric and positive
semi-definite for each $n$). We will apply \Cref{general_h_continuity_lemma}. To
verify the conditions of this lemma, first note that the modulus is strictly
increasing by translation invariance \citep[see Section C.2
in][]{ArKo18optimal}. The conditions on uniform convergence of $(L_n-L)h$ and
$(K_n-K)h$ follow since the constraint set for $h=g-f$ is compact. The condition
on $\mathcal{H}(\delta, K, \mathcal{F}, \mathcal{G})$ follows because
$(c-\Gamma s)'\Sigma^{-1}(c-\Gamma s)$ is continuous in $\Sigma^{-1}$ and
$\Gamma$ uniformly over $c$ and $s$ in any compact set, and there exists a
compact set that contains the constraint set for all $n$ large enough. We record
these results and some of their implications in a lemma.

\begin{lemma}\label{argmax_continuity_lemma}
  Let $\mathcal{C}$ and $\mathcal{D}$ be compact and let
  $c^*_{\delta, H, \Gamma, \Sigma}$, $s^*_{\delta, H, \Gamma, \Sigma}$ denote a
  solution to~\eqref{difference_modulus_eq}. Let $\mathcal{A}$ denote the set of
  $(\delta, H, \Gamma, \Sigma)$ such that $\delta>0$,
  $H\in\mathbb{R}^{d_\theta}\backslash \{0\}$, $\Gamma$ is a full rank
  $d_g\times d_\theta$ matrix and $\Sigma$ is a (strictly) positive definite
  $d_g\times d_g$ matrix. Then
  $\Sigma^{-1/2}(s^*_{\delta, H,\Gamma, \Sigma}-\Gamma
  c^*_{\delta, H,\Gamma, \Sigma})$ is defined uniquely for any
  $(\delta, H,\Gamma, \Sigma)\in\mathcal{A}$. Furthermore, the mappings
  $(\delta, H,\Gamma, \Sigma)\mapsto
  \Sigma^{-1/2}(s^*_{\delta, H,\Gamma, \Sigma}-\Gamma
  c^*_{\delta, H,\Gamma, \Sigma})$,
  \begin{equation*}
    k(\delta, H,\Gamma, \Sigma)'=\frac{(s^*_{\delta, H,\Gamma, \Sigma} - \Gamma
      c^*_{\delta, H,\Gamma, \Sigma})\Sigma^{-1}}{(s^*_{\delta, H,\Gamma, \Sigma} -
      \Gamma c^*_{\delta, H,\Gamma, \Sigma})\Sigma^{-1}\Gamma H/H'H}
    \quad\text{and}\quad
    \omega(\delta;H, \Gamma, \Sigma)=Hs^*_{\delta, H,\Gamma, \Sigma}
\end{equation*}
are continuous functions on $\mathcal{A}$.
\end{lemma}

\subsubsection{CLT and LLN for triangular arrays}\label{clt_lln_sec}

To verify the conditions of \Cref*{achieving_bound_sec}, a CLT and LLN for
triangular arrays (applied to the triangular arrays that arise from arbitrary
sequences $P_n\in\mathcal{P}$) are useful. We state them here for convenience.

\begin{lemma}\label{uniform_clt_lemma}
  Let $\varepsilon>0$ be given. Let $\{v_i\}_{i=1}^{n}$ be an iid sequence of
  scalar valued random variables and let $\mathcal{P}$ be a set of probability
  distributions with $E_{P}v_i^{2+\varepsilon}\le 1/\varepsilon$,
  $1/\varepsilon\le E_{P}v_i^2$ and $E_{P}v_i=0$ for all $P\in\mathcal{P}$. Then
\begin{equation*}
  \sup_{P\in\mathcal{P}}\sup_{t\in\mathbb{R}} \left|P\left(\frac{1}{\sqrt{n}}\sum_{i=1}^{n}
      v_i/\sqrt{\text{var}_P(v_i)}\le t\right)-\Phi(t)\right|\to 0.
\end{equation*}
\end{lemma}
\begin{proof}
  The result is immediate from Lemma 11.4.1 in \citet{LeRo05} applied to
  arbitrary sequences $P\in\mathcal{P}$ and the fact that convergence to a
  continuous cdf is always uniform over the point at which the cdf is evaluated
  \citep[Lemma 2.11 in][]{van_der_vaart_asymptotic_1998}.
\end{proof}
\begin{lemma}\label{uniform_lln_lemma}
  Let $\varepsilon>0$ be given. Let $\{v_i\}_{i=1}^{n}$ be an iid sequence of
  scalar valued random variables and let $\mathcal{P}$ be a set of probability
  distributions with $E_{P}|v_i|^{1+\varepsilon}\le 1/\varepsilon$ for all
  $P\in\mathcal{P}$. Then $\frac{1}{n}\sum_{i=1}^{n} v_i-E_{P}v_i=o_P(1)$
  uniformly over $P\in\mathcal{P}$.
\end{lemma}
\begin{proof}
  The stronger result
  $\sup_{P\in\mathcal{P}}E_{P}\left|\frac{1}{n}\sum_{i=1}^{n} v_{i} -
    E_{P}v_{i}\right|^{1+\min\{\varepsilon,2\}} \to 0$ follows from Theorem 3 in
  \citet{von_bahr_inequalities_1965}.
\end{proof}

\end{appendices}

\bibliography{../../np-testing-library}

%% file: preamble.tex
\usepackage[utf8]{inputenc}%
\usepackage[tracking=true, letterspace=50, expansion=false]{microtype}%
\usepackage{amsmath}%
\usepackage{amsfonts}%
\usepackage{amsthm}%
\usepackage[margin=1in]{geometry}%
\usepackage{setspace}%
\usepackage{natbib}%
\usepackage{booktabs}%
\usepackage{enumitem}%

\usepackage{tikz}

\newtheorem{theorem}{Theorem}[section]
\newtheorem{lemma}{Lemma}[section]%
\newtheorem{assumption}{Assumption}[section]

\theoremstyle{definition}%
\newtheorem{example}{Example}[section]
\newtheorem{remark}{Remark}[section]

\usepackage[title]{appendix}

\DeclareMathOperator{\cv}{cv}  
\DeclareMathOperator{\maxbias}{\overline{bias}} 
\DeclareMathOperator{\minbias}{\underline{bias}} 
\DeclareMathOperator*{\argmin}{argmin} 
\DeclareMathOperator{\sign}{sign} 
\DeclareMathOperator{\se}{se}
\newcommand{\indist}{\stackrel{d}{\to}} 
\newcommand{\inprob}{\stackrel{p}{\to}} 
\newcommand{\Cglob}{{\widetilde{\mathcal{C}}}}
\newcommand{\cglob}{{\tilde{c}}}

\RequirePackage{mathtools} %
\DeclarePairedDelimiter\abs{\lvert}{\rvert} 
\DeclarePairedDelimiter\norm{\lVert}{\rVert} 
\DeclarePairedDelimiter\hor{[}{)} 
\DeclarePairedDelimiter\hol{(}{]} 
\DeclarePairedDelimiter\indicatorfence{\{}{\}}
\newcommand\1{\operatorname{\mathbb{I}}\indicatorfence}

\newcommand{\tSigma}{S} %
\newcommand{\tlambda}{\lambda} %
\newcommand{\tGamma}{G} %
\newcommand{\tk}{\kappa} %
\newcommand{\Mbound}{M} %

\usepackage{xcolor}%
\definecolor{webbrown}{rgb}{.6,0,0}%
\usepackage{hyperref}

\usepackage{cleveref}
\Crefname{equation}{Eq.}{Eqs.}
\crefname{assumption}{assumption}{assumptions}
\crefname{remark}{remark}{remarks}
\crefname{proposition}{proposition}{propositions}
\crefname{sappsec}{Supplemental Appendix}{Supplemental Appendices}
\crefname{sappsubsec}{Supplemental Appendix}{Supplemental Appendices}
\crefname{sappsubsubsec}{Supplemental Appendix}{Supplemental Appendices}
\crefname{appsec}{Appendix}{Appendices}

%% file: blp_l2.tex
\begin{tikzpicture}[x=1pt,y=1pt]
\definecolor{fillColor}{RGB}{255,255,255}
\path[use as bounding box,fill=fillColor,fill opacity=0.00] (0,0) rectangle (411.94,303.53);
\begin{scope}
\path[clip] (  0.00,  0.00) rectangle (411.94,303.53);
\definecolor{drawColor}{RGB}{255,255,255}
\definecolor{fillColor}{RGB}{255,255,255}

\path[draw=drawColor,line width= 0.6pt,line join=round,line cap=round,fill=fillColor] (  0.00,  0.00) rectangle (411.94,303.53);
\end{scope}
\begin{scope}
\path[clip] ( 90.20, 30.40) rectangle (406.94,298.53);
\definecolor{fillColor}{RGB}{255,255,255}

\path[fill=fillColor] ( 90.20, 30.40) rectangle (406.94,298.53);
\definecolor{drawColor}{RGB}{81,81,81}

\path[draw=drawColor,line width= 0.6pt,line join=round] (203.44, 42.37) -- (239.96, 42.37);

\path[draw=drawColor,line width= 0.6pt,line join=round] (200.68, 66.31) -- (242.72, 66.31);

\path[draw=drawColor,line width= 0.6pt,line join=round] (198.88, 90.25) -- (244.52, 90.25);

\path[draw=drawColor,line width= 0.6pt,line join=round] (203.42,114.19) -- (239.98,114.19);

\path[draw=drawColor,line width= 0.6pt,line join=round] (184.51,138.13) -- (258.89,138.13);

\path[draw=drawColor,line width= 0.6pt,line join=round] (191.17,162.07) -- (252.23,162.07);

\path[draw=drawColor,line width= 0.6pt,line join=round] (167.60,186.01) -- (275.80,186.01);

\path[draw=drawColor,line width= 0.6pt,line join=round] (185.20,209.95) -- (258.20,209.95);

\path[draw=drawColor,line width= 0.6pt,line join=round] (168.71,233.89) -- (274.69,233.89);

\path[draw=drawColor,line width= 0.6pt,line join=round] (142.45,257.83) -- (300.95,257.83);

\path[draw=drawColor,line width= 0.6pt,line join=round] (104.60,281.78) -- (338.80,281.78);
\definecolor{drawColor}{RGB}{144,144,144}

\path[draw=drawColor,line width= 0.6pt,line join=round] (207.64, 47.15) -- (244.07, 47.15);

\path[draw=drawColor,line width= 0.6pt,line join=round] (217.74, 71.10) -- (255.65, 71.10);

\path[draw=drawColor,line width= 0.6pt,line join=round] (255.31, 95.04) -- (295.77, 95.04);

\path[draw=drawColor,line width= 0.6pt,line join=round] (208.05,118.98) -- (244.48,118.98);

\path[draw=drawColor,line width= 0.6pt,line join=round] (222.20,142.92) -- (260.78,142.92);

\path[draw=drawColor,line width= 0.6pt,line join=round] (155.55,166.86) -- (204.28,166.86);

\path[draw=drawColor,line width= 0.6pt,line join=round] (308.59,190.80) -- (353.98,190.80);

\path[draw=drawColor,line width= 0.6pt,line join=round] (268.39,214.74) -- (309.94,214.74);

\path[draw=drawColor,line width= 0.6pt,line join=round] (124.80,238.68) -- (178.20,238.68);

\path[draw=drawColor,line width= 0.6pt,line join=round] (311.10,262.62) -- (358.30,262.62);

\path[draw=drawColor,line width= 0.6pt,line join=round] (289.64,286.56) -- (392.54,286.56);
\definecolor{fillColor}{RGB}{81,81,81}

\path[fill=fillColor] (221.70, 42.37) circle (  2.85);

\path[fill=fillColor] (221.70, 66.31) circle (  2.85);

\path[fill=fillColor] (221.70, 90.25) circle (  2.85);

\path[fill=fillColor] (221.70,114.19) circle (  2.85);

\path[fill=fillColor] (221.70,138.13) circle (  2.85);

\path[fill=fillColor] (221.70,162.07) circle (  2.85);

\path[fill=fillColor] (221.70,186.01) circle (  2.85);

\path[fill=fillColor] (221.70,209.95) circle (  2.85);

\path[fill=fillColor] (221.70,233.89) circle (  2.85);

\path[fill=fillColor] (221.70,257.83) circle (  2.85);

\path[fill=fillColor] (221.70,281.78) circle (  2.85);
\definecolor{fillColor}{RGB}{144,144,144}

\path[fill=fillColor] (225.85, 51.59) --
	(229.70, 44.94) --
	(222.01, 44.94) --
	cycle;

\path[fill=fillColor] (236.70, 75.53) --
	(240.54, 68.88) --
	(232.85, 68.88) --
	cycle;

\path[fill=fillColor] (275.54, 99.47) --
	(279.38, 92.82) --
	(271.70, 92.82) --
	cycle;

\path[fill=fillColor] (226.26,123.41) --
	(230.11,116.76) --
	(222.42,116.76) --
	cycle;

\path[fill=fillColor] (241.49,147.36) --
	(245.33,140.70) --
	(237.65,140.70) --
	cycle;

\path[fill=fillColor] (179.91,171.30) --
	(183.76,164.64) --
	(176.07,164.64) --
	cycle;

\path[fill=fillColor] (331.28,195.24) --
	(335.12,188.58) --
	(327.44,188.58) --
	cycle;

\path[fill=fillColor] (289.16,219.18) --
	(293.01,212.52) --
	(285.32,212.52) --
	cycle;

\path[fill=fillColor] (151.50,243.12) --
	(155.34,236.46) --
	(147.66,236.46) --
	cycle;

\path[fill=fillColor] (334.70,267.06) --
	(338.55,260.40) --
	(330.86,260.40) --
	cycle;

\path[fill=fillColor] (341.09,291.00) --
	(344.94,284.34) --
	(337.25,284.34) --
	cycle;
\definecolor{drawColor}{RGB}{81,81,81}

\path[draw=drawColor,line width= 0.6pt,line join=round] (221.70, 38.78) --
	(221.70, 45.96);

\path[draw=drawColor,line width= 0.6pt,line join=round] (221.70, 42.37) --
	(221.70, 42.37);

\path[draw=drawColor,line width= 0.6pt,line join=round] (221.70, 38.78) --
	(221.70, 45.96);

\path[draw=drawColor,line width= 0.6pt,line join=round] (227.19, 62.72) --
	(227.19, 69.90);

\path[draw=drawColor,line width= 0.6pt,line join=round] (227.19, 66.31) --
	(216.21, 66.31);

\path[draw=drawColor,line width= 0.6pt,line join=round] (216.21, 62.72) --
	(216.21, 69.90);

\path[draw=drawColor,line width= 0.6pt,line join=round] (229.15, 86.66) --
	(229.15, 93.84);

\path[draw=drawColor,line width= 0.6pt,line join=round] (229.15, 90.25) --
	(214.25, 90.25);

\path[draw=drawColor,line width= 0.6pt,line join=round] (214.25, 86.66) --
	(214.25, 93.84);

\path[draw=drawColor,line width= 0.6pt,line join=round] (222.15,110.60) --
	(222.15,117.78);

\path[draw=drawColor,line width= 0.6pt,line join=round] (222.15,114.19) --
	(221.25,114.19);

\path[draw=drawColor,line width= 0.6pt,line join=round] (221.25,110.60) --
	(221.25,117.78);

\path[draw=drawColor,line width= 0.6pt,line join=round] (243.57,134.54) --
	(243.57,141.72);

\path[draw=drawColor,line width= 0.6pt,line join=round] (243.57,138.13) --
	(199.83,138.13);

\path[draw=drawColor,line width= 0.6pt,line join=round] (199.83,134.54) --
	(199.83,141.72);

\path[draw=drawColor,line width= 0.6pt,line join=round] (236.91,158.48) --
	(236.91,165.66);

\path[draw=drawColor,line width= 0.6pt,line join=round] (236.91,162.07) --
	(206.49,162.07);

\path[draw=drawColor,line width= 0.6pt,line join=round] (206.49,158.48) --
	(206.49,165.66);

\path[draw=drawColor,line width= 0.6pt,line join=round] (260.48,182.42) --
	(260.48,189.60);

\path[draw=drawColor,line width= 0.6pt,line join=round] (260.48,186.01) --
	(182.92,186.01);

\path[draw=drawColor,line width= 0.6pt,line join=round] (182.92,182.42) --
	(182.92,189.60);

\path[draw=drawColor,line width= 0.6pt,line join=round] (242.88,206.36) --
	(242.88,213.54);

\path[draw=drawColor,line width= 0.6pt,line join=round] (242.88,209.95) --
	(200.52,209.95);

\path[draw=drawColor,line width= 0.6pt,line join=round] (200.52,206.36) --
	(200.52,213.54);

\path[draw=drawColor,line width= 0.6pt,line join=round] (259.37,230.30) --
	(259.37,237.48);

\path[draw=drawColor,line width= 0.6pt,line join=round] (259.37,233.89) --
	(184.03,233.89);

\path[draw=drawColor,line width= 0.6pt,line join=round] (184.03,230.30) --
	(184.03,237.48);

\path[draw=drawColor,line width= 0.6pt,line join=round] (285.63,254.24) --
	(285.63,261.43);

\path[draw=drawColor,line width= 0.6pt,line join=round] (285.63,257.83) --
	(157.77,257.83);

\path[draw=drawColor,line width= 0.6pt,line join=round] (157.77,254.24) --
	(157.77,261.43);

\path[draw=drawColor,line width= 0.6pt,line join=round] (323.48,278.18) --
	(323.48,285.37);

\path[draw=drawColor,line width= 0.6pt,line join=round] (323.48,281.78) --
	(119.92,281.78);

\path[draw=drawColor,line width= 0.6pt,line join=round] (119.92,278.18) --
	(119.92,285.37);
\definecolor{drawColor}{RGB}{144,144,144}

\path[draw=drawColor,line width= 0.6pt,line join=round] (225.85, 43.56) --
	(225.85, 50.75);

\path[draw=drawColor,line width= 0.6pt,line join=round] (225.85, 47.15) --
	(225.85, 47.15);

\path[draw=drawColor,line width= 0.6pt,line join=round] (225.85, 43.56) --
	(225.85, 50.75);

\path[draw=drawColor,line width= 0.6pt,line join=round] (237.97, 67.50) --
	(237.97, 74.69);

\path[draw=drawColor,line width= 0.6pt,line join=round] (237.97, 71.10) --
	(235.42, 71.10);

\path[draw=drawColor,line width= 0.6pt,line join=round] (235.42, 67.50) --
	(235.42, 74.69);

\path[draw=drawColor,line width= 0.6pt,line join=round] (278.50, 91.45) --
	(278.50, 98.63);

\path[draw=drawColor,line width= 0.6pt,line join=round] (278.50, 95.04) --
	(272.58, 95.04);

\path[draw=drawColor,line width= 0.6pt,line join=round] (272.58, 91.45) --
	(272.58, 98.63);

\path[draw=drawColor,line width= 0.6pt,line join=round] (226.38,115.39) --
	(226.38,122.57);

\path[draw=drawColor,line width= 0.6pt,line join=round] (226.38,118.98) --
	(226.15,118.98);

\path[draw=drawColor,line width= 0.6pt,line join=round] (226.15,115.39) --
	(226.15,122.57);

\path[draw=drawColor,line width= 0.6pt,line join=round] (242.17,139.33) --
	(242.17,146.51);

\path[draw=drawColor,line width= 0.6pt,line join=round] (242.17,142.92) --
	(240.81,142.92);

\path[draw=drawColor,line width= 0.6pt,line join=round] (240.81,139.33) --
	(240.81,146.51);

\path[draw=drawColor,line width= 0.6pt,line join=round] (186.25,163.27) --
	(186.25,170.45);

\path[draw=drawColor,line width= 0.6pt,line join=round] (186.25,166.86) --
	(173.58,166.86);

\path[draw=drawColor,line width= 0.6pt,line join=round] (173.58,163.27) --
	(173.58,170.45);

\path[draw=drawColor,line width= 0.6pt,line join=round] (333.84,187.21) --
	(333.84,194.39);

\path[draw=drawColor,line width= 0.6pt,line join=round] (333.84,190.80) --
	(328.72,190.80);

\path[draw=drawColor,line width= 0.6pt,line join=round] (328.72,187.21) --
	(328.72,194.39);

\path[draw=drawColor,line width= 0.6pt,line join=round] (291.39,211.15) --
	(291.39,218.33);

\path[draw=drawColor,line width= 0.6pt,line join=round] (291.39,214.74) --
	(286.94,214.74);

\path[draw=drawColor,line width= 0.6pt,line join=round] (286.94,211.15) --
	(286.94,218.33);

\path[draw=drawColor,line width= 0.6pt,line join=round] (158.97,235.09) --
	(158.97,242.27);

\path[draw=drawColor,line width= 0.6pt,line join=round] (158.97,238.68) --
	(144.03,238.68);

\path[draw=drawColor,line width= 0.6pt,line join=round] (144.03,235.09) --
	(144.03,242.27);

\path[draw=drawColor,line width= 0.6pt,line join=round] (337.64,259.03) --
	(337.64,266.21);

\path[draw=drawColor,line width= 0.6pt,line join=round] (337.64,262.62) --
	(331.77,262.62);

\path[draw=drawColor,line width= 0.6pt,line join=round] (331.77,259.03) --
	(331.77,266.21);

\path[draw=drawColor,line width= 0.6pt,line join=round] (373.40,282.97) --
	(373.40,290.15);

\path[draw=drawColor,line width= 0.6pt,line join=round] (373.40,286.56) --
	(308.79,286.56);

\path[draw=drawColor,line width= 0.6pt,line join=round] (308.79,282.97) --
	(308.79,290.15);
\end{scope}
\begin{scope}
\path[clip] (  0.00,  0.00) rectangle (411.94,303.53);
\definecolor{drawColor}{RGB}{144,144,144}

\path[draw=drawColor,line width= 0.6pt,line join=round] ( 90.20, 30.40) --
	( 90.20,298.53);
\end{scope}
\begin{scope}
\path[clip] (  0.00,  0.00) rectangle (411.94,303.53);
\definecolor{drawColor}{RGB}{68,68,68}

\node[text=drawColor,anchor=base east,inner sep=0pt, outer sep=0pt, scale=  0.80] at ( 85.70, 42.01) {None};

\node[text=drawColor,anchor=base east,inner sep=0pt, outer sep=0pt, scale=  0.80] at ( 85.70, 65.95) {D/F: \# cars};

\node[text=drawColor,anchor=base east,inner sep=0pt, outer sep=0pt, scale=  0.80] at ( 85.70, 89.89) {S/F: \# cars};

\node[text=drawColor,anchor=base east,inner sep=0pt, outer sep=0pt, scale=  0.80] at ( 85.70,113.83) {Supply: Miles/dollar};

\node[text=drawColor,anchor=base east,inner sep=0pt, outer sep=0pt, scale=  0.80] at ( 85.70,137.77) {All D/F};

\node[text=drawColor,anchor=base east,inner sep=0pt, outer sep=0pt, scale=  0.80] at ( 85.70,161.71) {All D/R};

\node[text=drawColor,anchor=base east,inner sep=0pt, outer sep=0pt, scale=  0.80] at ( 85.70,185.65) {All S/F};

\node[text=drawColor,anchor=base east,inner sep=0pt, outer sep=0pt, scale=  0.80] at ( 85.70,209.59) {All S/R};

\node[text=drawColor,anchor=base east,inner sep=0pt, outer sep=0pt, scale=  0.80] at ( 85.70,233.53) {All excluded demand};

\node[text=drawColor,anchor=base east,inner sep=0pt, outer sep=0pt, scale=  0.80] at ( 85.70,257.47) {All excluded supply};

\node[text=drawColor,anchor=base east,inner sep=0pt, outer sep=0pt, scale=  0.80] at ( 85.70,281.41) {All excluded};
\end{scope}
\begin{scope}
\path[clip] (  0.00,  0.00) rectangle (411.94,303.53);
\definecolor{drawColor}{RGB}{144,144,144}

\path[draw=drawColor,line width= 0.6pt,line join=round] ( 87.70, 44.76) --
	( 90.20, 44.76);

\path[draw=drawColor,line width= 0.6pt,line join=round] ( 87.70, 68.70) --
	( 90.20, 68.70);

\path[draw=drawColor,line width= 0.6pt,line join=round] ( 87.70, 92.64) --
	( 90.20, 92.64);

\path[draw=drawColor,line width= 0.6pt,line join=round] ( 87.70,116.58) --
	( 90.20,116.58);

\path[draw=drawColor,line width= 0.6pt,line join=round] ( 87.70,140.52) --
	( 90.20,140.52);

\path[draw=drawColor,line width= 0.6pt,line join=round] ( 87.70,164.47) --
	( 90.20,164.47);

\path[draw=drawColor,line width= 0.6pt,line join=round] ( 87.70,188.41) --
	( 90.20,188.41);

\path[draw=drawColor,line width= 0.6pt,line join=round] ( 87.70,212.35) --
	( 90.20,212.35);

\path[draw=drawColor,line width= 0.6pt,line join=round] ( 87.70,236.29) --
	( 90.20,236.29);

\path[draw=drawColor,line width= 0.6pt,line join=round] ( 87.70,260.23) --
	( 90.20,260.23);

\path[draw=drawColor,line width= 0.6pt,line join=round] ( 87.70,284.17) --
	( 90.20,284.17);
\end{scope}
\begin{scope}
\path[clip] (  0.00,  0.00) rectangle (411.94,303.53);
\definecolor{drawColor}{RGB}{144,144,144}

\path[draw=drawColor,line width= 0.6pt,line join=round] ( 90.20, 30.40) --
	(406.94, 30.40);
\end{scope}
\begin{scope}
\path[clip] (  0.00,  0.00) rectangle (411.94,303.53);
\definecolor{drawColor}{RGB}{144,144,144}

\path[draw=drawColor,line width= 0.6pt,line join=round] (105.14, 27.90) --
	(105.14, 30.40);

\path[draw=drawColor,line width= 0.6pt,line join=round] (156.45, 27.90) --
	(156.45, 30.40);

\path[draw=drawColor,line width= 0.6pt,line join=round] (207.76, 27.90) --
	(207.76, 30.40);

\path[draw=drawColor,line width= 0.6pt,line join=round] (259.06, 27.90) --
	(259.06, 30.40);

\path[draw=drawColor,line width= 0.6pt,line join=round] (310.37, 27.90) --
	(310.37, 30.40);

\path[draw=drawColor,line width= 0.6pt,line join=round] (361.68, 27.90) --
	(361.68, 30.40);
\end{scope}
\begin{scope}
\path[clip] (  0.00,  0.00) rectangle (411.94,303.53);
\definecolor{drawColor}{RGB}{68,68,68}

\node[text=drawColor,anchor=base,inner sep=0pt, outer sep=0pt, scale=  0.80] at (105.14, 20.39) {10};

\node[text=drawColor,anchor=base,inner sep=0pt, outer sep=0pt, scale=  0.80] at (156.45, 20.39) {20};

\node[text=drawColor,anchor=base,inner sep=0pt, outer sep=0pt, scale=  0.80] at (207.76, 20.39) {30};

\node[text=drawColor,anchor=base,inner sep=0pt, outer sep=0pt, scale=  0.80] at (259.06, 20.39) {40};

\node[text=drawColor,anchor=base,inner sep=0pt, outer sep=0pt, scale=  0.80] at (310.37, 20.39) {50};

\node[text=drawColor,anchor=base,inner sep=0pt, outer sep=0pt, scale=  0.80] at (361.68, 20.39) {60};
\end{scope}
\begin{scope}
\path[clip] (  0.00,  0.00) rectangle (411.94,303.53);
\definecolor{drawColor}{RGB}{68,68,68}

\node[text=drawColor,anchor=base,inner sep=0pt, outer sep=0pt, scale=  1.00] at (248.57,  6.94) {Average Markup in \%};
\end{scope}
\begin{scope}
\path[clip] (  0.00,  0.00) rectangle (411.94,303.53);
\definecolor{drawColor}{RGB}{255,255,255}
\definecolor{fillColor}{RGB}{255,255,255}

\path[draw=drawColor,line width= 0.6pt,line join=round,line cap=round,fill=fillColor] (337.57, 50.63) rectangle (381.29, 90.60);
\end{scope}
\begin{scope}
\path[clip] (  0.00,  0.00) rectangle (411.94,303.53);
\definecolor{drawColor}{RGB}{68,68,68}

\node[text=drawColor,anchor=base west,inner sep=0pt, outer sep=0pt, scale=  0.80] at (337.57, 84.32) {Estimate:};
\end{scope}
\begin{scope}
\path[clip] (  0.00,  0.00) rectangle (411.94,303.53);
\definecolor{drawColor}{RGB}{255,255,255}
\definecolor{fillColor}{RGB}{255,255,255}

\path[draw=drawColor,line width= 0.6pt,line join=round,line cap=round,fill=fillColor] (337.57, 65.08) rectangle (352.02, 79.54);
\end{scope}
\begin{scope}
\path[clip] (  0.00,  0.00) rectangle (411.94,303.53);
\definecolor{drawColor}{RGB}{81,81,81}

\path[draw=drawColor,line width= 0.6pt,line join=round] (344.79, 66.53) -- (344.79, 78.09);
\definecolor{fillColor}{RGB}{81,81,81}

\path[fill=fillColor] (344.79, 72.31) circle (  2.85);
\end{scope}
\begin{scope}
\path[clip] (  0.00,  0.00) rectangle (411.94,303.53);
\definecolor{drawColor}{RGB}{81,81,81}

\path[draw=drawColor,line width= 0.6pt,line join=round] (339.01, 72.31) -- (350.57, 72.31);
\end{scope}
\begin{scope}
\path[clip] (  0.00,  0.00) rectangle (411.94,303.53);
\definecolor{drawColor}{RGB}{255,255,255}
\definecolor{fillColor}{RGB}{255,255,255}

\path[draw=drawColor,line width= 0.6pt,line join=round,line cap=round,fill=fillColor] (337.57, 50.63) rectangle (352.02, 65.08);
\end{scope}
\begin{scope}
\path[clip] (  0.00,  0.00) rectangle (411.94,303.53);
\definecolor{drawColor}{RGB}{144,144,144}

\path[draw=drawColor,line width= 0.6pt,line join=round] (344.79, 52.08) -- (344.79, 63.64);
\definecolor{fillColor}{RGB}{144,144,144}

\path[fill=fillColor] (344.79, 62.29) --
	(348.64, 55.64) --
	(340.95, 55.64) --
	cycle;
\end{scope}
\begin{scope}
\path[clip] (  0.00,  0.00) rectangle (411.94,303.53);
\definecolor{drawColor}{RGB}{144,144,144}

\path[draw=drawColor,line width= 0.6pt,line join=round] (339.01, 57.86) -- (350.57, 57.86);
\end{scope}
\begin{scope}
\path[clip] (  0.00,  0.00) rectangle (411.94,303.53);
\definecolor{drawColor}{RGB}{68,68,68}

\node[text=drawColor,anchor=base west,inner sep=0pt, outer sep=0pt, scale=  0.70] at (356.02, 69.90) {Initial};
\end{scope}
\begin{scope}
\path[clip] (  0.00,  0.00) rectangle (411.94,303.53);
\definecolor{drawColor}{RGB}{68,68,68}

\node[text=drawColor,anchor=base west,inner sep=0pt, outer sep=0pt, scale=  0.70] at (356.02, 55.45) {Optimal};
\end{scope}
\end{tikzpicture}

%% file: blp_l2_lI.tex
\begin{tikzpicture}[x=1pt,y=1pt]
\definecolor{fillColor}{RGB}{255,255,255}
\path[use as bounding box,fill=fillColor,fill opacity=0.00] (0,0) rectangle (411.94,303.53);
\begin{scope}
\path[clip] (  0.00,  0.00) rectangle (411.94,303.53);
\definecolor{drawColor}{RGB}{255,255,255}
\definecolor{fillColor}{RGB}{255,255,255}

\path[draw=drawColor,line width= 0.6pt,line join=round,line cap=round,fill=fillColor] (  0.00,  0.00) rectangle (411.94,303.53);
\end{scope}
\begin{scope}
\path[clip] ( 90.20, 30.40) rectangle (406.94,298.53);
\definecolor{fillColor}{RGB}{255,255,255}

\path[fill=fillColor] ( 90.20, 30.40) rectangle (406.94,298.53);
\definecolor{drawColor}{RGB}{144,144,144}

\path[draw=drawColor,line width= 0.6pt,line join=round] (188.46, 47.15) -- (225.33, 47.15);

\path[draw=drawColor,line width= 0.6pt,line join=round] (198.68, 71.10) -- (237.06, 71.10);

\path[draw=drawColor,line width= 0.6pt,line join=round] (236.72, 95.04) -- (277.67, 95.04);

\path[draw=drawColor,line width= 0.6pt,line join=round] (188.87,118.98) -- (225.75,118.98);

\path[draw=drawColor,line width= 0.6pt,line join=round] (203.20,142.92) -- (242.25,142.92);

\path[draw=drawColor,line width= 0.6pt,line join=round] (135.73,166.86) -- (185.05,166.86);

\path[draw=drawColor,line width= 0.6pt,line join=round] (290.64,190.80) -- (336.59,190.80);

\path[draw=drawColor,line width= 0.6pt,line join=round] (249.95,214.74) -- (292.01,214.74);

\path[draw=drawColor,line width= 0.6pt,line join=round] (104.60,238.68) -- (158.66,238.68);

\path[draw=drawColor,line width= 0.6pt,line join=round] (293.19,262.62) -- (340.97,262.62);

\path[draw=drawColor,line width= 0.6pt,line join=round] (271.47,286.56) -- (375.62,286.56);
\definecolor{drawColor}{RGB}{81,81,81}

\path[draw=drawColor,line width= 0.6pt,line join=round] (188.46, 42.37) -- (225.33, 42.37);

\path[draw=drawColor,line width= 0.6pt,line join=round] (198.68, 66.31) -- (237.06, 66.31);

\path[draw=drawColor,line width= 0.6pt,line join=round] (236.72, 90.25) -- (277.67, 90.25);

\path[draw=drawColor,line width= 0.6pt,line join=round] (188.87,114.19) -- (225.75,114.19);

\path[draw=drawColor,line width= 0.6pt,line join=round] (203.17,138.13) -- (242.15,138.13);

\path[draw=drawColor,line width= 0.6pt,line join=round] (146.07,162.07) -- (193.78,162.07);

\path[draw=drawColor,line width= 0.6pt,line join=round] (285.81,186.01) -- (331.38,186.01);

\path[draw=drawColor,line width= 0.6pt,line join=round] (249.88,209.95) -- (291.64,209.95);

\path[draw=drawColor,line width= 0.6pt,line join=round] (156.05,233.89) -- (204.79,233.89);

\path[draw=drawColor,line width= 0.6pt,line join=round] (286.96,257.83) -- (333.84,257.83);

\path[draw=drawColor,line width= 0.6pt,line join=round] (318.03,281.78) -- (392.54,281.78);
\definecolor{fillColor}{RGB}{144,144,144}

\path[fill=fillColor] (206.90, 51.59) --
	(210.74, 44.94) --
	(203.05, 44.94) --
	cycle;

\path[fill=fillColor] (217.87, 75.53) --
	(221.71, 68.88) --
	(214.03, 68.88) --
	cycle;

\path[fill=fillColor] (257.19, 99.47) --
	(261.03, 92.82) --
	(253.35, 92.82) --
	cycle;

\path[fill=fillColor] (207.31,123.41) --
	(211.15,116.76) --
	(203.47,116.76) --
	cycle;

\path[fill=fillColor] (222.72,147.36) --
	(226.56,140.70) --
	(218.88,140.70) --
	cycle;

\path[fill=fillColor] (160.39,171.30) --
	(164.24,164.64) --
	(156.55,164.64) --
	cycle;

\path[fill=fillColor] (313.61,195.24) --
	(317.46,188.58) --
	(309.77,188.58) --
	cycle;

\path[fill=fillColor] (270.98,219.18) --
	(274.82,212.52) --
	(267.14,212.52) --
	cycle;

\path[fill=fillColor] (131.63,243.12) --
	(135.47,236.46) --
	(127.79,236.46) --
	cycle;

\path[fill=fillColor] (317.08,267.06) --
	(320.92,260.40) --
	(313.23,260.40) --
	cycle;

\path[fill=fillColor] (323.55,291.00) --
	(327.39,284.34) --
	(319.70,284.34) --
	cycle;
\definecolor{fillColor}{RGB}{81,81,81}

\path[fill=fillColor] (206.90, 42.37) circle (  2.85);

\path[fill=fillColor] (217.87, 66.31) circle (  2.85);

\path[fill=fillColor] (257.19, 90.25) circle (  2.85);

\path[fill=fillColor] (207.31,114.19) circle (  2.85);

\path[fill=fillColor] (222.66,138.13) circle (  2.85);

\path[fill=fillColor] (169.93,162.07) circle (  2.85);

\path[fill=fillColor] (308.59,186.01) circle (  2.85);

\path[fill=fillColor] (270.76,209.95) circle (  2.85);

\path[fill=fillColor] (180.42,233.89) circle (  2.85);

\path[fill=fillColor] (310.40,257.83) circle (  2.85);

\path[fill=fillColor] (355.29,281.78) circle (  2.85);
\definecolor{drawColor}{RGB}{144,144,144}

\path[draw=drawColor,line width= 0.6pt,line join=round] (206.90, 43.56) --
	(206.90, 50.75);

\path[draw=drawColor,line width= 0.6pt,line join=round] (206.90, 47.15) --
	(206.90, 47.15);

\path[draw=drawColor,line width= 0.6pt,line join=round] (206.90, 43.56) --
	(206.90, 50.75);

\path[draw=drawColor,line width= 0.6pt,line join=round] (219.16, 67.50) --
	(219.16, 74.69);

\path[draw=drawColor,line width= 0.6pt,line join=round] (219.16, 71.10) --
	(216.58, 71.10);

\path[draw=drawColor,line width= 0.6pt,line join=round] (216.58, 67.50) --
	(216.58, 74.69);

\path[draw=drawColor,line width= 0.6pt,line join=round] (260.19, 91.45) --
	(260.19, 98.63);

\path[draw=drawColor,line width= 0.6pt,line join=round] (260.19, 95.04) --
	(254.20, 95.04);

\path[draw=drawColor,line width= 0.6pt,line join=round] (254.20, 91.45) --
	(254.20, 98.63);

\path[draw=drawColor,line width= 0.6pt,line join=round] (207.43,115.39) --
	(207.43,122.57);

\path[draw=drawColor,line width= 0.6pt,line join=round] (207.43,118.98) --
	(207.20,118.98);

\path[draw=drawColor,line width= 0.6pt,line join=round] (207.20,115.39) --
	(207.20,122.57);

\path[draw=drawColor,line width= 0.6pt,line join=round] (223.41,139.33) --
	(223.41,146.51);

\path[draw=drawColor,line width= 0.6pt,line join=round] (223.41,142.92) --
	(222.03,142.92);

\path[draw=drawColor,line width= 0.6pt,line join=round] (222.03,139.33) --
	(222.03,146.51);

\path[draw=drawColor,line width= 0.6pt,line join=round] (166.81,163.27) --
	(166.81,170.45);

\path[draw=drawColor,line width= 0.6pt,line join=round] (166.81,166.86) --
	(153.98,166.86);

\path[draw=drawColor,line width= 0.6pt,line join=round] (153.98,163.27) --
	(153.98,170.45);

\path[draw=drawColor,line width= 0.6pt,line join=round] (316.20,187.21) --
	(316.20,194.39);

\path[draw=drawColor,line width= 0.6pt,line join=round] (316.20,190.80) --
	(311.02,190.80);

\path[draw=drawColor,line width= 0.6pt,line join=round] (311.02,187.21) --
	(311.02,194.39);

\path[draw=drawColor,line width= 0.6pt,line join=round] (273.23,211.15) --
	(273.23,218.33);

\path[draw=drawColor,line width= 0.6pt,line join=round] (273.23,214.74) --
	(268.73,214.74);

\path[draw=drawColor,line width= 0.6pt,line join=round] (268.73,211.15) --
	(268.73,218.33);

\path[draw=drawColor,line width= 0.6pt,line join=round] (139.19,235.09) --
	(139.19,242.27);

\path[draw=drawColor,line width= 0.6pt,line join=round] (139.19,238.68) --
	(124.07,238.68);

\path[draw=drawColor,line width= 0.6pt,line join=round] (124.07,235.09) --
	(124.07,242.27);

\path[draw=drawColor,line width= 0.6pt,line join=round] (320.05,259.03) --
	(320.05,266.21);

\path[draw=drawColor,line width= 0.6pt,line join=round] (320.05,262.62) --
	(314.11,262.62);

\path[draw=drawColor,line width= 0.6pt,line join=round] (314.11,259.03) --
	(314.11,266.21);

\path[draw=drawColor,line width= 0.6pt,line join=round] (356.24,282.97) --
	(356.24,290.15);

\path[draw=drawColor,line width= 0.6pt,line join=round] (356.24,286.56) --
	(290.85,286.56);

\path[draw=drawColor,line width= 0.6pt,line join=round] (290.85,282.97) --
	(290.85,290.15);
\definecolor{drawColor}{RGB}{81,81,81}

\path[draw=drawColor,line width= 0.6pt,line join=round] (206.90, 38.78) --
	(206.90, 45.96);

\path[draw=drawColor,line width= 0.6pt,line join=round] (206.90, 42.37) --
	(206.90, 42.37);

\path[draw=drawColor,line width= 0.6pt,line join=round] (206.90, 38.78) --
	(206.90, 45.96);

\path[draw=drawColor,line width= 0.6pt,line join=round] (219.16, 62.72) --
	(219.16, 69.90);

\path[draw=drawColor,line width= 0.6pt,line join=round] (219.16, 66.31) --
	(216.58, 66.31);

\path[draw=drawColor,line width= 0.6pt,line join=round] (216.58, 62.72) --
	(216.58, 69.90);

\path[draw=drawColor,line width= 0.6pt,line join=round] (260.19, 86.66) --
	(260.19, 93.84);

\path[draw=drawColor,line width= 0.6pt,line join=round] (260.19, 90.25) --
	(254.20, 90.25);

\path[draw=drawColor,line width= 0.6pt,line join=round] (254.20, 86.66) --
	(254.20, 93.84);

\path[draw=drawColor,line width= 0.6pt,line join=round] (207.43,110.60) --
	(207.43,117.78);

\path[draw=drawColor,line width= 0.6pt,line join=round] (207.43,114.19) --
	(207.20,114.19);

\path[draw=drawColor,line width= 0.6pt,line join=round] (207.20,110.60) --
	(207.20,117.78);

\path[draw=drawColor,line width= 0.6pt,line join=round] (223.30,134.54) --
	(223.30,141.72);

\path[draw=drawColor,line width= 0.6pt,line join=round] (223.30,138.13) --
	(222.01,138.13);

\path[draw=drawColor,line width= 0.6pt,line join=round] (222.01,134.54) --
	(222.01,141.72);

\path[draw=drawColor,line width= 0.6pt,line join=round] (176.22,158.48) --
	(176.22,165.66);

\path[draw=drawColor,line width= 0.6pt,line join=round] (176.22,162.07) --
	(163.63,162.07);

\path[draw=drawColor,line width= 0.6pt,line join=round] (163.63,158.48) --
	(163.63,165.66);

\path[draw=drawColor,line width= 0.6pt,line join=round] (311.60,182.42) --
	(311.60,189.60);

\path[draw=drawColor,line width= 0.6pt,line join=round] (311.60,186.01) --
	(305.58,186.01);

\path[draw=drawColor,line width= 0.6pt,line join=round] (305.58,182.42) --
	(305.58,189.60);

\path[draw=drawColor,line width= 0.6pt,line join=round] (272.96,206.36) --
	(272.96,213.54);

\path[draw=drawColor,line width= 0.6pt,line join=round] (272.96,209.95) --
	(268.56,209.95);

\path[draw=drawColor,line width= 0.6pt,line join=round] (268.56,206.36) --
	(268.56,213.54);

\path[draw=drawColor,line width= 0.6pt,line join=round] (187.25,230.30) --
	(187.25,237.48);

\path[draw=drawColor,line width= 0.6pt,line join=round] (187.25,233.89) --
	(173.60,233.89);

\path[draw=drawColor,line width= 0.6pt,line join=round] (173.60,230.30) --
	(173.60,237.48);

\path[draw=drawColor,line width= 0.6pt,line join=round] (313.53,254.24) --
	(313.53,261.43);

\path[draw=drawColor,line width= 0.6pt,line join=round] (313.53,257.83) --
	(307.27,257.83);

\path[draw=drawColor,line width= 0.6pt,line join=round] (307.27,254.24) --
	(307.27,261.43);

\path[draw=drawColor,line width= 0.6pt,line join=round] (372.20,278.18) --
	(372.20,285.37);

\path[draw=drawColor,line width= 0.6pt,line join=round] (372.20,281.78) --
	(338.37,281.78);

\path[draw=drawColor,line width= 0.6pt,line join=round] (338.37,278.18) --
	(338.37,285.37);
\end{scope}
\begin{scope}
\path[clip] (  0.00,  0.00) rectangle (411.94,303.53);
\definecolor{drawColor}{RGB}{144,144,144}

\path[draw=drawColor,line width= 0.6pt,line join=round] ( 90.20, 30.40) --
	( 90.20,298.53);
\end{scope}
\begin{scope}
\path[clip] (  0.00,  0.00) rectangle (411.94,303.53);
\definecolor{drawColor}{RGB}{68,68,68}

\node[text=drawColor,anchor=base east,inner sep=0pt, outer sep=0pt, scale=  0.80] at ( 85.70, 42.01) {None};

\node[text=drawColor,anchor=base east,inner sep=0pt, outer sep=0pt, scale=  0.80] at ( 85.70, 65.95) {D/F: \# cars};

\node[text=drawColor,anchor=base east,inner sep=0pt, outer sep=0pt, scale=  0.80] at ( 85.70, 89.89) {S/F: \# cars};

\node[text=drawColor,anchor=base east,inner sep=0pt, outer sep=0pt, scale=  0.80] at ( 85.70,113.83) {Supply: Miles/dollar};

\node[text=drawColor,anchor=base east,inner sep=0pt, outer sep=0pt, scale=  0.80] at ( 85.70,137.77) {All D/F};

\node[text=drawColor,anchor=base east,inner sep=0pt, outer sep=0pt, scale=  0.80] at ( 85.70,161.71) {All D/R};

\node[text=drawColor,anchor=base east,inner sep=0pt, outer sep=0pt, scale=  0.80] at ( 85.70,185.65) {All S/F};

\node[text=drawColor,anchor=base east,inner sep=0pt, outer sep=0pt, scale=  0.80] at ( 85.70,209.59) {All S/R};

\node[text=drawColor,anchor=base east,inner sep=0pt, outer sep=0pt, scale=  0.80] at ( 85.70,233.53) {All excluded demand};

\node[text=drawColor,anchor=base east,inner sep=0pt, outer sep=0pt, scale=  0.80] at ( 85.70,257.47) {All excluded supply};

\node[text=drawColor,anchor=base east,inner sep=0pt, outer sep=0pt, scale=  0.80] at ( 85.70,281.41) {All excluded};
\end{scope}
\begin{scope}
\path[clip] (  0.00,  0.00) rectangle (411.94,303.53);
\definecolor{drawColor}{RGB}{144,144,144}

\path[draw=drawColor,line width= 0.6pt,line join=round] ( 87.70, 44.76) --
	( 90.20, 44.76);

\path[draw=drawColor,line width= 0.6pt,line join=round] ( 87.70, 68.70) --
	( 90.20, 68.70);

\path[draw=drawColor,line width= 0.6pt,line join=round] ( 87.70, 92.64) --
	( 90.20, 92.64);

\path[draw=drawColor,line width= 0.6pt,line join=round] ( 87.70,116.58) --
	( 90.20,116.58);

\path[draw=drawColor,line width= 0.6pt,line join=round] ( 87.70,140.52) --
	( 90.20,140.52);

\path[draw=drawColor,line width= 0.6pt,line join=round] ( 87.70,164.47) --
	( 90.20,164.47);

\path[draw=drawColor,line width= 0.6pt,line join=round] ( 87.70,188.41) --
	( 90.20,188.41);

\path[draw=drawColor,line width= 0.6pt,line join=round] ( 87.70,212.35) --
	( 90.20,212.35);

\path[draw=drawColor,line width= 0.6pt,line join=round] ( 87.70,236.29) --
	( 90.20,236.29);

\path[draw=drawColor,line width= 0.6pt,line join=round] ( 87.70,260.23) --
	( 90.20,260.23);

\path[draw=drawColor,line width= 0.6pt,line join=round] ( 87.70,284.17) --
	( 90.20,284.17);
\end{scope}
\begin{scope}
\path[clip] (  0.00,  0.00) rectangle (411.94,303.53);
\definecolor{drawColor}{RGB}{144,144,144}

\path[draw=drawColor,line width= 0.6pt,line join=round] ( 90.20, 30.40) --
	(406.94, 30.40);
\end{scope}
\begin{scope}
\path[clip] (  0.00,  0.00) rectangle (411.94,303.53);
\definecolor{drawColor}{RGB}{144,144,144}

\path[draw=drawColor,line width= 0.6pt,line join=round] (136.64, 27.90) --
	(136.64, 30.40);

\path[draw=drawColor,line width= 0.6pt,line join=round] (188.58, 27.90) --
	(188.58, 30.40);

\path[draw=drawColor,line width= 0.6pt,line join=round] (240.51, 27.90) --
	(240.51, 30.40);

\path[draw=drawColor,line width= 0.6pt,line join=round] (292.45, 27.90) --
	(292.45, 30.40);

\path[draw=drawColor,line width= 0.6pt,line join=round] (344.38, 27.90) --
	(344.38, 30.40);

\path[draw=drawColor,line width= 0.6pt,line join=round] (396.32, 27.90) --
	(396.32, 30.40);
\end{scope}
\begin{scope}
\path[clip] (  0.00,  0.00) rectangle (411.94,303.53);
\definecolor{drawColor}{RGB}{68,68,68}

\node[text=drawColor,anchor=base,inner sep=0pt, outer sep=0pt, scale=  0.80] at (136.64, 20.39) {20};

\node[text=drawColor,anchor=base,inner sep=0pt, outer sep=0pt, scale=  0.80] at (188.58, 20.39) {30};

\node[text=drawColor,anchor=base,inner sep=0pt, outer sep=0pt, scale=  0.80] at (240.51, 20.39) {40};

\node[text=drawColor,anchor=base,inner sep=0pt, outer sep=0pt, scale=  0.80] at (292.45, 20.39) {50};

\node[text=drawColor,anchor=base,inner sep=0pt, outer sep=0pt, scale=  0.80] at (344.38, 20.39) {60};

\node[text=drawColor,anchor=base,inner sep=0pt, outer sep=0pt, scale=  0.80] at (396.32, 20.39) {70};
\end{scope}
\begin{scope}
\path[clip] (  0.00,  0.00) rectangle (411.94,303.53);
\definecolor{drawColor}{RGB}{68,68,68}

\node[text=drawColor,anchor=base,inner sep=0pt, outer sep=0pt, scale=  1.00] at (248.57,  6.94) {Average Markup in \%};
\end{scope}
\begin{scope}
\path[clip] (  0.00,  0.00) rectangle (411.94,303.53);
\definecolor{drawColor}{RGB}{255,255,255}
\definecolor{fillColor}{RGB}{255,255,255}

\path[draw=drawColor,line width= 0.6pt,line join=round,line cap=round,fill=fillColor] (340.27, 50.63) rectangle (378.58, 90.60);
\end{scope}
\begin{scope}
\path[clip] (  0.00,  0.00) rectangle (411.94,303.53);
\definecolor{drawColor}{RGB}{68,68,68}

\node[text=drawColor,anchor=base west,inner sep=0pt, outer sep=0pt, scale=  0.80] at (340.27, 84.32) {Estimate:};
\end{scope}
\begin{scope}
\path[clip] (  0.00,  0.00) rectangle (411.94,303.53);
\definecolor{drawColor}{RGB}{255,255,255}
\definecolor{fillColor}{RGB}{255,255,255}

\path[draw=drawColor,line width= 0.6pt,line join=round,line cap=round,fill=fillColor] (340.27, 65.08) rectangle (354.73, 79.54);
\end{scope}
\begin{scope}
\path[clip] (  0.00,  0.00) rectangle (411.94,303.53);
\definecolor{drawColor}{RGB}{81,81,81}

\path[draw=drawColor,line width= 0.6pt,line join=round] (347.50, 66.53) -- (347.50, 78.09);
\definecolor{fillColor}{RGB}{81,81,81}

\path[fill=fillColor] (347.50, 72.31) circle (  2.85);
\end{scope}
\begin{scope}
\path[clip] (  0.00,  0.00) rectangle (411.94,303.53);
\definecolor{drawColor}{RGB}{81,81,81}

\path[draw=drawColor,line width= 0.6pt,line join=round] (341.72, 72.31) -- (353.28, 72.31);
\end{scope}
\begin{scope}
\path[clip] (  0.00,  0.00) rectangle (411.94,303.53);
\definecolor{drawColor}{RGB}{255,255,255}
\definecolor{fillColor}{RGB}{255,255,255}

\path[draw=drawColor,line width= 0.6pt,line join=round,line cap=round,fill=fillColor] (340.27, 50.63) rectangle (354.73, 65.08);
\end{scope}
\begin{scope}
\path[clip] (  0.00,  0.00) rectangle (411.94,303.53);
\definecolor{drawColor}{RGB}{144,144,144}

\path[draw=drawColor,line width= 0.6pt,line join=round] (347.50, 52.08) -- (347.50, 63.64);
\definecolor{fillColor}{RGB}{144,144,144}

\path[fill=fillColor] (347.50, 62.29) --
	(351.34, 55.64) --
	(343.66, 55.64) --
	cycle;
\end{scope}
\begin{scope}
\path[clip] (  0.00,  0.00) rectangle (411.94,303.53);
\definecolor{drawColor}{RGB}{144,144,144}

\path[draw=drawColor,line width= 0.6pt,line join=round] (341.72, 57.86) -- (353.28, 57.86);
\end{scope}
\begin{scope}
\path[clip] (  0.00,  0.00) rectangle (411.94,303.53);
\definecolor{drawColor}{RGB}{68,68,68}

\node[text=drawColor,anchor=base west,inner sep=0pt, outer sep=0pt, scale=  0.70] at (358.73, 69.90) {$p=\infty$};
\end{scope}
\begin{scope}
\path[clip] (  0.00,  0.00) rectangle (411.94,303.53);
\definecolor{drawColor}{RGB}{68,68,68}

\node[text=drawColor,anchor=base west,inner sep=0pt, outer sep=0pt, scale=  0.70] at (358.73, 55.45) {$p=2$};
\end{scope}
\end{tikzpicture}

%% file: blp_l2_M.tex
\begin{tikzpicture}[x=1pt,y=1pt]
\definecolor{fillColor}{RGB}{255,255,255}
\path[use as bounding box,fill=fillColor,fill opacity=0.00] (0,0) rectangle (419.17,274.63);
\begin{scope}
\path[clip] (  0.00,  0.00) rectangle (419.17,274.63);
\definecolor{drawColor}{RGB}{255,255,255}
\definecolor{fillColor}{RGB}{255,255,255}

\path[draw=drawColor,line width= 0.6pt,line join=round,line cap=round,fill=fillColor] (  0.00,  0.00) rectangle (419.17,274.63);
\end{scope}
\begin{scope}
\path[clip] ( 31.33, 30.40) rectangle (414.17,269.63);
\definecolor{fillColor}{RGB}{255,255,255}

\path[fill=fillColor] ( 31.33, 30.40) rectangle (414.17,269.63);
\definecolor{fillColor}{RGB}{81,81,81}

\path[fill=fillColor,fill opacity=0.10] ( 48.73, 59.12) --
	( 50.47, 61.98) --
	( 52.21, 68.06) --
	( 53.95, 76.06) --
	( 55.69, 84.75) --
	( 57.43, 93.31) --
	( 59.17,101.28) --
	( 60.91,108.48) --
	( 62.65,114.88) --
	( 64.39,120.53) --
	( 66.13,125.50) --
	( 67.87,129.89) --
	( 69.61,133.78) --
	( 71.35,137.25) --
	( 73.09,140.36) --
	( 74.83,143.16) --
	( 76.57,145.70) --
	( 78.31,148.02) --
	( 80.05,150.13) --
	( 81.79,152.07) --
	( 83.53,153.87) --
	( 85.27,155.53) --
	( 87.01,157.07) --
	( 88.75,158.51) --
	( 90.49,159.86) --
	( 92.24,161.12) --
	( 93.98,162.31) --
	( 95.72,163.44) --
	( 97.46,164.51) --
	( 99.20,165.53) --
	(100.94,166.50) --
	(102.68,167.43) --
	(104.42,168.32) --
	(106.16,169.18) --
	(107.90,170.01) --
	(109.64,170.81) --
	(111.38,171.59) --
	(113.12,172.34) --
	(114.86,173.07) --
	(116.60,173.79) --
	(118.34,174.48) --
	(120.08,175.17) --
	(121.82,175.83) --
	(123.56,176.48) --
	(125.30,177.12) --
	(127.04,177.75) --
	(128.78,178.37) --
	(130.52,178.97) --
	(132.26,179.57) --
	(134.00,180.16) --
	(135.74,180.73) --
	(137.48,181.30) --
	(139.22,181.87) --
	(140.96,182.42) --
	(142.70,182.97) --
	(144.44,183.51) --
	(146.18,184.04) --
	(147.92,184.57) --
	(149.66,185.09) --
	(151.40,185.61) --
	(153.14,186.12) --
	(154.88,186.62) --
	(156.62,187.13) --
	(158.36,187.62) --
	(160.10,188.11) --
	(161.84,188.60) --
	(163.58,189.08) --
	(165.32,189.56) --
	(167.06,190.03) --
	(168.80,190.51) --
	(170.54,190.97) --
	(172.28,191.44) --
	(174.02,191.90) --
	(175.76,192.35) --
	(177.50,192.80) --
	(179.24,193.25) --
	(180.98,193.70) --
	(182.72,194.15) --
	(184.46,194.59) --
	(186.20,195.03) --
	(187.94,195.46) --
	(189.68,195.89) --
	(191.42,196.33) --
	(193.16,196.75) --
	(194.90,197.18) --
	(196.65,197.60) --
	(198.39,198.03) --
	(200.13,198.44) --
	(201.87,198.86) --
	(203.61,199.28) --
	(205.35,199.69) --
	(207.09,200.10) --
	(208.83,200.51) --
	(210.57,200.92) --
	(212.31,201.32) --
	(214.05,201.73) --
	(215.79,202.13) --
	(217.53,202.53) --
	(219.27,202.93) --
	(221.01,203.33) --
	(222.75,203.72) --
	(224.49,204.12) --
	(226.23,204.51) --
	(227.97,204.90) --
	(229.71,205.29) --
	(231.45,205.68) --
	(233.19,206.07) --
	(234.93,206.46) --
	(236.67,206.84) --
	(238.41,207.22) --
	(240.15,207.61) --
	(241.89,207.99) --
	(243.63,208.37) --
	(245.37,208.75) --
	(247.11,209.13) --
	(248.85,209.50) --
	(250.59,209.88) --
	(252.33,210.25) --
	(254.07,210.63) --
	(255.81,211.00) --
	(257.55,211.37) --
	(259.29,211.74) --
	(261.03,212.12) --
	(262.77,212.48) --
	(264.51,212.85) --
	(266.25,213.22) --
	(267.99,213.59) --
	(269.73,213.95) --
	(271.47,214.32) --
	(273.21,214.68) --
	(274.95,215.05) --
	(276.69,215.41) --
	(278.43,215.77) --
	(280.17,216.13) --
	(281.91,216.49) --
	(283.65,216.85) --
	(285.39,217.21) --
	(287.13,217.57) --
	(288.87,217.93) --
	(290.61,218.29) --
	(292.35,218.64) --
	(294.09,219.00) --
	(295.83,219.36) --
	(297.57,219.71) --
	(299.31,220.07) --
	(301.06,220.42) --
	(302.80,220.77) --
	(304.54,221.13) --
	(306.28,221.48) --
	(308.02,221.83) --
	(309.76,222.18) --
	(311.50,222.53) --
	(313.24,222.88) --
	(314.98,223.23) --
	(316.72,223.58) --
	(318.46,223.93) --
	(320.20,224.28) --
	(321.94,224.62) --
	(323.68,224.97) --
	(325.42,225.32) --
	(327.16,225.66) --
	(328.90,226.01) --
	(330.64,226.36) --
	(332.38,226.70) --
	(334.12,227.05) --
	(335.86,227.39) --
	(337.60,227.73) --
	(339.34,228.08) --
	(341.08,228.42) --
	(342.82,228.76) --
	(344.56,229.10) --
	(346.30,229.45) --
	(348.04,229.79) --
	(349.78,230.13) --
	(351.52,230.47) --
	(353.26,230.81) --
	(355.00,231.15) --
	(356.74,231.49) --
	(358.48,231.83) --
	(360.22,232.17) --
	(361.96,232.51) --
	(363.70,232.85) --
	(365.44,233.19) --
	(367.18,233.52) --
	(368.92,233.86) --
	(370.66,234.20) --
	(372.40,234.54) --
	(374.14,234.87) --
	(375.88,235.21) --
	(377.62,235.55) --
	(379.36,235.88) --
	(381.10,236.22) --
	(382.84,236.55) --
	(384.58,236.89) --
	(386.32,237.22) --
	(388.06,237.56) --
	(389.80,237.89) --
	(391.54,238.23) --
	(393.28,238.56) --
	(395.02,238.89) --
	(396.76,239.23) --
	(396.76,114.75) --
	(395.02,115.03) --
	(393.28,115.31) --
	(391.54,115.58) --
	(389.80,115.86) --
	(388.06,116.14) --
	(386.32,116.41) --
	(384.58,116.69) --
	(382.84,116.96) --
	(381.10,117.24) --
	(379.36,117.51) --
	(377.62,117.79) --
	(375.88,118.06) --
	(374.14,118.33) --
	(372.40,118.61) --
	(370.66,118.88) --
	(368.92,119.15) --
	(367.18,119.43) --
	(365.44,119.70) --
	(363.70,119.97) --
	(361.96,120.24) --
	(360.22,120.52) --
	(358.48,120.79) --
	(356.74,121.06) --
	(355.00,121.33) --
	(353.26,121.60) --
	(351.52,121.87) --
	(349.78,122.14) --
	(348.04,122.41) --
	(346.30,122.68) --
	(344.56,122.95) --
	(342.82,123.21) --
	(341.08,123.48) --
	(339.34,123.75) --
	(337.60,124.02) --
	(335.86,124.28) --
	(334.12,124.55) --
	(332.38,124.82) --
	(330.64,125.08) --
	(328.90,125.35) --
	(327.16,125.61) --
	(325.42,125.88) --
	(323.68,126.14) --
	(321.94,126.41) --
	(320.20,126.67) --
	(318.46,126.93) --
	(316.72,127.20) --
	(314.98,127.46) --
	(313.24,127.72) --
	(311.50,127.98) --
	(309.76,128.24) --
	(308.02,128.50) --
	(306.28,128.76) --
	(304.54,129.02) --
	(302.80,129.28) --
	(301.06,129.54) --
	(299.31,129.80) --
	(297.57,130.05) --
	(295.83,130.31) --
	(294.09,130.57) --
	(292.35,130.82) --
	(290.61,131.08) --
	(288.87,131.33) --
	(287.13,131.58) --
	(285.39,131.84) --
	(283.65,132.09) --
	(281.91,132.34) --
	(280.17,132.59) --
	(278.43,132.84) --
	(276.69,133.09) --
	(274.95,133.34) --
	(273.21,133.59) --
	(271.47,133.84) --
	(269.73,134.09) --
	(267.99,134.33) --
	(266.25,134.58) --
	(264.51,134.82) --
	(262.77,135.07) --
	(261.03,135.31) --
	(259.29,135.55) --
	(257.55,135.79) --
	(255.81,136.03) --
	(254.07,136.27) --
	(252.33,136.51) --
	(250.59,136.75) --
	(248.85,136.99) --
	(247.11,137.22) --
	(245.37,137.46) --
	(243.63,137.69) --
	(241.89,137.92) --
	(240.15,138.15) --
	(238.41,138.38) --
	(236.67,138.61) --
	(234.93,138.84) --
	(233.19,139.07) --
	(231.45,139.30) --
	(229.71,139.52) --
	(227.97,139.74) --
	(226.23,139.97) --
	(224.49,140.19) --
	(222.75,140.41) --
	(221.01,140.62) --
	(219.27,140.84) --
	(217.53,141.05) --
	(215.79,141.27) --
	(214.05,141.48) --
	(212.31,141.69) --
	(210.57,141.90) --
	(208.83,142.11) --
	(207.09,142.31) --
	(205.35,142.51) --
	(203.61,142.72) --
	(201.87,142.92) --
	(200.13,143.11) --
	(198.39,143.31) --
	(196.65,143.50) --
	(194.90,143.69) --
	(193.16,143.88) --
	(191.42,144.07) --
	(189.68,144.26) --
	(187.94,144.44) --
	(186.20,144.62) --
	(184.46,144.79) --
	(182.72,144.97) --
	(180.98,145.14) --
	(179.24,145.31) --
	(177.50,145.48) --
	(175.76,145.64) --
	(174.02,145.80) --
	(172.28,145.96) --
	(170.54,146.11) --
	(168.80,146.26) --
	(167.06,146.40) --
	(165.32,146.55) --
	(163.58,146.69) --
	(161.84,146.82) --
	(160.10,146.95) --
	(158.36,147.08) --
	(156.62,147.20) --
	(154.88,147.32) --
	(153.14,147.43) --
	(151.40,147.53) --
	(149.66,147.64) --
	(147.92,147.73) --
	(146.18,147.82) --
	(144.44,147.91) --
	(142.70,147.98) --
	(140.96,148.05) --
	(139.22,148.12) --
	(137.48,148.18) --
	(135.74,148.22) --
	(134.00,148.27) --
	(132.26,148.30) --
	(130.52,148.32) --
	(128.78,148.33) --
	(127.04,148.34) --
	(125.30,148.33) --
	(123.56,148.31) --
	(121.82,148.27) --
	(120.08,148.23) --
	(118.34,148.16) --
	(116.60,148.09) --
	(114.86,147.99) --
	(113.12,147.88) --
	(111.38,147.74) --
	(109.64,147.58) --
	(107.90,147.40) --
	(106.16,147.18) --
	(104.42,146.94) --
	(102.68,146.66) --
	(100.94,146.34) --
	( 99.20,145.98) --
	( 97.46,145.58) --
	( 95.72,145.12) --
	( 93.98,144.59) --
	( 92.24,144.00) --
	( 90.49,143.34) --
	( 88.75,142.59) --
	( 87.01,141.74) --
	( 85.27,140.78) --
	( 83.53,139.70) --
	( 81.79,138.48) --
	( 80.05,137.10) --
	( 78.31,135.54) --
	( 76.57,133.77) --
	( 74.83,131.76) --
	( 73.09,129.47) --
	( 71.35,126.85) --
	( 69.61,123.85) --
	( 67.87,120.41) --
	( 66.13,116.45) --
	( 64.39,111.89) --
	( 62.65,106.66) --
	( 60.91,100.68) --
	( 59.17, 93.96) --
	( 57.43, 86.58) --
	( 55.69, 78.80) --
	( 53.95, 71.15) --
	( 52.21, 64.50) --
	( 50.47, 60.10) --
	( 48.73, 59.12) --
	cycle;

\path[] ( 48.73, 59.12) --
	( 50.47, 61.98) --
	( 52.21, 68.06) --
	( 53.95, 76.06) --
	( 55.69, 84.75) --
	( 57.43, 93.31) --
	( 59.17,101.28) --
	( 60.91,108.48) --
	( 62.65,114.88) --
	( 64.39,120.53) --
	( 66.13,125.50) --
	( 67.87,129.89) --
	( 69.61,133.78) --
	( 71.35,137.25) --
	( 73.09,140.36) --
	( 74.83,143.16) --
	( 76.57,145.70) --
	( 78.31,148.02) --
	( 80.05,150.13) --
	( 81.79,152.07) --
	( 83.53,153.87) --
	( 85.27,155.53) --
	( 87.01,157.07) --
	( 88.75,158.51) --
	( 90.49,159.86) --
	( 92.24,161.12) --
	( 93.98,162.31) --
	( 95.72,163.44) --
	( 97.46,164.51) --
	( 99.20,165.53) --
	(100.94,166.50) --
	(102.68,167.43) --
	(104.42,168.32) --
	(106.16,169.18) --
	(107.90,170.01) --
	(109.64,170.81) --
	(111.38,171.59) --
	(113.12,172.34) --
	(114.86,173.07) --
	(116.60,173.79) --
	(118.34,174.48) --
	(120.08,175.17) --
	(121.82,175.83) --
	(123.56,176.48) --
	(125.30,177.12) --
	(127.04,177.75) --
	(128.78,178.37) --
	(130.52,178.97) --
	(132.26,179.57) --
	(134.00,180.16) --
	(135.74,180.73) --
	(137.48,181.30) --
	(139.22,181.87) --
	(140.96,182.42) --
	(142.70,182.97) --
	(144.44,183.51) --
	(146.18,184.04) --
	(147.92,184.57) --
	(149.66,185.09) --
	(151.40,185.61) --
	(153.14,186.12) --
	(154.88,186.62) --
	(156.62,187.13) --
	(158.36,187.62) --
	(160.10,188.11) --
	(161.84,188.60) --
	(163.58,189.08) --
	(165.32,189.56) --
	(167.06,190.03) --
	(168.80,190.51) --
	(170.54,190.97) --
	(172.28,191.44) --
	(174.02,191.90) --
	(175.76,192.35) --
	(177.50,192.80) --
	(179.24,193.25) --
	(180.98,193.70) --
	(182.72,194.15) --
	(184.46,194.59) --
	(186.20,195.03) --
	(187.94,195.46) --
	(189.68,195.89) --
	(191.42,196.33) --
	(193.16,196.75) --
	(194.90,197.18) --
	(196.65,197.60) --
	(198.39,198.03) --
	(200.13,198.44) --
	(201.87,198.86) --
	(203.61,199.28) --
	(205.35,199.69) --
	(207.09,200.10) --
	(208.83,200.51) --
	(210.57,200.92) --
	(212.31,201.32) --
	(214.05,201.73) --
	(215.79,202.13) --
	(217.53,202.53) --
	(219.27,202.93) --
	(221.01,203.33) --
	(222.75,203.72) --
	(224.49,204.12) --
	(226.23,204.51) --
	(227.97,204.90) --
	(229.71,205.29) --
	(231.45,205.68) --
	(233.19,206.07) --
	(234.93,206.46) --
	(236.67,206.84) --
	(238.41,207.22) --
	(240.15,207.61) --
	(241.89,207.99) --
	(243.63,208.37) --
	(245.37,208.75) --
	(247.11,209.13) --
	(248.85,209.50) --
	(250.59,209.88) --
	(252.33,210.25) --
	(254.07,210.63) --
	(255.81,211.00) --
	(257.55,211.37) --
	(259.29,211.74) --
	(261.03,212.12) --
	(262.77,212.48) --
	(264.51,212.85) --
	(266.25,213.22) --
	(267.99,213.59) --
	(269.73,213.95) --
	(271.47,214.32) --
	(273.21,214.68) --
	(274.95,215.05) --
	(276.69,215.41) --
	(278.43,215.77) --
	(280.17,216.13) --
	(281.91,216.49) --
	(283.65,216.85) --
	(285.39,217.21) --
	(287.13,217.57) --
	(288.87,217.93) --
	(290.61,218.29) --
	(292.35,218.64) --
	(294.09,219.00) --
	(295.83,219.36) --
	(297.57,219.71) --
	(299.31,220.07) --
	(301.06,220.42) --
	(302.80,220.77) --
	(304.54,221.13) --
	(306.28,221.48) --
	(308.02,221.83) --
	(309.76,222.18) --
	(311.50,222.53) --
	(313.24,222.88) --
	(314.98,223.23) --
	(316.72,223.58) --
	(318.46,223.93) --
	(320.20,224.28) --
	(321.94,224.62) --
	(323.68,224.97) --
	(325.42,225.32) --
	(327.16,225.66) --
	(328.90,226.01) --
	(330.64,226.36) --
	(332.38,226.70) --
	(334.12,227.05) --
	(335.86,227.39) --
	(337.60,227.73) --
	(339.34,228.08) --
	(341.08,228.42) --
	(342.82,228.76) --
	(344.56,229.10) --
	(346.30,229.45) --
	(348.04,229.79) --
	(349.78,230.13) --
	(351.52,230.47) --
	(353.26,230.81) --
	(355.00,231.15) --
	(356.74,231.49) --
	(358.48,231.83) --
	(360.22,232.17) --
	(361.96,232.51) --
	(363.70,232.85) --
	(365.44,233.19) --
	(367.18,233.52) --
	(368.92,233.86) --
	(370.66,234.20) --
	(372.40,234.54) --
	(374.14,234.87) --
	(375.88,235.21) --
	(377.62,235.55) --
	(379.36,235.88) --
	(381.10,236.22) --
	(382.84,236.55) --
	(384.58,236.89) --
	(386.32,237.22) --
	(388.06,237.56) --
	(389.80,237.89) --
	(391.54,238.23) --
	(393.28,238.56) --
	(395.02,238.89) --
	(396.76,239.23);

\path[] (396.76,114.75) --
	(395.02,115.03) --
	(393.28,115.31) --
	(391.54,115.58) --
	(389.80,115.86) --
	(388.06,116.14) --
	(386.32,116.41) --
	(384.58,116.69) --
	(382.84,116.96) --
	(381.10,117.24) --
	(379.36,117.51) --
	(377.62,117.79) --
	(375.88,118.06) --
	(374.14,118.33) --
	(372.40,118.61) --
	(370.66,118.88) --
	(368.92,119.15) --
	(367.18,119.43) --
	(365.44,119.70) --
	(363.70,119.97) --
	(361.96,120.24) --
	(360.22,120.52) --
	(358.48,120.79) --
	(356.74,121.06) --
	(355.00,121.33) --
	(353.26,121.60) --
	(351.52,121.87) --
	(349.78,122.14) --
	(348.04,122.41) --
	(346.30,122.68) --
	(344.56,122.95) --
	(342.82,123.21) --
	(341.08,123.48) --
	(339.34,123.75) --
	(337.60,124.02) --
	(335.86,124.28) --
	(334.12,124.55) --
	(332.38,124.82) --
	(330.64,125.08) --
	(328.90,125.35) --
	(327.16,125.61) --
	(325.42,125.88) --
	(323.68,126.14) --
	(321.94,126.41) --
	(320.20,126.67) --
	(318.46,126.93) --
	(316.72,127.20) --
	(314.98,127.46) --
	(313.24,127.72) --
	(311.50,127.98) --
	(309.76,128.24) --
	(308.02,128.50) --
	(306.28,128.76) --
	(304.54,129.02) --
	(302.80,129.28) --
	(301.06,129.54) --
	(299.31,129.80) --
	(297.57,130.05) --
	(295.83,130.31) --
	(294.09,130.57) --
	(292.35,130.82) --
	(290.61,131.08) --
	(288.87,131.33) --
	(287.13,131.58) --
	(285.39,131.84) --
	(283.65,132.09) --
	(281.91,132.34) --
	(280.17,132.59) --
	(278.43,132.84) --
	(276.69,133.09) --
	(274.95,133.34) --
	(273.21,133.59) --
	(271.47,133.84) --
	(269.73,134.09) --
	(267.99,134.33) --
	(266.25,134.58) --
	(264.51,134.82) --
	(262.77,135.07) --
	(261.03,135.31) --
	(259.29,135.55) --
	(257.55,135.79) --
	(255.81,136.03) --
	(254.07,136.27) --
	(252.33,136.51) --
	(250.59,136.75) --
	(248.85,136.99) --
	(247.11,137.22) --
	(245.37,137.46) --
	(243.63,137.69) --
	(241.89,137.92) --
	(240.15,138.15) --
	(238.41,138.38) --
	(236.67,138.61) --
	(234.93,138.84) --
	(233.19,139.07) --
	(231.45,139.30) --
	(229.71,139.52) --
	(227.97,139.74) --
	(226.23,139.97) --
	(224.49,140.19) --
	(222.75,140.41) --
	(221.01,140.62) --
	(219.27,140.84) --
	(217.53,141.05) --
	(215.79,141.27) --
	(214.05,141.48) --
	(212.31,141.69) --
	(210.57,141.90) --
	(208.83,142.11) --
	(207.09,142.31) --
	(205.35,142.51) --
	(203.61,142.72) --
	(201.87,142.92) --
	(200.13,143.11) --
	(198.39,143.31) --
	(196.65,143.50) --
	(194.90,143.69) --
	(193.16,143.88) --
	(191.42,144.07) --
	(189.68,144.26) --
	(187.94,144.44) --
	(186.20,144.62) --
	(184.46,144.79) --
	(182.72,144.97) --
	(180.98,145.14) --
	(179.24,145.31) --
	(177.50,145.48) --
	(175.76,145.64) --
	(174.02,145.80) --
	(172.28,145.96) --
	(170.54,146.11) --
	(168.80,146.26) --
	(167.06,146.40) --
	(165.32,146.55) --
	(163.58,146.69) --
	(161.84,146.82) --
	(160.10,146.95) --
	(158.36,147.08) --
	(156.62,147.20) --
	(154.88,147.32) --
	(153.14,147.43) --
	(151.40,147.53) --
	(149.66,147.64) --
	(147.92,147.73) --
	(146.18,147.82) --
	(144.44,147.91) --
	(142.70,147.98) --
	(140.96,148.05) --
	(139.22,148.12) --
	(137.48,148.18) --
	(135.74,148.22) --
	(134.00,148.27) --
	(132.26,148.30) --
	(130.52,148.32) --
	(128.78,148.33) --
	(127.04,148.34) --
	(125.30,148.33) --
	(123.56,148.31) --
	(121.82,148.27) --
	(120.08,148.23) --
	(118.34,148.16) --
	(116.60,148.09) --
	(114.86,147.99) --
	(113.12,147.88) --
	(111.38,147.74) --
	(109.64,147.58) --
	(107.90,147.40) --
	(106.16,147.18) --
	(104.42,146.94) --
	(102.68,146.66) --
	(100.94,146.34) --
	( 99.20,145.98) --
	( 97.46,145.58) --
	( 95.72,145.12) --
	( 93.98,144.59) --
	( 92.24,144.00) --
	( 90.49,143.34) --
	( 88.75,142.59) --
	( 87.01,141.74) --
	( 85.27,140.78) --
	( 83.53,139.70) --
	( 81.79,138.48) --
	( 80.05,137.10) --
	( 78.31,135.54) --
	( 76.57,133.77) --
	( 74.83,131.76) --
	( 73.09,129.47) --
	( 71.35,126.85) --
	( 69.61,123.85) --
	( 67.87,120.41) --
	( 66.13,116.45) --
	( 64.39,111.89) --
	( 62.65,106.66) --
	( 60.91,100.68) --
	( 59.17, 93.96) --
	( 57.43, 86.58) --
	( 55.69, 78.80) --
	( 53.95, 71.15) --
	( 52.21, 64.50) --
	( 50.47, 60.10) --
	( 48.73, 59.12);
\definecolor{drawColor}{RGB}{81,81,81}

\path[draw=drawColor,line width= 0.6pt,line join=round] ( 48.73, 76.97) --
	( 50.47, 78.99) --
	( 52.21, 84.49) --
	( 53.95, 92.19) --
	( 55.69,100.79) --
	( 57.43,109.38) --
	( 59.17,117.48) --
	( 60.91,124.84) --
	( 62.65,131.40) --
	( 64.39,137.20) --
	( 66.13,142.31) --
	( 67.87,146.82) --
	( 69.61,150.82) --
	( 71.35,154.37) --
	( 73.09,157.56) --
	( 74.83,160.43) --
	( 76.57,163.02) --
	( 78.31,165.39) --
	( 80.05,167.55) --
	( 81.79,169.53) --
	( 83.53,171.36) --
	( 85.27,173.06) --
	( 87.01,174.63) --
	( 88.75,176.10) --
	( 90.49,177.48) --
	( 92.24,178.77) --
	( 93.98,179.99) --
	( 95.72,181.15) --
	( 97.46,182.24) --
	( 99.20,183.29) --
	(100.94,184.28) --
	(102.68,185.24) --
	(104.42,186.15) --
	(106.16,187.03) --
	(107.90,187.88) --
	(109.64,188.71) --
	(111.38,189.50) --
	(113.12,190.28) --
	(114.86,191.03) --
	(116.60,191.77) --
	(118.34,192.48) --
	(120.08,193.18) --
	(121.82,193.87) --
	(123.56,194.54) --
	(125.30,195.19) --
	(127.04,195.84) --
	(128.78,196.47) --
	(130.52,197.09) --
	(132.26,197.71) --
	(134.00,198.31) --
	(135.74,198.90) --
	(137.48,199.49) --
	(139.22,200.06) --
	(140.96,200.63) --
	(142.70,201.20) --
	(144.44,201.75) --
	(146.18,202.30) --
	(147.92,202.84) --
	(149.66,203.38) --
	(151.40,203.91) --
	(153.14,204.43) --
	(154.88,204.95) --
	(156.62,205.47) --
	(158.36,205.98) --
	(160.10,206.48) --
	(161.84,206.98) --
	(163.58,207.47) --
	(165.32,207.97) --
	(167.06,208.45) --
	(168.80,208.94) --
	(170.54,209.41) --
	(172.28,209.89) --
	(174.02,210.36) --
	(175.76,210.83) --
	(177.50,211.29) --
	(179.24,211.76) --
	(180.98,212.21) --
	(182.72,212.67) --
	(184.46,213.12) --
	(186.20,213.57) --
	(187.94,214.02) --
	(189.68,214.46) --
	(191.42,214.91) --
	(193.16,215.35) --
	(194.90,215.78) --
	(196.65,216.22) --
	(198.39,216.65) --
	(200.13,217.08) --
	(201.87,217.51) --
	(203.61,217.93) --
	(205.35,218.36) --
	(207.09,218.78) --
	(208.83,219.20) --
	(210.57,219.61) --
	(212.31,220.03) --
	(214.05,220.44) --
	(215.79,220.86) --
	(217.53,221.27) --
	(219.27,221.68) --
	(221.01,222.08) --
	(222.75,222.49) --
	(224.49,222.89) --
	(226.23,223.29) --
	(227.97,223.70) --
	(229.71,224.10) --
	(231.45,224.49) --
	(233.19,224.89) --
	(234.93,225.29) --
	(236.67,225.68) --
	(238.41,226.07) --
	(240.15,226.47) --
	(241.89,226.86) --
	(243.63,227.25) --
	(245.37,227.63) --
	(247.11,228.02) --
	(248.85,228.41) --
	(250.59,228.79) --
	(252.33,229.18) --
	(254.07,229.56) --
	(255.81,229.94) --
	(257.55,230.32) --
	(259.29,230.70) --
	(261.03,231.08) --
	(262.77,231.46) --
	(264.51,231.83) --
	(266.25,232.21) --
	(267.99,232.59) --
	(269.73,232.96) --
	(271.47,233.33) --
	(273.21,233.71) --
	(274.95,234.08) --
	(276.69,234.45) --
	(278.43,234.82) --
	(280.17,235.19) --
	(281.91,235.56) --
	(283.65,235.93) --
	(285.39,236.29) --
	(287.13,236.66) --
	(288.87,237.03) --
	(290.61,237.39) --
	(292.35,237.76) --
	(294.09,238.12) --
	(295.83,238.48) --
	(297.57,238.85) --
	(299.31,239.21) --
	(301.06,239.57) --
	(302.80,239.93) --
	(304.54,240.29) --
	(306.28,240.65) --
	(308.02,241.01) --
	(309.76,241.37) --
	(311.50,241.73) --
	(313.24,242.08) --
	(314.98,242.44) --
	(316.72,242.80) --
	(318.46,243.15) --
	(320.20,243.51) --
	(321.94,243.86) --
	(323.68,244.22) --
	(325.42,244.57) --
	(327.16,244.93) --
	(328.90,245.28) --
	(330.64,245.63) --
	(332.38,245.98) --
	(334.12,246.33) --
	(335.86,246.69) --
	(337.60,247.04) --
	(339.34,247.39) --
	(341.08,247.74) --
	(342.82,248.09) --
	(344.56,248.44) --
	(346.30,248.78) --
	(348.04,249.13) --
	(349.78,249.48) --
	(351.52,249.83) --
	(353.26,250.18) --
	(355.00,250.52) --
	(356.74,250.87) --
	(358.48,251.22) --
	(360.22,251.56) --
	(361.96,251.91) --
	(363.70,252.25) --
	(365.44,252.60) --
	(367.18,252.94) --
	(368.92,253.29) --
	(370.66,253.63) --
	(372.40,253.97) --
	(374.14,254.32) --
	(375.88,254.66) --
	(377.62,255.00) --
	(379.36,255.34) --
	(381.10,255.69) --
	(382.84,256.03) --
	(384.58,256.37) --
	(386.32,256.71) --
	(388.06,257.05) --
	(389.80,257.39) --
	(391.54,257.73) --
	(393.28,258.07) --
	(395.02,258.41) --
	(396.76,258.75);

\path[draw=drawColor,line width= 0.6pt,line join=round] ( 48.73, 41.27) --
	( 50.47, 43.10) --
	( 52.21, 48.07) --
	( 53.95, 55.02) --
	( 55.69, 62.77) --
	( 57.43, 70.51) --
	( 59.17, 77.76) --
	( 60.91, 84.33) --
	( 62.65, 90.14) --
	( 64.39, 95.22) --
	( 66.13, 99.64) --
	( 67.87,103.48) --
	( 69.61,106.82) --
	( 71.35,109.73) --
	( 73.09,112.27) --
	( 74.83,114.49) --
	( 76.57,116.45) --
	( 78.31,118.17) --
	( 80.05,119.69) --
	( 81.79,121.03) --
	( 83.53,122.21) --
	( 85.27,123.26) --
	( 87.01,124.18) --
	( 88.75,124.99) --
	( 90.49,125.72) --
	( 92.24,126.35) --
	( 93.98,126.91) --
	( 95.72,127.41) --
	( 97.46,127.84) --
	( 99.20,128.23) --
	(100.94,128.56) --
	(102.68,128.85) --
	(104.42,129.11) --
	(106.16,129.33) --
	(107.90,129.52) --
	(109.64,129.68) --
	(111.38,129.82) --
	(113.12,129.94) --
	(114.86,130.03) --
	(116.60,130.11) --
	(118.34,130.17) --
	(120.08,130.21) --
	(121.82,130.24) --
	(123.56,130.25) --
	(125.30,130.26) --
	(127.04,130.25) --
	(128.78,130.23) --
	(130.52,130.20) --
	(132.26,130.16) --
	(134.00,130.11) --
	(135.74,130.06) --
	(137.48,129.99) --
	(139.22,129.92) --
	(140.96,129.84) --
	(142.70,129.75) --
	(144.44,129.66) --
	(146.18,129.56) --
	(147.92,129.46) --
	(149.66,129.35) --
	(151.40,129.23) --
	(153.14,129.11) --
	(154.88,128.99) --
	(156.62,128.86) --
	(158.36,128.72) --
	(160.10,128.58) --
	(161.84,128.44) --
	(163.58,128.29) --
	(165.32,128.14) --
	(167.06,127.99) --
	(168.80,127.83) --
	(170.54,127.67) --
	(172.28,127.50) --
	(174.02,127.33) --
	(175.76,127.16) --
	(177.50,126.99) --
	(179.24,126.81) --
	(180.98,126.63) --
	(182.72,126.44) --
	(184.46,126.26) --
	(186.20,126.07) --
	(187.94,125.88) --
	(189.68,125.69) --
	(191.42,125.49) --
	(193.16,125.29) --
	(194.90,125.09) --
	(196.65,124.89) --
	(198.39,124.69) --
	(200.13,124.48) --
	(201.87,124.27) --
	(203.61,124.06) --
	(205.35,123.85) --
	(207.09,123.64) --
	(208.83,123.42) --
	(210.57,123.20) --
	(212.31,122.98) --
	(214.05,122.76) --
	(215.79,122.54) --
	(217.53,122.32) --
	(219.27,122.09) --
	(221.01,121.87) --
	(222.75,121.64) --
	(224.49,121.41) --
	(226.23,121.18) --
	(227.97,120.95) --
	(229.71,120.72) --
	(231.45,120.48) --
	(233.19,120.25) --
	(234.93,120.01) --
	(236.67,119.77) --
	(238.41,119.54) --
	(240.15,119.30) --
	(241.89,119.05) --
	(243.63,118.81) --
	(245.37,118.57) --
	(247.11,118.33) --
	(248.85,118.08) --
	(250.59,117.84) --
	(252.33,117.59) --
	(254.07,117.34) --
	(255.81,117.09) --
	(257.55,116.84) --
	(259.29,116.59) --
	(261.03,116.34) --
	(262.77,116.09) --
	(264.51,115.84) --
	(266.25,115.59) --
	(267.99,115.33) --
	(269.73,115.08) --
	(271.47,114.82) --
	(273.21,114.57) --
	(274.95,114.31) --
	(276.69,114.05) --
	(278.43,113.80) --
	(280.17,113.54) --
	(281.91,113.28) --
	(283.65,113.02) --
	(285.39,112.76) --
	(287.13,112.50) --
	(288.87,112.23) --
	(290.61,111.97) --
	(292.35,111.71) --
	(294.09,111.45) --
	(295.83,111.18) --
	(297.57,110.92) --
	(299.31,110.65) --
	(301.06,110.39) --
	(302.80,110.12) --
	(304.54,109.86) --
	(306.28,109.59) --
	(308.02,109.32) --
	(309.76,109.05) --
	(311.50,108.79) --
	(313.24,108.52) --
	(314.98,108.25) --
	(316.72,107.98) --
	(318.46,107.71) --
	(320.20,107.44) --
	(321.94,107.17) --
	(323.68,106.90) --
	(325.42,106.62) --
	(327.16,106.35) --
	(328.90,106.08) --
	(330.64,105.81) --
	(332.38,105.54) --
	(334.12,105.26) --
	(335.86,104.99) --
	(337.60,104.71) --
	(339.34,104.44) --
	(341.08,104.17) --
	(342.82,103.89) --
	(344.56,103.61) --
	(346.30,103.34) --
	(348.04,103.06) --
	(349.78,102.79) --
	(351.52,102.51) --
	(353.26,102.23) --
	(355.00,101.96) --
	(356.74,101.68) --
	(358.48,101.40) --
	(360.22,101.12) --
	(361.96,100.85) --
	(363.70,100.57) --
	(365.44,100.29) --
	(367.18,100.01) --
	(368.92, 99.73) --
	(370.66, 99.45) --
	(372.40, 99.17) --
	(374.14, 98.89) --
	(375.88, 98.61) --
	(377.62, 98.33) --
	(379.36, 98.05) --
	(381.10, 97.77) --
	(382.84, 97.49) --
	(384.58, 97.21) --
	(386.32, 96.92) --
	(388.06, 96.64) --
	(389.80, 96.36) --
	(391.54, 96.08) --
	(393.28, 95.80) --
	(395.02, 95.51) --
	(396.76, 95.23);

\path[draw=drawColor,line width= 0.6pt,dash pattern=on 2pt off 2pt ,line join=round] ( 48.73, 59.12) --
	( 50.47, 61.04) --
	( 52.21, 66.28) --
	( 53.95, 73.60) --
	( 55.69, 81.78) --
	( 57.43, 89.95) --
	( 59.17, 97.62) --
	( 60.91,104.58) --
	( 62.65,110.77) --
	( 64.39,116.21) --
	( 66.13,120.98) --
	( 67.87,125.15) --
	( 69.61,128.82) --
	( 71.35,132.05) --
	( 73.09,134.91) --
	( 74.83,137.46) --
	( 76.57,139.74) --
	( 78.31,141.78) --
	( 80.05,143.62) --
	( 81.79,145.28) --
	( 83.53,146.79) --
	( 85.27,148.16) --
	( 87.01,149.40) --
	( 88.75,150.55) --
	( 90.49,151.60) --
	( 92.24,152.56) --
	( 93.98,153.45) --
	( 95.72,154.28) --
	( 97.46,155.04) --
	( 99.20,155.76) --
	(100.94,156.42) --
	(102.68,157.05) --
	(104.42,157.63) --
	(106.16,158.18) --
	(107.90,158.70) --
	(109.64,159.20) --
	(111.38,159.66) --
	(113.12,160.11) --
	(114.86,160.53) --
	(116.60,160.94) --
	(118.34,161.32) --
	(120.08,161.70) --
	(121.82,162.05) --
	(123.56,162.40) --
	(125.30,162.72) --
	(127.04,163.04) --
	(128.78,163.35) --
	(130.52,163.65) --
	(132.26,163.93) --
	(134.00,164.21) --
	(135.74,164.48) --
	(137.48,164.74) --
	(139.22,164.99) --
	(140.96,165.24) --
	(142.70,165.48) --
	(144.44,165.71) --
	(146.18,165.93) --
	(147.92,166.15) --
	(149.66,166.36) --
	(151.40,166.57) --
	(153.14,166.77) --
	(154.88,166.97) --
	(156.62,167.16) --
	(158.36,167.35) --
	(160.10,167.53) --
	(161.84,167.71) --
	(163.58,167.88) --
	(165.32,168.05) --
	(167.06,168.22) --
	(168.80,168.38) --
	(170.54,168.54) --
	(172.28,168.70) --
	(174.02,168.85) --
	(175.76,169.00) --
	(177.50,169.14) --
	(179.24,169.28) --
	(180.98,169.42) --
	(182.72,169.56) --
	(184.46,169.69) --
	(186.20,169.82) --
	(187.94,169.95) --
	(189.68,170.07) --
	(191.42,170.20) --
	(193.16,170.32) --
	(194.90,170.44) --
	(196.65,170.55) --
	(198.39,170.67) --
	(200.13,170.78) --
	(201.87,170.89) --
	(203.61,171.00) --
	(205.35,171.10) --
	(207.09,171.21) --
	(208.83,171.31) --
	(210.57,171.41) --
	(212.31,171.51) --
	(214.05,171.60) --
	(215.79,171.70) --
	(217.53,171.79) --
	(219.27,171.88) --
	(221.01,171.98) --
	(222.75,172.06) --
	(224.49,172.15) --
	(226.23,172.24) --
	(227.97,172.32) --
	(229.71,172.41) --
	(231.45,172.49) --
	(233.19,172.57) --
	(234.93,172.65) --
	(236.67,172.73) --
	(238.41,172.80) --
	(240.15,172.88) --
	(241.89,172.96) --
	(243.63,173.03) --
	(245.37,173.10) --
	(247.11,173.17) --
	(248.85,173.24) --
	(250.59,173.31) --
	(252.33,173.38) --
	(254.07,173.45) --
	(255.81,173.52) --
	(257.55,173.58) --
	(259.29,173.65) --
	(261.03,173.71) --
	(262.77,173.78) --
	(264.51,173.84) --
	(266.25,173.90) --
	(267.99,173.96) --
	(269.73,174.02) --
	(271.47,174.08) --
	(273.21,174.14) --
	(274.95,174.19) --
	(276.69,174.25) --
	(278.43,174.31) --
	(280.17,174.36) --
	(281.91,174.42) --
	(283.65,174.47) --
	(285.39,174.53) --
	(287.13,174.58) --
	(288.87,174.63) --
	(290.61,174.68) --
	(292.35,174.73) --
	(294.09,174.78) --
	(295.83,174.83) --
	(297.57,174.88) --
	(299.31,174.93) --
	(301.06,174.98) --
	(302.80,175.03) --
	(304.54,175.07) --
	(306.28,175.12) --
	(308.02,175.17) --
	(309.76,175.21) --
	(311.50,175.26) --
	(313.24,175.30) --
	(314.98,175.34) --
	(316.72,175.39) --
	(318.46,175.43) --
	(320.20,175.47) --
	(321.94,175.52) --
	(323.68,175.56) --
	(325.42,175.60) --
	(327.16,175.64) --
	(328.90,175.68) --
	(330.64,175.72) --
	(332.38,175.76) --
	(334.12,175.80) --
	(335.86,175.84) --
	(337.60,175.88) --
	(339.34,175.91) --
	(341.08,175.95) --
	(342.82,175.99) --
	(344.56,176.03) --
	(346.30,176.06) --
	(348.04,176.10) --
	(349.78,176.13) --
	(351.52,176.17) --
	(353.26,176.20) --
	(355.00,176.24) --
	(356.74,176.27) --
	(358.48,176.31) --
	(360.22,176.34) --
	(361.96,176.38) --
	(363.70,176.41) --
	(365.44,176.44) --
	(367.18,176.48) --
	(368.92,176.51) --
	(370.66,176.54) --
	(372.40,176.57) --
	(374.14,176.60) --
	(375.88,176.63) --
	(377.62,176.67) --
	(379.36,176.70) --
	(381.10,176.73) --
	(382.84,176.76) --
	(384.58,176.79) --
	(386.32,176.82) --
	(388.06,176.85) --
	(389.80,176.88) --
	(391.54,176.91) --
	(393.28,176.93) --
	(395.02,176.96) --
	(396.76,176.99);
\end{scope}
\begin{scope}
\path[clip] (  0.00,  0.00) rectangle (419.17,274.63);
\definecolor{drawColor}{RGB}{144,144,144}

\path[draw=drawColor,line width= 0.6pt,line join=round] ( 31.33, 30.40) --
	( 31.33,269.63);
\end{scope}
\begin{scope}
\path[clip] (  0.00,  0.00) rectangle (419.17,274.63);
\definecolor{drawColor}{RGB}{68,68,68}

\node[text=drawColor,anchor=base east,inner sep=0pt, outer sep=0pt, scale=  0.80] at ( 26.83, 38.63) {30};

\node[text=drawColor,anchor=base east,inner sep=0pt, outer sep=0pt, scale=  0.80] at ( 26.83, 88.91) {40};

\node[text=drawColor,anchor=base east,inner sep=0pt, outer sep=0pt, scale=  0.80] at ( 26.83,139.20) {50};

\node[text=drawColor,anchor=base east,inner sep=0pt, outer sep=0pt, scale=  0.80] at ( 26.83,189.48) {60};

\node[text=drawColor,anchor=base east,inner sep=0pt, outer sep=0pt, scale=  0.80] at ( 26.83,239.77) {70};
\end{scope}
\begin{scope}
\path[clip] (  0.00,  0.00) rectangle (419.17,274.63);
\definecolor{drawColor}{RGB}{144,144,144}

\path[draw=drawColor,line width= 0.6pt,line join=round] ( 28.83, 41.38) --
	( 31.33, 41.38);

\path[draw=drawColor,line width= 0.6pt,line join=round] ( 28.83, 91.67) --
	( 31.33, 91.67);

\path[draw=drawColor,line width= 0.6pt,line join=round] ( 28.83,141.95) --
	( 31.33,141.95);

\path[draw=drawColor,line width= 0.6pt,line join=round] ( 28.83,192.24) --
	( 31.33,192.24);

\path[draw=drawColor,line width= 0.6pt,line join=round] ( 28.83,242.52) --
	( 31.33,242.52);
\end{scope}
\begin{scope}
\path[clip] (  0.00,  0.00) rectangle (419.17,274.63);
\definecolor{drawColor}{RGB}{144,144,144}

\path[draw=drawColor,line width= 0.6pt,line join=round] ( 31.33, 30.40) --
	(414.17, 30.40);
\end{scope}
\begin{scope}
\path[clip] (  0.00,  0.00) rectangle (419.17,274.63);
\definecolor{drawColor}{RGB}{144,144,144}

\path[draw=drawColor,line width= 0.6pt,line join=round] ( 48.73, 27.90) --
	( 48.73, 30.40);

\path[draw=drawColor,line width= 0.6pt,line join=round] (135.74, 27.90) --
	(135.74, 30.40);

\path[draw=drawColor,line width= 0.6pt,line join=round] (222.75, 27.90) --
	(222.75, 30.40);

\path[draw=drawColor,line width= 0.6pt,line join=round] (309.76, 27.90) --
	(309.76, 30.40);

\path[draw=drawColor,line width= 0.6pt,line join=round] (396.76, 27.90) --
	(396.76, 30.40);
\end{scope}
\begin{scope}
\path[clip] (  0.00,  0.00) rectangle (419.17,274.63);
\definecolor{drawColor}{RGB}{68,68,68}

\node[text=drawColor,anchor=base,inner sep=0pt, outer sep=0pt, scale=  0.80] at ( 48.73, 20.39) {0.0};

\node[text=drawColor,anchor=base,inner sep=0pt, outer sep=0pt, scale=  0.80] at (135.74, 20.39) {0.5};

\node[text=drawColor,anchor=base,inner sep=0pt, outer sep=0pt, scale=  0.80] at (222.75, 20.39) {1.0};

\node[text=drawColor,anchor=base,inner sep=0pt, outer sep=0pt, scale=  0.80] at (309.76, 20.39) {1.5};

\node[text=drawColor,anchor=base,inner sep=0pt, outer sep=0pt, scale=  0.80] at (396.76, 20.39) {2.0};
\end{scope}
\begin{scope}
\path[clip] (  0.00,  0.00) rectangle (419.17,274.63);
\definecolor{drawColor}{RGB}{68,68,68}

\node[text=drawColor,anchor=base,inner sep=0pt, outer sep=0pt, scale=  1.00] at (222.75,  6.94) {$M$};
\end{scope}
\begin{scope}
\path[clip] (  0.00,  0.00) rectangle (419.17,274.63);
\definecolor{drawColor}{RGB}{68,68,68}

\node[text=drawColor,rotate= 90.00,anchor=base,inner sep=0pt, outer sep=0pt, scale=  1.00] at ( 11.89,150.01) {Average Markup in \%};
\end{scope}
\end{tikzpicture}